\documentclass[a4paper,UKenglish,leqno]{article}

\usepackage[utf8]{inputenc}
\usepackage[T1]{fontenc}

\usepackage{alphabeta}
\usepackage{relsize,xspace}
 \usepackage{xcolor}
 \usepackage{mathtools}
\usepackage{microtype}
\usepackage{amsmath}
 \usepackage{amssymb}
 \usepackage{mathrsfs}  
 \usepackage{amsfonts}
\usepackage{hyperref}
\usepackage{bm}
\usepackage{enumerate,comment,environ,etoolbox}
\usepackage{fullpage}

\usepackage{caption}
\usepackage{subcaption}

\newcommand{\AAA}{\mathcal{A}}

\newcommand{\CCC}{\mathcal{C}}

\newcommand{\III}{\mathcal{I}}

\newcommand{\LLL}{\mathcal{L}}
\newcommand{\MMM}{\mathcal{M}}

\newcommand{\PPP}{\mathcal{P}}

\newcommand{\RRR}{\mathcal{R}}
\newcommand{\SSS}{\mathcal{S}}
\newcommand{\TTT}{\mathcal{T}}

\newcommand{\WWW}{\mathcal{W}}

\newcommand{\ZZZ}{\mathcal{Z}}

\newcommand{\CC}{\mathfrak{C}}

\newcommand{\LL}{\mathfrak{L}}

\newcommand{\NN}{\mathfrak{N}}

\newcommand{\bd}{\textup{bd}}

\newcommand{\Zone}{\textit{zone}}
\newcommand{\Split}{\textup{split}}
\newcommand{\Shore}{\textit{shore}}
\newcommand{\Port}{\textit{port}}
\newcommand{\tail}{\textit{tail}}
\newcommand{\head}{\textit{head}}
\newcommand{\sth}{\mathrel{ : } }
\newcommand{\nin}{\not\in} 
\newcommand{\Oof}{\mathcal{O}}

\newcommand{\N}{\mathbb{N}}
\newcommand{\dtw}{\mathrm{dtw}}
\newcommand{\tw}{\mathrm{tw}}

\newcommand{\WMI}{{\small\upshape(WM1)}\xspace}
\newcommand{\WMII}{{\small\upshape(WM2)}\xspace}
\newcommand{\WMIII}{{\small\upshape(WM3)}\xspace}
\newcommand{\WMIV}{{\small\upshape(WM4)}\xspace}
\newcommand{\WMV}{{\small\upshape(WM5)}\xspace}
\newcommand{\WMVI}{{\small\upshape(WM6)}\xspace}

\newcommand\restr[2]{{% we make the whole thing an ordinary symbol
  \left.\kern-\nulldelimiterspace % automatically resize the bar with \right
  #1 % the function
  \vphantom{\big|} % pretend it's a little taller at normal size
  \right|_{#2} % this is the delimiter
  }}

\definecolor{gruen}{rgb}{0,0.8,0.2}
\definecolor{rot}{rgb}{0.7,0,0}

\newcommand{\black}{\color{black}}

\providebool{showproofs}

\usepackage{empheq}

\usepackage{refcount}
\usepackage{wrapfig}
\usepackage{marginnote}
\usepackage[inline]{enumitem}
\usepackage{latexsym}
\usepackage{pdfpages}
\usepackage[disable]{todonotes}
\presetkeys{todonotes}{inline}{}

\definecolor{dark-blue}{rgb}{0.05,0.25,0.85}
\hypersetup{colorlinks=true,linkcolor=dark-blue,citecolor=dark-blue}

\makeatletter
\renewcommand\footnotesize{%
   \@setfontsize\footnotesize\@ixpt{11}%
   \abovedisplayskip 8\p@ \@plus2\p@ \@minus4\p@
   \abovedisplayshortskip \z@ \@plus\p@
   \belowdisplayshortskip 4\p@ \@plus2\p@ \@minus2\p@
   \def\@listi{\leftmargin\leftmargini
               \topsep 4\p@ \@plus2\p@ \@minus2\p@
               \parsep 2\p@ \@plus\p@ \@minus\p@
               \itemsep \parsep}%
   \belowdisplayskip \abovedisplayskip
}
\makeatother

\usepackage[amsmath,thmmarks,hyperref]{ntheorem}
\usepackage{cleveref}

\crefformat{page}{#2page~#1#3}
\Crefformat{page}{#2Page~#1#3}
\crefformat{equation}{#2(#1)#3}
\Crefformat{equation}{#2(#1)#3}
\crefformat{figure}{#2Figure~#1#3}
\Crefformat{figure}{#2Figure~#1#3}
\crefformat{section}{#2Section~#1#3}
\Crefformat{section}{#2Section~#1#3}
\crefformat{chapter}{#2Chapter~#1#3}
\Crefformat{chapter}{#2Chapter~#1#3}
\crefformat{chapter*}{#2Chapter~#1#3}
\Crefformat{chapter*}{#2Chapter~#1#3}
\crefformat{part}{#2Part~#1#3}
\Crefformat{part}{#2Part~#1#3}
\crefformat{enumi}{#2(#1)#3}
\Crefformat{enumi}{#2(#1)#3}

\theoremnumbering{arabic}
\theoremstyle{plain}
\theoremsymbol{}
\theorembodyfont{\itshape}
\theoremheaderfont{\normalfont\scshape}
\theoremseparator{.}

\newtheorem{theorem}{Theorem}[section]
\crefformat{theorem}{#2Theorem~#1#3}
\Crefformat{theorem}{#2Theorem~#1#3}

\newcommand{\newtheoremwithcrefformat}[2]{%
  \newtheorem{#1}[theorem]{\textsc{#2}}%
  \crefformat{#1}{##2\MakeUppercase#1~##1##3}%
  \Crefformat{#1}{##2\MakeUppercase#1~##1##3}%
}

  \newtheorem{lemma}[theorem]{Lemma}
  \crefformat{lemma}{#2Lemma~#1#3}
  \Crefformat{lemma}{#2Lemma~#1#3}

  \crefformat{conjecture}{#2Conjecture~#1#3}
  \Crefformat{conjecture}{#2Conjecture~#1#3}
  
  \crefformat{proposition}{#2Proposition~#1#3}
  \Crefformat{proposition}{#2Proposition~#1#3}

  \newtheorem{corollary}[theorem]{Corollary}
  \crefformat{corollary}{#2Corollary~#1#3}
  \Crefformat{corollary}{#2Corollary~#1#3}
\theorembodyfont{\upshape}
\newtheoremwithcrefformat{example}{Example}
\newtheoremwithcrefformat{remark}{Remark}

  \newtheorem{definition}[theorem]{Definition}
  \crefformat{definition}{#2Definition~#1#3}
  \Crefformat{definition}{#2Definition~#1#3}
\newtheoremwithcrefformat{notation}{Notation}
\newtheoremwithcrefformat{observation}{Observation}

\newtheorem{reduction}{Reduction}
\crefformat{reduction}{#2Reduction~#1#3}
\Crefformat{reduction}{#2Reduction~#1#3}

\theoremsymbol{\ensuremath{\square}}
\theoremheaderfont{\scshape}
\theorembodyfont{\normalfont}
\theoremsymbol{\ensuremath{\square}}
\newtheorem{proof}{Proof.}

\theoremstyle{nonumberplain}

\def\cqedsymbol{\ifmmode$\lrcorner$\else{\unskip\nobreak\hfil
\penalty50\hskip1em\null\nobreak\hfil$\lrcorner$
\parfillskip=0pt\finalhyphendemerits=0\endgraf}\fi}

\usepackage[inline]{enumitem}

\newcommand*{\Abs}[1]{| #1 |}

\newenvironment{cenv}{\begin{list}{}{%
      \setlength{\labelwidth}{1.5em}%
      \setlength{\leftmargin}{\labelwidth}%
      \addtolength{\leftmargin}{\labelsep}%
      \setlength{\listparindent}{0em}%
      \setlength{\topsep}{10pt}%
      \setlength{\itemsep}{5pt}%
      \setlength{\parsep}{0pt}%
    }
  }{
  \end{list}
}

\newcounter{claimcounter}
\newcounter{conditioncounter}

\newenvironment{Claim}{
  
  \refstepcounter{claimcounter}
  \begin{cenv}
  \item[{Claim \arabic{claimcounter}.}]
  }{
  \end{cenv}
}

\newenvironment{claim}{
  
  \refstepcounter{claimcounter}
  \begin{cenv}
  \item[{Claim \arabic{claimcounter}.}]
  }{
  \end{cenv}
}
\crefformat{claimcounter}{#2Claim~#1#3}

\newenvironment{ClaimProof}[1][]{\noindent{%
\ifthenelse{\equal{#1}{}}{{\itshape Proof.\ }}{{\itshape #1.\ }}%
}}{\hspace*{1em}\nobreak\hfill$\dashv$\endtrivlist\addvspace{2ex plus
0.5ex minus0.1ex}}

\RenewEnviron{proof}[1][]
{%
  \ifbool{showproofs}{%
    \setcounter{claimcounter}{0}%
    \ifthenelse{\equal{#1}{}}{\noindent\textit{Proof. }}{\noindent\textit{#1. }}%
    \BODY  \hspace*{1pt}\hfill$\Box$\par\bigskip
  }%
}%

\renewenvironment{proof}[1][]%%
{\setcounter{claimcounter}{0}\ifthenelse{\equal{#1}{}}{\noindent\textit{Proof.
    }}{\noindent\textit{#1. }}}%
{\hspace*{1pt}\hfill$\Box$\par\bigskip}

\excludecomment{comment}
\excludecomment{alternative}
\includecomment{full}
\includecomment{short}

\booltrue{showproofs}

\presetkeys{todonotes}{fancyline, color=red!30}{}

\usepackage{authblk,subfiles}

\usepackage{index}
\makeindex
\newindex{Symbol}{symbol.idx}{symbol.ind}{Symbols}

\title{Edge-Disjoint Paths in Eulerian Digraphs}

\DeclareRobustCommand{\authorthing}{
	\begin{center}
		Dario Cavallaro\thanks{\texttt{d.cavallaro@tu-berlin.de} (D. Cavallaro)} \\
		{\small Technical University Berlin, Germany} \\
		  \medskip
		Ken-ichi Kawarabayashi \thanks{\texttt{k\_keniti@nii.ac.jp} (K. Kawarabayashi)}\\
		{\small National Institute of Informatics, Japan}\\
            {\small The University of Tokyo, Japan}\\
		\medskip
		Stephan Kreutzer \thanks{\texttt{stephan.kreutzer@tu-berlin.de} (S. Kreutzer)} \\
		{\small Technical University Berlin, Germany} \\
\end{center}}
\author{\authorthing}

\let\oldtext\text
\renewcommand{\text}[1]{\oldtext{\black{#1}}}
 \date{}

\begin{document}
\maketitle
\begin{abstract}
  Disjoint paths problems are among the most prominent problems in combinatorial optimization. The edge- as well as vertex-disjoint paths problem, are NP-complete on directed and undirected graphs. But on undirected graphs, Robertson and Seymour \cite{GMXIII} developed an algorithm for the vertex- and the edge-disjoint paths problem that runs in cubic time for every fixed number~$p$ of terminal pairs, i.e. they proved that the problem is fixed-parameter tractable on undirected graphs.

  On directed graphs, Fortune, Hopcroft, and Wyllie proved that both problems are NP-complete already for~$p=2$ terminal pairs. 
  
  In this paper, we study the edge-disjoint paths problem (EDPP) on Eulerian digraphs, a problem that has received significant attention in the literature. Marx \cite{Marx2004} proved that the Eulerian EDPP is NP-complete even on structurally very simple Eulerian digraphs. On the positive side, polynomial time algorithms are known only for very restricted cases, such as~$p\leq 3$  or where the demand graph is a union of two stars (see e.g.~\cite{IbarakiP1991,Frank1988,FrankIN1995}). 

   The question of which values of~$p$ the edge-disjoint paths problem can be solved in polynomial time on Eulerian digraphs has already been raised by Frank, Ibaraki, and Nagamochi \cite{FrankIN1995} almost 30 years ago. 
   But despite considerable effort, the complexity of the problem is still wide open and is considered to be the main open problem in this area (see \cite[Chapter 4]{Bang-JensenG2018} for a recent survey).

   In this paper, we solve this long-open problem by showing that the Edge-Disjoint Paths Problem is fixed-parameter tractable on Eulerian digraphs in general (parameterized by the number of terminal pairs). The algorithm itself is reasonably simple but the proof of its correctness requires a deep structural analysis of Eulerian digraphs. 
\end{abstract}

\clearpage

\tableofcontents

\clearpage
\setcounter{page}{1}
\section{Introduction}

The disjoint paths problem, that is, the problem of deciding for a given graph or digraph~$G$ and a set~$S \coloneqq \{s_1, \dots, s_p\}$ of sources and~$T \coloneqq \{t_1, \dots, t_p\}$ of targets, whether there is a set of~$p$ mutually edge- or vertex-disjoint paths connecting the sources to the targets, is one of the most fundamental problems in the area of graph algorithms. 

By Menger's theorem, or network flow algorithms, it can be solved in polynomial time on undirected and directed graphs if we are only interested in a set of disjoint paths each having one end in~$S$ and the other in~$T$. But the situation changes completely if we require that the paths connect each source~$s_i$ to its corresponding target~$t_i$: this problem is NP-complete for edge-disjoint and vertex-disjoint paths, on directed and undirected graphs. 

On undirected graphs, Robertson and Seymour developed an algorithm for the~$p$-Vertex-Disjoint-Paths and the~$p$-Edge-Disjoint-Paths problem which runs in time~$\Oof(n^3)$ for any fixed number~$p$ of terminal pairs \cite{GMXIII}. Rephrased in the terminology of parameterized complexity, they showed that the problem is fixed-parameter tractable parameterized by the number~$p$ of terminal pairs. The complexity has subsequently been improved to quadratic time in \cite{kkr2012}. Robertson and Seymour developed the algorithm as part of their celebrated series of papers on graph minors. While the correctness proof is still long and difficult and relies on large parts of the graph minors series (even with the new unique linkage problem in \cite{unique}), the algorithm itself is much less complicated and essentially facilitates a reduction rule that reduces any input instance to an equivalent instance of bounded tree-width which can then be solved easily.

On directed graphs, the~$p$-Vertex- and~$p$-Edge-Disjoint-Paths problems are considerably more difficult. As shown by Fortune, Hopcroft, and Wyllie  \cite{FortuneHW1980}, both problems are NP-complete already for~$p=2$ terminal pairs. This implies that they are not fixed-parameter tractable and not even in the class XP (under the usual complexity theoretical assumptions that we tactically assume throughout the introduction). Furthermore, Slivkins  \cite{Slivkins2010} showed that the problems are W[1]-hard---and therefore (presumably) not fixed-parameter tractable---already on acyclic digraphs.
On the positive side, Cygan, Marx, Pilipczuk, and Pilipczuk \cite{CyganMPP2013} proved that the~$p$-Vertex-Disjoint paths problem is fixed-parameter tractable on planar digraphs, but interestingly, the edge-disjoint version remains~$W[1]$-hard \cite{Chi2023}. 

\smallskip

\noindent{\itshape Eulerian digraphs.} A well-studied class of digraphs whose complexity often turns out to be somewhere between undirected and general directed graphs is the class of Eulerian digraphs.  
A digraph is \emph{Eulerian} if the in-degree of each vertex equals its out-degree or -- equivalently -- if it is the union of a set of edge disjoint cycles. See \cite[Chapter 4]{Bang-JensenG2018} for a recent survey on Eulerian digraphs. 
It has been observed in \cite{sodaKK10} that the correctness proof of the algorithm for the~$p$-Edge-Disjoint-Paths problem can be simplified for undirected Eulerian graphs. This already suggests that the `Eulerian' property could make a difference for the~$p$-Edge-Disjoint-Paths problem also on digraphs. Indeed, Johnson \cite{Johnson2002} pioneered the structural analysis of Eulerian digraphs with emphasis on solving~$p$-edge disjoint paths problems in his dissertation. He proved a structure theorem for internally~$6$-connected Eulerian digraphs in the same flavour as the undirected structure theorem proved by Robertson and Seymour, following the same line of argumentation as their proof. Unfortunately, his results have never been published.

In the literature on Eulerian digraphs the~Edge-Disjoint-Paths problem is often studied in the following formulation.
\begin{definition}[Edge-Disjoint Paths Problem]
  The \emph{edge-disjoint paths problem} is the problem to decide, given two digraphs~$G$ and~$D$ with~$V(D) \subseteq V(G)$ as input, whether~$G$ contains a set~$\LLL$ of pairwise edge-disjoint paths which contains for each edge~$(t_i, s_i) \in E(D)$ an~$s_i{-}t_i$-path. 
\end{definition}
When fixing the number of terminal pairs we are interested in, i.e., for fixed~$p \coloneqq \Abs{E(D)}$ we refer to the problem as the~$p$-Edge-disjoint paths problem. 
   
An equivalent formulation is to decide, given~$G$ and~$D$ as above, whether~$G+D$ contains a set of pairwise edge-disjoint cycles each containing exactly one edge of~$D$.  

We call~$G$ the \emph{supply} and~$D$ the \emph{demand} digraph. The vertices incident to an edge in~$D$ are called \emph{terminals}. It is easily seen that this formulation is (qualitatively) equivalent to the specification of the disjoint paths problem by a single digraph~$G$ and~$p$ pairs~$(s_1, t_1), \dots, (s_p, t_p)$ of terminals. One advantage of the presentation with separate demand and supply graphs is that this makes it possible to classify the complexity of the problem relative to the structure  of the demand graph.

Unfortunately, the Edge-Disjoint-Paths problem remains NP-complete on Eulerian digraphs. In fact, Marx \cite{MARX04} proved that the problem is already NP-complete if~$G$ is an acyclic directed grid graph and~$G+D$ is Eulerian. 

On the positive side, Frank \cite{Frank1988} showed that the problem can be solved in polynomial time if~$G+D$ are Eulerian and~$D$ consists of two sets of parallel edges or is the union of two stars. Moreover, he showed that in these cases the \emph{directed cut criterion} is sufficient for the existence of a solution.
Polynomial time algorithms for a few other special cases (for~$p\leq 3$) have been developed in \cite{FrankIN1995,Vygen1995,IbarakiP1991}, but in the last nearly~$30$ years no significant progress on determining the complexity of the general~$p$-Edge-Disjoint-Paths problem on Eulerian digraphs has been made. However, Johnson \cite{Johnson2002} proved in his dissertation that given an Eulerian digraph \emph{(Euler-)embedded} in some surface~$\Sigma$---we will make this precise shortly---such that there exists a disc~$\Delta \subset \Sigma$ containing many concentric edge-disjoint cycles of alternating orientation, then the most deeply nested cycles are \emph{irrelevant} to the instance. That is, one may equally well delete the most deeply nested cycle from the graph resulting in an equivalent instance. Unfortunately Johnson's work has never been published.

As stated in \cite[Problem 4.5.7]{Bang-JensenG2018} the status of the~$p$-Edge-Disjoint-Paths problem on Eulerian digraphs is wide open and could range from fixed-parameter tractable with respect to~$p$ to NP-complete already for~$p=4$. 

The main result of this paper is to settle this long open problem by showing that the~$p$-Edge-Disjoint-Paths problem is fixed-parameter tractable parameterized by~$p$ on the class of all Eulerian digraphs.

\begin{theorem}\label{thm:main}
  The~$p$-Edge-Disjoint-Paths problem in Eulerian digraphs is fixed-parameter tractable parameterized by the number of terminal pairs~$p$. 

  That is, there is a computable function~$f$ and an algorithm with running time~$f(p)\cdot n^{\Oof(1)}$, which, given an~$n$-vertex digraph~$G$ and a~$2p$-vertex digraph~$D$ with~$V(D) \subset V(G)$ with parameter~$p \coloneqq |E(D)|$ such that~$G+D$ is Eulerian, decides correctly whether or not~$G$ contains a set of pairwise  edge-disjoint paths which contains an~$s{-}t$-path for each edge~$(t, s) \in E(D)$. 
\end{theorem}

We give a high level overview of the algorithm and its correctness proof in \cref{sec:proof_structure}. Before that we recall some concepts and notation relevant for the exposition in \cref{sec:proof_structure}.

For further details we refer to \cref{sec:prelims} and the respective sections.

\smallskip

\paragraph{Some notation.} Throughout this paper we use standard graph theoretic notation as in \cite{Diestel2012} and \cite{Bang-JensenG2018}. In particular, we write~$G=(V,E)$ for directed graphs and we call the graph resulting from~$G$ when forgetting the edge directions its \emph{underlying undirected graph}. We call a graph \emph{Eulerian} if the in-degree and out-degree at every vertex is the same. Given two edge-disjoint graphs~$G,D$ with~$V(D) \subseteq V(G)$ we write~$G+D$ to mean the graph with vertex set~$V(G)$ and edge-set~$E(D) \cup E(G)$. Given such~$G,D$ where~$D$ has no isolated vertices and~$G+D$ is Eulerian we call~$D$ the \emph{demand graph} and~$G$ the \emph{supply graph}. Throughout the paper we will frequently work with digraphs embedded on a fixed surface~$\Sigma$. In these cases we will work with a specific type of graph-embeddings which we call \emph{Euler-embeddings}. Let~$G+D$ be Eulerian of degree at most four. Then~$G$ is called \emph{Euler-embedded} (in some surface~$\Sigma$) if the embedding contains no \emph{strongly planar vertex}. A vertex~$v \in V(G)$ of the embedded digraph is called \emph{strongly planar} if it is of degree four and we can draw a simple closed curve~$\gamma_v$ in the surface~$\Sigma$ around the vertex~$v$ such that~$\gamma_v$ intersects exactly all edges adjacent to~$v$ (exactly once and only these) such that~$\gamma_v$ visits first both in-edges and then both out-edges (up-to a cyclic rotation).  Given~$G,D$ such that~$G+D$ is Eulerian and~$D$ is a demand graph with~$\Abs{E(D)} = p$ for some~$p\in \N$ we say that~$G+D$ \emph{encodes} an instance of the Eulerian edge-disjoint paths problem. That is given~$G$ and~$D$ we are to decide whether there exist~$p$ edge-disjoint paths~$L_1,\ldots,L_p$ forming what we call a \emph{$p$-linkage}~$\LLL=\{L_1,\ldots,L_p\}$ such that~$L_i$ connects~$s_i$ to~$t_i$ for some~$(t_i,s_i) \in E(D)$ for every~$1 \leq i \leq p$. (Note that a priori~$s_i=s_j$,~$s_i = t_j$ and~$t_i=t_j$ are possible for any~$1 \leq i,j \leq p$.) Given such an instance, if the respective edge-disjoint paths exist we call it a \emph{YES}-instance, otherwise we call it a \emph{NO}-instance.

\smallskip

Further, we assume that the readers are familiar with the notion of undirected tree-width and will at times talk about directed tree-width (see \cite{DirTreewidth2001,KawarabayashiK2015a} for definitions and results) although the exact definitions (of either) will not be of importance. The reason is that Eulerian digraphs of bounded degree admit `high' undirected tree-width if and only if they admit `high' directed tree-width and thus both notions are qualitatively the same (see \cref{thm:undirected_vs_directed_tw_in_Eulerian_graphs}). An important result tied to high directed tree-width that is central to our arguments is that it guarantees the existence of a large \emph{cylindrical wall}~$\WWW$ in the graph~$G$ (see \cref{fig:wall}). So whenever we will say `high tree-width' this implies a `large cylindrical wall as a subgraph'.

\smallskip

As mentioned above, it was proven in \cite{KawarabayashiK2015a} that every directed graph of high directed tree-width contains a large cylindrical wall (see \cref{thm:dir-grid-thm}). Moreover, it was shown in \cite{CamposLMS2019} that computing such a wall is fixed-parameter tractable parameterized by the tree-width. In particular, this means that for Eulerian directed graphs we can find a large cylindrical wall in~$fpt$-time parameterized by the undirected tree-width.

\section{Structure of the proof}
\label{sec:proof_structure}

In this section we give a high-level description of the algorithm and the proof of its correctness. Let~$G+D$ be an Eulerian directed graph, where~$D$ is the demand graph encoding an instance of the edge-disjoint paths problem with~$\Abs{E(D)} = p \in \N$. That is, we want to decide whether there exists a~$p$-linkage~$\LLL=\{L_1,\ldots,L_p\}$ where~$L_i$ connects~$s_i$ to~$t_i$ for~$(t_i,s_i) \in E(D)$ for every~$1 \leq i \leq p$. As we will prove in \cref{lem:reduce-to-degree-4}, we may assume that every vertex in~$V(G)\setminus V(D)$ is of degree four and every vertex in~$V(D)$ is of degree two. 

\paragraph{The main idea} of the algorithm is to keep reducing the instance using the \emph{irrelevant vertex technique}---in our case irrelevant cycle technique---until the graph has bounded undirected tree-width for which the instance can be solved in~$fpt$-time using standard techniques (e.g. using Courcelle's Theorem \cite{Courcelle1990}, more precisely an adaptation of it for directed graph for example as proved in \cite{ArnSLJSD91}). The key observation is that, as long as the tree-width of~$G$ is not bounded by a function in~$p$, we are able to locate a cycle (or rather a closed walk) in the graph, whose deletion does not change the existence of a  solution. We can therefore  repeat this process until, eventually, the graph is of bounded tree-width and proceed as discussed. Our goal in this paper is to show that the problem is fixed-parameter tractable and not to optimise the running-time of our algorithm. In our correctness proof we establish several results  that can be used more explicitly in the algorithm to get a better running-time performance; see for example \cref{thm:Frank_router_reduction}. As this would make the correctness proof even more complicated we refrain from doing so in this paper though. 

\smallskip

Readers familiar with the work of Robertson and Seymour \cite{GMXIII,GMXXI,GMXXII} will see the similarities between our algorithm and our approach to prove its correctness and the line of argument in these papers. While some of the arguments and techniques we use are highly inspired and follow the same line of reasoning as theirs, we note that our proofs, the respective constructions, and ideas behind them are in no way easily derivable from the results presented in \cite{GMXIII,GMXXI,GMXXII}. For example the standard graph-minor structure theorem due to Robertson and Seymour \cite{GMXVI} is of no direct use to us. One reason being that there is no straightforward argument how undirected clique-minors help in routing directed edge-disjoint paths; in particular edge-contractions do change the instance. Also, knowing that there are no undirected clique-minors left is no real help either, as given a drawing of~$G+D$ in the plane does in no way yield enough structure to forbid certain edge-disjoint linkages in the graph. To see this note that strongly planar vertices may still help in \emph{crossing} paths, a phenomenon that in the undirected vertex-disjoint case only appears if the graph itself is non-planar. But taking any drawing of a non-planar graph and adding vertices at points of crossing-edges results in a planar graph which has at least the possible edge-disjoint paths as the non-planar graph. Hence it is not obvious how to leverage the undirected graph-minor structure theorem---and in fact we do not in this paper---neither is it obvious how to use clique-minors (not even directed clique minors) for routing, which is why we do not. Note further that many results in the graph-minor structure theory rely on inductive reasoning, separating the graphs into smaller graphs, deleting parts of the vertices and edges, contracting subgraphs; arguments that cannot easily be transferred to the Eulerian setting, for they may destroy the Eulerianness of the graph or in the latter case augment the set of solutions by creating new ways to route the paths edge-disjointly. This problem required us to develop new techniques that are tailored towards Eulerian graphs and the edge-disjoint case, bringing to light a deeper structural understanding of both.

\smallskip

More generally, the fact that the graphs in question are \emph{directed} makes algorithmic problems often harder, and especially in the field of structural graph theory the direction of edges has turned out to be a major nuisance in the past; for example there is no directed analogue of the \emph{two-paths theorem} which in the undirected setting is used in \cite{GMXIII} to prove the existence of a flat wall given the absence of large clique-minors. Also, perhaps somewhat counter-intuitively at first, and certainly in contrast to the undirected case, given a cylindrical wall with many disjoint non-planarities (that is, crosses) that are pairwise `far apart' on the wall does in general \emph{not} yield a directed clique-minor \cite{GiannopoulouKK2020}; whereas on undirected graphs it does, an observation that lies at the core of the undirected Flat-Wall theorem \cite{GMXIII,KawarabayashiTW2018}. Fortunately, most nuisances seem to disappear when focusing on Eulerian graphs and there are routing devices that help with routing paths (edge-)disjointly.

\smallskip

In our setting, the irrelevant cycles found by the algorithm are either cycles of some large \emph{Router}---a collection of edge-disjoint cycles that pairwise intersect---or cycles deeply nested inside an Euler-embedded \emph{flat swirl}---a collection of edge-disjoint concentric cycles that alternatingly change their orientation (with respect to the orientation of the plane they are embedded in). We will introduce routers in \cref{sec:Routing} and (flat) swirls in \cref{sec:Swirls}. We briefly introduce both types of graphs to give the reader an intuition for them prior to explaining the main ideas of the algorithm. The details can be found in the respective sections. 
\begin{itemize}
    \item A~\emph{$t$-router}~$\RRR=\bigcup_{1\leq i \leq t}C_i$ is a graph consisting of~$t$ edge-disjoint cycles~$C_i$ with~$1 \leq i \leq t$ such that they pairwise intersect. A Router can be thought of as an edge-disjoint model of a bi-directed clique.
    \item A~\emph{$t$-swirl}~$\mathcal{S} = \bigcup_{1\leq i \leq t}S_i$ is a (plane embedded) graph consisting of edge-disjoint concentric cycles~$S_i$ for~$1 \leq i \leq t$ such that~$S_i\cap S_{j} = \emptyset$ if~$\Abs{i-j} > 1$ and two consecutive cycles~$S_{i}$ and~$S_{i+1}$ have different orientation with respect to a given orientation of the plane for~$1 \leq i < j \leq t$. 
\end{itemize} 
A graph~$G$ is said to contain a \emph{flat} swirl~$\SSS \subset G$ (with outer-cycle~$S_t$) if the component containing~$V(\SSS)$ in~$G-S_t$ does not contain two paths~$P_1,P_2$ such that for~$i=1,2$ the path~$P_i$ connects~$s_i$ to~$t_i$ for distinct vertices~$s_1,s_2,t_1,t_2 \in V(C)$ appearing on~$S_t$ in that order. Our algorithm exploits that, given a large router in our graph then that router contains some cycle whose deletion does not change the instance. We prove this in \cref{sec:Routing} as \cref{thm:irrelevant_cycle_in_router_general}. In particular we describe an algorithmic approach in \cref{sec:Swirls} (see \cref{thm:flat_swirl_theorem}) that either finds said router or a flat swirl. Then, after at most~$\Abs{V(G)}$ steps, either there is no Router left and the tree-width of the graph is low (bounded in~$p$) or the tree-width of the graph remains high. If the tree-width of the graph is low (with respect to~$p$) the problem can be solved using standard dynamic programming techniques, for example using a well-known theorem due to Courcelle \cite{Courcelle1990} which can be adapted to more general structures such ad directed graphs as pursued in \cite{ArnSLJSD91}. Hence assume the tree-width is still high. The Flat-Swirl \cref{thm:flat_swirl_theorem} then implies that we find a large flat swirl in the graph. Given a large flat swirl and no large router we finally show that there exists an irrelevant cycle deeply nested inside the flat swirl, i.e., a cycle whose deletion does not change the instance. This is proved in \cref{thm:irrelevant_cycle}, but the correctness proof needs a lot of preparation and machinery established and pursued in \cref{sec:charting_eulerian_digraphs,sec:structure_of_min_examples,sec:shippings,sec:structure_thms}. In either case, router or flat swirl, we are able to inductively reduce the input instance to an equivalent instance of low tree-with in~$fpt$-time on~$p$, which in turn is an instance that we can solve in~$fpt$-time. We proceed by giving the algorithm, proving the main \cref{thm:main} of this paper.

\medskip  

\begin{proof}[Proof of \cref{thm:main}]
Let~$G+D$ be an instance of the directed Eulerian edge-disjoint paths problem.  Let~$p \coloneqq \Abs{E(D)}$.  Let~$g_1(p) \coloneqq t_{\ref{thm:irrelevant_cycle_in_router_general}}(p)$
and~$g_2(p) \coloneqq 2h_{\ref{thm:irrelevant_cycle}}(p)$. Given~$g_1,g_2$ define~$f_1(p) \coloneqq f_{\ref{thm:flat_swirl_away_from_D}}(p;g_1,g_2)$.  And finally let~$f(p) \coloneqq f_{\ref{thm:dir_wall_away_from_D}}(f_1(p))$. The following algorithm decides the instance in~$fpt$-time on~$p$.
    \begin{itemize}
    \item[1. ] Determine whether~$\tw(G+D) \leq 6 \cdot f(p)$, which can be done in~$fpt$-time on~$p$. If this is the case, then we can solve the instance using Courcelle's theorem \cite{Courcelle1990} in~$fpt$-time. Otherwise continue with Step 2.
    \item[2. ] Since~$\tw(G+D) \geq 6 \cdot f(p)$, \cref{thm:undirected_vs_directed_tw_in_Eulerian_graphs} implies that~$\dtw(G+D) \geq f(p)$. Using \cref{thm:dir_wall_away_from_D} we deduce that there is an~$f_1(p)\times f_1(p)$-cylindrical wall~$\WWW$ in~$G$ (away from~$D$) which can be found in~$fpt$-time as proved in \cite{CamposLMS2019}. Using \cref{thm:flat_swirl_away_from_D} we deduce that either we can find a \emph{flat}~$g_2(p)$-\emph{swirl}~$\mathcal{S}$ or a~$g_1(p)$-\emph{router}~$\RRR$ in~$G$ away from~$V(D)$ in~$fpt$-time on~$p$. If we find a router, go to Step 3; else proceed with Step 4.
    \item[3. ] If we have found a~$g_1(p)$-router~$\RRR$, use \cref{thm:irrelevant_cycle_in_router_general} to find and delete a cycle~$C \subset \RRR$ that is irrelevant to the instance in~$fpt$-time on~$p$. That is, after deletion of the cycle, the graph~$(G-C) +D$ is Eulerian and an equivalent instance to~$G+D$. After having successfully reduced the instance, go back to Step 1 and start over.
    \item[4. ] If we have found a flat~$g_2(p)$-swirl~$\mathcal{S}$ use \cref{thm:embedded_flat_swirl} to construct an equivalent instance~$G' + D$ with~$\Abs{G'} \leq \Abs{G}$ together with a cut~$(A,B)$ with~$A\cup B = G'$ in polynomial-time such that~$G'[B]$ contains a flat~$g_2(p)$-swirl~$\mathcal{S}'$ with~$E(D) \cap G'[B] = \emptyset$ and~$G'[B]$ can be \emph{Euler-embedded} in a disc, that is, the embedding contains no \emph{strongly planar vertices}. Finally, using the embedding of~$G'[B]$, \cref{thm:irrelevant_cycle} proves that we can find an irrelevant cycle~$C \subset G'$ to the instance nested deep inside the swirl~$\SSS'$ in~$fpt$-time on~$p$. Thus~$(G'-C) + D$ is a reduced and equivalent instance; go back to Step 1 and start over with the new instance.
\end{itemize}

Since each time we enter Step 3 or Step 4 we reduce~$\Abs{G}\coloneqq \Abs{V(G)}+\Abs{E(G)}$ by at least~$1$, the above algorithm stops after at most~$\Abs{G}$ many recursive steps, each of which run in~$fpt$-time. This concludes the proof.
\end{proof}

 We continue with some further details concerning each of the steps, before diving into the full details of the proof starting with \cref{sec:flat_wall}.

\paragraph{Step 1.} The first step of the algorithm is rather straightforward and self-explanatory. It is noteworthy that for general directed graphs the directed tree-width must not be bounded by the undirected tree-width at all, see for example acyclic grids which have directed tree-width~$1$, but contain a large underlying undirected  grid and hence have high undirected tree-width. Thus, the Eulerianness is crucial for the first step of the algorithm. Note further that there is no analogue to Courcelle's Theorem \cite{Courcelle1990} for directed tree-width, again highlighting that the Eulerianness is crucial for our algorithm.

\paragraph{Step 2.} The main ingredient for the second step of the proof is the \emph{Flat-Swirl Theorem}, i.e., \cref{thm:flat_swirl_theorem}. This theorem is in the same spirit as the undirected Flat-Wall Theorem as seen in \cite{GMXIII} or the directed Flat-Wall Theorem as pursued in \cite{GiannopoulouKK2020}
and the many more \emph{grid-like theorems} lying at the heart of graph-structure Theorems \cite{HatzelRW2019,GeelenJ2016}. In a nutshell the Flat-Swirl Theorem states that given high (un-)directed tree-width we either find a large router or a large flat swirl. The proof of the theorem has three major steps. In a first instance, given high directed tree-width we find a large cylindrical wall~$\WWW$ in~$fpt$-time on~$p$ using \cref{thm:dir-grid-thm,thm:dir_wall_away_from_D} as proved in \cite{KawarabayashiK2015a} and \cite{CamposLMS2019} respectively.

\smallskip

\textbf{Finding and untangling a swirl.} Given the wall~$\WWW$, we first prove that either we find a large swirl \emph{grasped by the wall}, that is, containing a large \emph{tile} of the wall---a tile can be thought of as a subgraph whose underlying undirected graph is an undirected wall---or a large router, see \cref{thm:swirl_theorem_tangled}. The wall~$\WWW\subset G$ is~$3$-regular, while~$G$ is~$4$-regular (up-to the terminals). In order to find the large swirl (or router) in the wall, we analyse how the remaining paths that start (or end) in the wall---for every vertex~$v \in V(\WWW)$ must have to be the start- or end-vertex of a path~$P \subset G-\WWW$ by the degree condition---attach to the wall. To remove ambiguity, use Whitney's Theorem \cite{Whitney32} to embed the wall in the only possible way (after choosing an outer-face) in some plane which can be used to define coordinates of the wall and orientations more easily. Given the embedding, the cycles of the wall---denoted by~$W_1,\ldots,W_{t}$---all run in the same direction, say clock-wise. Let~$H_1,\ldots,H_{2t}$ be the~$2t$ horizontal paths as in \cref{fig:wall}. Finally let~$x_{i,j} \in V(W_i)\cap V(H_j)$ be some \emph{coordinate-vertex}, i.e., a vertex of degree three. Suppose a path~$P_{i,j}$ starts in~$x_{i,j}$, then either the path~$P_{i,j}$ ends in a vertex `above'~$x_{i,j}$ and close to it, given the embedding, say~$x_{i,j-1} \in V(C_i) \cap V(H_{j-1})$, then~$P_{i,j}$ is what we will call an \emph{up-path}. Or it ends somewhere further away or below, say in~$x_{i+2,j+1} \in V(C_{i+2}) \cap V(H_{j+1})$, then~$P_{i,j}$ is what we will call a \emph{jump}; see \cref{fig:wall-jumps} for a schematic representation of both. Carefully analysing the possible types of up-paths and jumps we then deduce that a lot of up-paths prove the existence of a swirl, while a lot of jumps can be used to build routers. If we have found a router we can reduce the instance as proposed in the algorithm (see Step 3), thus assume for now that we have found a swirl~$\SSS$. Note that the swirl we find will be what we call \emph{tangled}, i.e., there may still be ways to cross paths \emph{inside} the swirl itself: for example the up-paths could pair-wise intersect in vertices still being edge-disjoint. To get rid of this nuisance we will \emph{untangle} the swirl in \cref{thm:swirl_theorem}, resulting in a swirl~$\SSS$ where the swirl itself is \emph{cross-less}. This last step turns out to be rather easy: if the starting wall is huge we either find many tangled swirls far apart, or a router by simply applying the previous result to different tiles of the huge wall. Then either one of the swirls is untangled, or each of the swirls is tangled and contains a cross, where the swirls and crosses are pair-wise edge-disjoint. Then this again can be used to build a router using the huge wall. 

\smallskip

\textbf{Flattening a swirl.} In a third step then, we will refine the analysis proving that either we find a large router on top of the untangled swirl~$\SSS$ we just found, or we find what we will call a large \emph{flat} swirl, see \cref{thm:flat_swirl_theorem}. To this extent we repeat the analysis as in the previous step---we analyse the different paths attaching to~$\SSS$---but with the additional information that we already \emph{have} a swirl which provides even more structure to build routers (recall that in the last step we started from a cylindrical grid not a swirl). Since swirls are highly similar to bi-directed grids, they come with a richer structure than cylindrical walls: it turns out that almost any path starting and ending in a swirl can be used to \emph{cross} paths. This way then, we will be able to guarantee that either we find a router attached to the swirl~$\SSS$, or the swirl is \emph{flat}, i.e., one cannot cross two paths in the flat swirl. Finally, in \cref{subsec:flattening_a_flat_swirl} we prove that we can modify the instance in such a way---keeping it equivalent---that the flat swirl (together with its attachments) can be planar embedded; until now their could be two-cuts or four-cuts that do not help in building crosses but may attach highly non-planar graphs to the flat swirl that can in turn not be embdedded but neither be used to find crossing paths starting and ending in the swirl. Indeed we prove something stronger, namely that we may equally well assume that the flat swirl can be Euler-embedded by getting rid of attachments that are loosely connected to the `faces' of the swirl (technically it is not only faces but one may think of it that way intuitively). This relies on a \cref{thm:two_paths_Frank} due to Frank, Ibaraki, and Nagamochi \cite{FrankIN1995}, proving a version of the aforementioned \emph{two-paths theorem} tailored to Eulerian digraphs.

Note here that in his dissertation Johnson \cite{Johnson2002} implicitly proved the existence of such a flat swirl (in the case that~$G+D$ is internally~$6$ edge-connected removing the problem of loosely connected attachments); the techniques we use are however very different from his and have more of an algorithmic flavor. Further, the results proved in \cref{sec:Swirls} will be of independent interest for future work we are pursuing as we will highlight later.

\paragraph{Step 3.}
This relies on a single result, namely \cref{thm:irrelevant_cycle_in_router_general} proved in \cref{sec:Routing}. The theorem says that, given a large enough router~$\RRR$ there is a cycle~$C\subset \RRR$ in the router that is irrelevant to the instance, and we can find it in~$fpt$-time on~$p$. The proof of this relies on a \cref{thm:Frank_Algorithm} due to Frank \cite{Frank1988} which states that the edge-disjoint paths problem for Eulerian digraphs can be solved in polynomial time if the demand graph~$D$ consists of two \emph{directed stars}; a directed star is a graph where all the edges have the same vertex as a head (or as a tail). Given a large router~$\RRR$ away from~$V(D)$ we then change the graph~$G$ by adding a new central vertex~$v$ and choosing a vertex~$b_i$ of each of the router-cycles~$C_i \in \RRR$ forming what we will call \emph{a branching set}~$B_\RRR \subset V(\RRR)$. Then we construct~$G^v_\RRR$ from~$G$ by adding an edge from each~$b_i$ to~$v$ and from~$v$ back to~$b_i$; the graph~$G^v_\RRR + D$ remains Eulerian by construction. Now we change the demand graph~$D=\{(t_i,s_i) \mid 1 \leq i \leq p\}$ to the graph~$D^v \coloneqq D_s \cap D_t$ consisting of the two directed stars~$D_s \coloneqq \{(v,s_i) \mid 1 \leq i \leq p\}$ and~$D_t \coloneqq \{(t_i,v) \mid 1 \leq i \leq p\}$. Again~$G^v_\RRR + D^v$ is clearly Eulerian. Finally, we can decide the instance~$G^v_\RRR + D^v$ in~$fpt$-time using Franks \cref{thm:Frank_Algorithm} (and Algorithm) and find the respective paths solving the instance (if it is a \emph{YES}-instance) in~$fpt$-time \cite{Frank1988}. In a nutshell the instance~$G^v_\RRR + D^v$ asks whether there exist~$p$ edge-disjoint paths routing from~$s_1,\ldots,s_p$ to~$t_1,\ldots,t_p$ that \emph{use} the router cycles. If the answer is \emph{NO}, then this means that not all paths can simultaneously use the router---and thus not all paths enter it---which in turn can be used to apply inductive reasoning. If the instance is a \emph{YES}-instance, then we will show that it remains a \emph{YES}-instance even after deleting part of the router~$\RRR$, that is, after deleting a large sub-router~$\RRR'\cup C^\ast \subset \RRR$. Finally, we prove that the deleted sub-router~$\RRR'$ can be used in~$G$ to reroute the solution in~$G^v_\RRR$ omitting~$v$ and~$C^\ast$, and thus producing a solution in~$G$ omitting a cycle~$C^\ast \subset \RRR~$(see \cref{thm:rerouting_with_routers}). In either case we will show that there is some cycle of the router (it can be found in~$fpt$-time on~$p$) that could have been deleted from the start not changing the outcome of the instance~$G+D$.

\paragraph{Step 4.}
The last part of the algorithm is the most elaborate part again relying on a single result, namely \cref{thm:irrelevant_cycle} which, in a nutshell, proves that given a graph~$G+D$ and a large Euler-embedded swirl~$\SSS$, the most deeply nested cycle of~$\SSS$ is irrelevant to the instance. 

\smallskip

\textbf{Shifting the Paradigm.} Readers familiar with the \emph{Graph Minor Structure Theorem} due to Robertson and Seymour \cite{GMXVI} or the proof that the undirected (vertex-)disjoint paths problem can be solved in~$fpt$-time \cite{GMXIII} (which heavily relies on the irrelevant vertex technique \cite{GMXXI}) will see the inspiration in our line of argumentation, but also the differences in the obstacles we have to overcome stemming from the edge directions and the fact that we need to keep our graphs Eulerian. The general line of reasoning uses similar arguments as the respective proof of the undirected version in \cite{GMXXI}. However, many of our proofs differ from the proofs given in \cite{GMXXI} for many of the arguments do not transfer to our setting; for example we cannot cut the surface at lines going through vertices, for splitting vertices may result in non-Eulerian graphs (we could actually, but this would need some more argumentation that can be circumvented by simply not doing so). The arguments and ideas we provide heavily exploit the Eulerianness of the graphs and the fact that we are solving the edge-disjoint paths problem rather than the vertex-disjoint paths problem. The latter turns out to be far more impactful than anticipated: for one of the most notable differences being that we will \emph{shift the paradigm} from graphs to be seen as a set of vertices~$V$ where edges~$E$ are solely relations on vertices, to graphs to be seen as \emph{incidence structures}, that is vertices~$V$ and edges~$E$ are equals and may (in theory) exist without the need of each other. That is, an edge~$e \in E$ will be an element for itself where each edge in a (non-partial) graph is adjacent to exactly two vertices which is captured by the~$\operatorname{tail}\subset E \times V$ and~$\operatorname{head} \subset E \times V$ relations. This will be made more explicit and heavily used in \cref{sec:charting_eulerian_digraphs,sec:shippings,sec:structure_of_min_examples,sec:irrelevant_cycle_theorem}. It turns out that shifting the paradigm towards incidence graphs facilitates and smoothens a lot of the reasoning in \cref{sec:charting_eulerian_digraphs,sec:structure_of_min_examples,sec:shippings}. Note here that the conceptually similar idea of simply switching to the \emph{line-graph}---the graph obtained when taking edges~$e \in E(G)$ as vertices and adding an edge if there exists a~$2$-path between two edges (see \cref{def:linegraph})---and trying to leverage the arguments made for the vertex-disjoint case comes with a lot of obstacles. The key difference being, that when passing to the line-graph we create unnecessary edges that were of no use in the first place, for no solution may have ever used a two-path in~$G$ representing that edge. Note further that, even if the instance has now been reduced to a vertex-disjoint paths problem, many of the arguments made in \cite{GMXXI} do not transfer to this setting trivially: we cannot delete vertices nor contract any edge in the line-graph; operations heavily used to prove the main theorems in \cite{GMXXI,GMXXII} and in turn the irrelevant vertex theorem. Hence, since we are interested in edges, and edge-disjoint paths (rather walks, which allow to re-use vertices), it seemed only natural to take this step. We switch to viewing graphs as incidence structures and adapt the necessary notation in \cref{sec:charting_eulerian_digraphs}.

\smallskip

\textbf{A minimal counterexample.} The general idea is now as follows. Let~$G+D$ be an Eulerian graph and suppose that~$G$ contains an~$h$-swirl~$\SSS=\bigcup_{1\leq i \leq h}S_i$ that can be Euler-embedded in a disc (which we found in Step 2. of the algorithm). In order to prove the existence of an irrelevant cycle deeply nested inside the swirl~$\SSS$ we assume the contrary and let~$G+D$ be minimal towards our hypothesis, and~$\LLL$ a linkage in~$G+D$ witnessing that it is a \emph{YES}-instance and thus by assumption visiting all of the cycles deeply nested in~$\SSS$. In \cref{thm:irrelevant_cycle_minimal_counterexample} of \cref{sec:structure_of_min_examples} we prove that the graph adheres to a very restrictive structure, namely~$G+D$ is what we call an \emph{$h$-flower graph}; see \cref{fig:flower_graph}. In particular~$V(G)=V(\SSS)\cup V(D)$ and every other possible edge in~$G+D$ is either with vertices in~$V(D)$ or between two vertices lying on the outer-cycle~$S_h$. The most crucial observation being that the linkage~$\LLL$ must be what we call \emph{rigid}, i.e.,~$\LLL$ uses \emph{all} of the edges in~$G$ and there exists \emph{no} other linkage~$\LLL'$ solving the instance (see \cref{def:rigid_linkage}). This is very restrictive and a key ingredient to the following proofs in \cref{sec:shippings,sec:irrelevant_cycle_theorem}. Note that by \cref{thm:irrelevant_cycle_in_router_general} this implies that the graph cannot contain a large router anymore, for else there was another solution omitting a cycle of the router by Step 3. Using the rigidity we can derive even more restrictive patterns for the paths in~$\LLL$, namely it turns out that any restriction of a path~$L \in \LLL$ to the swirl~$\SSS$ is what we will call a \emph{level~$i$ path}. That is, the restriction---call it~$P \subset L$---is a path starting at the outer-cycle~$S_h$ which goes straight down to some swirl-cycle~$S_{i}$ and from there goes straight back up to~$S_h$ where it ends; in particular a level~$i$ path visits every swirl-cycle (up to the level~$i$) exactly twice. This is a rather powerful observation we will use in \cref{thm:shipping_in_open_sea}, highlighting that the rigid linkage~$\LLL$ is even more rigid than usual. 

The final ingredient is to show that, for~$h$ large enough, there exists no rigid linkage in a graph without large routers but containing a flat~$h$-swirl~$\SSS$ as given. This will use some heavy machinery that we describe next.

\smallskip

\textbf{Coastal Maps.} One can show that, after deleting at most~$\alpha(p)$ vertices in~$G$ (which is equivalent to cutting through at most~$4\alpha(p)$ edges for deleting a vertex does in partial graphs not really delete the edges), the~$h$-flower graph~$G+D$ admits what we will call a \emph{weak coastal map} (see \cref{def:weak_coastal_map}). Think of a coastal map of~$G'$---the graph resulting form cutting said edges, we describe this below---as~$G'$ being Euler-embedded on some surface~$\Sigma$ up to a few discs which do not contain large edge-disjoint linkages starting and ending in its boundary, where we get some nice linear decomposition of the graphs inside those discs (see \cref{subsec:charting_an_island} for a construction of coastal maps given the structure theorem); we mill make this more precise shortly. The proof of the existence of a coastal map relies on a structure \cref{thm:Johnsons_structure_theorem} for Eulerian directed graphs; actually we only need a structure theorem for the structurally very restricted class of flower-graphs. In his dissertation Johnson \cite{Johnson2002} proved a structure theorem for internally~$6$ edge-connected Eulerian digraphs that suits our needs. We will use said theorem in \cref{sec:structure_thms} since giving a rigorous proof ourselves (where we could reduce our interest to the case of flower-graphs) would take up quite some pages without much new insights. However, since the results in the dissertation have not yet been published, we will provide a proof---and a structure theorem for general Eulerian digraphs---in the future, using differing techniques similar to the ones pursued in \cite{KawarabayashiTW2021} and matching the techniques of the results we provide for finding the flat swirl in \cref{sec:flat_wall}.

\smallskip

 Thus let~$G+D$ be the~$h$-flower graph and let~$\LLL$ be a rigid~$p$-linkage on~$G$. Then cutting the respective~$4\alpha(p)$ edges results in a new Eulerian graph~$G' + D'$ admitting a~$p+2(\alpha(p))$-rigid linkage and a large enough Euler-embedded swirl, say the whole~$h$-swirl~$\SSS$ for simplicity (see \cref{lem:rigid_linkages_after_cutting_edges}), in a disc~$\Delta\subseteq \Sigma$. To make cutting edges more rigorous think of it as follows: let~$e \in E(G)$ be an edge, then since~$\LLL$ is rigid there exists~$L \in \LLL$ with~$e \in E(L)$. Since~$L$ is part of the solution to~$G+D$, there exists~$(t,s) \in E(D)$ such that~$L$ starts in~$s$ and ends in~$t$. Now cut~$e=(u,v)$ into two edges by introducing new vertices i.e., into two edges~$e_1 \coloneqq (u,u_e)$ and~$e_2 \coloneqq (v_e,v)$. Then replacing~$(t,s) \in E(D)$ via~$(t,v_e),(u_e,s)$ does the trick. Hence, for convenience, let us simply think of the resulting graph after the~$f\alpha(p)$ cuttings as~$G+D$ again and of the linkage as a~$p$-linkage; in short assume~$G+D$ is an~$h$-flower graph such that, as mentioned above, it admits a weak coastal map as in \cref{def:weak_coastal_map}. To make the intuition of coastal maps given above slightly more precise we propose the following: `admitting a coastal map' means that there exist graphs~$\Gamma,\III \subset G$ and a surface~$\Sigma$ (possibly with boundary) such that~$G=\Gamma \cup \III$ where~$\Gamma$ can be Euler-embedded into~$\Sigma$ with a few vertices (conceptually edges) drawn on the boundary~$\bd(\Sigma)$---rather the \emph{zone} (see \cref{def:ports_shores_zones})---of~$\Sigma$, with~$\SSS\subset \Gamma$. Further, each component~$I'$ of~$\III$ is what we will call an \emph{island} (see \cref{def:weak_island}), that is~$\Gamma \cap I \subset V(G)$ is drawn on a single cuff~$C \in c(\Sigma)$ of the boundary of~$\Sigma$, and there does not exist a~$d(p)$-linkage starting and ending in~$\Gamma \cap I$ that is other-wise contained in~$I$, i.e., the islands are of \emph{bounded depth}~$<d(p)$. (Readers familiar with the undirected graph-minor structure theorem may intuitively think of islands as vortices.) Then the coastal map assigns sort of a linear decomposition for each island, guaranteeing some linkedness properties of the respective parts; this is too technical to describe in detail, simply think of this as regrouping~$I$ into a bounded number of chunks that are well-connected and `tamed' so we have a good idea on how the linkage~$\LLL$ behaves inside the islands. Finally we prove that given a coastal map of a graph~$G$ of bounded depth, the instance~$G+D$ cannot admit a rigid linkage, a contradiction to the above, i.e., that the~$h$-flower graph admits a rigid linkage. The proof reduces the~$p$-linkage~$\LLL$ to a rigid~$(\rho(p))$-linkage~$\LLL'$ of~$\Gamma$ (a function independent of~$h$). The proof works via induction by cutting through a bounded number of edges in~$G$, whose endpoints result in new demand-edges for a new demand graph~$D'$ which by our assumption need not be embedded in~$\Sigma$ and neither be part of any island; this is crucial. Then, after doing some more skilled cuttings (we massage a weak coastal-map into a \emph{strong} one as given by \cref{def:strong_coastal_map}), if any cuff~$c(\Sigma)$ contains too many vertices of~$G$ (that is more than~$t(p)$ many for some~$t$) and thus too many \emph{ports} (see \cref{def:ports_shores_zones}) then the linkage can be rerouted inside an island to omit part of the edges lying on the zone, that is edges adjacent to the vertices lying on the boundary of~$\Sigma$ (the key-ingredient being \cref{thm:rigid_linkage_in_laminar_cuts_is_Menger}). Since~$\LLL$ was assumed to be rigid, this is impossible and thus only a bounded (in~$p$) number of vertices lie on the boundary of~$\Sigma$. This guarantees that our linkage~$\LLL$ does only enter the island a `bounded number of times' and thus the islands cannot be used to cross paths `too often'. But now, after cutting open the edges adjacent to those vertices, we get a new graph~$G'+D'$ where~$G'$ is \emph{completely} Euler-embedded in~$\Sigma$ together with a rigid~$g(p)$-linkage (for some function~$g$) in~$\Gamma$ containing an Euler-embedded~$h(p)$-swirl; we got rid of the islands.

\smallskip

\textbf{Shipping in the open Sea.} We reduced the problem to the case where we have a graph~$G+D$ together with an~$h$-swirl~$\SSS$ such that~$G$ is Euler-embedded in~$\Sigma$ and admits a rigid~$p$-linkage. To prove that this is impossible (that is the linkage cannot be unique) we again look at a minimal counterexample, which again is an~$h$-flower graph as above. We then leverage \cref{thm:directed_irr_vertex_Marx}, a result due to Cygan, Marx, Pilipczuk, and Pilipczuk \cite{CyganMPP2013} who have proved that given a graph embedded in a plane (not any surface~$\Sigma$) such that there exists a large embedded swirl (in their case nested concentric cycles of alternating orientation) then the most deeply nested cycle is irrelevant to the instance. Our proof uses induction on~$\Sigma$ and the structure of the minimal counterexample~$G+D$. The idea is that we are able to guarantee that we can cut the surface~$\Sigma$ along a curve that intersects the swirl~$\SSS$ in a bounded (say~$\ell(p)$) number of edges such that the resulting graph remains of high tree-width, i.e.,  resulting in a surface of lower genus, again containing a large swirl~$\SSS$ and a rigid~$p'=(p+\ell(p))$-linkage, where the base-case for the sphere can be reduced to the mentioned theorem proved in \cite{CyganMPP2013} which is given as \cref{thm:directed_irr_vertex_Marx} in this paper. To see how to use \cref{thm:directed_irr_vertex_Marx}, switch to the line-graph~$\LL(G)$ yielding a~$p'$-vertex-disjoint linkage using all of the vertices in~$\LL(\Gamma)$ (for the linkage was rigid in~$G$) where~$\LL(G)$ is planar embedded in the sphere (for the graph~$G$ was Euler-embedded in the sphere). It still contains a large swirl, for the line-graph of a swirl is a vertex-disjoint swirl by definition. But then all the requirements for \cref{thm:directed_irr_vertex_Marx} are met and the linkage can be rerouted to a vertex-disjoint linkage omitting a vertex (or a whole cycle) in the swirl of~$\LL(G)$, which in turn results in an edge-disjoint linkage omitting an edge (or a whole cycle) of the swirl in~$G$. This then concludes the proof for the \cref{thm:main} in Step 4. Again we note that in his dissertation Johnson \cite{Johnson2002} proved a similar theorem that in our setting implies that if~$G$ is embedded in~$\Sigma$ and~$G$ contains a large swirl, then the linkage cannot be rigid, i.e., he too proved the base case of our induction of Step 4. Note also that it is not clear at all how to extend the base-case result for Euler-embedded graphs to the general setting (which we capture by coastal maps). The proof we provide for the base-case is a consequence of preliminary work we have done to analyse the minimal counter-example---this is of independent importance to the proof of \cref{thm:irrelevant_cycle}---which is of its own interest, yielding a deeper understanding of the behaviour of rigid linkages in minimal counterexamples.

\medskip

Altogether this concludes the discussion of the main ideas used in the proof of \cref{thm:main}. We continue with introducing some more notation and preliminary results.

\section{Preliminaries}
\label{sec:prelims}
Throughout this paper we use standard graph-theoretic terminology following \cite{Bang-JensenG2018} and \cite{Diestel2012} unless stated otherwise. We write~$G=(V,E)$ to mean a \emph{directed} graph (or \emph{digraph}) with vertex set~$V$ and directed edges~$E$. Given a (directed) edge~$e=(u,v) \in E$ we call~$u$ the \emph{tail} of~$e$ and~$v$ the \emph{head} of~$e$. Throughout this paper we will tacitly assume graphs to be directed unless stated otherwise. In \cref{sec:charting_eulerian_digraphs} and the following sections we will switch notation to \emph{incidence graphs} and restate all necessary notation. Given a directed graph~$G=(V,E)$ we call the graph resulting from~$G$ when forgetting the directions of the edges its \emph{underlying undirected graph}. In this paper we will mainly focus on \emph{Eulerian} digraphs.

\begin{definition}[Eulerian Digraph]
    A graph~$G=(V,E)$ is \emph{Eulerian} if for every vertex~$v \in V(G)$ the number of incoming edges in~$v$ equals the number of outgoing edges in~$v$.
\end{definition}

Since we are interested in edge-disjoint paths, we will use the following definition of paths and cycles which deviates from the definition in the standard literature.

\begin{definition}[Paths and cycles]
    Let~$G$ be an Eulerian graph. Let~$t > 1$ and let~$(e_1,\ldots,e_t)$ be a sequence of distinct edges~$e_i=(u_i,v_i) \in E(G)$ with~$v_i = u_{i+1}$ for~$1 \leq i < t$.
     Then if~$u_1 \neq v_t$ we say that~$P \coloneqq (e_1,\ldots,e_t)$ is \emph{a path} in~$G$ and we call~$t$ its \emph{length}. We abuse notation and identify~$P$ with the subgraph~$P\subset G$ such that~$V(P) \coloneqq \bigcup_{1 \leq i \leq t}\{u_i,v_i\}$ and~$E(P) = \{e_1,\ldots, e_t\}$.

     Given a path~$P$ and two vertices~$u_i,u_j \in V(P)$ for~$1 \leq i < j \leq t$ we write~$u_iPu_j$ to mean the sub-path~$(e_i,\ldots,e_j)$ of~$P$ and we write~$(e_i,\ldots,e_j) \subseteq P$ to mean said sub-path.

    If~$u_1 = v_t$ we say that~$C \coloneqq (e_1,\ldots,e_t,e_1)$ is \emph{a cycle} in~$G$ and we call~$t$ its length. We again abuse notation and identify~$C$ with the respective subgraph~$C \subset G$. Similarly as above we define~$u_iCu_j$ (regardless whether~$i<j$ or~$j<i$) to mean the path connecting~$u_i$ to~$u_j$ along~$C$.

    We call a path~$P$ (or a cycle~$C$) \emph{vertex-disjoint} if it does not visit the same vertex twice.
    \label{def:paths_and_cycles}
\end{definition}

In the standard literature what we call paths is called \emph{trails}, while what we call vertex-disjoint paths are called paths, and what we call cycles are called \emph{circuits} while what we call vertex-disjoint cycles are called cycles. 

\smallskip

We call the directed graph~$G=(V,E)$ \emph{connected} if its underlying undirected graph is connected in the usual sense. Note that in the general literature on directed graphs what we call \emph{connected} is referred to as \emph{weakly connected}, where a directed graph~$G=(V,E)$ is called \emph{strongly connected} if for any two vertices~$u,v \in V(G)$ there exists a path connecting~$u$ to~$v$ and a path connecting~$v$ to~$u$. Throughout the rest of the paper we will assume the graphs to be connected unless stated otherwise.

\smallskip

\begin{definition}[Linkages]
    Let~$G$ be a graph and let~$L_1,\ldots,L_p$ be a collection of edge-disjoint paths in~$G$ for some~$p \in \N$. Then we call~$\LLL \coloneqq \{L_1,\ldots,L_p\}$ a~\emph{$p$-linkage (in~$G$)}.
\end{definition}
\begin{remark}
    We will write~$\bigcup_{1 \leq i \leq p}L_i$ or simply~$\bigcup \LLL$ to mean the subgraph obtained form combining all the paths in the linkage.
\end{remark}

Linkages come with a \emph{pattern}; depending on the context we are interested in the \emph{vertex-pattern} of a linkage---the pairs of vertices in~$V(G)$ that are the endpoints of a path in~$\LLL$---or the \emph{directed pattern}---that is the pairs of edges a path in~$\LLL$ starts and ends with. In a nutshell the first patterns are of importance in \cref{sec:Swirls,sec:Routing} while the latter are of importance in \cref{sec:charting_eulerian_digraphs,sec:shippings,sec:structure_of_min_examples}; they will be defined when needed.

We define what we mean by an~$A{-}B$-linkage.
\begin{definition}[$A{-}B$-linkage]
    Let~$G$ be a graph and let~$A,B \subset V(G)$ be vertex sets. Let~$\LLL$ be a linkage in~$G$ such that for every path~$L \in \LLL$ with endpoints~$s,t \in V(G)$ it holds that~$s\in A \iff t \in B$. In particular every path in~$\LLL$ either starts in~$A$ or~$B$ and ends in~$B$ or~$A$.
\end{definition}

In the context of linkages we will frequently use the following well-known theorem by Erd\H os and Szekeres \cite{ErdosS1935} (see also \cite[Theorem 4.4]{Jukna2001}). 
If $\bar{a} := (a_1, \dots, a_n)$ is a sequence of numbers, then a sub-sequence of length $s$ of $\bar{a}$ is a sequence $\bar{b} := (a_{i_1}, \dots, a_{i_s})$ such that $i_1 < \dots < i_s$.  $\bar{b}$ is strictly increasing if $a_{i_j} < a_{i_{j'}}$ for all $j < j'$ and it is strictly decreasing if $a_{i_j} > a_{i_{j'}}$ for all $j < j'$.

\begin{theorem}
    Let $\bar{a} := (a_1, \dots, a_n)$ be a sequence of distinct numbers. If $n \geq (s-1)\cdot (r-1) +1 $ then $\bar{a}$ contains a strictly increasing sub-sequence of length $s$ or a strictly decreasing sub-sequence of length $r$.   
\end{theorem}

We make rigorous what we mean by joining and deleting graphs.
\begin{definition}[$G\cup H$,~$G\cap H$ and~$G\setminus H$]
    Let~$G$ and~$H$ be graphs. We define~$G \cup H \coloneqq (V(G) \cup V(H) ,E(G) \cup E(H))$, and~$G\cap H \coloneqq (V(G) \cap V(H), E(G) \cap E(H))$, and~$G\setminus H \coloneqq (V(G) \setminus V(H) ,E(G) \setminus E(H))$.
\end{definition}
\begin{remark}
    Note that we do \emph{not} implicitly take a disjoint union and we do not rename any vertices in the case that~$V(G) \cap V(H) \neq \emptyset$ for we will need this at times.
\end{remark}

Next we define what we mean by adding and subtracting graphs.

\begin{definition}[$G+H$ and~$G-H$]
    Let~$G$ be a graph with vertex set~$V\coloneqq V(G)$ and let~$H$ be a graph with vertex set~$V(H) \subseteq V$. We define~$G-H \coloneqq (V,E(G) \setminus E(H))$.
    If~$E(H) \cap E(G) = \emptyset$, then we define~$G+H \coloneqq (V,E(G)+E(H))$.
\end{definition}

With this at hand we are ready to define the Eulerian edge-disjoint paths problem in the form we will study it.

\begin{definition}[Eulerian Edge-Disjoint Paths Problem]
    Let~$G$ and~$D$ be graphs where~$V(D)=\{s_1,\ldots,s_p,t_1,\ldots,t_p\} \subseteq V(G)$ and~$E(D) \coloneqq \{ (t_i,s_i) \mid 1 \leq i \leq p\}$ for some~$p \in \N$. Further assume that~$G+D$ is Eulerian, and in particular~$E(G) \cap E(D) = \emptyset$; we call~$D$ a \emph{demand graph}. The \textsc{Eulerian Edge-Disjoint Paths Problem} (encoded by~$G+D$) asks whether there exists a~$p$-linkage~$\LLL=\{L_1,\ldots,L_p\}$ such that~$L_i$ connects~$s_i$ to~$t_i$ for every~$ 1 \leq i \leq p$.

    If these paths exist we say that~$G+D$ is a \emph{YES}-instance, otherwise it is a \emph{NO}-instance.
\end{definition}
\begin{remark}
    Clearly if there exist edge-disjoint paths~$L_i$ connecting~$s_i$ to~$t_i$ for every~$1 \leq i \leq p$, then there exist edge-disjoint paths~$L_1',\ldots,L_p'$, each one of them a vertex-disjoint path, connecting~$s_i$ to~$t_i$ (simply reroute at self-intersection points). 
\end{remark}

Finally we note that given two digraphs~$H,G$, an \emph{immersion} of~$H$ in~$G$, is an injective map~$\gamma:V(H) \to V(G)$ such that for every edge~$e=(u,v) \in E(H)$ there exists a path~$P_e \subset G$ such that~$P_e$ connects~$\gamma(u)$ to~$\gamma(v)$ such that the paths~$P_e,P_{e'}$ are pairwise edge-disjoint for distinct~$e,e' \in E(H)$; we write~$\gamma: H \hookrightarrow G$. Although we will not explicitly state (or need) it at times, the structures we are analysing are always given in form of immersions.

\subsection{Notation}\label{subsec:notation}
In this subsection we define some of the notation used throughout the paper. Note that in \cref{sec:charting_eulerian_digraphs} we will slightly deviate from the notation we introduce here, which will be made precise in said section.

\paragraph{General concepts.}We start with the general concepts that will accompany us throughout most sections of the paper.

\begin{definition}[Induced subgraphs]
    Let~$G$ be some graph and let~$X \subseteq V(G)$. We write~$G[X]\coloneqq (X,E)$ to mean the subgraph with vertex set~$X$ and edge set~$E \coloneqq E(G) \cap (X\times X)$.
\end{definition}

Throughout this paper we will extensively work with \emph{(induced) cuts} in a graph.
\begin{definition}[Induced Cuts]    \label{Def:induced_cuts_prelims}
    Let~$G$ be a graph and~$X \subseteq V(G)$. We define~$\delta(X) \coloneqq \{ e \in E(G) \mid e$ is incident to a vertex in~$X$ and to a vertex in~$V(G) \setminus X \}$ and refer to it as the \emph{cut induced by~$X$}. We call~$\Abs{\delta(X)}$ the \emph{order} of the cut. 

    We write~$\delta^+(X) \coloneqq \{(u,v) \in \delta(X) \mid u \in X\}$ 
    and~$\delta^-(X) \coloneqq \delta(X) \setminus \delta^{+}(X)$. 
\end{definition}
\begin{remark}
    This definition will slightly change in \Cref{sec:charting_eulerian_digraphs,sec:shippings,sec:structure_of_min_examples}; see \cref{Def:induced_cuts}.
\end{remark}

Given~$X \subseteq V(G)$ we write~$\bar{X} \coloneqq V(G) \setminus X$ for shorthand notation. Note that~$\delta(X) = \delta(\bar{X})$ holds trivially by definition. Given a subset~$X\subseteq V(G)$ inducing a cut, we define the following extension of~$G[X]$.

\begin{definition}
    Let~$G$ be a graph and let~$X \subseteq V(G)$ induce some cut in~$G$. Then we define~$G[[[X]]] \coloneqq G[X] \cup \delta(X)$. 
    \label{def:G_with_cut_edges_full_prelims}
\end{definition}
\begin{remark}
    Note that in \cref{sec:charting_eulerian_digraphs} and the following sections we will see another restriction of graphs, namely~$G[[X]]$. But since we will not use it until we talk about incidence-graphs we omit it here and refer the interested reader to \cref{def:induced-partial-graphs}.
\end{remark}

In some situations it will turn out to be convenient to modify the graph~$G[[[X]]]$ slightly to assure that no two edges~$e_1,e_2 \in \delta(X)$ have a common end in~$V(G[[[X]]]) \setminus V(G[X])$.

\begin{definition}[Euler-restriction]
    Let~$G+D$ be an Eulerian graph and let~$X \subseteq V(G)$ such that~$D\subseteq G[\bar{X}]$. Then we define~$\delta(X)^{\text{ext}}$ as follows. For each edge~$e=(u,v) \in \delta(X)$ such that~$u \in X$ we define a vertex~$t_e$ and if~$v \in X$ we define a vertex~$s_e$. Then let~$\delta(X)^{\text{ext}} = \{(u,t_e),(s_{e'},v) \mid u,v \in X \text{ and } e=(u,w),e'=(z,v)\text{ width } e,e'
    \in \delta(X)\}$. Finally we define the \emph{Euler-restriction of~$G$ to~$X$} as~$\restr{G}{X} = G[X] \cup \delta(X)^{\text{ext}}$.
    \label{def:Euler-restriction}
\end{definition}

Throughout the paper, we will often work with cyclic coordinates, which is why we define the following.
\begin{definition}
    A \emph{cyclic ordering} on a set~$V$ is an enumeration of its elements~$v_1, \dots, v_n$, where addition in the indices is taken modulo~$n$, i.e.~$v_{i+l} = v_{({i+l \mod (n +1))+1}}$ if~$i+l > n$ and analogously defined for~$i-l$.   We call~$v_{i+1}$ the \emph{successor} of~$v_i$.
    
    A subset~$A \subseteq V$ is called \emph{consecutive} if~$A = V$ or there is exactly one~$v \in A$ whose successor is not in~$A$. We call~$v, v'$ the \emph{bounds of~$A$} with~$v$ being the \emph{end vertex} and~$v'$ being the \emph{start vertex}. 
    We say that~$v_i$ is \emph{between}~$v_j$ and~$v_l$ if there is a consecutive subset~$A$ of~$V$ starting at~$v$, ending at~$v_i$ and not containing~$v_l$.
\end{definition}
\begin{remark}
    Let~$A \subseteq V$ be consecutive, then, assuming~$V$ is finite, there is exactly one element~$v'$ whose predecessor is not in~$A$.
\end{remark}
Whenever we speak about cyclic orderings in the sequel we always assume implicitly that addition is modulo~$n = |V|$. Further the notion of consecutive sets allows us to unambiguously specify a partition of~$V$ into disjoint sets~$A, B$  by writing~$A = \{ v_i, \dots, v_j\}$ and~$B \coloneqq \{ v_{j+1}, \dots, v_{i-1}\}$. 

\smallskip

Henceforth, when talking about cycles, we will often assume the edges (or vertices) to be labelled by a cyclic ordering with respect to the cycle without further mention.

\paragraph{Surfaces (with boundary).}\label{par:surfaces}
We will use standard notation and results on surfaces (with boundary) used in graph theory (see \cite{DiestelKMW2012,MoharT2001}). We will commonly denote a surface (with boundary)~$\Sigma$ where we write~$\bd(\Sigma)$ for its boundary. We will write~$c(\Sigma)$ for the set of \emph{cuffs} in~$\Sigma$, where a cuff is a connected component of~$\bd(\Sigma)$ that can be traced by a simple closed curve. Note that~$\bigcup_{ C \in c(\Sigma)} C = \bd(\Sigma)$ where~$\bd(\Sigma)$ denotes its boundary. Given a surface~$\Sigma$ with boundary we write~$\hat{\Sigma}$ to mean the surface obtained when gluing a disc~$\Delta_C$ to each cuff~$C \in c(\Sigma)$ by choosing a simple closed curve~$\gamma_\Delta$ tracing the boundary of the disc~$\Delta_C$ and a respective simple closed curve~$\gamma_C$ tracing the cuff and gluing the disc to the cuff by gluing along these curves. 

\paragraph{Graph Drawings.}
We define what we mean by a \emph{drawing} (or \emph{embedding}) of a graph~$G$.
\begin{definition}[Drawing of a graph]
    Let~$G=(V,E)$ be a graph. Then a \emph{drawing} of~$G$ or an \emph{embedding} of~$G$ on a surface~$\Sigma$ is a subset~$U(G) \subset \Sigma$ together with an injective map~$\nu: V \to U(G)$ such that the components of~$U(G) \setminus \nu(V)$ are pairwise disjoint lines---injective images of open intervals on a surface---and if~$(u,v) \in E(G)$ then there exists a component~$\ell_e \in U(\Gamma)\setminus \nu(V)$ such that~$\ell_e = \gamma_i(]0,1[)$ for some injective map~$\gamma_i$ and the closure of~$\ell_e$ in~$U(\Gamma)$ contains~$\nu(u)$ and~$\nu(v)$. If a graph can be drawn on a surface in this way, and given some drawing~$(U,\nu)$, we say that the graph is \emph{embedded in~$\Sigma$}. 
\end{definition}
\begin{remark}
    This definition changes slightly in \cref{sec:charting_eulerian_digraphs} and the following sections; see \cref{def:embedding_incidence_graphs}.
\end{remark}
Given an embedding of a graph we define its faces as follows.

\begin{definition}[Faces and~$2$-cell embeddings]\label{def:faces_and_2cell}
    Let~$G$ be a graph and let~$(U,\nu)$ be an embedding of~$G$ in some surface~$\Sigma$ (possibly with boundary). Then we refer to the connected components of~$\Sigma \setminus (U(G)\cup \bd(\Sigma))$ as the \emph{faces} of the drawing. 

    If every face is homeomorphic to an open disc we call~$(U,\nu)$ a~$2$-cell embedding.
\end{definition}

It is well-known that every embedding of a connected graph in a disc, sphere or plane is~$2$-cell.

\begin{observation}\label{obs:disc_embedd_is_2cell}
    Let~$G$ be a connected graph and let~$(U,\nu)$ be an embedding of~$G$ in some surface~$\Delta$ where~$\Delta$ is (homeomorphic to) either a disc, a plane or a sphere. Then~$(U,\nu)$ is~$2$-cell.
\end{observation}
This is the reason why we will tacitly assume our embeddings to be~$2$-cell whenever drawn in a disc, sphere or plane (which will be the case most of the time).

\smallskip 

If a graph can be drawn in a plane (or disc) we call the graph \emph{planar}. We will refer to~$G$ as \emph{a plane graph} when we mean a graph~$G$ together with an embedding in a disc (or plane). When talking about embedded graphs we may abuse notation and refer to~$\nu(v)$ as~$v\in V(G)$ and to~$\ell_e$ as~$e \in E(G)$ for simplicity if it is clear from the context. In this paper we will be interested in a rather special form of embedding, namely \emph{Euler-embedding}. To this extent we call a vertex~$v \in V(G)$ \emph{Eulerian} if its in-degree and out-degree are equal in~$G$.

\begin{definition}[Euler-embedding and strongly planar vertices]
    Let~$G$ be an Eulerian graph of maximum degree four and let~$G' \subset G$ be a subgraph embedded in a surface~$\Sigma$ witnessed by a drawing~$(U,\nu)$. Let~$v \in V(G')$ be an Eulerian vertex of degree four in~$G'$ and let~$\gamma_v:[0,1]\to \Sigma$ be a simple closed curve winding around~$\nu(v)$ in~$\Sigma$ such that~$\gamma[0,1]$ intersects~$U(G)$ exactly in edges~$e \in E(G)$ incident to~$v$ and in each such edge exactly once. Let~$(e_1,e_2,e_3,e_4)$ be the edges visited in that order by~$\gamma_v$ (up to a cyclic rotation). Then we call~$v$ \emph{strongly planar} if~$e_1$ and~$e_2$ are edges with~$v$ as a head and~$e_3$ and~$e_4$ are edges with~$v$ as a tail (up to a cyclic rotation).

    If no Eulerian vertex is strongly planar given the embedding~$(U,\nu)$, then we call it a \emph{Euler-embedding} or \emph{Eulerian embedding}.
    \label{def:Euler-embedding}
\end{definition}

The following is a well-known fact regarding Euler-embedded graphs.
\begin{observation}\label{obs:faces_in_Euler_embeddings_bounded_by_cycle}
    Let~$G$ be a graph Euler-embedded in some surface~$\Sigma$ such that the embedding is~$2$-cell. Then every face in the embedding is bounded by a vertex-disjoint cycle.
\end{observation}

Both observations above are far more impactful than one may anticipate. As a first instance we will use it to define what we mean by \emph{concentric cycles}; recall that our cycles may not be vertex-disjoint. To this extent we define the \emph{outline of cycles.}

\begin{definition}\label{def:outline_of_cycle}
    Let~$G$ be a graph and let~$C \subseteq G$ be some cycle. Let~$(U,\nu)$ be an embedding of~$C$ in a disc~$\Delta$ such that~$U(C) \cap \bd(\Delta) = \emptyset$. Then the edges incident to the face containing~$\bd(\Sigma)$ form  vertex-disjoint cycle; we refer to it as the \emph{outline} of~$C$.
\end{definition}
\begin{remark}
    Note that our cycles may be non-injective closed curves in the drawing, for we allow cycles to not be vertex-disjoint (see \cref{def:paths_and_cycles}). In particular they do not necessarily bound discs without further restrictions; the outline in turn \emph{does} bound a disc.
\end{remark}
 Note that we may equally well define the outline for graphs drawn in a plane (or the sphere) by choosing an outer-face first, taking the role of~$\bd(\Delta)$ once an orientation has been fixed.

\subsection{General results}
In this subsection we gather a few general results needed throughout the paper. For ease of readability we do not gather all the results we need in this section, for it seems more convenient to introduce the results when needed. For example we will need some aforementioned results about finding edge-disjoint paths in Eulerian digraphs due to Frank, Ibaraki, and Nagamochi \cite{Frank1988,FrankIN1995}; we introduce said results and the needed notation in the respective sections (see \cref{subsec:flattening_a_flat_swirl,sec:Routing}).

Kawarabayashi and Kreutzer have proved that given a directed graph of high directed tree-width, there exists a large cylindrical grid-minor (by taking butterfly-minors) in the graph; see the lefthand side of \cref{fig:wall} for a representation of such a minor. 

 \begin{definition}[Elementary cylindrical grid]\label{def:cyl-grid}
 
  An \emph{elementary cylindrical grid} of order~$k$, for some~$k\geq 1$, is a
  digraph~$G_k$ consisting of~$k$ directed cycles~$W_1, \dots, W_k$
  of length~$2k$,
  pairwise vertex disjoint, together with a set of~$2k$ pairwise
  vertex disjoint paths~$H_1, \dots, H_{2k}$ of length~$k-1$ such that \parsep-10pt
  \begin{itemize}
  \item each path~$H_i$ has exactly one vertex in common with each
    cycle~$W_j$ and both endpoints of~$H_i$ are in~$V(W_1)\cup
    V(W_k)$
  \item the paths~$H_1, \dots, H_{2k}$ appear on each~$W_i$ in this
    order and
  \item for odd~$i$ the cycles~$W_1, \dots, W_k$ occur on all~$H_i$
    in this order and for even~$i$ they occur in reverse order~$W_k,
    \dots, W_1$.
  \end{itemize}
\end{definition}

See Figure~\ref{fig:grid} for an illustration of~$G_4$.

\begin{definition}\label{rem:grid-in-wall}
  Let us define an \emph{elementary cylindrical wall}~$\WWW_k$ of order~$k$ to be the
  digraph obtained from the elementary cylindrical grid~$G_k$ by replacing every
  vertex~$v$
  of degree~$4$ in~$G_k$ by two new vertices~$v_s, v_t$ connected by
  an edge~$(v_s, v_t)$ such that~$v_s$ has the same in-neighbours as
~$v$ and~$v_t$ has the same out-neighbours as~$v$.

  A \emph{cylindrical wall}~$\WWW$ of order~$k$ or \emph{$k$-wall} is a subdivision of~$W_k$.
  Clearly, every cylindrical wall of order~$k$ contains an elementary cylindrical grid of
  order~$k$ as a butterfly minor (deleting edges and contracting edges where the head has in-degree one or the tail has out-degree one). Conversely, an elementary cylindrical grid of
  order~$k$ contains an elementary cylindrical wall of order~$\frac k2$ as a
  subgraph.
\end{definition}

Again, see Figure~\ref{fig:wall} for an illustration. We may write~$\WWW=(W_{i_1},\ldots,W_{i_{t_1}},H_{j_1},\ldots,H_{j_{t_2}})$ meaning a~$t_1\times t_2$-wall specifying the respective wall-cycles and horizontal paths in the obvious way for some indices~$i_n,j_m$ with~$1 \leq n \leq t_1$ and~$1 \leq m \leq t_2$ and~$t_1,t_2 \in \N$.

\begin{theorem}[\cite{KawarabayashiK2015a}]\label{thm:dir-grid-thm}
  There is a function~$f\sth \N\rightarrow \N$  such that every digraph of directed tree-width at least~$f(k)$
  contains an elementary cylindrical grid of order~$k$ as a butterfly minor. 
\end{theorem}

The grid guaranteed in the previous theorem can be computed by a parameterized algorithm (parameterized by the grid size~$k$) as shown in \cite{CamposLMS2019}. The notion of butterfly-minors will not be of relevance in this paper, that is why we skip its definition. We will rather need the following direct consequence of the above (by skipping some of the paths).

\begin{theorem}
    There is a function~$f\sth \N\rightarrow \N$  such that every Eulerian digraph~$G+D$ of directed tree-width at least~$f(p)$ for~$p \coloneqq \Abs{E(D)}$ contains a cylindrical wall~$\WWW$ of order~$p$ with~$V(\WWW) \cap V(D)=\emptyset$ as a subgraph that can be found in~$fpt$-time on~$p$. 
\label{thm:dir_wall_away_from_D}
\end{theorem}

\begin{figure}

  \begin{center}
    \includegraphics[scale=1]{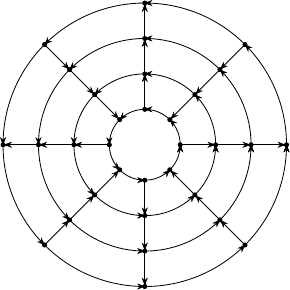}
	\includegraphics[scale=1]{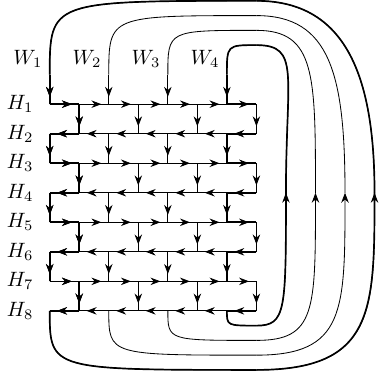}

    \addtolength{\textfloatsep}{-200pt}
    \caption{Cylindrical grid~$G_4$ and the cylindrical wall of order
~$4$. The perimeters of the wall are depicted using thick edges.}\label{fig:grid}\label{fig:wall}
  \end{center}
 
\end{figure}

Cygan, Marx, Pilipczuk, and Pilipczuk \cite{CyganMPP2013} have shown the following positive result towards finding irrelevant vertices for vertex-disjoint paths in directed graphs.
\begin{theorem}[Theorem 2.2 in \cite{CyganMPP2013}]
    For any integer~$k$, there exists~$d = d(k) = 2^{\mathcal{O}(k^2)}$ such that the following holds. Let~$G$ planar be an instance of the~$k$-vertex disjoint paths problem and let~$C_1, C_2,\ldots,C_d$ be a sequence of concentric vertex-disjoint cycles in~$G$ with alternating orientation, where~$C_1$ is the outermost cycle. Assume moreover that~$C_1$ does not enclose any terminal. Then any vertex of~$C_d$ is irrelevant.
    \label{thm:directed_irr_vertex_Marx}
\end{theorem}

\section{Eulerian Digraphs}
\label{sec:basic_results}
We continue by gathering some more definitions and results regarding Eulerian digraphs that will be needed throughout the paper.

\subsection{Basics of Eulerian digraphs}

The first operation, and the most needed operation throughout the paper is what is known as \emph{splitting off at vertices}. An operation widely known and commonly used when dealing with structural results on Eulerian digraphs \cite{Frank1988,Bang-JensenG2018}.

\begin{definition}[Splitting off at a vertex]
    Let~$G$ be an Eulerian graph and let~$u \in V(G)$ be a vertex together with two edges~$(v,u),(u,w) \in E(G)$. Then the (multi)-graph~$G' = (V(G),E')$ with~$E' = (E \setminus\{(v,u),(u,w)\}) \cup \{(v,u)\}$ is obtained from~$G$ by \emph{splitting off at~$u$ along the two-path~$(v,u,w)$} or simply by \emph{splitting off along~$(v,u),(u,w)$}, and, if needed, deleting any resulting isolated vertices.
    \label{def:splitting_off_vertex}
\end{definition}

The following observation is obvious.
\begin{observation}
    Let~$G$ be an Eulerian graph and~$v \in V(G)$ a vertex. Let~$G'$ be the graph obtained by splitting off at~$v$ along some two-path~$(u,v,w) \subset G$. Then~$G'$ is Eulerian.
    \label{obs:splitting_off_at_vertex_remains_Eulerian}
\end{observation}
For more results on splitting off at vertices and how it affects the connectivity in directed Eulerian graphs we refer the reader to \cite{Frank1988}.

As mentioned above we will be extensively working with cuts in graphs. The following is an obvious but crucial observation relying on the Eulerianness of the graph.

\begin{observation}
    Let~$G$ be Eulerian and let~$X\subseteq V(G)$ induce a cut in~$G$. Then~$\Abs{\delta^-(X)} = \Abs{\delta^+(X)}$ and in particular~$\Abs{\delta(X)}$ is even.
    \label{obs:eulerian_cut_condition}
\end{observation}

Another well-known property about cuts in directed (Eulerian) graphs is that they are \emph{submodular}---a property that is known to be wrong for vertex-separations in general directed graphs.

\begin{lemma}\label{lem:submodularity}
     Let~$A, B$ be cuts in~$G$. Then~$|\delta(A)| + |\delta(B)| \geq |\delta(A \cap B)| + |\delta(A \cup B)|$.     
\end{lemma}

The following theorem due to Johnson, Robertson, Seymour, and Thomas \cite{DirTreewidth2001} is an essential ingredient to our proof of the main \cref{thm:main}. Namely, it holds true that the directed tree-width of an Eulerian graph is lower-bounded by a linear function of the undirected tree-width (given a bounded degree). 

\begin{theorem}[Theorem 2.2 in \cite{DirTreewidth2001}]
    Let~$G+D$ be Eulerian and let~$t$ be its directed tree-width. Let~$d$ be the undirected tree-width of the underlying undirected graph. Let~$\Delta$ be the maximal out-degree of~$G+D$. Then~$d \leq (2\Delta+1)(t+1)-1$.
    \label{thm:undirected_vs_directed_tw_in_Eulerian_graphs}
\end{theorem}

We continue with a few lemmas that exploit the Eulerianness of the graph further.

\begin{lemma}
Let~$G$ be an Eulerian digraph. If there is a directed path~$P$ in~$G$ from~$u$ to~$v$ then there is also a directed path~$P'$ in~$G$ from~$v$ to~$u$ which is edge-disjoint from~$P$.
\end{lemma}
\begin{proof}
Suppose not. Let~$G' = G - E(P)$. So by assumption, there is no~$v{-}u$-path in~$G'$. By Menger, this implies that there is a partition~$(A, B)$ of~$V(G') = V(G)$ such that~$u \in A$,~$v \in B$ and all cross edges go from~$A$ to~$B$. 

Now consider the partition~$(A, B)$ in~$G$. As~$G$ has an~$u{-}v$-path, there is at least one edge from~$A$ to~$B$ and thus, as~$G$ is Eulerian, there also must be an edge from~$B$ to~$A$, a contradiction. 
\end{proof}
\begin{remark}
    This makes precise why we may restrict our attention to the underlying undirected graph when talking about connectivity of~$G$: it is not hard to see that---using standard notation from the literature---if~$G$ is an Eulerian digraph and if the underlying undirected graph of~$G$ is connected, then for evey pair of vertices~$u,v \in V(G)$ there is either a path from~$u$ to~$v$ or from~$v$ to~$u$ which together with the lemma implies that~$G$ is strongly connected. See \cite{Frank1988} for more details.
\end{remark}

Finally we get the following, which is one of the most crucial observations for our paper that is intrinsic to the Eulerianness of the graphs.

\begin{lemma}\label{lem:euler-separation}
  Let~$G$ be an Eulerian digraph and~$S_1, S_2 \subseteq V(G)$ be disjoint.
  \begin{enumerate}
  \item For every~$k$, either there is a linkage~$\LLL$ containing~$k$ paths from~$S_1$ to~$S_2$ and~$k$ paths from~$S_2$ to~$S_1$ or there is a partition~$(A, B)$ of~$V(G)$ such that~$S_1 \subseteq A$,~$S_2 \subseteq B$ and~$|\delta^+(A)| = |\delta^-(A)| < k$.
  \item The maximum number of disjoint paths from~$S_1$ to~$S_2$ is the same as the
    maximum number of disjoint paths from~$S_2$ to~$S_1$. 
  \end{enumerate}
\end{lemma}
\begin{proof}
  Follows immediately from the previous lemma and Menger's theorem.
\end{proof}
\begin{remark}
    Another way to see the above is that given any linkage~$\LLL$ from~$S_1$ to~$S_2$ we can leverage the fact that~$G$ is Eulerian and simply continue each path without using edges used in~$\LLL$ until either we are back in~$S_1$ or we close a cycle. Both will then yield the relevant collection of~$S_2$ to~$S_1$ paths edge-disjoint from~$\LLL$.
\end{remark}

The following lemma highlights the usefulness of cuts when dealing with linkages.

\begin{lemma}\label{lem:switching_linkages_at_cuts_prelims}
    Let~$G$ be a graph and let~$\LLL=\{L_1,\ldots,L_p\}$ be a~$p$-linkage in~$G$ for some~$p \in \N$ such that~$L_i$ connects~$s_i$ to~$t_i$ for every~$1\leq i \leq p$ and some~$s_i,t_i \in V(G)$. Let~$X$ induce a~$k$-cut in~$G$ for some~$k \in \N$. Let~$\LLL_X \coloneqq \{L_{1,\small \cap},\ldots,L_{\ell,\small \cap}\}$ be a~$\ell$-linkage for some~$\ell \geq 0$ where each~$L_{i,\small \cap}$ is a maximal sub-path of some~$L_j$ in~$G[[[X]]] \cap \bigcup_{i=1}^{p}L_i~$ for~$1 \leq i \leq \ell$ and~$1 \leq j \leq p$. Let~$S' \coloneqq \{s_1',\ldots,s_\ell'\}$ and~$T' \coloneqq \{t_1',\ldots,t_\ell'\}$ such that~$L_{i,\small \cap}$ connects~$s_i'$ to~$t_i'$ for~$1 \leq i \leq \ell$.
    
    Let~$\LLL^X \coloneqq \{L^{1,\tiny \cap},\ldots,L^{\ell,\tiny \cap}\}$ be an~$\ell$-linkage in~$G[[[X]]]$ such that~$L^{i,\tiny \cap}$ connects~$s_i'$ to~$t_i'$ for~$1 \leq i \leq \ell$. Finally for each~$1 \leq j \leq p$ and~$L_j \in \LLL$ let~$L^j$ be defined by replacing all possible sub-paths~$L_{i,\small \cap}$ of~$L_j$ by~$L^{i,\tiny \cap}$ for~$1 \leq i \leq \ell$. Then~$\LLL' \coloneqq \{L^1,\ldots,L^p\}$ is a~$p$-linkage in~$G$ and~$L^j$ connects~$s_j$ to~$t_j$ for every~$1 \leq j \leq p$.
\end{lemma}
\begin{proof}
   The proof is straightforward and clear for~$L_i$ and~$L^i$ agree away from~$G[[[X]]]$ and whenever~$L_i$ uses an edge in~$\delta(X)$ then~$L^i$ uses the same edge by definition, where~$L_{j,\cap}$ and~$L^{j,\cap}$ are paths connecting the same cut-edges in~$\delta(X)$ for~$1 \leq i \leq p$ and~$1\leq j \leq \ell$. (See also \cref{lem:switching_linkages_at_cuts} for another version of this lemma).
\end{proof}

\subsection{Reducing the instance}
\label{sec:Reducing_the_instance}
In this subsection we start with massaging the instance~$G+D$ into a form suited for the use of \cref{thm:undirected_vs_directed_tw_in_Eulerian_graphs} as well as for the use of the results in \cite{Frank1988} and \cite{FrankIN1995} that we introduce in \cref{subsec:flattening_a_flat_swirl,sec:Routing} respectively. That is, we prove that we can assume the graph~$G+D$ to be of degree four away from the terminals~$V(D)$ while all the terminals are distinct and of degree two in~$G+D$.
 
 \begin{lemma}\label{lem:reduce-to-degree-4}
     Let~$G + D$ be Eulerian where~$D$ is the demand graph with~$\Abs{E(D)} = p \in \N$. 
     Then~$G+D$ can be reduced in polynomial time to a new instance~$G', D'$ such that~$G'+D'$ is Eulerian and in  which each non-terminal vertex has degree~$4$, all terminals have degree~$2$, no terminal vertex is part of two edges in~$E(D)$ and~$G+D$ is a \emph{YES}-instance if, and only if,~$G'+D'$ is a \emph{YES}-instance.
 \end{lemma} 
 \begin{proof}
     Suppose~$v$ is a non-terminal vertex of degree at least~$6$ (if it is of degree~$2$ we can dissolve it, and if it is of degree~$0$ we can delete it). 

     Let~$s_1, \dots, s_k$ be the in-neighbours and~$t_1, \dots, t_k$ be the out-neighbours of~$v$.
    For~$k\geq 1$ let~$W_k$ be the graph consisting of~$k$ disjoint cycles~$C_1, \dots, C_k$, where~$C_i = (v^i_1, v^i_2, \dots, v^i_k, v^i_1)$, and the edges~$(v^i_j, v^{i+1}_j)$, for all~$1 \leq i < k$ and~$1 \leq j \leq k$.
    We replace~$v$ by a copy of~$W_k$ and replace each edge~$(s_i, v)$ by the edge~$(s_i, v^1_i)$ and each edge~$(v, t_i)$ by~$v^k_i, t_i)$. 

    Then for any matching~$M = (s_1, t_{i_1}), \dots, (s_k, t_{i_k})$ of the in-neighbours of~$v$ to its out-neighbours there is a set~$
    P_M$ of~$k$ pairwise edge-disjoint paths in~$W_k$ such that~$P_M$ contains a path starting in~$s_j$ and ending in~$t_{i_j}$, for all~$1 \leq j \leq k$. This shows that by replacing~$v$ by this gadget the connectivity remains the same.

     Finally, once all vertices in the gadget have degree at most~$4$ (i.e. in-degree at most~$2$), we eliminate all vertices of degree~$2$ by splitting off their incident edges (or dissolving the vertices). Further we may assume that each terminal has degree~$2$ in~$G+D$. For, say, a vertex of a terminal pair~$(s, t)$ has higher degree. The we add the edge~$(t, s)$ (which is in~$D$) to~$G$, split it into a path~$(t, t', s', s)$, remove the edge~$(t', s')$ and make~$(s', t')$ the new terminal pair. Thus the edge~$(t', s')$ appears in~$D$ and the resulting pair~$G + D$ is again Eulerian. Finally if a vertex~$s \in V(D)$ is adjacent to two edges in~$E(D)$, then simply subdivide one edge to get another terminal pair. 
 \end{proof}

Throughout the paper we will assume our graphs---and hence instances of the Eulerian edge-disjoint paths problem---to be as given by \cref{lem:reduce-to-degree-4} unless stated otherwise. Armed with these notions we are ready to start our journey proving \cref{thm:main}. The first result in this direction is the proof of the \emph{Flat-Swirl Theorem} as discussed next.

\section{Finding Routers and flat Swirls}
\label{sec:Swirls}
\label{sec:flat_wall}
In this section we prove a structural result in the spirit of the
well-known \emph{Flat-Wall Theorem}, one of the many results due to
Robertson and Seymour \cite{GMXIII} established in their proof of
the \emph{Graph Minors Structure Theorem} \cite{GMXVI}. Since their
well-known graph minor series, there have been many more results of
the same flavour pursuing similar ideas. For example Kawarabayashi, Thomas, and Wollan have provided a
concise proof of the Flat-Wall theorem \cite{KawarabayashiTW2018}. More recently,
Giannopoulou, Kawarabayashi, Kreutzer, and Kwon proved a Flat-Wall
Theorem for general directed graphs
\cite{GiannopoulouKK2020}. The theorem establishes a flat cylindrical wall in the absence of a bidirected clique as butterfly minor. Both outcomes are not strong enough for our purpose. Instead of a flat cylindrical wall we need a flat wall with cycles oriented in alternating orientations. And instead of a bidirected clique as butterfly minor we need pairwise touching edge-disjoint cycles. We therefore establish a flat-wall theorem that is more suitable to our purpose. 
Broadly speaking we prove that, given a large cylindrical wall in
an Eulerian graph, we either find a large bi-directed grid-minor
in~$G + D$---more precisely a large \emph{swirl}---or many disjoint \emph{usable crosses} that will yield a structure we call
\emph{a router}: a collection of edge-disjoint cycles that pairwise
intersect in vertices. The undirected analogue of routers for the
undirected Flat-Wall theorem are cliques. Similar in flavour to the
proof of the undirected Flat-Wall theorem
\cite{GMXIII,KawarabayashiTW2018}, \emph{routers} can in our setting
be used to reduce the instance~$G+D$ to an equivalent but simpler and
smaller one, making progress in the search for~$p$ edge-disjoint
paths. That is, given a large enough router we may either solve the
instance in polynomial time or find an \emph{irrelevant cycle} in the
router in~$fpt$-time, i.e., a cycle whose deletion yields an
equivalent instance of the edge-disjoint paths problem (this will be
proven as \cref{thm:irrelevant_cycle_in_router_general} in
\cref{sec:Routing}). We then continue (see
\cref{sec:charting_eulerian_digraphs,sec:structure_of_min_examples,sec:shippings})
with a more in-depth analysis of the case that we do \emph{not} have
any large routers in the graph. In that case, and given high
tree-width with respect to~$p$, the \emph{Flat-Swirl
  \cref{thm:flat_swirl_theorem}} yields that we may find and
(Euler-)embed a large \emph{flat} swirl: a collection of embedded
concentric edge-disjoint cycles such that two consecutive cycles have
opposing orientation in the plane. Finally, given the embedded swirl
we show that we can again find an irrelevant cycle to the instance
shrinking the instance further until the graph becomes of low
tree-width; this will be proven as \cref{thm:irrelevant_cycle} in
\cref{sec:irrelevant_cycle_theorem}. We emphasise that during our work
on this problem, and after we already proved the existence of a flat
swirl by excluding routers, we came across the dissertation of Johnson
\cite{Johnson2002} who himself already proved a structure theorem for
internally~$6$ edge-connected Eulerian graphs. In particular he
implicitly proved a version of what we call the Flat-Swirl theorem.
Unfortunately the results have never been published. It is note-worthy
that our approach to finding said swirl is completely different from his
and of its own interest, where Johnson pursued the ideas pioneered by Robertson and Seymour, using tangles in order to find highly connected pieces in the graph that he then proves contain a swirl as an immersion, while we start from a cylindrical wall guaranteeing us some structure where we then find the swirl `by hand' via analysing jumps and attachments to the swirl. This has the advantage that the techniques we introduce to find a swirl in the graph can be reused and generalised to find a flat swirl in the graph, which is very convenient.

\smallskip

The section is subdivided into five parts. In the first subsection we
introduce the required notions and state the main result of this
section, namely the \emph{Flat-Swirl Theorem}. In
\cref{subsec:finding-swirls-by-excluding-routers} we first prove that high tree-width implies the existence of a router or a swirl. The swirl we find may still be partly
\emph{tangled}, i.e., it could still contain a router as a subgraph.
In the rather short \cref{subsec:untangling_a_swirl} we untangle the
swirl, before we flatten it in \cref{subsec:taming_a_swirl}. The flat
swirl may however still have highly non-planar (or non
Euler-embeddable) subgraphs that are loosely attached to it (removable by
small cuts). In \cref{subsec:flattening_a_flat_swirl} we prove that
there exists an equivalent instance to~$G+D$ that contains a large
flat swirl that can be Euler-embedded together with all its
attachments---the components resulting after deleting the outer-cycle of
the swirl---respectively, getting rid of the just described nuisances.

\subsection{Swirls, routers, and crosses}
\label{subsec:swirl_structure}
The structure we are trying to exclude in this section is what we will
call a \emph{Router}; we commonly denote routers by~$\RRR$. Although
large bi-directed cliques seem to be the obvious candidates for
structures that help with routing paths disjointly, the technique we
use to prove the usefulness of routers relies on the fact that the
graph remains Eulerian after the deletion of~$\mathcal{R}$. This is why general
bi-directed clique minors may not always be of help, for it is not
immediately clear whether we always find an Eulerian model of a
bi-directed clique in an Eulerian graph even if a bi-directed clique
minor exists. Furthermore, since we are only interested in routing
\emph{edge-disjoint} paths, the standard butterfly minor model seems
unnecessarily strong; we are rather looking for immersions and
derivatives thereof. Note, however, that the graphs we will look at
are (almost)~$4$-regular and thus there cannot be a $K_5$-immersion in
our graphs. This is why we will use a different kind of routing
device: \emph{pairwise crossing cycles.}

\smallskip
Throughout this sub-section we assume an Eulerian digraph~$G+D$ to be given, where~$D$ is a demand-graph and every vertex in~$V(G) \setminus V(D)$ is of degree four, and every vertex in~$V(D)$ is of degree two in~$G+D$, unless stated otherwise (see \cref{sec:Reducing_the_instance}). Further let~$p \coloneqq \Abs{E(D)}$, i.e.,~$G+D$ encodes an instance of the~$p$-edge-disjoint paths problem. Recall the notions of paths and concentric cycles introduced in \cref{sec:prelims}. We start by defining the objects of interest for this section, \emph{Routers} and \emph{Swirls}.

\begin{definition}[Router]
    Let~$G$ be a digraph and~$C_1,\ldots,C_r$ be mutually edge-disjoint directed cycles that pairwise intersect in a vertex of~$G$. Then we call $$\RRR\coloneqq C_1\cup\ldots \cup C_r$$ a \emph{Router of order $r$}, or simply an \emph{$r$-Router}.

    Let~$b_1 \in C_1,\ldots,b_r \in C_r$ be a choice of distinct vertices in~$\RRR$ called \emph{branching vertices}. We define the \emph{branching set} $B_\RRR \coloneqq \{b_1,\ldots,b_r\}$. 
    \label{def:router}
\end{definition}
\begin{remark}
    Note that since our graph has maximum-degree four there cannot be three disjoint cycles intersecting in the same vertex.
\end{remark}
We continue with the relevant definitions to state the main theorem of this section.

\paragraph{Plane wall. } Let~$t \in \N$ with~$t>2$ (usually large with
respect to~$p$) and let $\WWW$ be a plane (cylindrical)
$t\times t$-wall, with a designated outer-face~$O$. We refer to
the~$t$ disjoint concentric wall-cycles---these are vertex-disjoint cycles---as $W_1,\ldots, W_t$, where the
component of the plane bounded by $W_j$ containing $O$ contains every
$W_i$ for $1\leq i<j$. We call that component (or \emph{side})
\emph{outside of $W_j$} and the other side \emph{inside of $W_j$}. The
face bounded by~$W_t$ that does not contain any $W_i$ is referred to
as $I$---the \emph{inner-face}. Note that the face bounded by $W_1$ that does not contain any
other $W_i$ is exactly $O$. We further distinguish the~$2t$ disjoint
\emph{horizontal paths}~$H_1,\ldots,H_{2t}$---paths between~$W_1$
and~$W_t$---of~$\WWW$ into~{$O$-to-$I$} and~{$I$-to-$O$} paths. That is,
$H_i$ is an~{$O$-to-$I$} path if $(i \mod 2) = 1$, and an~{$I$-to-$O$}
path otherwise, where~$O$-to-$I$ paths start in vertices of~$W_1$ and
end in vertices of~$W_t$, and~$I$-to-$O$-paths start in vertices
of~$W_t$ and end in vertices of~$W_1$. Finally denote the
\emph{wall-coordinates}, or simply \emph{coordinate vertices} or
\emph{coordinates}, by $x^-_{i,p}$ and $x^+_{i,p}$,
where~$x^-_{i,p},x^+_{i,p} \in W_i \cap H_p$ and~$x^-_{i,p}$ is the unique
vertex with~$\delta^{-}_{\WWW}(x_{i,p}^{-}) = 2$ and $x^+_{i,p}$ is the unique
vertex with~$\delta^{+}_{\WWW}(x_{i,p}^+) = 2$; we write~$\operatorname{Coord}(\WWW)$ for the set of coordinates. We refer to coordinates $x_{i,p}^-$ as \emph{in-coordinates} and to $x^+_{i,p}$ as \emph{out-coordinates}. After possibly skipping some
horizontal paths, we may further assume
that~$x^-_{i,p} \neq x^+_{i,p}$. Note that, after the choice of an
outer-face~$O$ satisfying the above, the above is well-defined for our
wall has a unique embedding on the sphere by Whitney's Theorem
\cite{Whitney32}.

We may ambiguously write~$x_{i,p}$ if it is irrelevant for the context or clear from the context which of both vertices~$x_{i,p}^-,x_{i,p}^+$ is meant. Also note that we will treat the coordinates as cyclic-coordinates, thus~$x_{i,1} = x_{i,2t+1}$ and so on. 

In what is to follow we will tacitly assume that the above setting is given with the respective notations unless stated otherwise; see \cref{fig:wall-coordinates}. 
\begin{figure}
    \centering
    \includegraphics{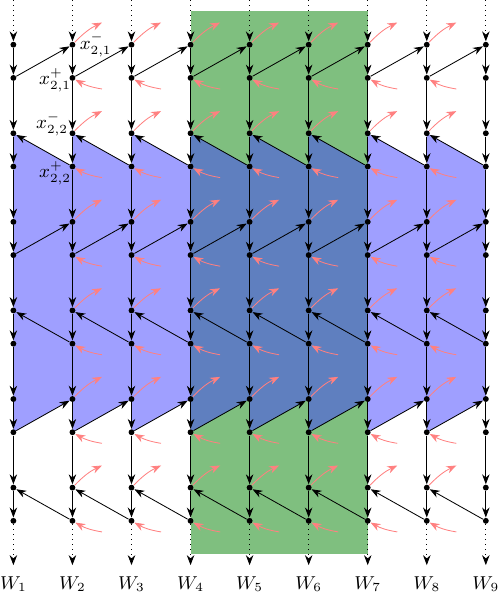}
    \caption{Wall with wall-coordinates. The paths $\PPP$ forming a
      matching are marked in red. The green area is a subwall, the
      blue area forms a band, and the intersection of both areas,
      marked in dark blue,  is a tile.} 
    \label{fig:wall-coordinates}
\end{figure}
\begin{definition}[Tiles, bands, and subwalls]
Let $1\leq i <  j \leq t$, and $1 \leq p < q \leq 2t$. We define the following:
\begin{itemize}
 \item[(i)] The \emph{$(i,j)$-subwall}~$W_{[i,j]}$ is defined as the induced subgraph of $\WWW$ bounded by $W_i$ and $W_j$, i.e., the induced graph resulting from $\WWW$ after deleting the vertices lying outside of $W_i$ and inside of~$W_j$. 
 \item[(ii)] The \emph{$(p,q)$-band}~$W^{[p,q]}$ is defined as the induced subgraph of $\WWW$ bounded by $H_p$ and $H_q$, i.e., the graph resulting from $\WWW$ after deleting the outside---the face containing~$O$---of the undirected cycle formed by~$H_p \cup x_{1,p}W_1x_{1,q} \cup H_q \cup x_{t,q} W_{t} x_{t,p}$. 
 
 \item[(iii)] The \emph{$((i,j),(p,q))$-tile}~$T_{[i,j]}^{[p,q]}$ is
   defined as the induced subgraph of $\WWW$ bounded by $W_i$, $W_j$,
   $H_p$ and~$H_q$, i.e., the induced graph resulting from the
   intersection of the~$(i,j)$-subwall and the~$(p,q)$-band. Delete
   the degree-one vertices in the tile (which arise at its corners,
   i.e.~for $H_\alpha \cap W_\beta$ for~$\alpha \in \{p,q\}$ and~$\beta
   \in \{i,j\}$ and refer to the remaining~$x_{i,p},x_{i,q},x_{j,p}$
   and~$x_{j,q}$ as the corners of the tile. 
\end{itemize}
\end{definition}

\begin{remark} 
Depending on the parity of~$p,q$ the corners are~$x_{i,p}^+$
or~$x_{i,p}^-$ and the same for the other corner vertices. For
simplicity we did not state all the cases as it is obvious which
coordinates would be the corners in which case; see~\cref{fig:wall-coordinates}.

\end{remark}

In what is to follow, we will only consider tiles and bands that are bounded by horizontal paths of opposite type, meaning that~$p \mod 2 \not\equiv q \mod 2$. This allows for more consistent notation with respect to~$t \times t$-walls having~$2t$ horizontal paths. Thus we fix the following notations.

\begin{definition}
    Let~$\WWW$ be a~$t\times t$-wall for some~$k \in \N$. Then we say
    that~$T \subset \WWW$ is a~$k \times s$-tile if~$T =
    T_{[i,j]}^{[p,q]}$ for~$\Abs{i-j}+1 = k$ and~$\Abs{p - q} +1 = 2s$
    where~$1 \leq i,j \leq t$ for~$1 \leq p,q \leq 2t$ and~$ p \mod 2
    \neq q \mod 2$. In particular a~$k \times s$-tile contains as
    many~$O$-to-$I$ paths as~$I$-to-$O$ paths. 

    Finally if~$k = s$ we say that~$T$ is a~$s$-tile.
\end{definition}

Note that, although when talking about coordinates of the wall we always talk about cyclic coordinates implicitly, when focusing on some fixed tile~$T_{[i,j]}^{[p,q]}$ the coordinates are \emph{not} cyclic anymore, and we may always assume that~$i<j$ and~$p<q$ after possibly relabelling the coordinates of the wall by rotations or shifts.

Given the above we refine our notion of routers to \emph{routers grasped by a wall}.
\begin{definition}[Router grasped by a wall]
  Let~$\WWW$ be a wall and let~$\RRR = C_1 \cup \ldots \cup C_r$ be
  an~$r$-router with some branching set~$B_\RRR$. We say that~$\RRR$ is \emph{grasped by the
    wall~$\WWW$} if~$B_\RRR\subset V(\WWW)$ are coordinate vertices of
  the wall and there
  exist~$c_{i,j} \in \Big(V(C_i) \cap V(C_j) \cap V(\WWW)\Big)$ for
  every~$i \neq j$,~$1\leq i,j \leq r$. We call~$B_\RRR$ a \emph{grasping branching-set}.
\end{definition}

\begin{figure}
\begin{subfigure}{.4\textwidth}
  \centering
  \includegraphics[scale=0.75]{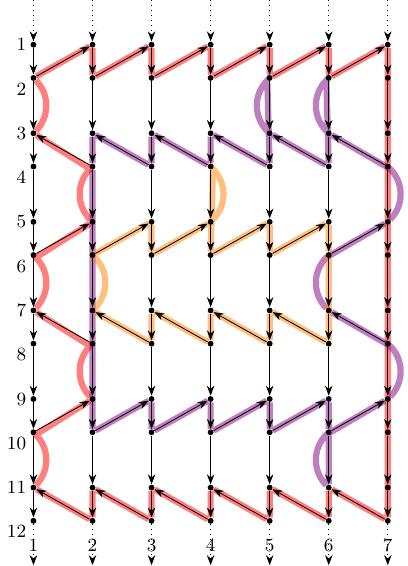}
  \caption{A~$3$-swirl in a wall.}
  \label{fig:swirl}
\end{subfigure}
\begin{subfigure}{.6\textwidth}
  \centering
  
  \includegraphics[scale=.9]{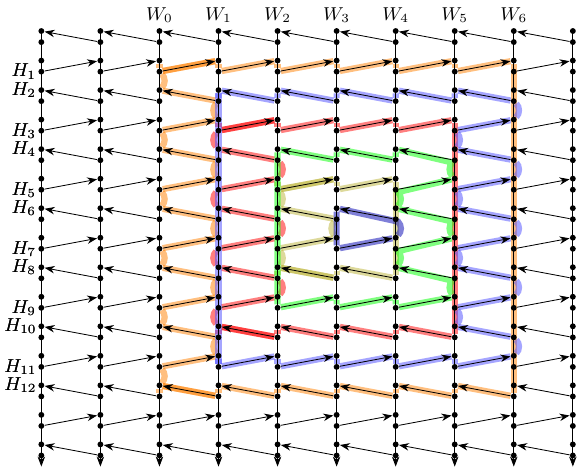}
  \caption{A~$6$-swirl grasped by a $6$-tile}
  \label{fig:grasped_swirl}
\end{subfigure}
\caption{The left side is a figure of a~$3$-swirl in the wall. The right figure represents a swirl grasped by a tile. }
\label{fig:swirl_and_grasped_swirl}
\end{figure}

Finally we make explicit what we mean by a \emph{Swirl}.
\begin{definition}[Swirls, grasped swirls and induced swirls]
    Let~$s \in \N$ with~$s \geq 2$. An~\emph{$s$-swirl} is a plane graph~$\mathcal{S} = S_1 \cup \ldots \cup S_s$ formed by~$s$ edge-disjoint concentric cycles~$S_i$ where~$S_1$ bounds the inner-face of the plane and~$S_s$ bounds the outer-face of the plane, such that
    \begin{itemize}
        \item[1.] $S_i$ and~$S_{i+1}$ have opposite orientation for every~$1\leq i < s$, and
        \item[2.]~$S_i \cap S_j = \emptyset$ if~$j \notin \{i-1,i+1\}$ for~$1<i<s$.
    \end{itemize}
    We call~$S_s$ the \emph{outer-cycle} of the swirl, and~$S_1$ the \emph{inner-cycle} of the swirl.

    Let $\WWW$ be a wall.
    We say that~$\mathcal{S}$ \emph{is grasped by an~$s'$-tile~$T \subset \WWW$ of the wall} for some~$s' \leq s$ if~$V(T) \cap \operatorname{Coord}(\WWW) \subset V(\mathcal{S})$. Similarly we say that~$\SSS$ is \emph{grasped by a wall~$\WWW$} if there exists some~$s$-tile~$T$ in~$\WWW$ grasping it.

    We say that~$\mathcal{S}$ \emph{is induced by an~$s'$-tile~$T \subset \WWW$ of the wall} for some~$s' \leq s$ if~$T\subseteq \mathcal{S}\cap \WWW$. We say that~$\SSS$ is \emph{induced by~$\WWW$} if there exists an~$s$-tile inducing it.
     
\label{def:swirl}
\end{definition}

\begin{remark}
    The notion of orientation in the definition of swirl is well-defined, since the graph is assumed to be plane and there is a canonical orientation of the plane, disc, or sphere that contains the swirl as mentioned in \cref{subsec:notation}.

    Further note that one may think of swirls as concentric cycles where~$S_i \cap S_j \neq \emptyset$ for~$\Abs{i-j} = 1$ where~$i,j \in \{1,\ldots,s\}$; we did not impose this since the results hold in this more general setting, where said structure is guaranteed when talking about induced swirls.
\end{remark}
See \cref{fig:swirl_and_grasped_swirl} for an example of (grasped) swirls.

\paragraph{Usable Crosses.} Let~$G$ be a digraph and~$C=(s_1,s_2,t_1,t_2)$ be a cyclic order on four vertices,~$s_i,t_i \in V(G)$ for~$i=1,2$. If there exist edge-disjoint paths~$P_1,P_2$ in~$G$ connecting~$s_1$ to~$t_1$ and~$s_2$ to~$t_2$, then we say that~\emph{$G$ has a~$C$-cross}.
In contrast to the vertex-disjoint case, even a plane graph may have a~$C$-cross for some embedded cycle~$C$ by making use of strongly planar vertices. However, not every strongly planar vertex in a plane graph induces a~$C$-cross,  see \cref{fig:not_usable_cross} for an exemplary representation. This leads to the definition of \emph{usable crosses}; for convenience we start by defining the \emph{vertex-pattern of a path}.

\begin{definition}[Vertex-pattern of a path]
  Let~$G$ be a graph and let~$P \subseteq G$ be a path starting
  in~$u\in V(G)$ and ending in~$v \in V(G)$. Then we
  define~$\operatorname{v-}\pi(P) = (u,v)$ and refer to it as the \emph{vertex-pattern of~$P$.}
\end{definition}

Using the vertex-pattern of paths we define what we mean by local usable crosses---there is two different usable crosses \emph{wall-local} and \emph{swirl-local usable crosses}.

\begin{figure}
\centering
\begin{subfigure}{.45\textwidth}
 \centering
\includegraphics[width=.7\linewidth,height=.7\linewidth]{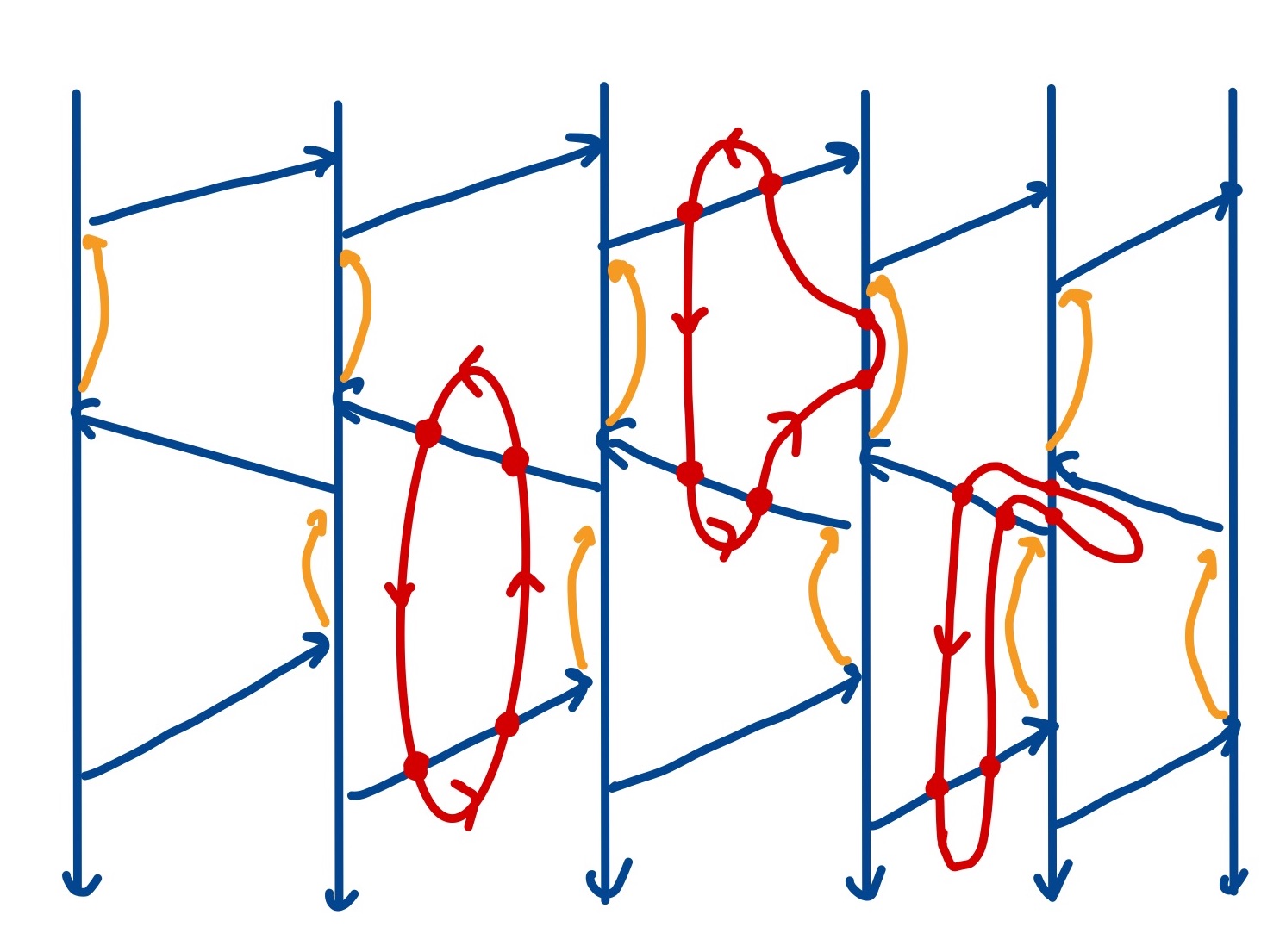}
\caption{Strongly planar vertices.}
\end{subfigure}
\begin{subfigure}{.45\textwidth}
  \centering
\includegraphics[width=.7\linewidth,height=.7\linewidth]{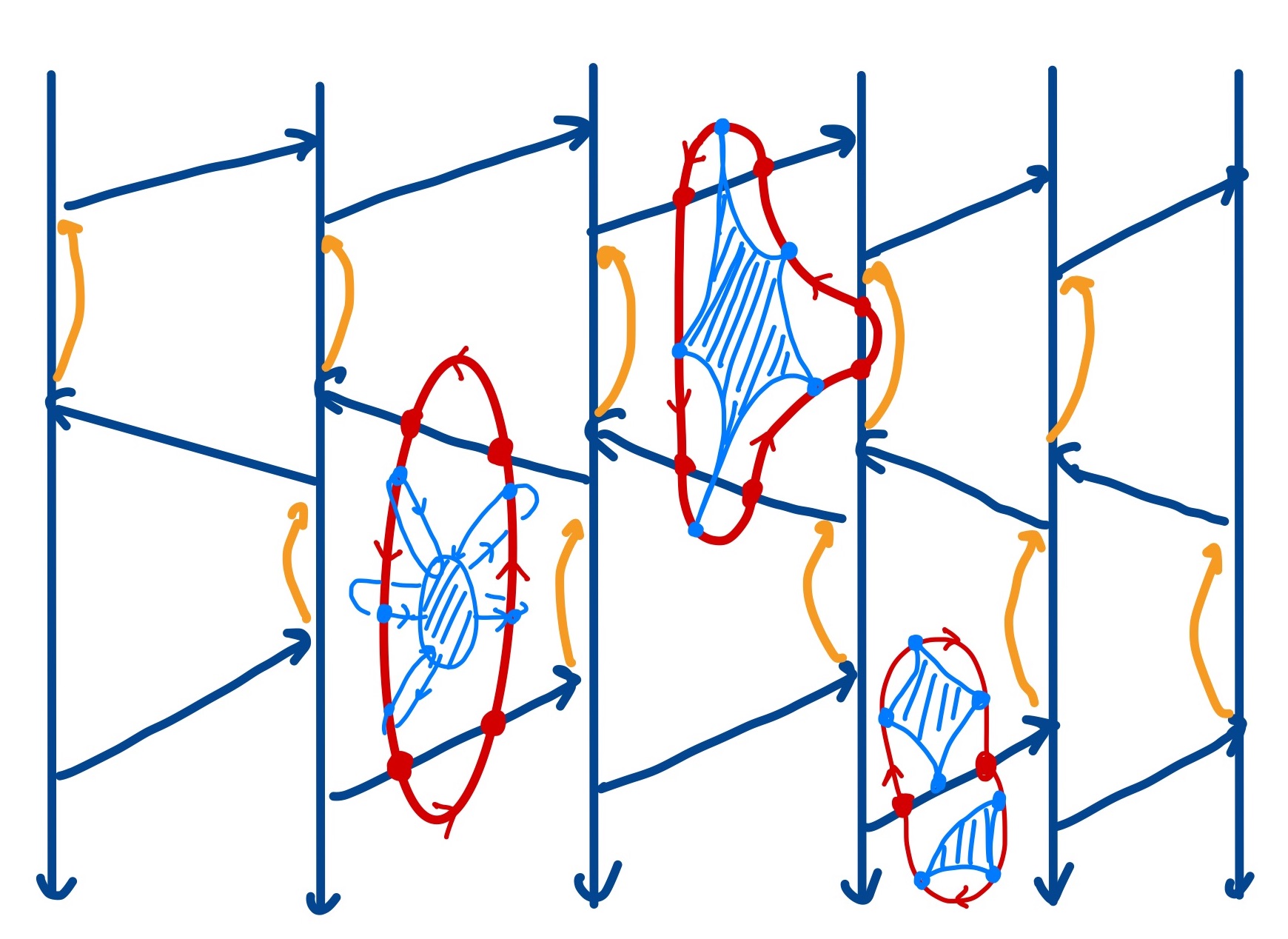}
  \caption{Possible non-planar attachments.}
\end{subfigure}
\caption{An example of three attachments to a tile containing strongly planar vertices and possible non-planarities, neither of which induce a usable cross in the tile.}
\label{fig:not_usable_cross}
\end{figure}

\begin{definition}[Wall-local usable cross]    \label{def:usable_cross}
    Let~$\WWW$ be a plane~$t \times t$-wall for~$t\geq 4$. 
    Let~$T \coloneqq \WWW_{[i,j]}^{[p,q]}$ be some tile in the wall
      for~$1<i \leq j<t$ and~$1\leq p \leq q \leq 2t$.
      Let~$s_1 = x_{i,p}$,~$s_2 = x_{j,p}$ as well as~$t_1 = x_{j,q}$
      and~$t_2 = x_{i,q}$ be the corners of the tile. Then we say that
      there is a \emph{wall-local usable~$T$-cross} if there exist
      edge-disjoint paths~$P_1, P_2 \subseteq G-(\WWW-T)$ such
      that~$\operatorname{v-}\pi(P_1) = (s_1,t_1)$
      and~$\operatorname{v-}\pi(P_2) = (s_2,t_2)$.

     We say that two paths~$P_1,P_2$ form a \emph{wall-usable cross} if they form a wall-local
      usable~$T$-cross with respect to some tile~$T$ in a plane wall~$\WWW$.
\end{definition}
When talking about crosses in a wall, they are only readily usable if they go \emph{with} the direction of the wall, which is why they are defined to depend on the coordinates and specific corners of the tile. A swirl on the other hand has no apparent `flow direction', it behaves similar to a bi-directed grid. This is why in this case any kind of cross on a swirl turns out to be a usable cross .

Prior to defining swirl-local usable crosses, we define what we mean by the \emph{attachment-extension to a swirl}.

\begin{definition}[The attachment-extension]\label{def:attachment_extension_of_swirl}
    Let~$G$ be Eulerian and let~$\WWW$ be a plane~$t\times t$-wall in~$G$. Let~$\mathcal{S}=S_1 \cup \ldots \cup S_s$ be an~$s$-swirl induced by some~$s$-tile~$T\subset\WWW$ for some~$t,s\in \N$ where~$S_s$ denotes the outer-cycle of~$\SSS$. Let~$\SSS^\star[G]$ denote the connected component of (the undirected underlying graph of)~$G-S_s$ that contains~$\SSS - S_s$. Then we call~$\SSS[G] \coloneqq \SSS^\star[G] \cup S_s$ the \emph{attachment-extension to~$\SSS$}. 
\end{definition}
\begin{remark}
 As an intuition to the definition note that if the graph~$G$ is Euler-embedded in some disc, then~$E(S_s)$ is an edge-cut separating~$\SSS-S_s$ from the rest of the wall for every face in a Euler-embedded graph in the plane is homeomorphic to a disc. Thus we want to quantify the rest of~$G$ that is attached to~$\SSS$ after deletion of its outer-cycle, which in the worst case could be all of~$G-S_s$ of course.
\end{remark}

Given~$\SSS[G]$ we are ready to define what we mean by swirl-local usable cross.

\begin{definition}[Swirl-local usable cross]\label{def:swirl_usable_cross}
Let~$G$ be Eulerian and let~$\WWW\subset G$ be a plane~$t \times t$-wall for~$t\geq 4$. Let~$\mathcal{S}=S_1 \cup \ldots \cup S_s$ be an~$s$-swirl induced by some~$s$-tile~$T \coloneqq \WWW_{[i,j]}^{[p,q]}$ with corners~$s_1,s_2,t_1,t_2$ appearing as listed in clock-wise order given the embedding. Then there is a \emph{swirl-local usable~$T$-cross} or simply \emph{swirl-local usable cross} if there exist two edge-disjoint paths~$P_1, P_2 \subseteq \SSS[G]$ such that~$\operatorname{v-}\pi(P_1) \in \{(s_1,t_1),(t_1,s_1)\}$ and~$\operatorname{v-}\pi(P_2) \in \{(s_2,t_2),(t_2,s_2)\}$.

 We say that two paths~$P_1,P_2$ form a \emph{swirl-usable cross} respectively if they form a swirl-local usable cross with respect to some induced swirl~$\SSS$ in a plane wall~$\WWW$.
\end{definition}
\begin{remark}
    The definition of usable crosses depends on an embedding of the wall in the plane, in particular we induce an orientation on the outer-cycle of the wall and every tile in it from said plane; this is crucial.

    Note further that for swirl-local usable crosses we allow the crossing paths to \emph{use} most of the wall and swirl. The idea is that we want the outer-cycle~$S_s$ of the swirl to separate the `inside' of the swirl from the outside in a sense that edge-disjoint paths need to use edges of~$S_s$ to get from one to the other. So the cross is really local with respect to~$\SSS$ where the rest of the wall~$\WWW$ can be seen as an attachment to it in a sense.
\end{remark}

We will simply talk about (local) usable crosses, when it is clear from the context if we are investigating crosses with respect to a wall or a swirl; by a rule of thumb wall-local crosses are investigated in \Cref{subsec:finding_the_swirl,subsec:untangling_a_swirl} and swirl-local crosses are investigated in \Cref{subsec:taming_a_swirl,subsec:flattening_a_flat_swirl}.

We are ready to define the notion of \emph{flatness} we are interested in. Recall \cref{def:G_with_cut_edges_full_prelims}.
\begin{definition}[Flat Swirl]
    Let~$G$ be an Eulerian graph. Let~$\WWW$ be a plane $t\times t$-wall in~$G$ for some~$t\geq 4$ and let~$\mathcal{S}$ be an~$s$-swirl induced by some~$s$-tile~$T$ for some~$s \leq t$. We say that~$\SSS$ is \emph{flat} if it does not admit any swirl-local usable~$T$-cross.
    
    \label{def:flat_swirl}
\end{definition}
\begin{remark}
  We emphasise that our notion of flatness does not guarantee any sort of planarity in the embedding besides the embedding of the swirl, for there may be highly (locally) non-planar parts that do not help in the sense of a usable cross as seen in \cref{fig:not_usable_cross}. We will however be able to get rid of these parts in \cref{subsec:flattening_a_flat_swirl} using a \cref{thm:two_paths_Frank} by Frank, Ibaraki, and Nagamochi \cite{FrankIN1995}.
\end{remark}

It is note-worthy that the swirl itself is Eulerian, which leads to~$\SSS[G]$ and~$G - \SSS[G]$ being Eulerian.

\begin{observation}\label{obs:both_sides_of_flat_are_eulerian}
    Let~$G$ be Eulerian and let~$\WWW\subset G$ be plane wall. Let~$\mathcal{S}=S_1\cup\ldots\cup S_s$ be a flat~$s$-swirl grasped by some~$s$-tile~$\WWW_{[i,j]}^{[p,q]}$ in~$\WWW$ for some~$s \in \N$. Then~$\SSS[G]$ as well as~$G-\SSS[G]$ are Eulerian.
\end{observation}

The main theorem of this section now reads as follows
\begin{theorem}[Flat-swirl Theorem]\label{thm:flat_swirl_theorem}
    
  For every pair of integers $t_1,t_2 \in \N$ there exists
  a computable function~$f:\N\times \N \to \N$ such that the following holds. Let
  $G+D$ be an Eulerian digraph of maximum degree~$4$ and
  let~$\WWW \subset G$ with~$V(\WWW) \cap V(D) = \emptyset$ be
  an~$f(t_1,t_2)\times f(t_1,t_2)$-wall. Then, either
    \begin{itemize}
        \item[(i)] $G$ contains a flat~$t_1$-swirl~$\mathcal{S}$ induced by a~$t_1$-tile~$T\subset \WWW$ that is edge-disjoint from~$D$, or
        \item[(ii)] $G$ contains a~$t_2$-router grasped by~$\WWW$ that is edge-disjoint from~$D$.
    \end{itemize}
    Moreover we can decide in~$fpt$-time on~$t\coloneqq \max(t_1,t_2)$ whether~$(i)$ or~$(ii)$ hold and output the relevant structure.
\end{theorem}
\begin{remark}
    For simplicity most of the proofs will assume~$t_1 = t_2$ since we can always set~$t \coloneqq \max(t_1,t_2)$.
\end{remark}

As mentioned above we prove the main theorem in several steps, the first of which deals with finding an immersion of a large swirl in a wall given the absence of large routers.

\subsection{Finding swirls by excluding routers}
\label{subsec:finding-swirls-by-excluding-routers}
As a first step towards \cref{thm:flat_swirl_theorem} we first prove a slightly weaker form of the theorem in this section. We show that given a large cylindrical wall~$\WWW$ in~$G$ we either find a large router or a large swirl induced by~$\WWW$. The proof will be done in two steps which we separate into two subsections. The first subsection deals with the case that our graph~$G$ is given by an elementary wall~$\WWW'$ together with a perfect matching $M$ on the degree three vertices~$\operatorname{Coord}(\WWW')$. The second rather short section deals with the general case given a cylindrical wall~$\WWW$ in~$G$, where the results follows easily from the ones provided in the first subsection. 

\smallskip

Throughout this section whenever we talk about (local) usable crosses, we mean wall-local usable crosses (recall \cref{def:usable_cross}), in particular we assume a \emph{plane} wall together with the respective labels introduced in \cref{subsec:swirl_structure} to be given. For what is to follow we will need what we will refer to as~$\WWW$-paths; a general definition reads as follows.

\begin{definition}
    Let~$G$ be a graph and let~$W \subseteq G$ be a sub-graph. A~\emph{$W$-path} is a path that has both its endpoints in~$V(W)$ and is otherwise edge-disjoint from~$W$.
    \label{def:W-paths}
\end{definition} 

Since we are interested in finding swirls and routers grasped by a
wall, we will, in a first instance, get rid of the seemingly
unnecessary attachments to the wall by splitting off vertices as
defined in \cref{def:splitting_off_vertex} until we have reduced our graph to
the core that is interesting to us; this will be an elementary wall together with a perfect matching on the degree-three vertices, whence the two cases. It is crucial that we do not
create any new connectivity by doing so and we need our results to be readily retrieved in the general graph. Since we are looking for large routers and swirls, we may equally well ignore the demand graph, that is we may
treat it as part of~$G$ and just talk about~$G$ being Eulerian. If a router or swirl then uses part of the
demand edges~$E(D)$ we can simply omit the at most~$p$ relevant cycles using
any of its edges and still get a large router or a large swirl in the
graph away from~$D$. We will make all of this more precise in what is
to follow.

\paragraph{Preparing the wall. } Let~$G+D$ be Eulerian and let~$\WWW$ be a ~$t\times t$-wall in~$G+D$,  for some~$t \geq 2$, such that~$V(D) \cap V(\WWW) = \emptyset$.  We continue by defining \emph{complete sets of coordinate-paths}.

\begin{definition}[coordinate-paths] Let $\WWW$ be a wall in an Eulerian graph $G$. A \emph{coordinate-path} (for $\WWW$) is a~$\WWW$-path $P$ that starts and ends at a coordinate vertex of $\WWW$.

    A \emph{complete set of coordinate-paths for $\WWW$} is a set $\PPP$ of pairwise edge-disjoint coordinate-paths for $\WWW$ such that every coordinate vertex of $\WWW$ is the endpoint of exactly one path in $\PPP$.
\label{def:coordinate_paths}
\end{definition}

\begin{lemma}\label{lem:compute_coord_paths}
  Let $\WWW$ be a wall in $G$. Then $G$ contains a complete set of
  coordinate-paths for $\WWW$ and there is a polynomial time algorithm
  that, given $\WWW$, computes such a set.
\end{lemma}
\begin{proof}
    By definition, the coordinate-vertices of $\WWW$ are of degree $3$. As $G$ is Eulerian and each non-terminal vertex has degree $4$, every coordinate vertex $x_{i,j}^-$ of $\WWW$ has exactly one additional outgoing edge $e_{i,j}^- \not\in E(\WWW)$ and every coordinate vertex $x_{i,j}^+$ has one additional incoming edge $e_{ij}^+ \not\in E(\WWW)$. As $G$ is Eulerian, this implies that there is a set $\PPP$ of pairwise edge-disjoint paths, which in addition are edge-disjoint from $\WWW$, such that each path in $\PPP$ starts at a coordinate vertex $x_{i,j}^-$ with the edge $e_{i,j}^-$ and ends with some edge $e_{i',j'}^+$ at a coordinate vertex $x_{i',j'}^+$ and, furthermore, each coordinate vertex of $\WWW$ is contained in exactly one path of $\PPP$.
    
    Clearly such a complete set of coordinate-paths can be found efficiently in the obvious way.
\end{proof}

Note that the complete set~$\PPP$ of coordinate-paths for~$\WWW$ is not necessarily unique.
It follows immediately from the construction that every coordinate-path must start at an in-coordinate and end at an out-coordinate. Thus, any set of coordinate-paths defines a matching from in-coordinates to out-coordinates. 

\begin{definition}[Coordinate matching]\label{def:coordinate-matching}
    Let $\WWW$ be a wall. A \emph{coordinate matching} $M$ is a matching between in-coordinates and out-coordinates. We consider the pairs $(x^-_{i,j}, x^+_{i',j'}) \in M$ as directed edges from an in-coordinate to an out-coordinate and define $\WWW+M$ as the digraph obtained from $\WWW$ by adding the edges from $M$.
    
    We call $M$ \emph{complete} if every in-coordinate is the tail and every out-coordinate is the head of an edge in $M$. Note here that if~$M$ is complete, then it is a perfect matching on the set of coordinates~$\operatorname{Coord}(\WWW)$ of~$\WWW$.
    
    If $\PPP$ is a set of coordinate-paths of $\WWW$ we define the \emph{coordinate matching $M(\PPP)$ induced by $\PPP$} as the matching $\{ (u,v) \mid $ there is $P \in \PPP$ starting at $u$ and ending at $v$ $\}$. 
\end{definition}

With this definition at hand, we start by analysing the case that we are given an elementary wall together with a complete coordinate matching~$M$. In a second step we will then reduce the general case to this one by splitting off at non-coordinate vertices with respect to the wall and a complete set of coordinate-paths.

\subsubsection{Swirls and Routers in Elementary Walls}
\label{subsubsec:elementary_wall_swirl}
As a first step towards finding routers and swirls, we prove a special case of \cref{thm:flat_swirl_theorem} where we restrict the setting to an elementary wall~$\WWW'$ together with a complete coordinate matching~$M$. 
We proceed with proving that, provided that ~$\WWW'$ is `large enough', we either find a `large' router or a `large' swirl in~$\WWW' + M$. This will be the base case for the general setting, where we are only given an immersion of~$\WWW'+M$, i.e., a cylindrical wall together with a complete set of coordinate-paths.

\begin{theorem}[Flat-Swirl Theorem for Elementary Walls]\label{thm:elementary-flat-swirl-theorem}
  There are computable functions
  $f_{\ref{thm:elementary-flat-swirl-theorem}},g_{\ref{thm:elementary-flat-swirl-theorem}}$ such that for all
  $t_1, t_2 \in \N$, if $\WWW'$ is an elementary cylindrical
  $f_{\ref{thm:elementary-flat-swirl-theorem}}(t_1, t_2) \times
  g_{\ref{thm:elementary-flat-swirl-theorem}}(t_1, t_2)$-wall with
  $V(\WWW') \cap V(D) = \emptyset$ and $M$ is a complete coordinate matching on
  $\WWW'$ then
  \begin{itemize}
  \item[(i)] $\WWW'+M$ contains a flat~$t_1$-swirl~$\mathcal{S}$ induced
    by~$\WWW'$, or
  \item[(ii)] $\WWW'+M$ contains a~$t_2$-router grasped by~$\WWW'$.
  \end{itemize}
  Moreover we can decide in~$\textit{fpt}$-time, parameterized by $t_1, t_2$, whether~$(i)$
  or~$(ii)$ holds and output the relevant structure.
\end{theorem}

Towards proving the theorem we first need to establish some notation
and auxiliary results.

\begin{figure}
\begin{subfigure}{.5\textwidth}
  \centering
  \includegraphics[width=.8\textwidth]{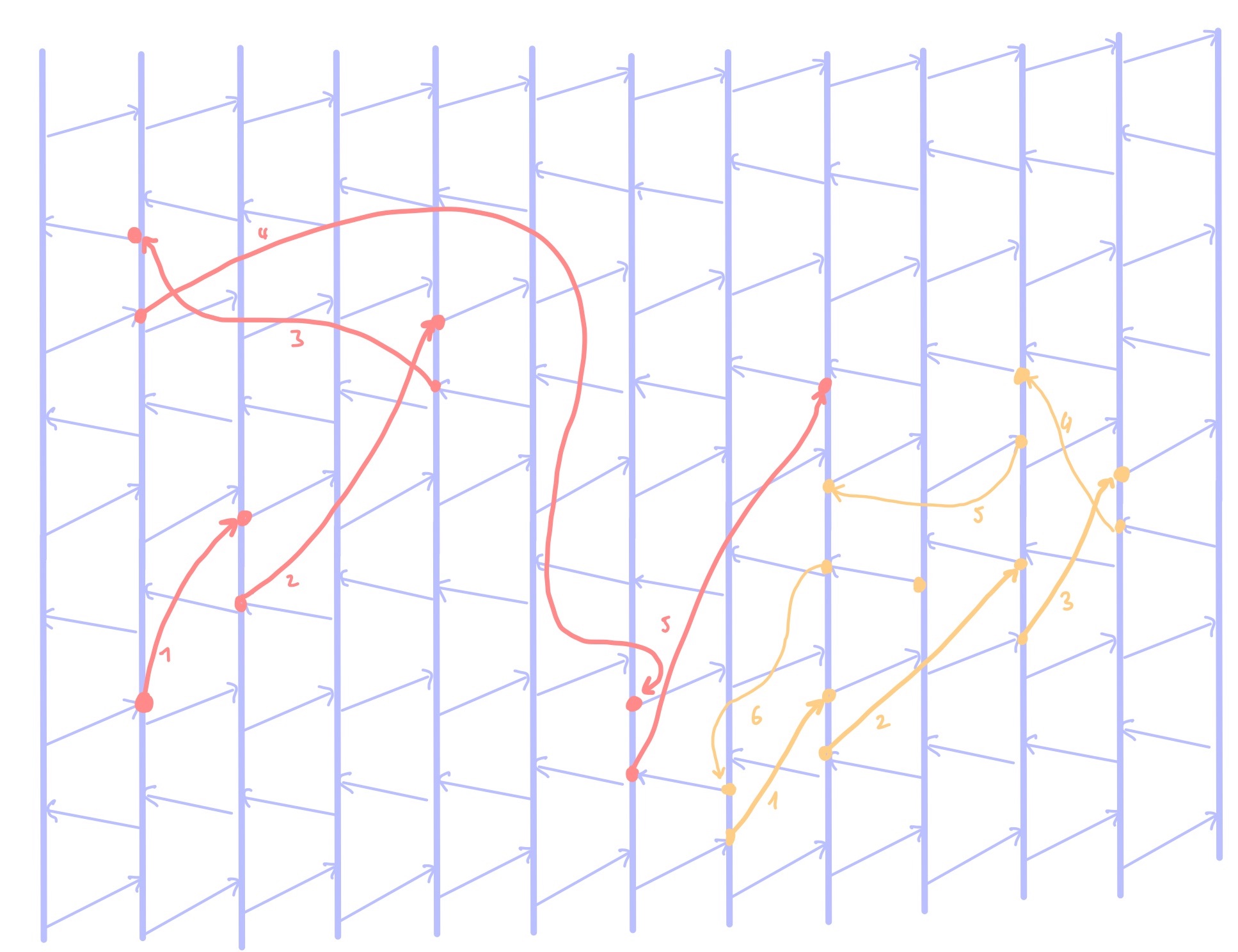}
  \caption{Two jump-sequences.}
  \label{fig:jump-seq}
\end{subfigure}
\begin{subfigure}{.5\textwidth}
  \centering
  
  \includegraphics[width=.8\textwidth]{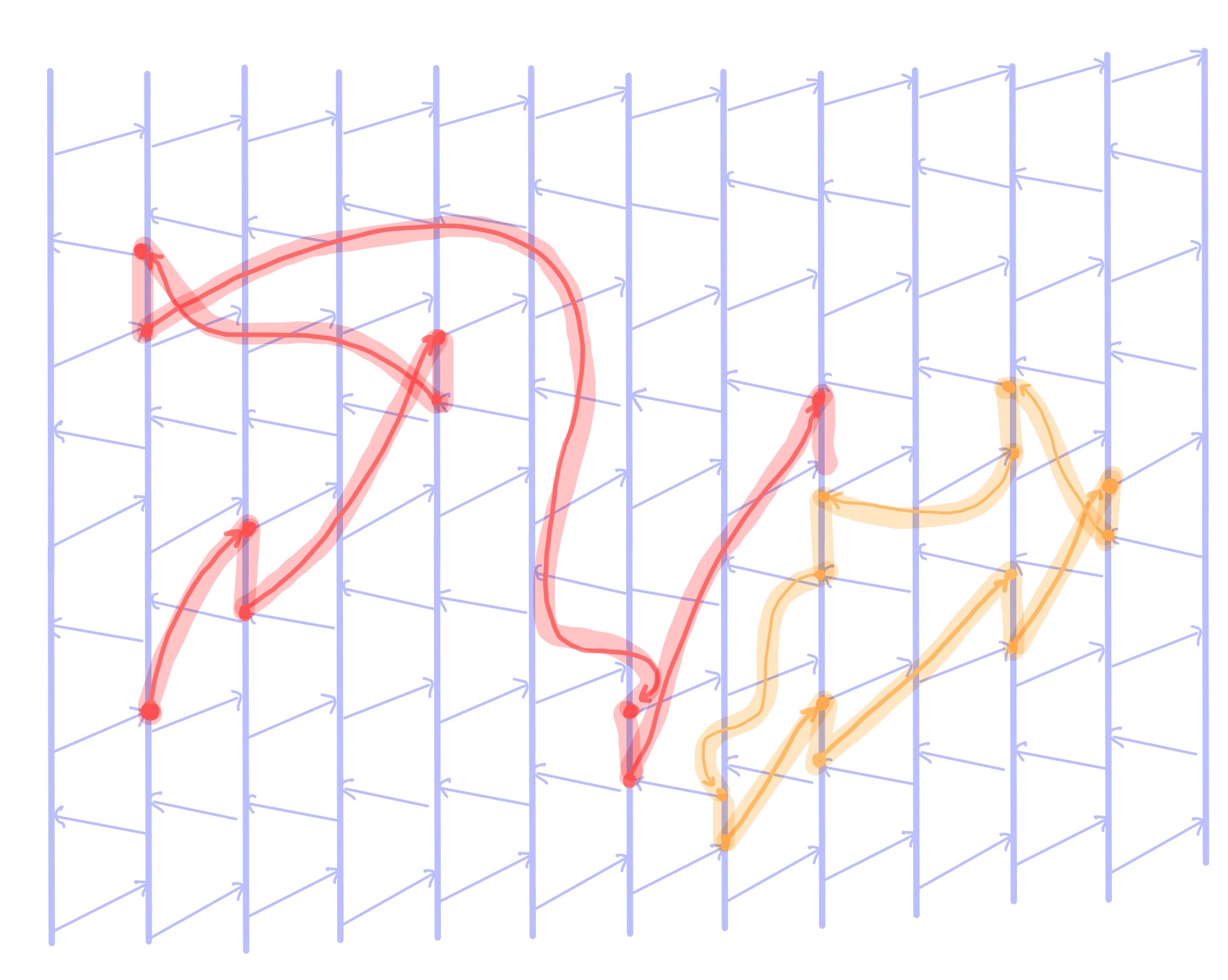}
  \caption{Two jump-paths.}
  \label{fig:jump-path}
\end{subfigure}
\caption{The left side depicts of a wall with two jump-sequences marked in red and orange. The right hand side depicts the associated jump-paths to the two jump-sequences where one of which is a jump cycle. }
\label{fig:jump_sequence_and_path}
\end{figure}

\begin{definition}\label{def:jump_path}
  Let $\WWW'$ be an elementary wall and let $M$ be a coordinate matching
  on $\WWW'$. An edge $e \in M$ of the form
  $e = (x^-_{i,j}, x^+_{i, j - 1})$ is called an \emph{up-path}
  of $\WWW'$. Any edge $e \in M$ which is not an up-path is called a \emph{jump}
  of $\WWW'$.

  A \emph{jump-sequence} of $\WWW'$ is a sequence
  $\bar{\iota} \coloneqq (\iota_1, \dots, \iota_l)$ of jumps
  $\iota_j = (x^-_{c_{1,j}, r_{1, j}}, x^+_{c_{2,j}, r_{2, j}}) \in M$ such
  that $c_{2, j} = c_{1,j+1}$ and $r_{2,j} = r_{1,j} + 1$ for all
  $1 \leq j < l$. We say that $\bar{\iota}$ starts at $\iota_1$ and ends at
  $\iota_l$. The length of $\bar{\iota}$ is $l$.

  With any jump-sequence $\bar{\iota} \coloneqq (\iota_1, \dots, \iota_l)$ of distinct matching edges $\iota_j = (x^-_{c_{1,j}, r_{1, j}}, x^+_{c_{2,j}, r_{2, j}}) \in M$ we associate a \emph{jump-path} $P(\bar\iota)$ defined as the path
  \[
     \big(x^-_{c_{1,1}, r_{1, 1}}, \iota_1, x^+_{c_{2,1}, r_{2, 1}}, Q_1, x^-_{c_{1,2}, r_{1, 2}}, \iota_2, x^+_{c_{2,2}, r_{2, 2}}, Q_2, \dots, \iota_l, x^+_{c_{2,l}, r_{2, l}}, Q_l\big)
   \] 
   where $Q_i$ is the subpath of the wall cycle $W_{c_{2,i}}$ between the coordinates $x^+_{c_{2,i}, r_{2, i}}$ and $x^-_{c_{2,i}, r_{2, i}+1}$.       
\end{definition}
See \cref{fig:jump_sequence_and_path} for exemplary Note here that jump-sequences are well-defined, i.e., given a jump-sequence starting with a jump, then no other edge in the sequence is an up-path; this is an easy but very crucial and fruitful observation.
\begin{observation}
    Let~$\bar{\iota} \coloneqq (\iota_1, \dots, \iota_l)$ be a sequence as in \cref{def:jump_path} and let~$\iota_1 \in M$ be a jump. Then~$\iota_i \in M$ is a jump for every~$1 \leq i \leq l$.
\end{observation}
\begin{proof}
    This is imminent since for any up-path~$\iota \in M$ we have that~$\iota= (x^-_{i,j}, x^+_{i, j - 1})$ for respective coordinates by definition, and thus by the coordinate restriction for consecutive jumps in a `jump-sequence', any jump-sequence containing~$\iota$ can only contain~$\iota$. To see this note that for the coordinates of two consecutive jumps~$\iota_j,\iota_{j+1}$ in the sequence where~$\iota_i =(x^-_{c_{1,i}, r_{1, i}}, x^+_{c_{2,i}, r_{2, i}})$, it holds~$c_{2, j} = c_{1,j+1}$ by definition and thus if~$\iota_j= \iota$ we have~$\iota_{j+1} = \iota$ and vice-versa for all~$1 \leq j<i \leq l$. 
\end{proof}

Intuitively, the jump-path of a jump-sequence $\bar{\iota} \coloneqq (\iota_1, \dots, \iota_l)$ is the path in $\WWW'+M$ starting at the tail of $\iota_1$, taking the jump $\iota_1$ to its head, which is an out-coordinate, then following the wall cycle down towards the next in-coordinate, which, by definition, is the tail of $\iota_2$, and then continue in the same way until we have reached the last jump $\iota_l$. The jump-path ends on a subpath of the wall cycle; see \cref{fig:jump-path}.

Observe that a jump-sequence is completely determined by the edge it starts at and its length. Thus, if $M$ is a complete coordinate matching on $\WWW'$, then for every $e \in M$ and
every $l \geq 1$ there is exactly one jump-sequence of length $l$
starting at $e$ (note that repetitions are allowed in jump-sequences). 

\begin{observation}\label{obs:jump_sequence_is_unique}
    Let~$\iota \in M$ be a jump and let~$l \in \N$ be fixed. Then there is a unique jump-sequence of length~$l$ starting with~$\iota$. 
\end{observation}

Consequently, for every jump $e \in M$ there is a minimal
number $l \geq 1$ such that the jump-sequence $(e = \iota_1, \dots, \iota_l=e)$ of length $l$ starting at
$e$ also ends at $e$; for an up-path~$e \in M$ we fix this number to be~$2$, which makes sense when trying to construct a jump-sequence starting at the up-path~$e$ in the same spirit, the sequence would consist solely of repetitions of~$e$. In either case, the jump-path of the sequence $(e = \iota_1, \dots, \iota_{l-1})$ is a cycle in $\WWW+M$ which we call the \emph{jump cycle} of $e$; see \cref{fig:jump-path}. Note that if~$M$ is complete, then~$\WWW'+M $ is Eulerian and thus every jump-path can be extended to a jump cycle; this together with \cref{obs:jump_sequence_is_unique} yields the following.

\begin{observation}\label{obs:jump_cycles_exist_if_Euler}
Let~$M$ be a complete coordinate matching on an elementary wall~$\WWW'$ and let~$e \in M$ be some jump. Then there exists a unique jump cycle starting at~$e$.
\end{observation}

Note, however, that if the coordinate matching $M$ is not complete, then the jump cycle of an edge $e\in M$ may not exist. This leads to the following definition.

\begin{definition}\label{def:r_saturated}
    Let~$\WWW$ be an elementary wall and let~$M$ be a coordinate matching. We say that $\iota \in M$ is \emph{$r$-saturated} in
    $M$ if the jump-sequence of length $r$ starting at $\iota$ exists.
\end{definition}

We first prove that if the wall $\WWW$ contains a tile such that $M$ is a
complete coordinate matching on $T$ and such that every matching edge
starting at an in-coordinate of $T$ is an up-path, then we obtain a
swirl induced by~$T$.

\begin{lemma}\label{lem:no_jumps_is_swirl}
  Let~$\WWW$ be a plane elementary wall and let $M$ be a coordinate matching
  in $\WWW$.

  Let~$T \subseteq \WWW$ be a~$(t+1) \times t$-tile and
  let~$M_T \subseteq M$ be the set of all matching edges in $M$ with an endpoint in $T$. 
  If~$M_T$ is complete for $T$ and solely consists of up-paths, then $T+M_T$ can (uniquely) be embedded so that it contains a~$t$-swirl~$\SSS$ grasped by the~$(t+1) \times t$-tile~$T$ and a~$t$-swirl~$\SSS'$ induced by a~$t$-tile~$T'\subset T$. In particular~$\SSS'$ can be constructed such that~$T'+M_{T'} \subseteq \SSS'$.
\end{lemma}
\begin{proof}
  To see this assume that~$T = \WWW_{[0,t]}^{[1,2t]}$
  i.e.,~$T \cap \WWW \subseteq W_0\cup \ldots \cup W_t\cup H_1 \cup
  H_{2t}$, where, to simplify notation, we assume the wall-cycles to start at~$W_0$. Since~$M_T$ only contains up-paths in~$T$ this
  means that we additionally have the
  paths~$U_1,\ldots,U_t \subset T + M_T$
  with~$U_i \subseteq \WWW_{[i-1,i]}^{[1,2t]}$ for~$1 \leq i \leq t$ defined as follows: 
  \begin{itemize}
  \item  If~$i$ is odd, then 
  \[
    U_i \coloneqq
    (x_{i,2t}^+,x_{i-1,2t}^-,x_{i-1,2t-1}^+,x_{i,2t-1}^-,x_{i,2t-2}^+,\ldots,x_{i,2}^+,x_{i-1,2}^-,x_{i-1,1}^+,x_{i,1}^-)
  \]
  is a path from~$x_{i,2t}^+$ to~$x_{i,1}^-$ alternating between edges
  in~$\bigcup_{1\leq j \leq t} E(H_j \cap \WWW_{[i,i-1]}^{[1,2t]})$ and up-paths~$(x_{j,p}^-,x_{j,p-1}^+) \in M_T$ for~$1 \leq i \leq t$
  and~$0\leq j \leq t$ and~$1 \leq p \leq 2t$.
  \item Similarly, if~$i$ is even, we get the following
  \[
    U_i \coloneqq
    (x_{i,2t}^+,x_{i+1,2t}^-,x_{i+1,2t-1}^+,x_{i,2t-1}^-,x_{i,2t-2}^+,\ldots,x_{i+1,2}^+,x_{i,2}^-,x_{i,1}^+,x_{i+1,1}^-).
  \]
\end{itemize}
    Note that the paths $U_i$ are edge-disjoint from the paths $W_j$  for~$1 \leq i \leq t$ and~$0 \leq j \leq t$. 

    Finally we can define the swirl cycles. For~$i$ odd these are as follows:
    \begin{align*}
        S_i \coloneqq \:& x^+_{\frac{(i-1)}{2},i}\,H_i\,x^-_{\frac{t+1-i}{2},i} \:\cup\: x^-_{\frac{t+1-i}{2},i} \,W_{\frac{t+i-1}{2}} \,x^+_{\frac{t+1-i}{2},2t+1-i} \;\cup \\
        &x^+_{\frac{t+1-i}{2},2t+1-i}\, H_{2t+1-i} \,x^-_{\frac{i-1}{2},2t+1-i} \:\cup \:x^-_{\frac{i-1}{2},2t+1-i} \,U_{\frac{i+1}{2}}\, x^+_{\frac{i-1}{2},i}.
    \end{align*} 
    For~$i$ even we define them similarly as follows:

    \begin{align*}
         S_i \coloneqq \: &x^-_{\frac{i}{2},i}\,W_{\frac{i}{2}} \,x^+_{\frac{i}{2},2t+1-i} \:\cup \:x^+_{\frac{i}{2},2t+1-i} \,H_{2t+1 -i} \,x^-_{t+1-\frac{i}{2}, 2t+1 -i} \:\cup \\
         &x^-_{t+1-\frac{i}{2},2t+1 -i}\,U_{t+1-\frac{i}{2}}\, x^+_{t+1-\frac{i}{2},i} \:\cup\: x^+_{t+1-\frac{i}{2},i} \,H_{i} \,x^-_{\frac{i}{2},i}.
    \end{align*}
    
    By construction, the cycles~$S_1,\ldots,S_t$ are all disjoint, they
    are concentric where the region bounded by~$S_t$ that does not contain the outer-face of~$\WWW$ is their centre,
    and they have alternating orientation with respect to a given
    orientation of the plane; then~$\SSS \coloneqq \bigcup_{i=1}^sS_i$ is a swirl grasped by~$T$ for all its coordinates are contained in~$V(\SSS)$. Note that a priori it does not necessarily hold true that~$T'\subseteq \SSS$ for any~$t$-tile~$T'$; we show how to achieve this next.

For every~$1 \leq i \leq s$,~$1 \leq p \leq 2s$ and every vertex~$x_{i,p} \in \mathcal{S}$, if the path~$x^+_{i,p}W_ix^-_{i,p+1}$ is not contained in~$\SSS$, then extend~$S_{\min(p,2s-p+1)}$ (which
  contains~$x^+_{i,p}$ by definition) by the
  cycle~$x^+_{i,p}W_ix^-_{i,p+1}\,\cup \,x^-_{i,p+1}U_ix^+_{i,p}$. (Note that
  we do not add these paths for~$W_0$). Let~$T' \coloneqq T \cap \WWW_{[1,t]}$ be the~$t$-tile resulting when restricting to the cycles~$W_1,\ldots,W_t$. Then by construction it holds that~$T' \subseteq \SSS$ and in particular~$M_{T'} \subset E(\SSS')$, concluding the proof.
\end{proof}

We refer to the swirl constructed in \cref{lem:no_jumps_is_swirl} as the \emph{canonical swirl induced by a tile~$T$} or simply \emph{the swirl induced by~$T$}. See \cref{fig:grasped_swirl} for an illustration. Note that in general a grasped swirl~$\SSS$ may not necessarily have high tree-width itself. This is why in the proof above we extended~$\SSS$ to an~$s$-swirl~$\SSS'$ such that~$T \subseteq \SSS'$. Induced swirls will be to Eulerian digraphs what walls are to undirected graphs: the centre piece to all further structural analysis.

\begin{figure}
    \centering

    \includegraphics[scale=.75]{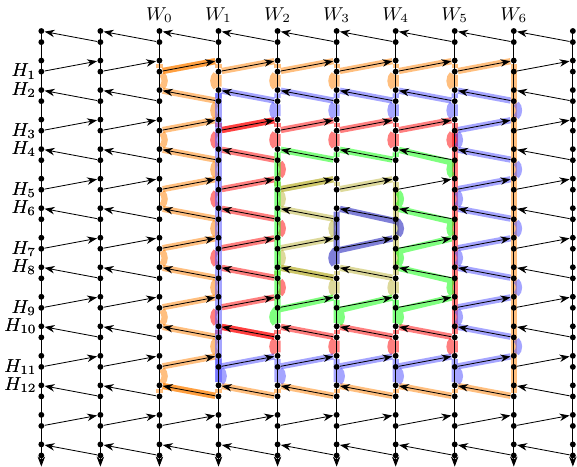}
    \caption{A canonical swirl induced by a tile.}
    \label{fig:induced_swirl}
\end{figure}

\begin{definition}[Canonical swirl]\label{def:canonical_swirl_induced_by_tile}
  Let~$T' \subseteq \WWW$ be an~$(s+1) \times s$-tile and~$T \subset T'$ an~$s$-tile.
  Let~$\SSS$ be the swirl induced by~$T$ as constructed in the proof of \cref{lem:no_jumps_is_swirl}. Then we call~$\SSS$ the \emph{canonical swirl (induced by~$T$)} or simply \emph{the swirl induced by~$T$}.

\end{definition}

\begin{remark}
  First note that a swirl induced by some tile is indeed a swirl, for
  the cycles remain concentric and pairwise edge-disjoint, thus the
  definition is sensible. Further, it is clear by definition that a
  swirl \emph{induced} by some tile~$T$ is, in particular, \emph{grasped} by
  the same tile~$T$. See
  \cref{fig:induced_swirl} for an exemplary representation of a
  canonical induced swirl.

The reason why we omitted the edges of~$W_0$ in the construction is that this way, when taking
an~$s\times 2s$-tile~$T \subseteq \WWW$ and looking at two
disjoint~$s$-tiles~$T_1,T_2 \subseteq T$ covering all the vertices
of~$T$, then the two induced swirls~$\SSS_1,\SSS_2$ are easily seen to be edge-disjoint satisfying~$T \subseteq \SSS_1 \cup \SSS_2$.

\label{rem:induced_swirls_and_tiling}
\end{remark}

We rephrase part of the remark as an observation as follows.
\begin{observation}\label{obs:subtiles_of_grasping_tiles_induce_swirls}
     Let~$\mathcal{S}$ be an~$s$-swirl induced by an~$s$-tile~$T \subseteq \WWW$ for some~$s \in \N$. Let~$T' \subseteq T$ be an~$s'$-tile in~$\WWW$. Then there exists an~$s'$-swirl~$\mathcal{S}'$ induced by~$T'$ such that~$\mathcal{S}' \subseteq \mathcal{S}$ (viewing both as graphs).
\end{observation}
\begin{proof}
    The observation is obvious by construction of canonical swirls induced by tiles, since the swirl~$\SSS$ contains all of the up-paths in~$M_{T'}$ by construction.
\end{proof}

We call the swirl~$\mathcal{S}'$ in \cref{obs:subtiles_of_grasping_tiles_induce_swirls} a \emph{sub-swirl induced by the tile~$T'$}.
It is an easy observation that a canonical swirl induced by a large tile contains a large cylindrical wall~$\WWW'$ as a subgraph. Note that this follows at once from \cref{thm:undirected_vs_directed_tw_in_Eulerian_graphs} using the fact that the underlying graph contains a large undirected wall (since the graph contains a tile). Furthermore one easily verifies that a large canonical~$s$-swirl contains a large bi-directed grid as a minor, highlighting its usefulness with respect to vertex-disjoint routing too. This is a qualitative result since any large bi-directed grid trivially contains a large swirl.

\smallskip

\cref{lem:no_jumps_is_swirl} covered the case where the edges of a coordinate matching allow us to find a swirl. We now consider the case where not all matching edges are up-paths and show that if there are sufficiently many jumps we can construct a large router. 

For the following exposition let $\WWW = (W_1, \dots, W_{t_w}, H_1, \dots, H_{2t_h})$  be a plane elementary $t_w \times t_h$-cylindrical wall and $M$ be a coordinate matching on $\WWW$.
It turns out that not every jump in a wall is equally useful for constructing a cross. In most cases a single jump is enough to yield a cross. But some jumps are too `short' to yield a cross on their own and we need to combine up to three of these jumps to construct a cross. This motivates the following definition capturing the different types of jumps. See \cref{fig:wall-jumps} for an illustration.

\begin{definition}[Types of Jumps] \label{def:jumps_on_wall}
    Let~$\iota \in M$. We distinguish between three types of jumps as follows.
    \begin{itemize}
    \item If $p$ is odd and $\iota = (x_{i,p}^-,x_{i-1,p-2}^+)$ or $p$ is even and $\iota = (x_{i,p}^-,x_{i+1,p-2}^+)$, then $\iota$ is a \emph{Type I jump}.
    \item If $\iota = (x_{i,p}^-, x_{i, p-3}^+)$, then $\iota$ is called a \emph{Type II jump}.
    \item Every jump $\iota$ that is not of Type I or II is a \emph{Type 0} jump.
    \end{itemize}
\end{definition}

\begin{figure}
  \centering
  \includegraphics[scale=.8]{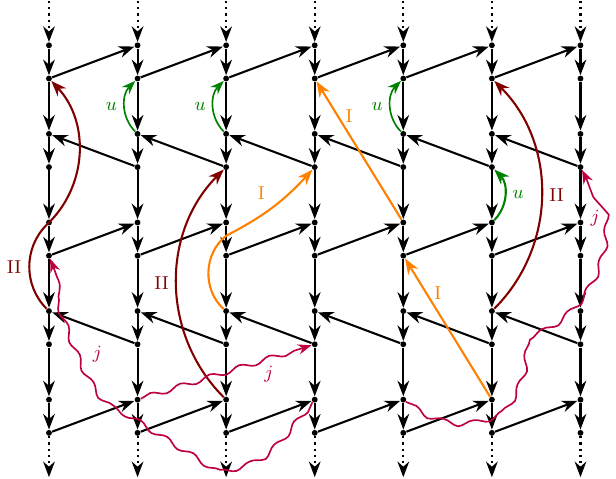}

    \caption{Wall with jumps marked by $j$, Type I jumps marked by~$I$, Type II jumps marked by~$II$ and up-paths marked by $u$.}
    \label{fig:wall-jumps}
\end{figure}

We define what we mean by jumps and paths inducing crosses.

\begin{definition}[Jump(-path) inducing a Cross]
    Let~$T$ be a~$t$-tile and let~$\iota$ be a jump with endpoints~$x,y \in V(T)$. We say that~$\iota$ \emph{induces a (wall-)local usable cross} if there are two paths~$P,P'$ forming a wall-local usable~$T$-cross as in \cref{def:swirl_usable_cross} such that~$\iota \in E(P)$. We say that a jump-path~$P(\iota)$ starting at~$\iota$ \emph{induces a (wall-)local usable cross} if~$P(\iota) \subseteq P$. Similarly we say that a jump-sequence~$\bar \iota$ \emph{induces a (wall-)local usable cross} if the respective jump-path~$P(\bar\iota)$ does.
    \label{def:jump_inducing_cross_wall}
\end{definition}

With these notions at hand we prove the following easy but rather technical lemma.

\begin{lemma}\label{lem:jumps_give_crosses}
 Let~$\iota\coloneqq(x_{i,p}^-,x_{j,q}^+) \in M$ be a Type 0 jump in $\WWW$, 
 for some~$3\leq i,j \leq t-3$ and~$5 \leq p,q \leq 2t-5$. Then the tile $T$ bounded by~$W_{\min(i,j)-2}$, $W_{\max(i,j)+2}$ and~$H_{\min(p,q)-4}$, $H_{\max(p,q)+4}$ contains a usable cross.
\end{lemma}
\begin{proof}
    Let $u = \min(p,q)$ and $l = \max(p,q)$ be the indices of the upper and lower rows containing the ends of the jump $\iota$. Let $b_l \coloneqq \min(i,j)$ and $b_r \coloneqq \max(i,j)$ be the left and right boundary of the jump $\iota$. Without loss of generality we assume that $u$ is odd and $l$ is even, so that $H_{u-4}$ is a left-to-right path and $H_{l+4}$ is a right-to-left path; otherwise extend the tile by another path above and (or) below which is feasible for~$5 \leq p,q \leq 2t-5$. 
    Let $s_1 \coloneqq x^+_{b_{l}-2, u-4}$ and $s_2 \coloneqq x^-_{b_{r}+2, u-4}$ as well as $t_2  \coloneqq x^-_{b_{l}-2, l+4}$ and  $t_1 \coloneqq x^-_{b_{r}+2, l+4}$ be the corners of the tile.

    By definition, $T$ contains the horizontal paths $H_{u-4}, \dots, H_{u-1}$ and $H_{l+1}, \dots, H_{l+4}$, and thus contains two horizontal paths in each direction `above' the jump $\iota$ and two paths in each direction `below' $\iota$. 
    Similarly, $T$ contains the vertical paths $W_{b_l-1}, W_{b_l-2}$ and $W_{b_r+1}, W_{b_r+2}$ connecting $H_{u-4}$ and $H_{l+4}$. We continue with a case distinction regarding the ends of the jump.

\medskip

    \noindent\textit{Case 1: $p < q$. } It is easy to see that if $p < q$ then $\iota$ induces a usable cross as follows.     
    Let $m_l$ be the vertex in the intersection of $H_p$ and $W_{b_l-2}$ and let $m_r$ be the vertex in the intersection of $H_p$ and $W_{b_r+2}$. Then $P_1 \coloneqq (s_1, W_{b_l-2}, m_l, H_p, m_r, W_{b_r+2}, t_1)$ is an $s_1{-}t_1$-path. We construct an $s_2{-}t_2$-path $P_2$ as follows. Let $P^1_2 \coloneqq (s_2, W_{b_r+2}, x^+_{b_r+2, u-3}, H_{u-3}, x^-_{i, u-3}, W_i, x_{i,p}^-)$. Then $P^1_2$ is an $s_2{-}x_{i,p}^-$-path which is edge-disjoint from $P_1$. Let $P_2^2 \coloneqq (x_{j,q}^+, W_j, x^+_{j, l+4}, H_{l+4}, t_2)$. Then $P_2^2$ is an  $x_{j,q}^+{-}t_2$-path which is edge-disjoint from $P_1$ and from $P_2^1$. Then $P_2 \coloneqq P_2^1 + \iota + P_2^2$ is an $s_2{-}t_2$-path edge-disjoint from  $P_1$. Thus, $P_1, P_2$ form a local usable cross in the tile $T$ as required. 

\medskip

    \noindent\textit{Case 2: $p = q$. } Now suppose $p=q$. We only show the case where $p$ is odd and $i < j$ as the other cases can be proven analogously using the symmetry of the tile. Again let $m_l$ and $m_r$ be the vertices in the intersection of $H_p$ and $W_{b_l-2}$ and $W_{b_r+2}$, resp.
    Let $P_1$ be the $s_1{-}t_1$-path $(s_1, W_{b_l-2}, m_l, H_p, x_{i,p}^-, \iota, x_{j,p}^+, H_p, m_r, W_{b_r+2}, t_1)$ and let $P_2$ be the path obtained by routing from $s_2$ down along $W_{b_r+2}$ to $x^+_{b_r+2, u-3}$ and then along  $H_{u-3}$ to $x^+_{i, u-3}$. From there $P_2$ continues along $W_i$ to $x^+_{i, l+4}$ and then along  $H_{l+4}$ to $t_2$. Then $P_1$ and $P_2$ are edge disjoint and form a cross in $T$ by construction.

\medskip

    \noindent\textit{Case 3: $p > q$. } The last case we consider is $p > q$, that is, the jump $\iota$ goes upwards; this is where we must differentiate from Type I and II jumps which \emph{do not} immediately yield crosses in contrast to the following. We distinguish between two sub-cases.
    
\smallskip

    \textit{Case 3a: $i \neq j$. } Suppose first that $i\neq j$. Without loss of generality we may assume that $i<j$, the other case is symmetric. If $i < j-1$ then we immediately get a cross as follows: Let $P_1$ be the path that starts at $s_1$, then follows  $H_{u-4}$ to $x^-_{i, u-4}$. From there it follows $W_i$ to $x^-_{i, p}$, then takes $\iota$ to $x^+_{j, q}$, and finally routes along  $W_j$ and $H_{l+3}$ to $x^-_{b_r+2, l+3}$ from where we route to~$t_1$ along $W_{b_r+2}$. 
    Let $P_2$ be the path $(s_2, W_{b_r+2}, x^+_{b_r+2, u-3}, H_{u-3}, x^-_{i+1, u-3}, W_{i+1}, x^+_{i+1, l+4}, H_{l+4}, t_2)$. Then $P_1$ and $P_2$ are edge-disjoint and yield a cross in $T$. 

    Thus assume that $i = j-1$. 
    We again distinguish two cases depending on the parity of $p$.
    \begin{itemize}
        \item[$p$ odd.] The easy case is when $p$ is odd, that is, $H_p$ is a left-to-right path. In this case let $m_l \coloneqq x^+_{b_l-2, p}$. We define $P_1 \coloneqq (s_1, W_{b_l-1}, m_l, H_p, x^-_{i,p}, \iota, x^+_{j,q}, W_j, x^+_{j, l+3}, H_{l+3}, x^-_{b_r+2, l+3}, W_{b_r+2}, t_1)$. We now take $P_2$ to be the path that starts at $s_2$, follows $W_{b_r+2}$ down to $x^+_{b_r+2, u-3}$, then follows $H_{u-3}$ to $x^-_{i, u-3}$ and then along $W_i$ to $x^+_{i, l+4}$. From there is follows $H_{l+4}$ to $t_2$. Then $P_1$ and $P_2$ are edge disjoint and yield a usable cross. 

        \item[$p$ even.]  Suppose that $p$ is even; we get two sub-cases. 
        If $q$ is odd, we define $P_1$ as follows. Starting at $s_1$ we first use  $H_{u-4}$  to route to $x^-_{i, u-4}$ on $W_i$, then continue downwards along $W_i$ to $x^-_{i, p}$, and then use $\iota$ to get to $x^+_{i+1, q}$. As $q$ is odd we can continue from $x^+_{i+1, q}$ along $H_q$ to $W_{b_r+2}$ and follow $W_{b_r+2}$ until we reach $t_1$. To define $P_2$ we start at $s_2$ and follow $W_{b_r+2}$ to $x^+_{b_r+2,u-3}$ and then use $H_{u-3}$ to route the path to $x^-_{i+1, u-3}$. We then follow $W_{i+1}$ all the way down to $x^+_{i+1, l+4}$ and then  $H_{l+4}$ towards $t_2$. As $P_1$ does not use any edge of $W_{i+1}$, $P_1$ and $P_2$ are vertex disjoint as required. 
        
        \smallskip
        
       Whence we may assume that $q$ is even as well. 
        As $\iota$ is not a Type I jump, this implies that $p \not= q+2$ and thus $p \geq q+4$. This time we need a slightly different way of constructing the paths $P_1, P_2$ inducing the cross on $T$.
        We define $P_1$ as follows. Starting at $s_1$ we first use  $H_{u-4}$ to route $P_1$ to $x^-_{j+1, u-4}$ on $W_{j+1}$ and then follow $W_{j+1}$ downwards until we reach $x^+_{j+1, p-2}$ on $H_{p-2}$. We then follow $H_{p-2}$ to the left until we reach $x^-_{i-1, p-2}$ on $W_{i-1}$. We then follow $W_{i-1}$ downwards until we reach $x^+_{i-1, l+3}$ on $H_{l+3}$  and then use $H_{l+3}$ and $W_{b_r+2}$ to route $P_1$ to $t_1$ where it ends. 
        Now we construct $P_2$. Starting at $s_2$ we first follow $W_{b_r+2}$ downwards until we reach $H_p$. We then follow $H_p$ to the left until we reach $x_{i,p}$ and then use $\iota$ to continue to $x_{j, q}$. From $x_{j, q}$ we use the path $H_q$ to go left until we reach $W_{i-2}$ which we follow all the way down until we reach $H_{l+4}$ which we use to reach $t_2$ where the path $P_2$ ends.  By construction, $P_1$ and $P_2$ are edge-disjoint and form a cross on $T$.  
    \end{itemize}
    This concludes the case where $i\not= j$.
\smallskip

    \textit{Case 3b: $i = j$. } The last case we need to consider is $i=j$ (recall that $p > q$). 
    As $\iota$ is not of Type II and also not an up-path, we have $p \not\in \{ q+1, q+3\}$. We proceed by analysing two sub-cases once more.
    \begin{itemize}
        \item[$p \equiv q$.] If $p$ and $q$ have the same parity we get a cross as follows. Suppose first that $p$ is odd. Then we define $P_1$ as the path from $s_1$ to $W_{i-1}$ using $H_{u-3}$ or $H_{u-4}$, then along $W_{i-1}$ to $x^+_{i-1, p}$ and then along $H_p$ to $x^-_{i,p}$. From there $P_1$ follows $\iota$ up to $x^+_{i, q}$ and then along $H_q$ to $x^-_{i+1, q}$ on $W_{i+1}$. From there $P_1$ follows $W_{i+1}$ to the bottom of the tile where it uses $H_{l+3}$ or $H_{l+4}$ to reach $t_1$. Note that $P_1$ is an $s_1{-}t_1$-path which is edge-disjoint from $W_i$. 
    
        The path $P_2$ starts at $s_2$, follows $W_{b_r+2}$ down until it reaches  $H_{u-3}$ which it follows to reach $W_i$. It then follows $W_i$ to $x^+_{i, l+4}$ and then continues along $H_{l+4}$ until it reaches $t_2$. By construction, $P_1$ and $P_2$ are edge-disjoint and thus form a cross on $T$. 
    
        If $p$ is even we construct $P_1$ and $P_2$ analogously but this time we route $P_1$ along $W_i$.

        \item[$p \not \equiv q$.] Now suppose that $p$ and $q$ have different parity. Without loss of generality we assume that $p$ is odd and $q$ is even. As $\iota$ is not of Type II this implies $q \leq p-5$. We define $P_1$ as the path that starts at $s_1$, then uses $H_{u-4}$ to route to $W_{i-2}$, then follows $W_{i-2}$ down to $H_p$ and then uses $H_p$ to route to $x^-_{i,p}$. From there we follow the jump $\iota$ to $x^+_{i, q}$ and then down along $W_i$ to $x^+_{i, q+1}$. From there we follow $H_{q+1}$ to the right until we reach $W_{b_r+2}$ and then down to $t_1$.
        Now $P_1$ is no longer edge-disjoint from $W_i$ and therefore we need to choose $P_2$ differently from the previous case. Instead we define $P_2$ as follows. Starting at $s_2$, $P_2$ first follows $W_{b_r+2}$ down to $x^+_{b_r+2, u-3}$ and then uses $H_{u-3}$ to reach $W_{i-1}$. It then follows $W_{i-1}$ down to $x^+_{i-1, p-2}$ and then uses $H_{p-2}$ to reach $W_i$. From there we proceed as before, that is, $P_2$ continues along $W_i$ until it reaches $H_{l+4}$ which it follows until it ends at $t_2$. Again $P_1$ and $P_2$ are edge-disjoint and form a cross on $T$. 
    \end{itemize}
Having discussed all the possible cases, this concludes the proof.

\end{proof}

The Type I and Type II jumps are degenerate in the sense that they do not directly yield a cross as in the proof above. It turns out however that they may easily be extended (and can in fact always be extended) to guarantee the existence of a cross by using jump-paths.

\begin{lemma}\label{lem:jumps_on_a_wall_yields_crosses}
  Let $\iota_1 = (x_{i,p}^-,x_{j,q}^+) \in M$ be a $3$-saturated jump of Type I or II for some $3 \leq i,j\leq t-3$ and $5 \leq p,q \leq 2t-5$. Let
  $\bar\iota  = (\iota_1, \iota_2, \iota_3)$ be the jump-sequence starting at 
  $\iota_1$ and let $P = P(\bar\iota)$ be the jump-path of the sequence. Then~$P$ induces a 
  local usable cross. Moreover, if neither $\iota_2$ nor $\iota_3$ is of Type 0, then the tile $T \coloneqq T^{[p-8, p+4]}_{[i-3, i+3]} \subseteq \WWW$ contains a usable cross.
\end{lemma}
\begin{proof}
  If $\iota_2$ or $\iota_3$ is of Type 0, then
  there is nothing to show as, by \cref{lem:jumps_give_crosses}, the Type $0$ jump alone induces a cross. Thus we may assume that $\iota_1, \iota_2, \iota_3$ are all of Type I or II (for they cannot be up-paths). In particular $\iota_2$ and $\iota_3$ exist and are pairwise distinct, as no sequence of Type I and II jumps can close a jump-cycle of length less than three.

  We first show that any combination of a Type I and a Type II jump yields a local usable cross. 
  Without loss of generality we assume that $p$ is odd; the case where $p$ is even is analogous using the symmetry of the tiles.
  
  Suppose $\iota_1 = (x_{i,p}^-, x_{i-1,p-2}^+)$ is of Type I
  and $\iota_2 = (x_{i-1,p-1}^-, x_{i-1,p-4}^+)$ is of Type II. Let $U_1 =
  (x_{i-1,p-2}^+, x_{i-1,p-1}^-)$ and $P = \iota_1 + U_1 + \iota_2$. Thus $P$ is a path
  from $x_{i,p}^-$ to $x_{i-1,p-4}^+$. Now let 
  \[
     P_1 \coloneqq (x_{i-3, p-8}^-, W_{i-3}, x^+_{i-3, p}, H_p, x_{i,p}^-, P, x_{i-1,p-4}^+, H_{p-4}, x^-_{i+3,p-4}, W_{i+3}, x^-_{i+3, p+4}).
   \]
   and let
   \[
   \begin{array}{rl}
    P_2 \coloneqq    &  (x^-_{i+3,p-8}, W_{i+3}, x^+_{i+3,p-7}, H_{p-7}, x^-_{i-2,p-7}, W_{i-2}, x^+_{i-2,p-2},H_{p-2}, x^-_{i+1, p-2}, \\[0.5em]
        & \hspace*{4cm} W_{i+1}, x^+_{i+1,p+3}, H_{p+3}, x^-_{i-3,p+3}, W_{i-3}, x^-_{i-3,p+4}).
   \end{array}
   \]
   See \cref{fig:crosses_consec_type_I_II_jumps:b} for an illustration, where the path $P_1$ is marked in red and the path $P_2$ is marked in green.
  Then $P_1$ and $P_2$ are edge-disjoint and form a cross on the corners of the tile $T$.

  The other combinations of a Type I and a Type II jump are completely analogous.

  \smallskip

One easily verifies that a jump-path resulting from two consecutive Type I jumps has the same endpoint as Type II jump. Thus if a jump-path using solely Type I jumps induces a cross, then so does a jump-path using solely Type II jumps, for all the wall edges used by the jump-path using solely Type II jumps are also used by the jump-path using solely Type I jumps; whence it suffices to prove the case that~$\iota_1,\iota_2$ and~$\iota_3$ are Type I jumps.

  Assume that $\iota_1, \iota_2$, and $\iota_3$
  are all of Type I. Without loss of generality assume that $p$ is odd and thus $\iota_1 = (x_{i,p}^-, x_{i-1,p-2}^+)$; the other case follows analogously using the symmetry of the tile.

  Let $P \coloneqq (x_{i,p}^-, \iota_1, x_{i-1,p-2}^+, W_{i-1}, x_{i-1,p-1}^-, \iota_2, x_{i,p-3}^+, W_{i}, x_{i,p-2}^-, \iota_3, x_{i-1,p-4}^+)$. Again, $P$ is a path from  $x_{i,p}^-$ to $x_{i-1,p-4}^+$. So we can define $P_1$ and $P_2$ as above and set 
  \[
     P_1 \coloneqq (x_{i-3, p-8}^-, W_{i-3}, x^+_{i-3, p}, H_p, x_{i,p}^-, P, x_{i-1,p-4}^+, H_{p-4}, x^-_{i+3,p-4}, W_{i+3}, x^-_{i+3, p+4}).
   \]
   and let
   \[
   \begin{array}{rl}
    P_2 \coloneqq    &  (x^-_{i+3,p-8}, W_{i+3}, x^+_{i+3,p-7}, H_{p-7}, x^-_{i-2,p-7}, W_{i-2}, x^+_{i-2,p-2},H_{p-2}, x^-_{i+1, p-2}, \\[0.5em]
        & \hspace*{4cm} W_{i+1}, x^+_{i+1,p+3}, H_{p+3}, x^-_{i-3,p+3}, W_{i-3}, x^-_{i-3,p+4}).
   \end{array}
   \]
   See \cref{fig:crosses_consec_type_I_II_jumps:a} for an illustration, where the path $P_1$ is marked in red and the path $P_2$ is marked in green.
  Then $P_1$ and $P_2$ are edge-disjoint and form a cross on the corners of the tile $T$.
\end{proof}

\begin{figure}
\centering
\begin{subfigure}{.45\textwidth}
  \centering
  
  \includegraphics[scale=0.7]{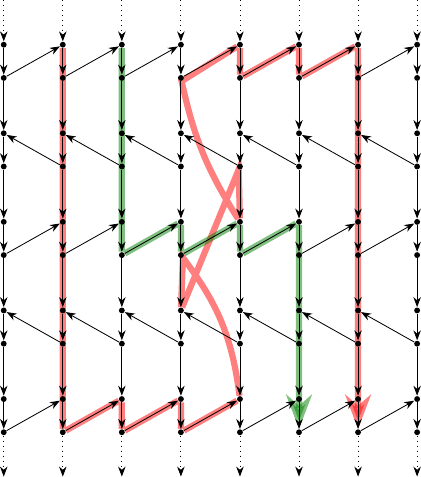}
  \caption{A cross using three consecutive Type I jumps.}
\label{fig:crosses_consec_type_I_II_jumps:a}
\end{subfigure}
\begin{subfigure}{.45\textwidth}
  \centering
    \includegraphics[scale=0.7]{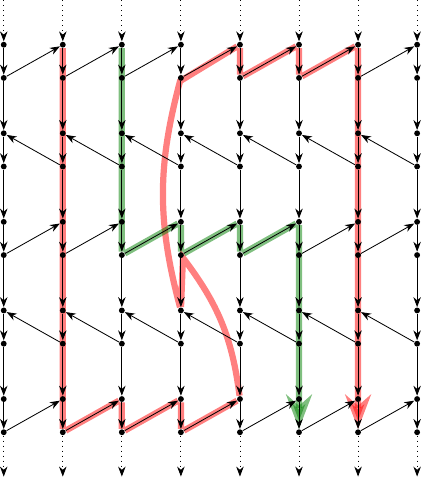}

  \caption{A cross using  Type I and Type II jumps.}
\label{fig:crosses_consec_type_I_II_jumps:b}
\end{subfigure}
\caption{Some examples on how to cross paths using a sequence of Type I and Type II jumps.}
\label{fig:crosses_consec_type_I_II_jumps}
\end{figure}

The previous two lemmas combined immediately imply the following.
\begin{corollary}\label{cor:jumps-induce-crosses}
    Let $\WWW$ be a $t_w \times t_h$-wall and let $M$ be a complete coordinate matching of $\WWW$.
    Let $e = (x^-_{i,j}, x^+_{p, q}) \in M$ be a jump in $\WWW$, for some~$3\leq i,j \leq t_w-3$ and~$9 \leq p,q \leq t_h-5$. Then the tile bounded by~$W_{\min(i,j)-2}$, $W_{\max(i,j)+2}$ and~$H_{\min(p,q)-8}$, $H_{\max(p,q)+4}$ contains a usable cross.
\end{corollary}

In the proof of \cref{thm:elementary-flat-swirl-theorem} we will use the previous result to construct large routers in various ways. To unify the presentation, we will reduce the  various cases to the existence of a particular structure in our graph that we call a \emph{cross column} which  will certify
the existence of a router.

\begin{definition}    \label{def:cross_column_general}
  Let $t, b \in \N$ and let $t_w \geq 2b+2$ and $t_h \geq (t+1)b$. 
  Let $\WWW \coloneqq (W_1, \dots, W_{t_w}, H_1, \dots, H_{2t_h})$ be a $t_w \times t_h$-wall. We say that $\WWW$ has
  a~$t$-\emph{cross column} with \emph{boundary $b$}, if there
  exist~$2b \leq p_1<\ldots< p_{t} \leq 2t_h-2b$
  where~$\Abs{p_{j} - p_{j+1}} \geq 2b$, such that there are pairwise
  edge-disjoint local usable crosses in
  the tiles $T_j \coloneqq T^{[b+1,t_w-b-1]}_{[p_j, p_{j+1}]}$, for
  $1 \leq j \leq t$, which, moreover, are edge-disjoint from the
  wall-cycles~$W_1,\ldots,W_{b},W_{t_w-b},\ldots,W_{t_w}$.
\end{definition}

We show next that every wall with a $t^2$-cross column and boundary $b$ contains a $t$-Router. To simplify the presentation we first prove a simpler version of the result which captures the construction. From this the result follows easily.

\begin{lemma}\label{lem:cross_column_yields_router_basecase:a}
  Let $h \coloneqq 2t^3$.
  Let $\WWW = (W_1, \dots,  W_{2t}, H_1, \dots,
  H_{2h})$ be a $(2t) \times h$-wall  and let $\PPP = (P_1, \dots,
  P_{t^2})$ be pairwise edge-disjoint paths which are also
  edge-disjoint from $\WWW$ such that $P_i$ links $x^-_{t+1, 2t\cdot (i-1)+1}$
  to $x^+_{t, 2t\cdot i}$. Then $\WWW \cup \PPP$ contains a $t$-Router.
\end{lemma}
\begin{proof}
  We will construct a~$t$-router by providing a procedure for threading $t$ paths in the wall using the paths $P_i$ to cross paths and finally close them to cycles. To this extent we will need some book-keeping on which paths we have already crossed---this is achieved via what we call a \emph{pattern}.
  
  A \emph{pattern} is any permutation $\pi \coloneqq (c_1, \dots, c_{t+1})$ of
  $\{ 1, \dots, t, \times\}$. The idea of a pattern is to store information on the~$t$ paths ultimately forming the~$t$ cycles of the router, where every~$c_i$ in the pattern may be thought of as one such cycle. The~$\times$ symbol in the pattern highlights which two cycles we want to `cross' next.
  
  We define a function $m$ which assigns to any pattern $\pi \coloneqq (c_1, \dots, c_i, \times, c_{i+1} \dots c_{t+1})$ and 
  $j \leq {t^2}$ the sequence
  \[
    m(\pi, j) \coloneqq (x_{1,
      2t}^-,\ldots, x^-_{i, 2t}, x^-_{t+1, 2t}, x^-_{t+(i+2), 2t}, \ldots, x^-_{2t, 2t})
  \]
  of vertices on row $2t$ of the wall, where~$x_{t+1,2t}^-$ is the starting point of a jump inducing a cross by the assumptions of the lemma.
  Note that $m(\pi, j)$ depends on $j$ and the position $i$ of the cross symbol $\times$ in $\pi$.

  We define the following two operations on patterns. 
  For $\pi \coloneqq (c_1, \dots, c_l, \times, c_{l+1}, \dots, c_t)$, with $l \geq 1$, 
  we define 
  $\textit{shift}(\pi) \coloneqq (c_1, \dots, c_{l-1}, \times, c_{l+1}, c_l,c_{l+2}, \dots, 
  c_t)$, switching~$c_l$ and~$c_{l+1}$---intuitively this symbolises the procedure of `crossing' the two cycles as we prove in the next claim.
  
  For $\pi \coloneqq (\times, r_1, \dots, r_t)$ 
  we define 
  $\textit{reset}(\pi) \coloneqq (r_1, \dots, r_{t-1}, \times, r_t)$.  

  \begin{claim}\label{lem:cross_column_yields_router_basecase:a:1}
    Let $\pi \coloneqq (c_1, \dots, c_l, \times, c_{l+1}, \dots, c_t)$, where
    $l \geq 1$, be a pattern and let $j < {t^2}$.  Let
    $(v_1, \dots v_t) \coloneqq m(\pi, j)$ and let
    $(u_1, \dots, u_t) \coloneqq m(\textit{shift}(\pi), j+1)$.  Then the tile
    $T_j = \WWW^{[2t(j-1)+1, 2tj+1]}_{[1,2t]}$ contains edge-disjoint paths
    $L_1, \dots, L_t$ such that $L_i$ links $v_i$ to $u_i$, for all
    $1 \leq i \leq r$, and $L_{l+1}$ and $L_l$ share a vertex.
  \end{claim}
  \begin{ClaimProof}
    Let $r \coloneqq 2t(j-1)+1$ and $r' \coloneqq 2tj+1$ so that the vertices of $\bar v$ are contained in row $r$ and the vertices of $\bar u$ are contained in row $r'$.
    For $i \neq l, (l+1)$ the vertices $v_i$ and $u_i$ are on the same wall cycle $W$ and we set $L_i$ to be the sub-path of $W$ from $v_i$ to $u_i$. For~$i = l,(l+1)$ we proceed as follows. 
    We define 
    \[
    \begin{array}{r@{\;}c@{\;}ll} 
      L_l &\coloneqq & x^-_{l, r} H_r x^-_{t+l+1,r} W_{t+l+1} x^-_{t+l+1,r'} & \text{and}\\
      L_{l+1}  & \coloneqq &  x^-_{t, r} P_j x^+_{t-1, r+3} W_{t-1} x^+_{t-1, r'} H_{r'} x^-_{t, r'}.   
    \end{array}
    \]
    Thus $L_1, \dots, L_t$ link $\bar v$ to $\bar u$ and $L_l$ and $L_{l+1}$ have the vertex $x^-_{t, r}$ but no edge in common.   
  \end{ClaimProof}
Informally speaking we have crossed the paths~$L_l$ and~$L_{l+1}$ where we `store'~$L_l$ in a wall-cycle of the right boundary and may continue with crossing~$L_{l+1}$ and~$L_{l-1}$ next.

The next claim follows easily from the routing properties of walls.

  \begin{claim}\label{lem:cross_column_yields_router_basecase:a:2}
    Let $1 \leq j \leq {t^2}$ and let $r \coloneqq 2t(j-1)+1$ and $r' \coloneqq 2tj+1$.
    Let $\pi \coloneqq(\times, c_1,  \dots, c_t)$, let $(v_1, \dots, v_t) = m(\pi, j)$, and
    let $(u_1, \dots, u_t) = m(\textit{reset}(\pi), j+1)$.  Then the tile
    $\WWW^{[r,r']}_{[1,2t]}$ contains edge-disjoint paths
    $L_1, \dots, L_t$ such that $L_i$ links $v_i$ to $u_i$, for all
    $1 \leq i \leq t$.
  \end{claim}  

  We define $\pi_1 = \pi_{t^2+1} \coloneqq (1, \dots, t-1, \times, t)$ and, for $1 \leq i < t$, $\pi_{i\cdot t +1} \coloneqq (t-i+1, \dots, t, 1, \dots, t-i-1, \times, t-i)$.
  For all $0 \leq i < t$ and $1 \leq j < t$ we define $\pi_{i\cdot t+1+j} \coloneqq \textit{shift}(\pi_{i\cdot t+j})$.

  Thus, we start with
  \[
    \begin{array}{rclrcll}
      \pi_1 & = & (&1, \dots, t-1, &\times,& t&),\\
      \pi_2 & = & (&1, \dots, t-2, &\times,& t, t-1&),\\
      \pi_3 & = & (&1, \dots, t-3, &\times,& t, t-2, t-1&),\\
      &&&&\vdots\\
      \pi_t & = & (&&\times,& t, 1, 2, \dots, t-1&),\\
      \pi_{t+1} & = & ( & t, t-1, t-2, \dots, 2, &\times,& 1 &).
    \end{array}
  \]
  Observe that
  $\pi_{j\cdot t+1} = \textit{reset}(\pi_{j\cdot t})$, for all
  $j < t$. Furthermore,  $\pi_{t\cdot t} = (\times, 1, 2, \dots, t)$ and
  therefore $\pi_1 = \textit{reset}(\pi_{t\cdot t})$.

  For $1 \leq i \leq t^2$ let $\bar{u}_{i}  =
  m(\pi_{i}, 4\cdot i)$.
  Then, by \cref{lem:cross_column_yields_router_basecase:a:1}
  and \cref{lem:cross_column_yields_router_basecase:a:2}, we obtain
  for each $1 \leq i \leq t^2$ a set $\LLL_i$ of pairwise edge-disjoint paths
  linking $\bar{u}_i$ to $\bar{u}_{i+1}$. The union $\bigcup \{\LLL_i \sth 1 \leq i \leq t^2\}$  yield a $t$-router as required. 
\end{proof}

\Cref{cor:cross_column_certificate} can be readily leveraged to the general setting of cross-columns reading as follows.

\begin{corollary}\label{cor:cross_column_certificate}
  Let~$\WWW$ be a~$t_w \times t_h$-wall with
  a~$t^2$-cross-column with boundary $t$. Then~$\WWW$ contains a~$t$-Router. 
\end{corollary}

\begin{figure}
    \centering
    \includegraphics[width=.9\linewidth, height=.45\linewidth]{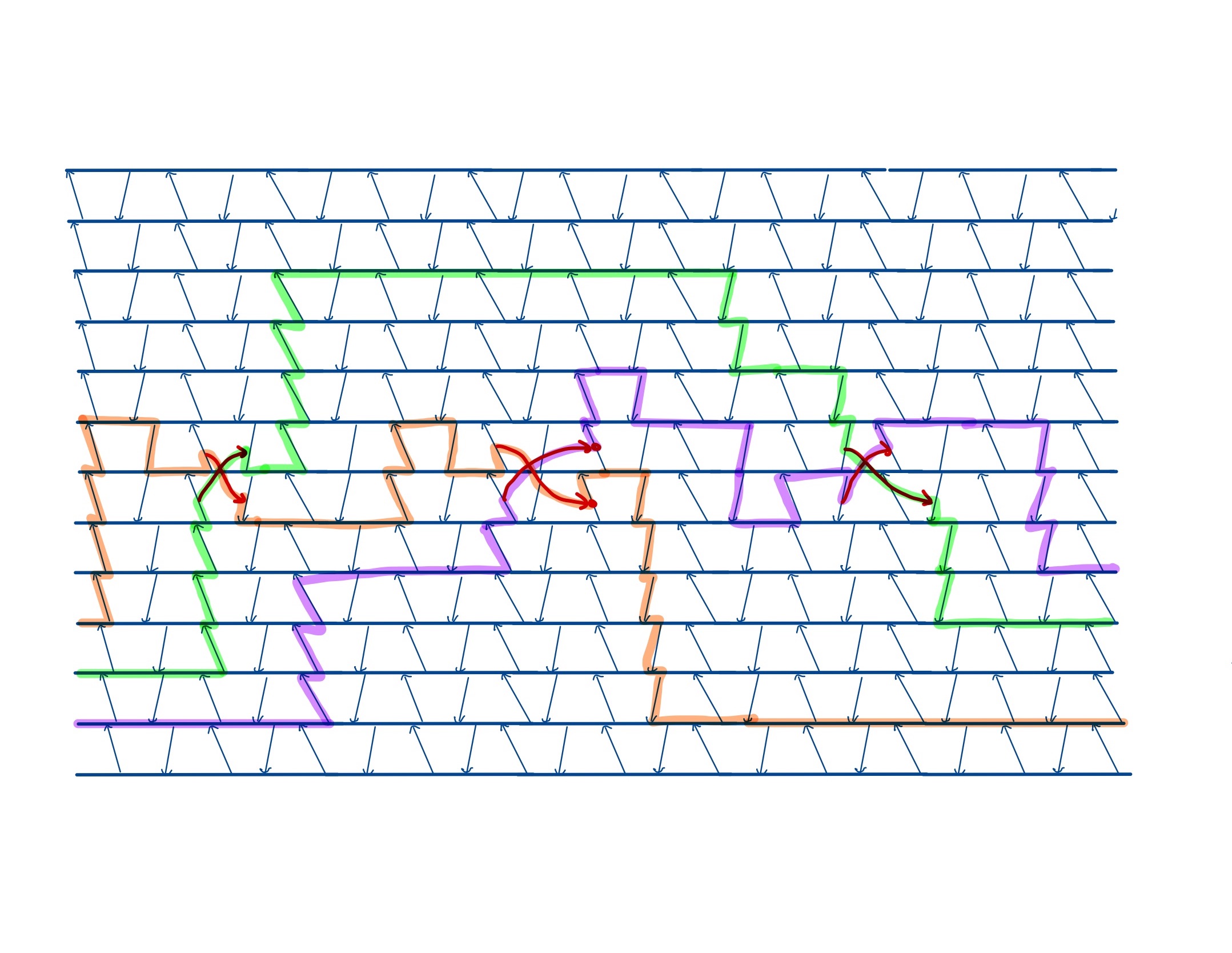}
    \caption{Constructing a router from a cross column, where the
      wall-cycles~$W_1,\ldots,W_t$ go from left to right, and the shift operation is highlighted by a red cross.}
    \label{fig:router_from_cross_column}
\end{figure}

See \cref{fig:router_from_cross_column} for a schematic representation on how to construct a router from a cross column.

In light of the results above we introduce what we mean by a \emph{tiling} of the wall. This will help to locate and book-keep crosses more easily, by assigning them to the tile they are `living in'.
The following definition helps in locating edge-disjoint usable crosses and keeping track of the tiles they are living in.

\begin{definition}\label{def:tiling}
    Let~$s=\rho \cdot t$ for some integers~$\rho,t > 0$. Let~$T$ be an~$s$-tile. Then the~\emph{$t$-tiling~$(T_{i,j})_{1\leq i,j \leq \rho}$ of~$T$} is a tuple of~$\rho^2$ disjoint~$t$-tiles~$T_{i,j} \subseteq T$ for~$1\leq i,j \leq \rho$ such that~$V(T) = \bigcup_{1\leq i,j \leq \rho} V(T_{i,j})$.

    Similarly any partition of~$T$ into edge-disjoint tiles~$(T_{i,j})_{1 \leq i,j \leq k}$ of some common size~$t_1 \times t_2$ for some~$k \in \N$ is called a~$(t_1\times t_2)$-tiling of~$T$.
\end{definition}
The concept of tiling will also be of importance in the following \cref{subsec:untangling_a_swirl,subsec:taming_a_swirl}.

\smallskip

We are almost ready to prove \cref{thm:elementary-flat-swirl-theorem}. We first need a simple variation of the famous Erd\H os-Szekeres theorem which relies on the following definition.

\begin{definition}\label{def:non-repetitive-sequences}
    Let $\bar a \coloneqq a_1, \dots, a_n$ be a sequence of numbers, not necessarily distinct. 
    
    A \emph{sub-sequence}  of $\bar a$ is a sequence $\bar b \coloneqq a_{i_1}, \dots, a_{i_l}$ with $i_1 < i_2 < \dots < i_l$.
    We say that $\bar b$ is \emph{monotone} if either $a_{i_j} < a_{i_{j'}}$ for all $1 \leq j < j' \leq l$ or $a_{i_j} > a_{i_{j'}}$ for all $1 \leq j < j' \leq l$.
    We say that $\bar b$ is \emph{non-repetitive} if $a_{i_j} \not= a_{i_{j''}}$ for all $1 \leq j \leq l$ and all $1 \leq j' \leq l$ with $j'\not= j-1$ and all $j' < j'' < j'+1$.
\end{definition}
\begin{lemma}\label{lem:es-with-repetition}
  Let $c \geq 1$ and let $\bar{a} \coloneqq a_1, \dots, a_n$ be a sequence of
  numbers such that no number appears more than $c$ times in
  $\bar{a}$.
  If $n \geq (s\cdot c^s-1)^2$, then $\bar{a}$ contains a non-repetitive,
  monotone 
  sub-sequence of length $s$.

\end{lemma}
\begin{proof}
  By the Erd\H os-Szekeres Theorem \cite[Theorem 4.4]{Jukna2001,ErdosS1935} or \cite{ErdosS1935} $\bar a$ contains a strictly
  monotone sub-sequence $\bar{b}_0 \coloneqq a_{i_1}, \dots, a_{i_{s\cdot c^s}}$ of length $s \cdot c^s$. We now construct a non-repetitive sub-sequence of $\bar b$  as
  follows. 
  If $a_{i_1}$ appears somewhere between elements of $\bar{b_0}$, then we can divide $\bar{b}_0-a_{i_1} \coloneqq a_{i_2}, \dots, a_{i_{s\cdot c^s}}$ into at most $c$
  sub-sequences which do not contain $a_{i_1}$ between any of its
  members. At least one of these must be of length $s\cdot c^{s-1}$. Applying induction on this sub-sequence yields a non-repetitive sequence to which we
  can add $a_{i_1}$ as prefix. 
\end{proof}

We continue with another definition helping us with book-keeping the jumps and the tiles they are related to.

\begin{definition}
    Let $T$ be a tile in $\WWW$.
    \begin{itemize}
    \item We define $M_{|T}$ as the set of pairs $(x^-_{c, r}, x^+_{c',r'})$ such that there is a jump-sequence $\bar{\iota} \coloneqq (\iota_1, \dots, \iota_i)$ such that $x^-_{c, r}$ is the tail of $\iota_1$, $x^+_{c',r'}$ is the head of $\iota_i$ and no other coordinate of $T$ is incident to a jump in $\bar{\iota}$.
    \item Let $\bar\iota = (\iota_1, \dots, \iota_i)$ be a jump-sequence such that the tail of $\iota_1$ and the head of $\iota_i$ are in $T$. 
    
    Let $v_1, \dots, v_{i'}$ be the sequence of coordinate vertices incident to jumps in $\bar\iota$ which are contained in $T$ in the order in which they appear on $\bar\iota$. 
    Then, for each $1 \leq j \leq \frac {i'}2$, $v_{2j-1} = x^-_{c_j, r_j}$, for some $c_j, r_j$ and $v_{2j} = x^+_{c_j, r_{j+1}}$. 
    
    The \emph{$T$-restriction} of $\bar\iota$ is defined as the sequence $\bar\iota_{|T} \coloneqq \big( (v_1, v_2), (v_3, v_4), \dots (v_{i'-1}, v_{i'})\big)$. 
    \end{itemize}
    We define the restrictions to a band and sub-wall analogously.
\end{definition}

Note that the $T$-restriction of $\bar\iota$ is a jump-sequence in $M_{|T}$.
The next result is a simple corollary of \cref{lem:no_jumps_is_swirl}. 
\begin{corollary}
    If there is a ~$(t+1) \times t$-tile $T$ in $\WWW$ such that $M_{|T}$ is complete  on $T$ and only contains up paths, then $T + M_{|T}$ induces a $t$-swirl.
\end{corollary}

We are now ready to prove the main theorem of this section.

\smallskip

\begin{proof}[Proof of \cref{thm:elementary-flat-swirl-theorem}]
Let $s \coloneqq 2t_2^2$, $l \coloneqq (s\cdot 3^s-1)^2$, and $b \coloneqq (t_2+4)$. 
   We define $g_{\ref{thm:elementary-flat-swirl-theorem}}(t_1, t_2) = t_h \coloneqq s \cdot h_1$ where 
   $h_1 \coloneqq 2(t_1 + o)$ and $o \coloneqq 8 + s\cdot l$. 
         Let $b' \coloneqq 4\cdot b \cdot t_h$.
    Finally, we define $g_{\ref{thm:elementary-flat-swirl-theorem}}(t_1, t_2) = w \coloneqq  
    (t_1+1)\cdot w_1$, where $w_1 \coloneqq w_2 + b'\cdot l$ and $w_2 \coloneqq  s\cdot (l\cdot (t_1+1)\cdot t_1)$.

    Let $\WWW' = (W_{-b+1}, \dots, W_0, W_1, \dots, W_{w}, W_{w+1}, \dots, W_{w+b}, H_1, \dots, H_{2t_h})$ be an elementary $(2b+w)\times t_h$-wall and let $M$ be a complete coordinate matching on $\WWW'$. Let $\WWW$ be the sub-wall induced by $W_1, \dots, W_w$. 
    Let $B$ be the set of coordinates of the cycles $W_{-(t_2+1)+1}, \dots, W_0, W_{w+1}, \dots, W_{w+(t_2+1)}$. We set $S_0 \coloneqq B$.
    
    By \cref{lem:no_jumps_is_swirl}, if there is a $(t_1+1)\times t_1$-tile $T \subseteq \WWW$ such that $M_{|T}$ is complete on $T$ and contains up-paths only, then we obtain a $t_1$-swirl and are done. Whence we may from now on assume that every $(t_1+1) \times t_1$-tile in $\WWW$ contains a jump in $M$.

    Let $I_0$ be the set of all $v = v^-_{i,j} \in \textit{coord}(\WWW)$ such that the jump at $v$ is not an up-path and the jump-sequence of length $l$ starting at $v$ is disjoint from $S_0$. Here, and below, we say that a jump-sequence $\bar\iota$ \emph{is disjoint from a set $S$ of coordinates} if no jump in $\bar\iota$ has an endpoint in $S$. 

    Observe that if $\WWW_i \coloneqq \WWW^{[i, i+t']}$ is a band of height $t' \geq t_1$ in $\WWW$ then $I_0$ contains at least~$w_2$ coordinates in~$\WWW_i$. For, each $\WWW_i$ has width $w$ and can therefore be partitioned into $w_1$ disjoint $(t_1+1)\times t_1$-tiles, i.e., we have a~$(t_1+1)\times t$-tiling~$(T_{i,j})_{1 \le1 i,j \leq w_1}$ of~$T$. As discussed above, each of the tiles~$T_{i,j}$ contains a vertex which is the tail of a jump in $M$. Of these, at most $l\cdot b'$ vertices can be on a jump-sequence of length $\leq l$ that contains a coordinate of $S_0$. Thus as least $w_2$ coordinates of $\WWW_i$ must be contained in $I_0$. 
    
    For every coordinate $v^-_{i,j} \in V(\WWW)$ we define the band 
    \[
       \WWW(v) \coloneqq \WWW^{[i-t_1-o, i+ t_1+o]}
    \]
    and the tile 
    \[
       T(v) \coloneqq \WWW^{[i-t_1-o, i+t_1+o]}_{[j-t_1-o, j+t_1+o]}.
    \]
    Set $J_0 = \emptyset$. We now construct sets $S_i, I_i, J_i$ inductively as follows, where~$S_i$ will keep track of regions that we want our jump-paths to avoid, i.e., it is a growing `safety zone', and~$I_i$ keeps track of vertices that can be the starting point of edge-disjoint paths inducing crosses without interfering too much with already found paths while~$J_i$ keeps track of the end-points of the~$i$ edge-disjoint paths inducing crosses we have found so far.
    
    Suppose $S_m, I_m, J_m$ have already been defined for some~$m \geq 0$. 
   
    If there is a vertex $v = v^-_{i,j} \in I_m$ such that there is a jump-sequence $\bar\iota$ starting at $v$ of length $l(v) \leq l$ which induces a cross on the band $\WWW(v)$, then choose such a $v = v^-_{i,j}$ which minimises $l(v)$. Here, using the minimality of $l(v)$, $\bar\iota$ induces a cross on $\WWW(v)$ if $\bar\iota$ starts with a Type 0 jump with both ends in $\WWW(v)$ or $\bar\iota$ starts with a Type I or a Type II jump $\iota_1$ and $\iota_1$ is $2$- or $3$-saturated in $M_{|\WWW(v)}$ and induces a cross as in \cref{lem:jumps_on_a_wall_yields_crosses}.
    By construction, $\bar\iota$ is a jump-sequence starting at $v^-_{i,j}$ and ending at a coordinate $v^+_{i', j'} \in \WWW(v)$ such that the edge $(v^-_{i,j}, v^+_{i',j'})$ is of Type $0$ (but it is not in $M$ unless $\iota_1$ is of Type $0$). 
    
    We define $J_{m+1} \coloneqq J_m \cup \{  (v^-_{i,j}, v^+_{i',j'}) \}$, where the edge~$(v^-_{i,j}, v^+_{i',j'})$ witnesses the new edge-disjoint path inducing a cross we found in this iteration. We define $S_{m+1} \coloneqq S_m \cup V(T(v^-_{i,j})) \cup V(T(v^+_{i',j'}))$, where we add the tiles containing the end-points of said path to our `safety zone' so that no new jump-paths will collide with the tile and thus possibly `kill' horizontal wall-cycles we would need for routing the crosses disjointly. 
    Furthermore, we define $I_{m+1} \coloneqq I_m \setminus
    \big(V(\WWW(v^-_{i,j})\cup \WWW(v^+_{i',j'})) \cup 
    \bigcup \{ V(T(v)) \mid v$ is a coordinate occurring on $P(\bar\iota) \} \big)$ with the intention that no jump-path we consider in a next iteration starts too close to any vertex of the jump-paths we already considered for again these paths could `kill' vertical paths needed for routing the crosses disjointly. All of this will be made precise shortly. Note that in this way we remove at most $(t_1+1)\cdot t \cdot l$ elements from $I_m$ which are not in one of the bands $\WWW(v^-_{i,j})$ or $\WWW(v^+_{i',j'})$. As the height $\WWW$ is $\geq s\cdot h_1$, after $m$ iterations we still find $s-m$ disjoint bands of height $h_1$ such that each band contains at least $(s-m)\cdot   ((t_1+1)\cdot t \cdot l)$ elements in $I_{m}$. 
    This allows us to continue the construction for at least $s$ iterations, if possible. 
    
    If there is no such vertex $v = v^-_{i,j} \in I_m$ as above or if $m=s$, the construction stops. 

    \begin{Claim}
        If $m=s$ then $G$ contains a $t_2$-router grasped by $\WWW'$.
    \end{Claim}
    \begin{ClaimProof}
        Suppose $J_s = \{ \iota_1, \dots, \iota_s\}$ with $\iota_r \coloneqq (v^-_{i_r, j_r}, v^+_{i'_r, j'_r})$, for all $1 \leq r \leq s$. By construction, $|j_r - j'_r| \leq t_1+o$. Let $l_r \coloneqq \max\{j_r, j'_r\}+o$ and $u_r \coloneqq \min\{j_r, j'_r\}-o$ and let $\WWW_r \coloneqq \WWW^{[u_r, l_r]}$. Furthermore, by construction, $\WWW_r \cap \WWW_{r'} = \emptyset$ for all $1 \leq r < r' \leq s$, where $\iota_r$ induces a cross in the band $\WWW_r$. 

        If each $\iota_r$ was an edge in $M$ then these crosses together with the boundary $B$ and the walls $\WWW_r$ would yield a wall with a $t^2$ cross column and boundary $b$ and thus, by \cref{cor:cross_column_certificate}, we would obtain a $t_2$-router as required. However, the edges $\iota_r$ only represent jump-sequences $\bar\iota_r$ of length $l_r \leq l$ and we need to show that the respective jump-path~$P(\bar\iota_r)$ does not interfere with a cross generated from another $\iota_{r'}$. 

        Fortunately, by construction, if $r < r'$ then $\bar\iota_{r'}$ does not start or end in a tile $T(v)$ for some coordinate $v$ appearing in $\bar\iota_r$, as these tiles are removed from the set $I_{r+1}$. Conversely, the jump-path $P(\bar\iota_{r'})$ does not use a coordinate from $T(v^-_{i_r, j_r}) \cup T(v^-_{i'_r, j'_r})$, as these were added to the `safety zone' $S_{r+1}$. Furthermore, the jump-paths $P_1, \dots, P_s$ in total use at most $2\cdot s\cdot l < o-8$ coordinates. Thus, in the band $\WWW_r$ there are still $8$ horizontal paths above $H_{u_r}$ and $8$ horizontal paths below $H_{l_r}$ disjoint from the jump-paths $P(\bar\iota_1, \dots, P(\bar\iota_s)$ which we can use to construct a cross on the corners of $\WWW_r$.
        
        Thus, the bands $\WWW_1, \dots,\WWW_s$ and the crosses they contain together with the boundary $B$ yield a wall with a $t^2$ cross column and boundary $b$. By \cref{cor:cross_column_certificate}, we obtain a $t_2$-router as required. 
    \end{ClaimProof}

    By the previous claim we may now assume that the construction stops after $m<s$ iterations. 
    Let $v^-_{i,j} \in I_m$ and consider the jump-sequence $\bar\iota$ of length $l$ starting at $v^-_{i,j}$. As the construction cannot be continued, this implies that no sub-sequence of $\bar\iota$ induces a cross on the band of its initial coordinate. By \cref{cor:jumps-induce-crosses}, this implies that for every $\iota_i$ in $\bar\iota$ either $\iota_i$ is a jump of Type I or II or the `vertical distance'---that is the distance between the rows~$H_i$ the end-points of the jump lie on with respect to that index---between its endpoints is bigger than $t_1+2o>8$. 
    Now suppose $\bar\iota$ contains a Type I or II jump $\iota_i$. Let $i$ be minimal such that $\iota_i$ is a Type I or II jump. Let $v^-_{c_i, r_i}$ be the tail of $\iota_i$. Then the jump-sequence $\bar\iota$ can contain at most  one other jump starting in the rows between $r_i$ and $r_i-4$ as otherwise,  by \cref{cor:jumps-induce-crosses}, this would induce a local cross. 

    With this in my mind, let $\bar\iota = (\iota_1, \dots, \iota_l)$ and let $r'_1, \dots, r'_l$ be the row indices of the tails of $\iota_1, \dots, \iota_l$. If there are $i_1 < i_2$ such that $|r'_{i_1}- r'_{i_2}| \leq 8$ then we set $r'_{i_2} = r'_{i_1}$. 
    Let $\bar r = r_1, \dots, r_l$ be the resulting sequence. By construction---there is no vertex~$v \in I_m$ left such that there is a jump-sequence starting t~$v$ of length~$\leq l$ that induces a cross on the band~$\WWW(v)$---and using our analysis regarding the Type I and Type II jumps above, no number in  $\bar r$ repeats more than $3$ times. Thus, by \cref{lem:es-with-repetition}, there is a non-repetitive monotone sub-sequence $r_{i_1}, \dots, r_{i_s}$ of length $s$, where $i_1 < \dots < i_s$. 
    Without loss of generality we may assume that $r_{i_1} < r_{i_2} < \dots < r_{r_s}$. Let $v^-_{c_{i_j}, h_{i_j}}$ be the tail of $\iota_{i_j}$, for all $1 \leq j \leq s$. For $1 \leq j \leq \frac12s = t_2^2$, let  $P_j(\iota_{i_{2j-1}} \iota_{i_{2j-1}+1} \dots \iota_{i_{2j}})$ be the jump-path of the sequence of jumps in $\bar\iota$ between $\iota_{i_{2j-1}}$ and $\iota_{i_{2j}}$. $P_j$ starts at $v^-_{c_{i_{2j-1}}, h_{i_{2j-1}}}$ and ends at $v^-_{c_{i_{2j}}, h_{i_{2j}}}$. Furthermore, $\WWW(v^-_{c_{i_{2j-1}}, h_{i_{2j-1}}})$ is disjoint from all other jump-paths. Thus, the bands $W^{[h_{i_{2j-1}}+8, h_{i_{2j-1}}-8]}$ each contain a usable cross.
    This yields a wall with a $t_2^2$ cross column and boundary $b$ and thus a $t_2$-router.
\end{proof}

In the following section we will lift the ideas presented in this section to the general setting, given a general wall---which is an immersion of an elementary wall---as well as a complete set of coordinate paths---which is an immersion of the complete set of coordinates. The general result is easily derived by splitting off to the \emph{core} being the elementary wall and the coordinate matching, at the price of losing embedding assumptions since immersion may behave very wild, leaving us with an untamed swirl.

\subsubsection{Swirls and Routers in General Walls}
\label{subsec:finding_the_swirl}
As mentioned above, in this section we leverage \cref{thm:elementary-flat-swirl-theorem} to prove the main theorem of \cref{subsec:finding_the_swirl}. Namely, we prove that given a large wall~$\WWW$ in an Eulerian graph~$G$ we either find a large (tangled) swirl or a large router. Although we are now finding the structures in a more general setting there is a slight drawback: we lose the flatness of the swirl, and even worse, the swirl may be \emph{tangled}.

\begin{definition}\label{def:tangled_swirl}
    Let~$G$ be some Eulerian graph and let~$S_1,\ldots,S_s \subset G$ be a collection of embedded edge-disjoint cycles such that~$\SSS \coloneqq S_1\cup \ldots \cup S_s$ is a plane graph. Further assume that there exists an~$s$-swirl~$\SSS'$ such that there is an immersion~$\gamma:\SSS' \hookrightarrow \SSS$. Then we call~$\SSS$ a \emph{tangled}~$s$-swirl.

    We say that~$\SSS$ is \emph{induced} by an~$s$-tile~$T \subset \WWW$ of some plane wall~$\WWW$, if~$T \subset \SSS$.
\end{definition}
\begin{remark}
    The difference of a tangled swirl to swirls in general is that all the cycles may intersect in common vertices, whereas for normal swirls any swirl cycle is `trapped' between its two neighbouring cycles. Note further that for tangled swirls the cycles may not be concentric anymore, another difference to standard swirls as given by \cref{def:swirl}. We still chose to call them swirls, for we will get rid of these pathologies in the next section, and we want the reader to think of them as being swirls up to the fact that the up-paths we found may destroy the assumptions on the cycles being concentric and disjoint up-to neighbouring cycles. So in a sense they are `almost swirls' when focusing on~$\SSS \cap T$ which is neglecting the up-paths.
\end{remark}
\begin{figure}
    \centering
\includegraphics[width=.4\linewidth,height=.4\linewidth]{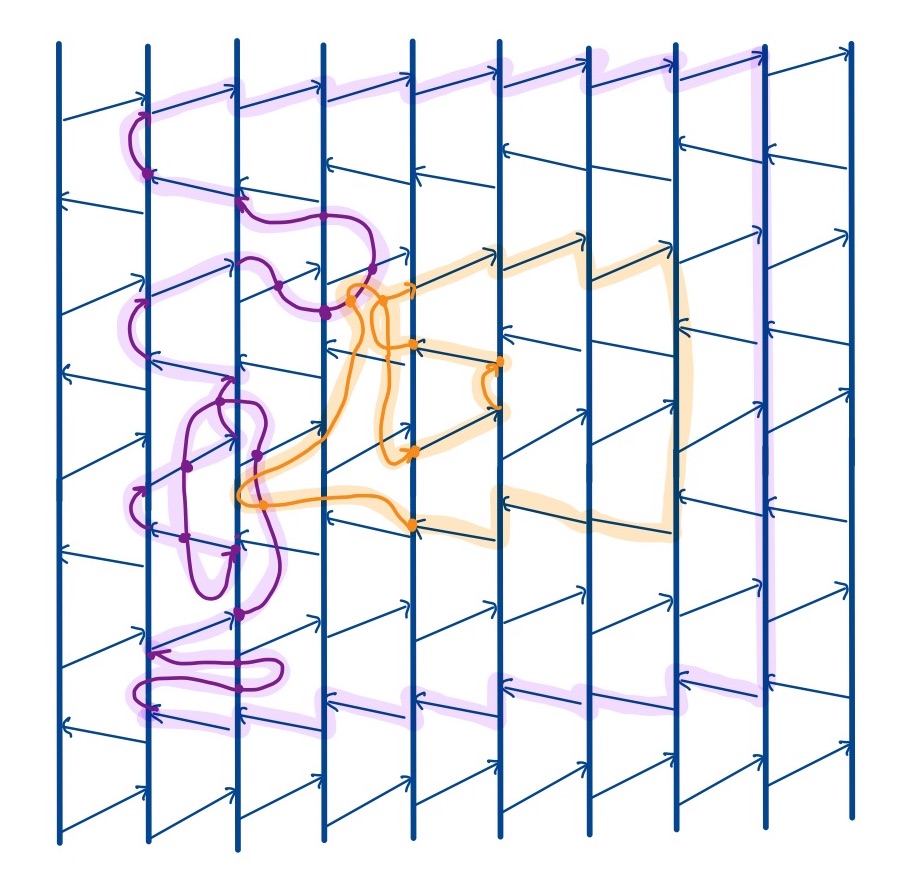}
    \caption{An example of two swirl cycles with tangled up-paths that may yield usable crosses in the tile.}
    \label{fig:tangled_uppaths}
\end{figure}

The key observation is that for tangled swirls we lose some restrictive embedding properties imposed on swirls. This comes with a cost, for tangled swirls may still contain wall-local usable crosses, a nuisance we will deal with in the next section; see \cref{fig:tangled_uppaths} for an illustration. Note that the definition of tangled swirls is rather loose, for technically any set of pairwise edge-disjoint cycles (even vertex disjoint cycles) is a tangled swirl; we will mainly work with \emph{induced} tangled swirls imposing more structure.

The main theorem now reads as follows.

\begin{theorem}
    Let~$G+D$ be an Eulerian graph with~$\Abs{D} = p \in \N$. Then there exist maps~$f(p)$ and~$g(p)$ satisfying the following. Let~$\WWW \subseteq G$ be a cylindrical wall of size~$f(p)\times f(p)$ with~$\WWW \cap V(D) = \emptyset$. Then at least one of the following holds.
       \begin{itemize}
           \item[(i)] $G$ contains a tangled~$g(p)$-swirl induced by a~$g(p)$-tile that is edge-disjoint from~$D$, or
           \item[(ii)] $G$ contains a~$g(p)$-router grasped by~$\WWW$ that is edge-disjoint from~$D$.
       \end{itemize}
       In particular we can find a tangled swirl satisfying~$(i)$ or a router satisfying~$(ii)$ in~$fpt$-time on~$p$.
   \label{thm:swirl_theorem_tangled}
\end{theorem}

The proof of \cref{thm:swirl_theorem_tangled} is a rather straightforward consequence 
of \cref{thm:elementary-flat-swirl-theorem}. To see this we introduce an operation we call the \emph{Eulerian closure}. Recall the \cref{def:coordinate_paths} of complete sets of coordinate-paths and the fact that such a complete set~$\PPP$ induces a complete coordinate matching~$M(\PPP)$.

\begin{definition}[Skeleton of a wall and Eulerian Closure] \label{def:eulerian_closure}
  Let~$\WWW$ be a cylindrical wall in~$G+D$
  with~$V(\WWW) \cap V(D) = \emptyset$. Let $\PPP$ be a complete set of coordinate 
  paths for $\WWW$.
  \begin{itemize}
  \item The \emph{skeleton} of $\WWW$ is the digraph on the set of coordinate vertices of $\WWW$ which contains an edge $(u, v)$ between two
    coordinate vertices if $\WWW$ contains a path from $u$ to $v$ whose internal vertices all have degree $2$ in $\WWW$.
    \item Let $\WWW'$ be the skeleton of $\WWW$. We define the \emph{Eulerian closure of~$\WWW$ with respect
      to~$\mathcal{P}$}, denoted by $\operatorname{eCl}(\WWW;\mathcal{P})$, as the graph with vertex set $V(\WWW')$ and edge set $E(\WWW') \cup M(\PPP)$.
  \end{itemize}
\end{definition}
\begin{remark}
    Note that the skeleton of~$\WWW$ is an \emph{elementary} wall and~$M(\PPP)$ is a complete coordinate matching which propels us to the setting of \cref{thm:elementary-flat-swirl-theorem}.
\end{remark}

The following is obvious.
\begin{lemma}\label{lem:compute-wall-closure}
    Let $\WWW$ be a cylindrical wall in $G+D$ and let $\PPP$ be a complete set of coordinate-paths. Then
    $\operatorname{eCl}(\WWW;\mathcal{P})$ can be obtained from $\WWW \cup \bigcup\PPP$ by splitting off at vertices and deleting remaining isolated vertices. Furthermore, $\operatorname{eCl}(\WWW;\mathcal{P})$ can be computed in polynomial time given $\WWW$ and $\PPP$ as input. 
\end{lemma}
\begin{proof}
    Recall that the paths in $\PPP$ are pairwise edge disjoint and also edge-disjoint from $\WWW$. Let $P \in \PPP$ be a path starting at $u$ and ending at $v$. Then $(u, v)$ is an edge in $M(\PPP)$ and it can be obtained by repeatedly splitting off at the internal vertices of $P$.
    Likewise we can obtain the skeleton of $\WWW$ by repeatedly splitting off at vertices of degree $2$ in $\WWW$ respecting the wall-paths. Then we split off at all vertices in~$V(G)\setminus V(\WWW \cup \bigcup \PPP)$. As a last step we remove all isolated vertices generated by this process. 

    It is easily seen that this procedure can be implemented in polynomial time. 
\end{proof}

\cref{lem:compute-wall-closure} is the main ingredient to the proof of \cref{thm:swirl_theorem_tangled}. To see this recall that splitting off at vertices keeps a graph Eulerian by \cref{obs:splitting_off_at_vertex_remains_Eulerian} and splitting off can efficiently be reversed (see \cite{Frank1988} for more
details) which is all we need for the proof as highlighted by the following.

\begin{lemma}\label{obs:from_eCl_to_swirls_and_routers_in_the_graph}
    Let~$\WWW$ be wall in~$G+D$ and let~$\mathcal{P}$ be a complete set of coordinate-paths. 
    \begin{enumerate}
        \item If~$\operatorname{eCl}(\WWW;\mathcal{P})$ contains a~$t$-Router grasped by~$\WWW'$, then~$G$ contains a~$(t-p)$-Router grasped by~$\WWW$ that is edge-disjoint from~$D$.
        \item If~$\operatorname{eCl}(\WWW;\mathcal{P})$ contains a~$t$-swirl induced by~$\WWW'$, then~$G$ contains a tangled~$\lceil\frac{t}{2^p}\rceil$-swirl induced by~$\WWW$ that is edge-disjoint from~$D$.
    \end{enumerate}
\end{lemma}
\begin{proof}
The proof follows by reversing the splitting off and noting that the resulting structures and paths remain edge-disjoint, more precisely, immersions. Then finally we remove the cycles using edges in~$D$ from the respective structures yielding the claim. Note that for the swirl, since it is an embedded structure, in order for it to be induced by some tile we need the whole tile to be part of the swirl and thus deleting a single swirl-cycle results in two smaller swirls where at most one has at least half its original size.
\end{proof}

In particular may extend \cref{def:canonical_swirl_induced_by_tile} to the general setting where we are given a cylindrical wall and a complete set of coordinate-paths. That is the \emph{canonical swirl induced by a tile~$T$} is obtained by taking the Eulerian closure, using \cref{def:canonical_swirl_induced_by_tile} and then reversing the splitting off procedure. Clearly said swirl contains the whole tile~$T$. We will without further ado talk about canonical swirls induced by tiles in the general setting meaning the obvious. In the same way we can extend the definition of \emph{sub-swirls induced by sub-tiles} and thus we will abuse notation and lift those definitions to the general setting.

\smallskip

As a last result before tackling the proof of \cref{thm:swirl_theorem_tangled} we give a generalisation of \cref{lem:no_jumps_is_swirl}; we say that a path~$P \in \PPP$ is an \emph{up-path} if the respective edge~$\iota_p \in M(\PPP)$ is an up-path in the skeleton.

\begin{corollary}\label{lem:no_jumps_is_swirl_tangled}
    Let~$G$ be Eulerian, let~$\WWW \subset G$ be some plane cylindrical wall and let~$\PPP$ be a complete set of coordinate-paths. Let~$T \subset \WWW$ be a~$t$-tile in~$\WWW$ for some~$t \in \N$. Let~$\PPP' \subset \PPP$ be the set of coordinate-paths with at least one endpoint in~$T$. If~$M(\PPP')$ is complete for~$T$ and every path~$P \in \PPP'$ is an up-path, then~$T \cup \bigcup \PPP'$ is a plane tangled swirl induced by~$T$.
\end{corollary}
\begin{proof}
    The proof follows at once by switching to~$\operatorname{eCl}(\WWW,\PPP) = \WWW' + M(\PPP)$ where we find a respective~$t$-tile~$T'$ such that~$M(\PPP')$ is a complete coordinate matching for~$T'$ and every~$\iota \in M(\PPP')$ is an up-path. By \cref{lem:no_jumps_is_swirl} we derive that~$T' + M(\PPP')$ is a swirl induced by~$\WWW$ and reversing the splitting off we derive that~$T \cup \bigcup \PPP$ is a tangled swirl. 
\end{proof}

In accordance with \cref{def:canonical_swirl_induced_by_tile} we will extend the definition to the general case and call the~$t$-swirl obtained in \cref{lem:no_jumps_is_swirl_gen} following the construction provided in the proof of \cref{lem:no_jumps_is_swirl} a \emph{tangled canonical swirl (induced by a~$t$-tile~$T$)}. 

Finally we provide the proof of \cref{thm:swirl_theorem_tangled}.

\begin{proof}[Proof of \cref{thm:swirl_theorem_tangled}]
    The existence of~$f$ and~$g$ follows at once from \cref{thm:elementary-flat-swirl-theorem} and \cref{obs:from_eCl_to_swirls_and_routers_in_the_graph}; recall that the canonical swirl is induced by a tile. To find the router or swirl in~$fpt$-time first find a directed cylindrical wall~$\WWW$ in~$fpt$ time which is possible by \cref{thm:dir_wall_away_from_D}. Then determine a complete set of coordinate-paths~$\mathcal{P}$ which works in polynomial time as seen in \cref{lem:compute_coord_paths}. Then build an Eulerian closure~$\WWW^{\mathcal{P}} \in \operatorname{eCl}(\WWW,\mathcal{P})$ with respect to~$\PPP$ which again works in~$fpt$-time by \cref{lem:compute-wall-closure}. The Eulerian closure now satisfies~$\Abs{E(\WWW^{\mathcal{P}})},\Abs{V(\WWW^{\mathcal{P}})} \in \mathcal{O}(p^{\tilde{c}})$ for some~$\tilde{c} \in \N$. In particular we can check for the routers (or cross-columns) and swirls in the Eulerian closure by brute-force. Use \cref{obs:from_eCl_to_swirls_and_routers_in_the_graph} and the fact that reversing the splitting off works in polynomial time to conclude the proof. 
\end{proof}
\begin{remark}
  It is crucial to note that although \cref{lem:no_jumps_is_swirl}
  yields a swirl in the Eulerian closure (and thus a swirl in the
  graph), when analysing the up-paths in the original graph, there may
  still be paths that intersect many parts of the swirl---for example
  wall-cycles and even other up-paths; see \cref{fig:tangled_uppaths}.
  This is why in this section, using
  \cref{obs:from_eCl_to_swirls_and_routers_in_the_graph}, we only
  find a \emph{tangled} swirl: a swirl without jumps starting and
  ending in coordinates but that can still contribute wall-local usable crosses coming from its up-paths. 
\end{remark}

In a next step we will `untangle' the up-paths of the tangled swirl found in \cref{thm:swirl_theorem_tangled} with the aim to regain the desired embedding properties of swirls imposed by \cref{def:swirl}. In particular we prove that we can find a \emph{cross-less} swirl---note that all the paths but the up-paths were already nicely embedded as we started with a plane wall---that in turn is a swirl as given by \cref{def:swirl}.

\subsection{Untangling a swirl}
\label{subsec:untangling_a_swirl}
In this section we assume a large tangled canonical swirl~$\SSS$ induced by some tile~$T$ in a wall~$\WWW$ to be given. The goal of this sub-section is to make sure the swirl can be (almost) Eulerian-embedded, i.e., we want to make sure that there is no wall-usable crosses left in the graph~$\WWW \cup \SSS$, say,  which may exist due to tangled-up up-paths as mentioned in the previous section. This leads to the following definition.

\begin{definition}[Cross-less swirl]
    Let~$T \subseteq \WWW$ be an~$s$-tile inducing a tangled canonical~$s$-swirl~$\SSS$. Then we say that~$\SSS$ is \emph{cross-less} if there exists no wall-local usable~$T$-cross using solely edges of~$T \cup \SSS$.
\end{definition}

Note that the swirl~$\SSS$ found in \cref{thm:swirl_theorem_tangled} does not admit any~$T$-cross \emph{after} splitting off the up-paths by \cref{lem:no_jumps_is_swirl} and the way we found it in the first place. The up-paths may however still be used to yield wall-local usable~$T$-crosses inside that swirl (see \cref{fig:tangled_uppaths}). We will untangle those paths, by proving that there cannot be many sub-tiles~$T'\subseteq T$ such that the canonical tangled swirls they induce have up-paths that yield~$T'$-crosses. That is we prove the following main theorem of this section---a refinement of \cref{thm:swirl_theorem_tangled}.

\begin{theorem}\label{thm:swirl_theorem}
    Let~$G+D$ be an Eulerian graph and let~$p \coloneqq \Abs{D}$ and let~$g(p)$ be some function. Then there exists a function~$f(p)$ satisfying the following. Let~$\WWW \subseteq G$ be a cylindrical wall of size~$f(p)\times f(p)$ with~$\WWW \cap V(D) = \emptyset$. Then at least one of the following holds true:
    \begin{itemize}
        \item[(i)] $G$ contains a cross-less canonical~$g(p)$-swirl that is edge-disjoint from~$D$,or
         \item[(ii)] $G$ contains a~$g(p)$-router grasped by~$\WWW$ that is edge-disjoint from~$D$.
    \end{itemize}
    In particular we can find a swirl satisfying~$(i)$ or a router satisfying~$(ii)$ in~$fpt$-time on~$p$.
\end{theorem}
\begin{proof}
    Let~$g(p) = g_{\ref{thm:swirl_theorem_tangled}}(p)^5$ and let~$f(p)$ be the respective~$f_{\ref{thm:swirl_theorem_tangled}}(p)$ from \cref{thm:swirl_theorem_tangled}. We claim that~$f$ satisfies the theorem. To this extent let~$\SSS$ be a tangled canonical~$g(p)$-swirl induced by some~$g(p)$-tile~$T$ as given by \cref{thm:swirl_theorem_tangled}. For convenience write~$s = g_{\ref{thm:swirl_theorem_tangled}}(p)$. Let~$\TTT=(T_{i,j})_{1 \leq i,j\leq s^4}$ be an~$s$-tiling of~$T$---i.e., each~$T_{i,j}$is an~$s$-tile as given by \cref{def:tiling}---then for every~$1 \leq i,j \leq s^4$, each~$T_{i,j}$ induces a canonical tangled~$s$-swirl~$\SSS_{i,j}$ using \cref{lem:no_jumps_is_swirl_tangled}. Assume none of the swirls~$\SSS_{i,j}$ is cross-less. Then, by definition each tile~$T_{i,j}$ contains some~$T_{i,j}$-usable cross using solely paths in~$T_{i,j} \cup \SSS_{i,j}$. But then all these crosses are pairwise edge-disjoint and thus we find an~$s^4$ cross-column (with boundary~$s$) which together with \cref{cor:cross_column_certificate} concludes that there exists an~$s$-router; thus assume that there is a cross-less~$s$-swirl~$\SSS=S_1\cup\ldots \cup S_s$ induced by some~$s$-tile~$T \subset \WWW$. One easily verifies that~$\SSS$ adheres to the embedding properties of swirls given by \cref{def:swirl}, i.e., the graph can be planar embedded where all the cycles are concentric and only consecutive cycles of the swirl may intersect. The proof of this is slightly tedious but straightforward without yielding much new insight for it uses the same arguments as provided in \cref{subsubsec:elementary_wall_swirl}, more precisely a slight variation of \cref{lem:jumps_give_crosses} and \cref{cor:jumps-induce-crosses} where the jumps we analyse are given by coordinate-paths (so they are immersions of jumps ) admitting the same `behaviour' as general jumps. That is, we may construct crosses in the exact same way as we did in \cref{lem:jumps_give_crosses} and \cref{cor:jumps-induce-crosses}. In a nutshell we start with partly embedded swirl-cycles~$S_1,\ldots,S_s$ such that~$S_i \cap T$ yields a `nested embedded sequence of paths' in the obvious way, i.e.,~$S_1 \cap T$ lies at the the centre of~$T$ and~$S_s\cap T$ at the perimeter (so~$S_s$ will be the outer-cycle of the swirl); take Eulerian closure of~$T\cup \PPP$, look at the embedding of the resulting swirl given by \cref{thm:elementary-flat-swirl-theorem} and reverse the splitting keeping most of the embedding in tact where only the up-paths may induce crosses. Then we claim that no up-path~$P$ used in~$S_i$ can intersect~$S_j$ for~$1 \leq i < i+1 <j \leq s$ and neither can it intersect any other edge outside of the~$3$-tile the ends of the path are in, for this would yield a wall-local usable cross in the same way as given by \cref{cor:jumps-induce-crosses}. Finally then the up-paths are bounded between the neighbouring swirl-cycles and thus the swirl-cycles are concentric and they may only intersect neighbouring cycles, i.e., they form a swirl.
\end{proof}

We derive a final generalisation of \cref{lem:no_jumps_is_swirl} which is of its own interest.

\begin{corollary}
    \label{lem:no_jumps_is_swirl_gen}
    Let~$G$ be Eulerian, let~$\WWW \subset G$ be some plane cylindrical wall and let~$\PPP$ be a complete set of coordinate-paths. Let~$T \subset \WWW$ be a~$t$-tile in~$\WWW$ for some~$t \in \N$. Let~$\PPP' \subset \PPP$ be the set of coordinate-paths with at least one endpoint in~$T$. If every path~$P \in \PPP'$ is an up-path, and if the canonical tangled swirl~$T \cup \PPP'$ is cross-less, then~$T \cup \bigcup \PPP'$ is a plane swirl induced by~$T$.
\end{corollary}

In accordance with \cref{def:canonical_swirl_induced_by_tile} we will extend the definition to the general case and call the~$t$-swirl obtained in \cref{lem:no_jumps_is_swirl_gen} following the construction provided in the proof of \cref{lem:no_jumps_is_swirl} a \emph{canonical (cross-less) swirl (induced by a~$t$-tile~$T$)}. Note here that swirls induced by tiles are cross-less swirls which follows at ones form their embedding restrictions as imposed by their \cref{def:swirl}.

In a next step we will analyse the structure of an Eulerian digraph containing a large induced cross-less swirl but no large routers. That is, we are left to analyse whether there exist any other paths that do neither start nor end in coordinate vertex of the underlying tile of the swirl that may lead to usable crosses. Put differently, we are left to prove that a cross-less swirl in~$G$ either contains many swirl-usable crosses that yield a router, or it contains a tile inducing a flat swirl. This leads to what we call the \emph{Flat-Swirl Theorem}.

\subsection{Taming an untangled swirl}
\label{subsec:taming_a_swirl}
In this subsection we prove that, given a large cross-less swirl induced (and thus grasped) by a wall in a graph without an~$f(p)$-router, we find a large \emph{flat} swirl as in \cref{def:flat_swirl}. Formulated differently, if there is too many well-spread jumps on a swirl we can find a router grasped by the swirl (more accurately by the underlying wall). Note that flat here really means that there are no jumps at all, for any jump starting at a coordinate yields a cross as we have seen in \cref{lem:jumps_on_a_wall_yields_crosses}, and we will show something similar for paths starting and ending in the swirl with both end-points `far apart' (not as far as one may think) with respect to the coordinates of the swirl. To do so we will adapt the definition of jump to the setting of swirls in what is to follow. However, even after flattening the swirl~$\SSS$ it may still not be Euler-embeddable together with its \emph{attachments}---the components of~$G-C$ that are not disjoint from~$\SSS$ where~$C$ denotes its outer-cycle---for the attachments may be loosely connected to the swirl (removable by~$\leq 4$-cuts) and thus contain crosses that cannot be readily transferred to swirl-usable crosses; see \cref{fig:not_usable_cross}.

Note that, when talking about usable crosses in this section, we always mean \emph{swirl}-usable crosses unless stated otherwise; see \cref{def:swirl_usable_cross}. Also we will not write \emph{cross-less} swirl and tacitly assume the swirl we work with to be cross-less; in particular every swirl-usable cross must use some edge disjoint from the swirl. Further, we assume some meaningful labelling of the swirl to be given, that is, we use the induced labelling of the wall it is grasped by. Since we may restrict our attention to canonical swirls induced by tiles the definition of the coordinates is straight forward.

Note that the \emph{Type I and II} jumps from the previous section (see \cref{def:jumps_on_wall}) would indeed yield usable crosses on a swirl as they are standard jumps after a rotation of the swirl which is symmetric under rotation. That is: upwards and downwards make no difference on a swirl. Of course there is no in- or out-degree left at coordinate vertices hence there are no Type I or II jumps left, but there could be jumps starting shortly before or after those coordinate vertices. We mention this to highlight that some of the corner cases for crosses on a wall become obsolete when looking at swirls which behave rather similar to bi-directed grids and thus the nauseous pathological cases for non-usable crosses on directed cylindrical walls due to the fixed direction of the wall-cycles disappear. 

\paragraph{The general setting:} The majority of the arguments in this section are very similar in flavor to the arguments given in \cref{subsec:finding_the_swirl} but slightly smoother, for we do not need to care about directions anymore. Throughout this section we will assume the following setting: We are given an Eulerian graph~$G$ and an~$f(p)\times f(p)$-wall~$\WWW \subseteq G$ as well as a large canonical swirl~$\SSS$ in~$G$ induced by some tile~$T\subseteq \WWW$ of the wall. We proceed with the analysis of jumps with respect to such an embedded swirl, that is, from here on we assume a plane canonical~$t$-swirl~$\mathcal{S}$ \emph{induced} by a~$t$-tile~$T \subseteq \WWW$ to be given for some~$t \in \N$. Recall that~$T \subseteq \SSS$ and the underlying undirected graph of the tile~$T$ is itself an undirected wall (not necessarily a cylindrical wall). Hence we may designate coordinate-vertices to~$V(T) \subseteq V(\mathcal{S})$ in the obvious way using the one from the wall.

 Similarly we get an obvious definition of a metric---induced by the maximum-norm---on~$\mathcal{S}$ with respect to the coordinate vertices; we refer to it as~$\operatorname{dist}_{\mathcal{S}}(\cdot, \cdot)$.

\begin{definition}[Distance on a swirl]
    Let~$x_{i,p},y_{j,q} \in \mathcal{S}$ for~$1\leq i,j \leq t$ and~$1\leq p,q\leq 2t$. Then~$\operatorname{dist}_{\mathcal{S}}(x_{i,p},y_{j,q}) \coloneqq \max(\Abs{i-j},\Abs{p-q})$.
\end{definition}

We continue with defining jumps on a swirl similar to the \cref{def:jumps_on_wall} of jumps on a wall. The main difference being that the jumps we have looked at in \cref{subsec:finding_the_swirl,subsec:untangling_a_swirl} were jumps between coordinate vertices. In our new setting there is no coordinate vertex left that could have such a jump, that is, every coordinate vertex has degree four in~$\mathcal{S}$ and in particular, since~$\SSS$ is an induced swirl, each such vertex is the start of an up-path resulting in the swirl structure that does not yield a cross in the swirl for the swirl is assumed to be cross-less. Thus what we are really interested in is~$V(\mathcal{S})$-paths disjoint from~$E(\mathcal{S}$):~$\SSS$-paths using \cref{def:W-paths}. Note that since~$\mathcal{S}$ is Eulerian, the graph~$G - \mathcal{S}$ is again Eulerian. Hence we will look at~$\mathcal{S}$-paths in~$G - \mathcal{S}$, where each such path can, by Eulerianness, be closed to a cycle. So in a sense, the jumps we are interested in form cycles in~$G - \mathcal{S}$ \emph{grasped by~$\mathcal{S}$}. (Note here the similarity to the important notion of jump-cycles following the \cref{def:jump_path} of jump-paths in \cref{subsec:finding_the_swirl}).

\begin{definition}[(Long and short) Jumps]
    Let~$\mathcal{S} = S_1\cup\ldots S_s$ be an~$s$-swirl. Let~$P$ be an~$\mathcal{S}$-path in~$G- \mathcal{S}$ with endpoints~$x,y \in V(\mathcal{S})$. Then we say that~$P$ is a \emph{jump} on~$\mathcal{S}$. 
    
    We say that~$P$ is a \emph{long} jump if~$\operatorname{dist}_{\mathcal{S}}(x,y) \geq 3$. Otherwise we say that~$P$ is a \emph{short} jump.
    \label{def:jump_on_swirl}
\end{definition}
\begin{remark}
    A jump is long if its endpoints do not lie inside some common~$3\times 3$-tile---they may lie on the boundary of it---and otherwise it is short.
\end{remark}

The idea behind \cref{def:jump_on_swirl} stems from the fact that long jumps always yield usable crosses, while short jumps may not yield usable crosses, even if they impose strongly planar vertices in a possible embedding of the swirl and its attachments. This follows from \cref{lem:jumps_give_crosses} using the symmetry of a swirl~$\SSS$ induced by a tile~$T$. That is, we find some tile~$T' \subseteq \SSS$ (after a possible rotation of the swirl) for which the long jump induces a wall-local usable cross; note that~$T'$ might not be a tile of the wall~$\WWW$ we used to find~$\SSS$ but any tile we get from~$\SSS$.

In order to classify the end-points of jumps given a canonical swirl induced by a tile, we define the notion of \emph{covering vertices}.

\begin{definition}[Covering vertices with tiles]
    Let~$\SSS$ be a~$t$-swirl induced by some~$t$-tile~$T \subseteq \WWW$ for some~$t \in \N$. Let~$x 
\in V(S)$ be a vertex of the swirl. Further let~$T'\subseteq T$ be a~$t'$-tile for some~$t' \leq t$ with induced swirl~$\SSS'$ (see \cref{obs:subtiles_of_grasping_tiles_induce_swirls}). Then, if~$x \in V(\SSS')$ we say that \emph{$x$ is covered by~$T'$}.
\label{Def:tile_covering_vertices}
\end{definition}

Using this we define the following.

\begin{definition}
Let~$\SSS\coloneqq S_1 \cup \ldots \cup S_t$ be a~$t$-swirl. Let~$P$ be a jump on~$S$ with endpoints~$x,y \in V(S)$.
    We define~$\pi(P; S) \coloneqq \{j \mid S_j \cap \{x,y\} \neq \emptyset,\:1 \leq j \leq t\}$.

    If~$\mathcal{S}$ is induced by some~$t$-tile~$T \subseteq \WWW$ of some plane wall~$\WWW$ we define~$\pi_\WWW(P;S) \coloneqq \{(j,p) \mid T_{j,j+1}^{p,p+1} \text{ covers } x \text{ or } y\}$.

\end{definition}
\begin{remark}
    Note that it is an easy observation (see \cref{def:canonical_swirl_induced_by_tile}, \cref{rem:induced_swirls_and_tiling} and the proof of \cref{lem:no_jumps_is_swirl_gen}) that every vertex in a swirl induced by a tile is covered by a sub-$2$-tile. In particular it holds~$\pi_\WWW(P;S) \neq \emptyset$. 
\end{remark}

Given a jump on a swirl, we can easily verify whether it can be used to get a usable cross (see \cite{Frank1988,FrankIN1995}), leading to the definition of \emph{jumps inducing crosses}.
\begin{definition}[Jumps inducing Crosses]
    Let~$\mathcal{S}$ be an~$s$-swirl grasped by an~$s$-tile and let~$\iota$ be a jump with endpoints~$x,y \in V(\mathcal{S})$. We say that~$\iota$ \emph{induces a swirl-usable cross} if there are two paths~$P_1,P_2$ forming a swirl-local usable cross as in \cref{def:swirl_usable_cross} such that~$P \subseteq P_1$ and~$(P_1-\iota),P_2 \subseteq \mathcal{S}$.
    \label{def:jump_inducing_cross}
\end{definition}

The idea behind the above definition is that a jump~$\iota$ induces a cross if and only if the graph~$\mathcal{S} \cup \iota$ contains a cross, that is,~$\iota$ is crucial for the existence of said cross.

There are two main cases to consider: either we find `many' disjoint tiles~$T_i \subseteq \mathcal{S} \subseteq \WWW$ each covering the ends of jumps inducing crosses---in which case we easily get a router using that~$\mathcal{S}$ is a swirl---or there do not exist many such tiles. 
For ease of argumentation we define what we mean by a \emph{tile-flip with respect to~$\mathcal{S}$}.

\begin{definition}[Tile-flips of~$T$ with respect to~$\mathcal{S}$]
    Let~$\mathcal{S}=S_1\cup\ldots \cup S_{t}$ be a swirl induced by a~$t$-tile~$T'$. Let~$T$ be the respective~$(t+1)\times t$-tile as in \cref{lem:no_jumps_is_swirl_gen} with the respective labelling derived from the respective proof (see \cref{lem:no_jumps_is_swirl}); in particular let~$U_1,\ldots,U_t$ be as in that proof. We define four directions with respect to~$T$ namely \emph{up}, \emph{down} \emph{left} and \emph{right} abbreviated via~$u,d,l$ and~$r$ respectively.
    
     Then we define~$T'' \coloneqq \operatorname{flip}(\mathcal{S},T';\operatorname{dir}) \subseteq \mathcal{S}$ for~$\operatorname{dir} \in \{u,l,d,r\}$ to be the unique~$t$-tile given by the following.
     \begin{itemize}
         \item[($u$)] If~$\operatorname{dir} = u$ then the tile~$T''$ is given by 
         $$ T'' = \bigcup_{1\leq i \leq t} U_i \: \cup \: \bigcup_{1\leq i \leq 2t} H_i,$$
         where the coordinates are defined as always when identifying~$W_i$ as~$U_{t-i}$ and~$H_p$ as~$H_{2t-p}$. 

        \item[($d$)] If~$\operatorname{dir} = d$ then the we set~$T''=T'$ i.
         \item[($l$)] If~$\operatorname{dir} = l$ then the tile~$T''$ is given by 
         $$ T'' = \bigcup_{1\leq i \leq t} H_{2i} \: \cup \: \bigcup_{1\leq i \leq t} W_i\cup U_i,$$
         where the coordinates are defined as always when identifying~$W_i$ as~$H_{2i}$ and~$H_p$ as~$W_{\frac{p+1}{2}}$ for~$p \mod 2 =1$ as well as~$H_p = U_{\frac{p}{2}}$ else. 
         
         \item[($r$)] If~$\operatorname{dir} = r$ then the tile~$T''$ is given by 
         $$ T'' = \bigcup_{1\leq i \leq t} H_{2i-1} \: \cup \: \bigcup_{1\leq i \leq t} W_i\cup U_i,$$
         where the coordinates are defined as always when identifying~$W_i$ as~$H_{2i-1}$ and~$H_p$ as~$U_{\frac{p+1}{2}}$ for~$p \mod 2 =1$ as well as~$H_p = W_{\frac{p}{2}}$ else. 
     \end{itemize}
    \label{def:tile_flip}
\end{definition}
\begin{remark}
    A tile-flip can be thought of as rotating or mirroring the underlying tile of the swirl and simply renaming horizontal paths to be vertical paths and vice-versa; for in an induced swirl we have alternating left-to-right paths~$H_i$ and $H_{i+1}$ and alternating up-and-down paths~$W_i$ and~$U_i$.
\end{remark}

Clearly any flipped tile~$T''$ still induces the swirl~$\SSS$ and using the obvious relabelling for the tile~$T$ (coming from~$T''$) it even induces it.

We now make precise an observation already mentioned earlier, where we sketch the proof for the details are cumbersome but highly analogous to previous results as we will highlight.
\begin{observation}
    Let~$\mathcal{S}$ be an~$s$-swirl induced by an~$s$-tile~$T$. Let~$P$ be a long jump on~$\SSS$ such that its ends are covered by a tile~$T'\subseteq T$. Let~$\SSS'$ be the swirl induced by~$T'$ and assume further that the ends of~$P$ are at distance at least~$4$ from the outer-cycle of~$\SSS'$ with respect to~$\operatorname{dist}_{\SSS}(\cdot,\cdot)$. Then~$P$ induces a cross in~$\SSS'$. 
\label{obs:long_jumps_induce_crosses_on_swirl}
\end{observation}
\begin{proof}
    The proof follows by flipping the tile~$T'$ to a tile~$T''$ so that the jump is no Type I or Type II jump on~$T'$ (again rigorously it cannot be either for the jump cannot start neither end at coordinates, but it may \emph{behave} like either when starting and ending close to said coordinates). Since it is no up-path either for it is too long (it cannot behave like an up-path either) there exists a wall-local usable cross using~$P$. Revisiting the proofs of \cref{lem:jumps_give_crosses} and \cref{cor:jumps-induce-crosses} concludes the proof.
\end{proof}

In the same flavour we derive the following---we omit some technical details in the proof without compromising its integrity for the needed details do not provide any new insight and require more unnecessary notation resulting mainly in clutter.

\begin{lemma}
    There exists a function~$f(t)$ such that the following holds. Let~$\mathcal{S}$ be an~$s(t)$-swirl induced by some~$s(t)$-tile~$T \subseteq \WWW$ for some~$s(t) \geq f(t)$. Let~$T_1,\ldots,T_{f(t)}\subseteq T$ be disjoint $t$-tiles in~$T$ inducing swirls~$\SSS_1,\ldots,\SSS_{f(t)}$, such that they are pairwise~$t^2$ apart with respect to~$\operatorname{dist}_\SSS(\cdot,\cdot)$ and all~$t^2$ apart from the outer-cycle of~$\SSS$. Let~$P_1,\ldots,P_{f(t)}$ be edge-disjoint jumps on~$\SSS$ such that~$P_i$ induces a swirl-usable cross in~$\SSS$ and has both its endpoints in~$\SSS_i$ for~$1 \leq i \leq f(t)$. Then there exists a~$t$-Router~$\RRR$ grasped by~$T$ (and in particular by~$\WWW)$.
    \label{lem:disj_crosses_on_swirl_yield_router}
\end{lemma}
\begin{proof}
    Let~$f(t) = t^{64}$, then we claim that~$f$ satisfies the claim. 
    To each path~$P_i$ we may now assign a direction with respect to the coordinates of~$T$ and the underlying wall~$\WWW$ in the obvious way using~$\pi_\WWW(P_i;\SSS)$; i.e., the jump can go upwards, downwards, to the right or to the left.  By the pigeonhole principle at least~$t^4$ of the tiles, say~$T_{i_1},\ldots,T_{i_{t^4}}$ cover jumps inducing crosses in the same direction~$\operatorname{dir} \in \{u,d,l,r\}$; we call the directions~$u,d,l,r$ as in \cref{def:tile_flip}. Using a tile-flip with respect to that direction we get an~$f(t)$-tile~$T'$ with~$t^4$ disjoint wall-local usable~$T'$-crosses. Again using the pigeonhole principle there are~$t^2$ many usable crosses in different rows or in different columns. We can extend said tile to a wall using the edges and paths of~$\SSS$. Then, if the crosses are in different rows the claim follows similar to \cref{cor:cross_column_certificate} where we get a cross-row with boundary at least~$t$ for all the crosses are at least~$t^2$ away from the outer-cycle of~$\SSS$ and thus at distance at least~$t$ (almost~$t^2$) of the boundary of the wall we get when extending~$T$. Thus assume we get a \emph{cross-row} with an obvious and analogous definition to \cref{def:cross_column_general}---the crosses all lie in a common row instead of column. It is straightforward to check that a~$t^2$-cross-row of wall-local usable~$T'$-crosses in a swirl induced by the tile yields a~$t$-router in the same spirit as in \cref{lem:cross_column_yields_router_basecase:a} where, opposed to that proof, the fact that~$\SSS$ is a swirl is \emph{crucial}. (Note that a~$t^2$-cross-row in a tile does in general \emph{not} witness the existence of a~$t$-router in the wall). One way to see it is that we can find an immersion of a cylindrical wall in~$\mathcal{S}$ for which the crosses become a~$t^2$-cross column with boundary~$t$ certifying the existence of the router using \cref{cor:cross_column_certificate}: take a bundle of paths of width~$4t$ and thread it from cross to cross in the swirl using that they are each~$t^2$ apart, then close it using the fact that~$\SSS$ is a swirl; see \cref{fig:cross_row_threading} for a schematic illustration.
\end{proof}
\begin{remark}
    The numbers in this lemma are far too large---compare to \cref{cor:cross_column_certificate}---but the exact numbers can be easily determined following the analysis of \cref{lem:cross_column_yields_router_basecase:a}. In general the proof is highly analogous using standard tricks trying to find a wall in the swirl for which the crosses become a~$t$-cross column with boundary~$t$; a case we already dealt with explicitly.
\end{remark}

\begin{figure}
    \centering
    \includegraphics[height=.4\linewidth]{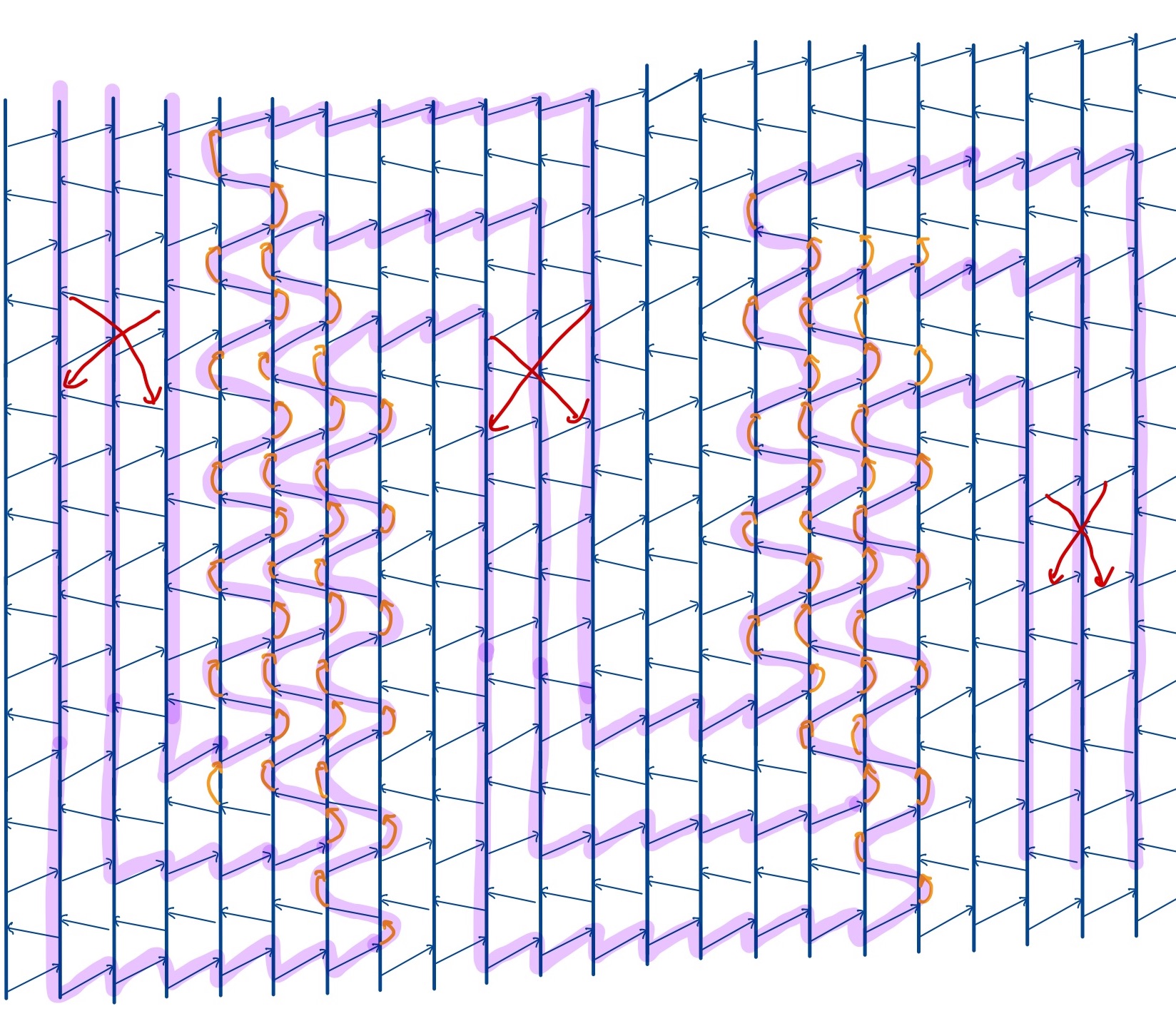}
    \caption{A schematic representation of how to thread paths in a swirl to construct a cylindrical wall with a cross-column. The highlighted paths will be the outer-cycle of the wall which can be closed to a cycle using the swirl. The up-paths (which may result in new~$H_p$ paths in the new wall) have not been drawn unless needed for the threading to not overload the figure.}
    \label{fig:cross_row_threading}
\end{figure}
Thus we may assume that there do not exist many disjoint tiles inducing disjoint swirl-local usable crosses. Next, to simplify what is to come, we define \emph{tiles tailored to jumps}.

\begin{definition}[Jump-tailored Tiles]
    Let~$\mathcal{S}$ be an~$s$-swirl induced by an~$s$-tile~$T$. Let~$P$ be a jump on~$\mathcal{S}$ (with its ends at distance~$4$ from the outer-cycle of~$\SSS$) and let~$T' \subseteq \mathcal{S}$ be a $t$-tile for some~$t \leq s$ covering both endpoints of~$P$ such that~$P$ induces a swirl-local usable~$T'$-cross with respect to the sub-swirl~$\SSS'$ induced by~$T'$. Let~$\mathcal{T}_P$ be the set of all covering tiles~$T'$ such that~$P$ induces a swirl-local usable~$T$-cross with respect to the induced sub-swirl~$\SSS'$.

    We call a tile~$T' \in \mathcal{T}_P$ \emph{tailored to}~$P$ if~$T'$ is inclusion minimal in~$\mathcal{T}_P$.
\end{definition}

We state another obvious observation.

\begin{observation}
    Let~$P$ be some jump in~$\SSS$ inducing a cross with endpoints~$x,y$ covered by a~$t$-tile~$T$ such that the ends are at distance at least~$4$ from the outer-cycle of~$\SSS$. Let~$T_P$ be a tile tailored to~$P$. Then~$x$ and~$y$ are in the~$2$-neighbourhood with respect to~$\operatorname{dist}_\SSS$ of the corners of the tile~$T_P$.
\label{obs:tailored_tiles_ends_are_corners}
\end{observation}

From here on let~$\mathcal{S}$ be an~$s$-swirl induced by some~$s$-tile~$T$ where~$T \subseteq \WWW$ for some~$t\times t$-wall~$\WWW$. Let~$\mathcal{P}$ be the set of all jumps (not necessarily edge-disjoint) on~$\mathcal{S}$ that induce some swirl-local usable cross. Further let~$$\mathcal{T}_{\mathcal{P}} \coloneqq \{T_P \subseteq T \mid T_P \text{ is tailored to the jump } P \in \mathcal{P}\},$$ be the set of respective jump-tailored tiles.

Recall the \cref{def:attachment_extension_of_swirl} of~$\SSS[G]$ the attachment-extension of a swirl. We continue with defining the \emph{attachments} to a swirl. 

\begin{definition}[Attachments to a swirl]
    Let~$G$ be an Eulerian graph and let~$\mathcal{S}$ be an~$s$-swirl in~$G$ for some~$s\in \N$. Let~$\SSS[G]$ denote the attachment-extension to~$\SSS$. Let~$\AAA(\SSS) \coloneqq \{A_1,\ldots,A_p\}$ be the set of connected components of~$\SSS[G] - \mathcal{S}$, then we call~$\AAA(\SSS)$ the \emph{attachment-set of~$\mathcal{S}$} and each~$A_i\in \AAA(\SSS)$ an \emph{attachment of~$\mathcal{S}$} for~$1 \leq i \leq p$ and some~$p \in \N$.
\end{definition}

The following is obvious.
\begin{observation}
    The graph~$\SSS[G]$ and every attachment~$A\in \AAA(\SSS)$ of~$\mathcal{S}$ is Eulerian.
    \label{obs:attachments_are_eulerian}
\end{observation}

In order to prove \cref{thm:flat_swirl_theorem}, we first note that the absence of a flat swirl implies that the attachment-set of~$\mathcal{S}$ attaches to every tile in a~$t$-tiling of~$\mathcal{T}$.

\begin{observation}
Let~$\mathcal{S}$ be a~$f(t)=(\rho(t)\cdot t)$-swirl in an Eulerian graph~$G$ induced by some~$f(t)$-tile~$T$ for some~$\rho(t) \geq t>1$. Further assume that~$\mathcal{S}$ contains no flat~$t$-sub-swirl~$\mathcal{S}'$.
Let~$(T_{i,j})_{1\leq i,j \leq \rho(t)}$ be the~$t$-Tiling of~$T$ and let~$(\SSS_{i,j})_{1\leq i,j \leq \rho(t)}$ be the respective induced sub-swirls. Let~$A \coloneqq \bigcup_{A' \in \AAA(\SSS)}A'$.
Then~$V(A) \cap V(\SSS_{i,j}) \neq \emptyset$ for every~$1\leq i,j \leq \rho(t)$ and moreover every~$T_{i,j}$ covers an end of some path inducing a cross in~$\mathcal{S}$.
\label{obs:every_tile_has_a_jump}
\end{observation}
\begin{proof}
    Assume the contrary, and let~$T_{i,j}$ witness it for some~$1 \leq i,j \leq \rho(t)$. Let~$\mathcal{S}_{i,j}$ be the~$t$-swirl induced by~$T_{i,j}$ as in \cref{obs:subtiles_of_grasping_tiles_induce_swirls}. Then~$\SSS_{i,j}[G] \subseteq \SSS[G]$ by definition, and by assumption~$\SSS_{i,j}[G] \cap A = \emptyset$. But this implies that~$\SSS_{i,j}[G] \subseteq \SSS_{i,j}$ and thus it is flat as given by \cref{def:flat_swirl}, for it is cross-less by assumption and contains no other jumps. Now since there is no flat~$t$-swirl in~$\mathcal{S}$, every~$T_{i,j}$ induces a non-flat~$t$-sub-swirl~$\mathcal{S}_{i,j}$ for~$1\leq i , j \leq \rho(t)$. Since~$\mathcal{S}_{i,j}$ is not flat, there exists a swirl-local usable~$T_{i,j}$-cross (swirl-local with respect to~$\mathcal{S}_{i,j}$). Let~$P$ be the respective path inducing it, starting in~$x \in V(\SSS_{i,j}[G])$ and ending in~$y \in V(\SSS_{i,j}[G])$. Now~$P$ is either edge-disjoint from~$\mathcal{S}$, or otherwise let~$e=(u,v) \in E(\mathcal{S})\cap E(P)$ be the first edge along~$P$ that is not disjoint from~$\mathcal{S}$. Since~$P$ is a path in~$\SSS_{i,j}[G]$ we deduce that the sub-path~$x P u$ must be a long-jump, that is a jump of length at least~$3$ with respect to~$\operatorname{dist}_{\SSS}(\cdot,\cdot)$;~$x$ is covered by~$T_{i,j}$ away from the outer-cycle of~$\SSS_{i,j}$ and~$u$ lies outside of~$\SSS_{i,j}$. In particular it induces a jump in~$\mathcal{S}$ by \cref{obs:long_jumps_induce_crosses_on_swirl}. This proves the claim.
\end{proof}

Suppose now that there exists some huge component of~$\AAA(\SSS)$ attached to many tiles of the tiling (i.e., having endpoints covered by the respective tiles). We prove that this helps in finding a~$t$-router.

\begin{lemma}
Let~$\SSS$ be an~$f(t)=(\rho(t)\cdot t)$-swirl for~$\rho(t) \geq t^3$ induced by an~$f(t)$-tile~$T$ for some~$t \geq 4$. Let~$A\in \AAA(\SSS)$ and define~$\mathcal{T}_{A} \coloneqq \{T_{i,j} \mid V(T_{i,j}) \text{ covers a vertex of } V(A) ,\: 1 \leq i,j \leq \rho(t)\}$ for some~$t$-tiling~$(T_{i,j})_{1\leq i,j \leq \rho(t)}$ of~$T$. 
If~$\Abs{\mathcal{T}_{A}} \geq t^8$, then there exist~$t$ edge-disjoint jumps~$P_1,\ldots,P_t$ on~$\mathcal{S}$ each inducing swirl-usable crosses with pairwise disjoint jump-tailored tiles~$T_1,\ldots,T_t$ at pairwise distance at least~$t$ with respect to~$\operatorname{dist}_{\mathcal{S}}(\cdot,\cdot)$.
    \label{lem:attachment_components_are_small}
\end{lemma}
\begin{proof}
    By \cref{obs:attachments_are_eulerian} we know that~$A$ is Eulerian. Let~$C_{A}$ be an Euler-cycle of~$A$, then~$C_{A}$ has vertices covered by each of the tiles~$T' \in \mathcal{T}_{A'}$. Let~$\tau=((i_1,j_1,),(i_2,j_2),\ldots,(i_{t^8},j_{t^8}))$ be the sequence of the respective tiles visited in order where~$(i_\alpha,j_\alpha) \neq (i_{\beta},j_{\beta})$ for~$1 \leq \alpha\neq \beta \leq t^8$. (Note here that each vertex is covered by some tile~$T_{i,j}$ for the respective induced swirls~$\SSS_{i,j}$ cover~$\SSS$, see \cref{rem:induced_swirls_and_tiling}). Using the well-known Erdös-Szekeres theorem \cite{ErdosS1935} with respect to the first component---the~$i$-component---of the elements in~$\tau$, the sequence either contains a weakly monotonically increasing sub-sequence (where equality is seen as increasing) of length~$t^4$ or a strictly monotonically decreasing sub-sequence in the~$i$-component of length at least~$t^4$. Call the sequence~$\tau'$. Then again using the Erdös-Szekeres theorem \cite{ErdosS1935} we get a monotonically increasing or decreasing sub-sequence~$\tau''$ in the second component of~$\tau'$---in the~$j$-component---of length at least~$t^2$. This in turn yields~$\Theta(t^2)$ edge-disjoint jumps each inducing a cross, for their ends lie in different tiles of the~$t$-tiling as given by~$\tau''$. Note that by taking every fourth element of~$\tau''$ we can guarantee the jumps to induce crossings using that the sequence is always strictly monotonic in at least one component and at most four tiles of the~$t$-tiling are pair-wise close with respect to~$\operatorname{dist}_\SSS$, thus taking every fourth element results in jumps of length at least~$3$ with respect to~$\operatorname{dist}_{\mathcal{S}}$ and thus it induces some swirl-usable cross by \cref{obs:long_jumps_induce_crosses_on_swirl}.
    
    For each of the~$\theta \in \Theta(t^2)$ jumps~$P_1,\ldots,P_{\theta}$ let~$T_1,\ldots,T_{\theta}$ denote their jump-tailored tiles. Then, since~$\tau''$ does not contain any element twice, and since the jumps may be assumed to have length at least~$t$ (larger than~$4$), at most two consecutive jump-tailored tiles may pairwise intersect and they do so close to their corners by \cref{obs:tailored_tiles_ends_are_corners}. Since the different tiles in the~$t$-tiling each have diameter~$t\geq 4$, and since every point in the sequence lies in a different such~$t$-tile, we easily find a sub-sequence of length~$t$ satisfying the lemma.  
\end{proof}

Using \cref{lem:attachment_components_are_small} together with \cref{lem:disj_crosses_on_swirl_yield_router} implies that, assuming there is no~$t$-router grasped by~$\mathcal{S}$, the components of~$\AAA(\SSS)$ do not cover many tiles of the~$t$-tiling. This, together with \cref{obs:every_tile_has_a_jump} implies that there exists a large subset of~$\mathcal{P}$ that is indeed edge-disjoint, and even more, that has no ends covered by common tiles; think of those jumps as choices of connection for two attachment points of each component of~$A$. Formally this means the following.

\begin{observation}
    Let~$\SSS$ be an~$f(t)=(\rho(t)\cdot t)$-swirl for~$\rho(t) \geq (\xi(t)+t^4)$ induced by an~$f(t)$-tile~$T$ for some function~$\rho(t)$ and $\xi(t)$ and~$t \geq 4$. Let~$A\in \AAA(\SSS)$ and define~$$\mathcal{T}_{A} \coloneqq \{T_{i,j} \mid  V(T_{i,j})\text{ covers a vertex in } V(A),\: 1 \leq i,j \leq \rho(t)\}$$ for some~$t$-tiling~$(T_{i,j})_{1\leq i,j \leq \rho(t)}$ of~$T$. 
    If~$\Abs{\mathcal{T}_{A}} \leq t^8$, for every~$A \in \AAA(\SSS)$, then there exist~$\xi(t)^2$ edge-disjoint jumps~$P_1,\ldots,P_{\xi(t)^2}$ such that they have ends covered by pairwise disjoint tiles of~$(T_{i,j})_{1\leq i,j \leq \rho(t)}$.
    \label{obs:small_attachments_imply_disjoint_crossing_paths}
\end{observation}
\begin{proof}
    By the pigeon-hole principle we immediately see that~$\AAA(\SSS)$ partitions~$(T_{i,j})_{1\leq i,j \leq \rho(t)}$ into blocks of size~$t^8$. Thus we have at least~$\xi(t)^2$ many components in~$\AAA(\SSS)$. For each of the components choose a path that induces some cross (if it exist). Note that for each tile~$T_{i,j}$ there exists at least one component in~$\AAA(\mathcal{S})$ attached to the swirl~$\SSS_{i,j}$ induced by~$T_{i,j}$, such that it contains a swirl-local usable cross for otherwise~$T_{i,j}$ induces a flat swirl by \cref{obs:every_tile_has_a_jump}. Thus, one easily verifies that we can use the disjoint components~$\AAA(\SSS)$ to derive the needed set of edge-disjoint paths inducing a matching on the respective tiles.
\end{proof}

We will use this insight to prove the last missing piece in the proof of \cref{thm:flat_swirl_theorem}. To the readers familiar with the respective proofs of finding undirected flat-walls (as in \cite{KawarabayashiTW2018} or \cite{GMXIII} for example), it is the next result that takes (implicitly) care of the corner cases such as having many disjoint jumps with their ends very close together, or many disjoint jumps ending somewhere close to the boundary of the underlying wall. The main difference to the undirected case is that we already \emph{know} from the above that for a large set of edge-disjoint jumps there is at least one end of each jump so that those ends are pair-wise far apart. If all ends are pair-wise far apart then things work nicely, and if not, then there is a small chunk of the swirl containing most of the other half of the ends. But then taking that chunk out of the swirl gives rise to a connected attachment well-spread on the remainder of the swirl, which reduces to a case we already discussed: \cref{lem:attachment_components_are_small} which in turn exploited the Eulerianness of the attachments and the fact that we look for edge-disjoint paths.

\begin{lemma}
    There exists a function~$f(t)$ such that the following holds. Let~$\mathcal{S}$ be an~$s(t)$-swirl induced by an~$s(t)$-tile~$T$ for some~$t \geq 2$ and~$s(t) \geq f(t)$. Let~$J \subseteq \mathcal{P}$ be a set of edge-disjoint jumps, so that no two jumps have ends covered by some common tile of the~$t$-tiling of~$T$. If~$\Abs{J} \geq f(t)$ then~$\mathcal{S}$ contains a sub-swirl~$\mathcal{S}'$ such that there exist~$t$ edge-disjoint jumps~$P_1,\ldots,P_t$ on~$\mathcal{S}'$ each inducing swirl-local usable crosses (with respect to~$\mathcal{S}'$) with pairwise disjoint jump-tailored tiles each at distance at least~$t$ with respect to~$\operatorname{dist}_{\mathcal{S}'}(\cdot,\cdot)$.
    \label{lem:many_attchments_on_a_swirl_yield_good_subswirl}
\end{lemma}
\begin{proof}
    We claim that~$f(t) = h(t)^8\cdot t$ for~$h(t) \geq t^8$ satisfies the lemma.
    To this extent let~$\mathcal{T}_{J} \subseteq \mathcal{T}_{\mathcal{P}}$ be the respective jump-tailored tiles to the~$f(t)$ edge-disjoint paths in~$J$ and let~$(T_{i,j})_{1\leq i,j \leq \rho(t)}$ be the~$t$-tiling of~$T$ such that each pair of jumps in~$J$ have their ends covered by distinct tiles where~$\rho(t) = h(t)^8$. Next look at the intersection graph~$\TTT_J^\cap$ induced by~$\TTT_J$ with respect to~$\cap$, that is for every tile~$T \in \TTT_J$ we have a vertex~$v_T$ and two vertices~$v_T,v_T' \in V(\TTT_J^\cap)$ are adjacent if the respective tiles intersect, i.e.,~$T \cap T' \neq \emptyset$. If~$\TTT_J^\cap$ contains a~$9t$-independent set, then there exist~$9t$ disjoint jump-tailored tiles in~$\TTT_J$ and finally, noting that there is at most~$9$ disjoint~$t$-tiles that cover a~$3t$-tile, we find~$t$ jump-tailored tiles pairwise at least~$t$ apart. Thus the lemma follows for~$\mathcal{S} = \mathcal{S}'$.

    Henceforth assume we do not find a~$9t$-independent set in~$\TTT_J^\cap$. We claim the following.
    \begin{claim}
        There exists an~$s(t)$-tile~$T' \subseteq T$ for some~$s(t) \geq \frac{f(t)}{3}-1$ such that~$T'$ induces a swirl~$\mathcal{S}'\subseteq \mathcal{S}$ and a set~$J' \subseteq J$ with~$\Abs{J'} \geq \frac{h(t)^2}{9}-1$ such that~$T'$ covers exactly one end of each jump in~$J'$. 
        \label{lem:many_attchments_on_a_swirl_yield_good_subswirl_claim1}
    \end{claim}
    \begin{ClaimProof}
        First we define~$\rho(t)$ columns of the~$t$-tiling via~$T_i \coloneqq \bigcup_{1 \leq j \leq \rho(t)}T_{i,j}$ for every~$1 \leq i \leq \rho(t)$. Then~$T = \bigcup_{1 \leq i \leq \rho(t)} T_i$ by construction. The notion of covering ends transmits in the obvious way to the columns. 

    \smallskip
    
        Start by sorting the columns via~$T_1 \prec \ldots \prec T_{\rho(t)}$. Using~$\prec$ we may now sort the ends of jumps in~$J$, since the jumps all have their ends covered by pairwise disjoint tiles this means that for every tile in the ordering given by~$\prec$ at most one jump has its end covered by that tile. Let~$\iota \in J$ be a jump and write~$\iota^-,\iota^+$ for its first and last vertex with respect to the direction of~$\iota$. Now let~$T_{i_{\iota^-}}$ be the column containing the~$t$-tile from the tiling covering~$\iota^-$ and similarly let~$T_{i_{\iota^+}}$ be the column containing the~$t$-tile from the tiling covering~$\iota^+$, where~$1 \leq i_{\iota^-}, i_{\iota^+} \leq \rho(t)$. We order the ends of the jumps via~$\prec_{J}$ where~$\iota^x \prec_J {\iota'}^y$ if and only if~$T_{i_{\iota^x}} \prec T_{i_{{\iota'}^y}}$, for~$x,y \in \{-,+\}$ and~$\iota^x \neq {\iota'}^y$ for some~$\iota,\iota' \in J$. 

        Note that~$\prec_{J}$ yields a partial order for there exist jumps with ends in the same column (we may see those as being incomparable). Using the pigeon-hole principle then either there is an~$h(t)^4$ sub-sequence~$\tau$ inducing a strict order with respect to~$\prec$ (a \emph{chain}), or a~$h(t)^4$ sub-sequence~$\tau'$ of pairwise incomparable elements (an \emph{anti-chain}). Note that we may without loss of generality assume the first outcome to be true, for if the latter were true, then we may change the setting to rows~$T^j \coloneqq \bigcup_{1\leq i \leq \rho(t)}T_{i,j}$ for~$1 \leq j \leq \rho(t)$ which in turn yields a new ordering for which the incomparable elements forming the anti-chain now becomes a chain, since no two endpoints lie in the same tile, i.e., in the intersection~$T_i \cap T^j$ of a column~$T_i$ and a row~$T^j$.
        
        Thus let~$\tau$ be the~$h(t)^4$ sub-sequence forming a chain with respect to~$\prec_J$. Let~$J_{\tau} \subseteq J$ be the set of jumps having an end-point covered by some column in~$\tau$, then~$\Abs{J_{\tau}} \geq \lfloor \frac{h(t)^4}{2} \rfloor$ for every jump has two distinct endpoints.

\smallskip

        Next, given~$\prec_{J}$ we define a strict order~$\ll$ on the jumps in~$J_{\tau}$ as follows. Let~$\iota,\iota' \in J_{\tau}$ be two jumps then extend~$\prec_J$ by setting incomparable elements equal and writing~$\preceq_J$. Clearly~$\preceq_J$ is a linear order on~$\tau$ (even on all the ends of jumps as above). Without loss of generality we will from here on tacitly assume that~$\iota^- \preceq_{J} \iota^+$---where equality holds if they are incomparable, i.e., they lie in the same column---as well as~${\iota'}^- \preceq_{J} {\iota'}^+$, otherwise relabel them. Then we say that~$\iota \ll \iota'$ if and only if
        \begin{equation}
            \iota^- \preceq_{J} \iota^+ \preceq_J {\iota'}^- \preceq_J {\iota'}^+,
            \label{eq:partial_order_on_jumps}
        \end{equation} otherwise we say~$\iota' \ll \iota $. By the Erdöz-Szekeres theorem \cite{ErdosS1935} we either get a monotonically increasing sub-sequence (with respect to~$\ll$) of length at least~$h(t)^2$ or a monotonically decreasing sub-sequence of length at least~$h(t)^2$. Let~$\sigma$ be said sub-sequence. 

        If~$\sigma$ is monotonically increasing, then this immediately implies the existence of~$h(t)^2 > t^2$ many edge-disjoint jumps in~$J_\tau$ with---by taking every second one with respect to~$\sigma$---pairwise disjoint jump-tailored tiles. These are a witness for a~$9t$-independent set in~$\TTT_J^\cap$; contradiction. Note here that the way~$\tau$ was defined, for two jumps~$\iota \ll \iota'$ it holds that at least one of the inequalities~$\preceq_J$ in \cref{eq:partial_order_on_jumps} is \emph{strict}.

\smallskip

        Thus~$\sigma$ is monotonically decreasing whence \cref{eq:partial_order_on_jumps}  does not hold for any two jumps~$\iota,\iota' \in \sigma$. In particular, for every two jumps~$\iota,\iota' \in \sigma$ it holds~${\iota'}^- \prec_J \iota^+$. To see this recall that by assumption it holds~$\iota^- \preceq_J \iota^+$ for every~$\iota \in J_\tau$, then it follows trivially by analysing the possible cases not violating the assumption. Finally this implies that the sequence~$\tau$ is of the following type.
        \begin{equation*}
            \tau: \; \iota_1^- \prec_J \iota_2^- \prec_J \ldots \prec_J \iota_{h(t)^2}^- \prec_J \iota_1^+ \prec_J \iota_2^+ \prec_J \ldots \prec_J \iota_{h(t)^2}^+,
        \end{equation*}
        where~$\iota_1,\ldots,\iota_{h(t)^2}$ are the distinct jumps in~$\tau$. This in turn implies that there exist two sequences~$\tau_1:\;1\leq \ldots \leq \ell$ and~$\tau_2: \: \ell+1 \leq \ldots, \leq \rho(t)$ such that~$T_1 \coloneqq \bigcup_{1 \leq i \leq \ell} T_i$ covers all of~$\iota_1^- \prec_J \iota_2^- \prec_J \ldots \prec_J \iota_{h(t)^2}^- $ and~$T_2\coloneqq \bigcup_{\ell+1 \leq i \leq \rho(t)} T_i$ covers all of~$\iota_1^+ \prec_J \iota_2^+ \prec_J \ldots \prec_J \iota_{h(t)^2}^+ $. Without loss of generality assume~$\tau_1$ is the longer sequence of both (or they are equal). 
        
        Then~$T_1$ is an~$f(t) \times f'(t)$-tile where~$\frac{f(t)}{2} \leq f'(t)<f(t)$. In particular we can cover~$T_1$ with nine~$\lceil\frac{f'(t)}{3}\rceil$-tiles~$T_1^1,\ldots,T_1^9$ which are not necessarily disjoint. Thus at least one of the tiles must cover~$h'(t) \geq \lfloor \frac{h(t)^2}{9} \rfloor$ many elements of~$\{\iota_1^-,\ldots, \iota_{h(t)^2}^-\}$; let it be~$T_1^1$. Then~$T_1^1$ induces an~$s(t)$-swirl~$\mathcal{S}'$ for some~$s(t) \geq \lceil\frac{f'(t)}{3}\rceil$ by \cref{obs:subtiles_of_grasping_tiles_induce_swirls} satisfying the claim. 
     \end{ClaimProof}

    Let~$\mathcal{S}'$ be as in \cref{lem:many_attchments_on_a_swirl_yield_good_subswirl_claim1}. Then~$\AAA(\SSS')$ contains at least one component~$A \in \AAA(\SSS')$ such that all the jumps in~$J'$ have an end in~$A$ for they are jumps on~$\mathcal{S}$ and thus we have that~$(\SSS - \SSS') \subseteq A$. Finally since~$\frac{h(t)^2}{9} -1 \geq t^8$ and~$h(t) \geq t^8$ and~$t \geq 2$, the claim follows by \cref{lem:attachment_components_are_small}.
\end{proof}

We have finally gathered all the pieces needed for the proof of the Flat-Swirl \cref{thm:flat_swirl_theorem}.

\begin{proof}[Proof of \cref{thm:flat_swirl_theorem}]
    Let~$t \coloneqq \max(t_1,t_2)$ and let~$s(t) \coloneqq \max\left(f_{\ref{lem:many_attchments_on_a_swirl_yield_good_subswirl}}\big(f_{\ref{lem:disj_crosses_on_swirl_yield_router}}(t^2)\big), (f_{\ref{lem:many_attchments_on_a_swirl_yield_good_subswirl}}(t^2)+t^4)\cdot t\right)^2$ and let~$f(t) \coloneqq f_{\ref{thm:swirl_theorem}}(t;s)$ with respect to~$s$. Then let~$\mathcal{S}$ be a cross-less~$s(t)$-swirl induced by some~$s(t)$-tile~$T\subseteq \WWW$ of some~$f(t)\times f(t)$-wall~$\WWW$ using \cref{thm:swirl_theorem}. Now suppose that there is no~$t$-router grasped by~$\WWW$ and no flat~$t$-swirl~$\SSS' \subseteq \SSS$ induced by a respective~$t$-tile~$T'\subseteq T$. Let~$\AAA(\SSS)$ be the attachment-set of~$\SSS$ and denote by~$\TTT$ a~$t$-tiling as usual. Then by \cref{lem:attachment_components_are_small} we deduce that every component in~$\AAA(\SSS)$ is adjacent to at most~$f_{\ref{lem:disj_crosses_on_swirl_yield_router}}(t)^8$ tiles in the~$t$-tiling~$\TTT$ of~$\SSS$, for otherwise \cref{lem:disj_crosses_on_swirl_yield_router} implies the existence of a~$t$-Router. In particular applying \cref{obs:small_attachments_imply_disjoint_crossing_paths} and our choice of~$s(t)$, there exist at least~$\xi(t^2)^2\coloneqq f_{\ref{lem:many_attchments_on_a_swirl_yield_good_subswirl}}(f_{\ref{lem:disj_crosses_on_swirl_yield_router}}(t^2))^2$ many edge-disjoint jumps~$P_1,\ldots,P_{\xi(t^2)^2}$ with their end-points covered by pairwise disjoint~$t$-tiles in~$\TTT$; in particular no path has both its ends in the same tile. Now~$\xi(t^2) \geq f_{\ref{lem:many_attchments_on_a_swirl_yield_good_subswirl}}((f_{\ref{lem:disj_crosses_on_swirl_yield_router}}(t^2))$ and thus there exists a~$s'(t)$-swirl~$\SSS'$ induced by an~$s'(t)$-tile~$T' \subseteq T$ such that there exist~$(f_{\ref{lem:disj_crosses_on_swirl_yield_router}}(t^2))^2$ edge-disjoint jumps on~$\SSS'$ each inducing swirl-usable crosses with disjoint jump-tailored tiles at pairwise distance at least~$t^2$ with respect to~$\operatorname{dist}_{\SSS'}(\cdot,\cdot)$ by \cref{lem:many_attchments_on_a_swirl_yield_good_subswirl}. But this implies the existence of a~$t$-router using \cref{lem:disj_crosses_on_swirl_yield_router}; contradiction. Finally to see that the structures can be found in~$fpt$-time note that the proofs were of algorithmic nature, each providing an algorithm to find the routers or the flat swirl. Also note that deciding whether a swirl contains a usable cross can be done in polynomial time (for example using a \cref{thm:Frank_Algorithm} due to \cite{Frank1988}).
\end{proof}
\begin{remark}
    The functions in the proof of \cref{thm:flat_swirl_theorem} are far too large and can be simplified with a more careful analysis of the different cases. Since this is irrelevant for our cause we omit the analysis here.
\end{remark}

The following is a slightly stronger version of the Flat-Swirl theorem following analogously to the above by enlarging the numbers so we can guarantee the existence of~$\Abs{E(D)}+1$ flat swirls with edge-disjoint attachment-extensions.

\begin{theorem}\label{thm:flat_swirl_away_from_D}
Let~$t_1,t_2 \in \N$ , then there exists a computable function~$f:\N\times\N \to \N$ such that the following holds. Let $G+D$ be an Eulerian digraph of maximum degree~$4$ and let~$\WWW \subset G$ with~$V(\WWW) \cap V(D) = \emptyset$ be an~$f(t_1,t_2)\times f(t_1,t_2)$-wall~$\WWW \subset G$. Then, either
    \begin{itemize}
        \item[(i)] $G$ contains a flat~$t_1$-swirl~$\mathcal{S}$ induced by a~$t_1$-tile~$T\subset \WWW$ such that~$\SSS[G+D] \cap E(D) = \emptyset$, or
        \item[(ii)] $G$ contains a~$t_2$-router grasped by~$\WWW$ that is edge-disjoint from~$D$.
    \end{itemize}
    Moreover we can decide in~$fpt$-time on~$t\coloneqq \max(t_1,t_2)$ whether~$(i)$ or~$(ii)$ hold and output the relevant structure.
\end{theorem}
\begin{proof}
    Let~$t \coloneqq \max(t_1,t_2)$ and use \cref{thm:flat_swirl_theorem} to get a flat swirl~$\SSS$ of order~$\rho\coloneqq (p+1)f_{\ref{lem:disj_crosses_on_swirl_yield_router}}(t)^8$ induced by some~$\rho$-tile~$T'$ by choosing~$f$ accordingly. Now look at a~$t$-tiling~$\TTT \coloneqq (T_{i,j})_{1 \leq i,j \leq \rho}$ of~$T'$ and the respective sub-swirls denoted by~$\SSS_{i,j}$ induced by the respective tiles. Since~$\SSS$ is flat so are all of the sub-swirls, for any swirl-local usable cross in a sub-swirl~$\SSS_{i,j}$ can be readily extended to a swirl-local usable cross in~$\SSS$. Let~$D(\SSS)$ be the set containing all attachments~$A \in \AAA(\SSS)$ such that~$E(A) \cap E(D) \neq \emptyset$. Then~$\Abs{D(\SSS)} \leq p$; let~$A \in D(\SSS)$. Again we define $\mathcal{T}_{A} \coloneqq \{T_{i,j} \mid V(T_{i,j}) \text{ covers a vertex of } V(A) ,\: 1 \leq i,j \leq \rho(t)\}$. Using \cref{lem:attachment_components_are_small} together with \cref{lem:disj_crosses_on_swirl_yield_router} we deduce that~$\Abs{\TTT_A} \leq f_{\ref{lem:disj_crosses_on_swirl_yield_router}}(t)^8$ for every~$A \in D(\SSS)$. Hence by the pigeonhole principle there is at least one pair~$1 \leq i,j \leq \rho$ such that~$\AAA(\SSS_{i,j}) \cap D(\SSS) = \emptyset$. This concludes the proof.
\end{proof}

Next we will prove a refinement of the Flat-Swirl Theorem, i.e., \cref{thm:flat_swirl_theorem}; namely that we can change the graph~$G$ locally, so that not only the swirl is flat in the sense of \cref{def:flat_swirl}, but that~$\SSS[G+D]$ can be Euler-embedded in a plane. Clearly, as seen with previous examples this is not true for arbitrary Eulerian graphs and it is not true for general flat swirls (see \cref{fig:not_usable_cross}). It turns out that, using a result similar to the \emph{two-paths theorem} due to Frank, Ibaraki, and Nagamochi \cite{FrankIN1995} we can find an equivalent instance~$G'+D$, that gets rid of the strongly planar vertices and crossing edges. That is, the new graph~$G'$ is equivalent to the original graph~$G$ with respect to the given disjoint paths instance, but it has a flat  swirl that can be \emph{Euler-embedded}. As already discovered by Johnson \cite{Johnson2002} as well as discussed by Cygan, Marx, Pilipczuk and Pilipczuk (although not refering to it as a Euler-embedding) in \cite{CyganMPP2013}, it turns out that, when trying to prove the existence of irrelevant vertices, it is nice to have Euler-embeddings of the graph. 

\subsection{Flattening a flat swirl}
\label{subsec:flattening_a_flat_swirl}

As mentioned above, given a canonical flat swirl~$\mathcal{S}$ induced by some wall~$\WWW$ in an Eulerian graph~$G+D$ we will prove that there exists an equivalent instance to~$G'+D$ containing with a flat swirl~$\mathcal{S}'$ that can be Euler-embedded in the plane. In a nutshell we have seen that given a plane embedding of the swirl, all the edge-disjoint attachments to the swirl must lie inside~$3\times 3$-tiles of the underlying tile that grasps it, for otherwise it contains a long-jump which in turn induces a usable cross on the swirl induced by the tile as seen in \cref{obs:long_jumps_induce_crosses_on_swirl}. It may happen that the attachment is not Euler-embeddable together with a given Euler-embedding of the swirl (see \cref{fig:not_usable_cross}). Elaborating the analysis in \cref{subsec:finding_the_swirl} and \cref{subsec:taming_a_swirl} further one can show that for a flat swirl this only happens if the attachment can be separated from the swirl (and thus the tile) by a cut of size at most four. This is similar in a sense to the idea of~$3$-separations and~$C$-reductions in flat walls in the undirected case \cite{GMXIII,KawarabayashiTW2021}. In the undirected case the so-called \emph{two-paths theorem} guarantees that the remaining crosses attached to a wall that are non-usable can be disconnected from the wall via~$3$-separations. It turns out that, under some minimality assumptions, there exists an equivalent two-paths theorem for Eulerian graphs as shown by Frank, Ibaraki, and Nagamochi \cite{FrankIN1995}, where the role of~$3$-separations is taken by~$4$-cuts and~$2$-cuts. Note that for general directed graphs there is no analogue to the two-paths theorem which is one of the main obstacles in proving (and finding) a general structure theorem for digraphs.

\medskip

Recall that for swirl-usable crosses the directions of the paths inducing the cross was not relevant, i.e., it is irrelevant whether we want to connect~$s_1$ to~$t_1$ or vice-versa as given in \cref{def:swirl_usable_cross}. This is the same setting as the one for which Frank, Ibaraki, and Nagamochi \cite{FrankIN1995} proved their version of the two-paths theorem. We start by defining the problem in question. The following definitions agree with the ones in \cite{FrankIN1995}. 

\begin{definition}[Unordered edge-disjoint Eulerian paths]
    Let~$G$ be an Eulerian graph and let~$s_i,t_i \in V(G)$ for~$i=1,2$ be four distinct vertices in~$G$ and let~$S_1 = \{s_1,t_1\}$ and~$S_2 = \{s_2,t_2\}$. We call~$(G,S_1,S_2)$ an \emph{instance of the unordered edge-disjoint Eulerian paths problem}. We say that~$(G,S_1,S_2)$ is \emph{feasible} if there exist two edge disjoint paths~$P_1,P_2$ with~$P_1$ connecting~$\{s_1,t_1\}$ and~$P_2$ connecting~$\{s_2,t_2\}$. Otherwise we call the instance \emph{infeasible}. 
\end{definition}
 The following reductions were presented and proved in \cite{FrankIN1995}.
 
\begin{reduction}[Section 3 in \cite{FrankIN1995}]
Let~$G$ be an Eulerian graph and let~$S_i=\{s_i,t_i\} \subseteq V(G)$ for~$i=1,2$ and let~$T=S_1\cup S_2$ be the four distinct terminals. We define the following reductions.
    \begin{itemize}
    \item[1] Let~$X \subseteq V(G)$ be some~$2$-cut such that~$X\cap T = \emptyset$. Let~$u$ be the tail of the edge from~$\bar{X}$ to~$X$ and~$v$ the head of the edge from~$X$ to~$\bar{X}$. Then delete~$X$ and (if~$u \neq v$) add the edge~$(u,v)$.
    \item[2] Let~$X\subseteq V(G)$ be a~$2$-cut such that~$\Abs{X\cap T} = 1$. Contract~$X$ to the terminal~$t \in T$ deleting any resulting loops. (The resulting terminal~$t$ has degree two).
    \item[3] Let~$X$ be a~$4$-cut such that the subgraph~$G[[[X]]]$ is connected and,~$|X| \geq 2$ and~$X \cap V(T) = \emptyset$. Then contract~$G[X]$ to a single vertex of degree four, and delete possible loops.
\end{itemize}
\label{red:almost_6_connected}
\end{reduction}

Using the above reductions, Frank, Ibaraki, and Nagamochi \cite{FrankIN1995} define \emph{minimal instances} as follows.

\begin{definition}
    Let~$G$ be an Eulerian digraph and~$(G,S_1,S_2)$ an instance of the unordered Eulerian edge-disjoint paths problem. Then we say that~$(G,S_1,S_2)$ is a \emph{minimal instance} if none of the reductions from \cref{red:almost_6_connected} are applicable to~$G$.
\end{definition}

The same authors showed that applying any of the reductions in \cref{red:almost_6_connected} results in an equivalent instance. 

\begin{lemma}[Lemma 3.1 in \cite{FrankIN1995}]
    An instance~$(G,S_1,S_2)$ is feasible if and only if it is feasible after performing any of the reductions in \cref{red:almost_6_connected}.
    \label{lem:Frank_IN_equiv}
\end{lemma} 

This remains true in the setting of edge-disjoint paths where the pairs are ordered. We will only need reductions~$1.$ and~$3.$ for the following lemma, which prepares our instance and massages it into the setting of \cite{FrankIN1995}. 
\begin{lemma}
    Let~$G+D$ be an instance of the Eulerian edge-disjoint-paths problem. Let~$G'$ be obtained from~$G$ after applying reduction~$1.$ or~$3.$ from \cref{red:almost_6_connected} by setting~$T = V(D)$. Then~$G'+D$ is Eulerian and an equivalent instance.
    \label{lem:reducing_to_no_small_cuts}
\end{lemma}
\begin{proof}
    It is clear that if~$G+D$ is a \emph{YES}-instance, then so is~$G'+D$.
    Thus assume~$G'+D$ is obtained from~$G+D$ after applying~$1.$ of \cref{red:almost_6_connected} and assume that~$G'+D$ is a \emph{YES}-instance. Let~$\LLL$ be a linkage solving the edge-disjoint paths problem. Then if no~$L \in \LLL$ uses the edge~$(u,v)$---see 1. of \cref{red:almost_6_connected}---the claim is clear. If~$L \in \LLL$ uses~$(u,v)$, then in~$G$ we can replace~$(u,v)$ by a path inside~$G[[[X]]]$ linking~$u$ to~$v$ to get a solution in~$G+D$. 
    
    Finally assume~$G'+D$ is obtained from~$G+D$ after applying~$3.$ of \cref{red:almost_6_connected}. Let~$\delta^-(X) = \{e_{\text{in}},e_{\text{in}}'\}$ and~$\delta^+(X) = \{e_{\text{out}},e_{\text{out}}'\}$. Then since~$G[X]$ is connected we find a path~$P_1$ in~$G[[[X]]]$ whose first edge is~$e_\text{in}$ and its last edge is~$e_{\text{out}}$ (this will later be called the directed pattern of~$P$,i.e.,~$\pi(P) = (e_\text{in},e_{\text{out}})$). Since~$G$ is Eulerian, starting a path~$P'$ in~$e_{\text{in}}'$ we can extend it to a cycle in~$G$ omitting~$P$, but the only possible edge to leave~$G[X]$ is now~$e_{\text{out}}'$ and thus there exists a path~$P'$ with pattern~$\pi(P') = (e_\text{in}',e_{\text{out}}')$. Renaming the edges we see that all the possible connections are feasible in~$X$ and thus the instances are equivalent.
\end{proof}

With \cref{red:almost_6_connected} at hand we can state the needed two-path theorem as presented in \cite{FrankIN1995}.
\begin{theorem}[Theorem 3.7 in \cite{FrankIN1995}]
    Let~$(G,S_1,S_2)$ be a minimal infeasible instance with~$\Abs{V(G)} \neq 6$, then~$G$ can be Euler-embedded.
    \label{thm:two_paths_Frank}
\end{theorem}

Given a wall away from~$V(D)$ together with an induced flat swirl~$\SSS$, we can apply \cref{lem:reducing_to_no_small_cuts} to massage our graph into an instance suited for the application of \cref{thm:two_paths_Frank} yielding the following crucial theorem.

\begin{theorem}
  Let~$G+D$ be Eulerian and let~$\mathcal{S}\subseteq G$ be a flat~$t$-swirl induced by a~$t$-tile~$T$ of some plane wall~$\WWW$ such that~$\SSS[G]$ is edge-disjoint from~$D$. Then there exists an equivalent instance~$G'+D$ with~$\Abs{V(G')}+\Abs{E(G')} \leq \Abs{V(G)}+\Abs{E(G)}$ that can be computed in polynomial time on~$\Abs{G} + \Abs{D}$ containing a flat~$t$-swirl~$\mathcal{S}' \subseteq G'$ such that~$\SSS'[G'+D]$ can be Euler-embedded in the plane.

    \label{thm:embedded_flat_swirl}
\end{theorem}
\begin{proof}
Let~$\SSS = S_1 \cup \ldots \cup S_s$ be a flat swirl induced by some tile~$T \subset \WWW$ such that the attachment-extension~$\SSS[G+D]$ is edge-disjoint from~$D$ and let~$S_s$ be its outer-cycle given the embedding of the wall. Henceforth let~$G^+ \coloneqq \SSS[G+D] $ for some~$G^+ \subseteq G$ and let~$R \coloneqq \big((G+D) - (G^+-S_s)\big)$; recall that~$G^+$ is Eulerian by \cref{obs:both_sides_of_flat_are_eulerian} and~$G^+ \cap R = S_s$ where~$V(G^+)$ induces a cut with~$\delta(V(G^+) \subseteq E(S_s)$; we write~$\operatorname{Coord}(\SSS)\coloneqq \operatorname{Coord}(\WWW) \cap V(\SSS)$. Apply \cref{lem:reducing_to_no_small_cuts} to~$G^+$ to get an equivalent instance~$(G^+)'$ without~$2$-cuts or~$4$-cuts away from~$D$. Let~$\SSS' \subset (G^+)'$ be the induced swirl resulting from~$\SSS \subset G^+$ after applying said reductions; it is an easy observation that after renaming the vertices resulting form the contractions accordingly, it holds~$\operatorname{Coord}(\SSS) = \operatorname{Coord}(\SSS')$. To see this recall that~$T\subset \SSS$ for the swirl is induced by~$T$, every coordinate vertex of~$T$ for is of degree four in~$\SSS$, and at most one coordinate vertex can be cut off from the rest of the tile by a~$4$-cut. Thus we can rename the vertex introduced when applying 3. of \cref{red:almost_6_connected} to be the new coordinate of a new~$s$-tile~$T'$ resulting from~$T$ where~$T'$ induces~$\SSS'$. Further~$\SSS'$ is clearly flat, for we do not create any new connections. Let~$s_1,s_2,t_1,t_2 \in V(\mathcal{S}') \cap V(T')$ be the four corners of the tile inducing the swirl appearing in said order given the embedding induced by the embedding of~$\SSS$. We set~$S_1 = \{s_1,t_1\}$ and~$S_2 = \{s_2,t_2\}$. 

Since~$\mathcal{S}'$ is flat, there is no swirl-local usable cross with respect to~$\SSS'$; in particular the reduced instance~$((G^+)',S_1,S_2)$ is infeasible. By the above reduction we know that~$((G^+)',S_1,S_2)$ cannot be
  further reduced using~$1.$ or~$3.$ from \cref{red:almost_6_connected}. One
  easily sees that~$2.$ is not applicable either by our choice
  of~$S_1,S_2$ and the fact that~$\mathcal{S}'$ is four-regular
  in~$(G^+)'$ after the reduction (else we could contract degree-two vertices) and thus the corners have degree four in the swirl. Combining the above we deduce that the instance is minimal infeasible and by
  \cref{thm:two_paths_Frank}~$(G^+)$ can be Euler-embedded; in particular~$\SSS' \subseteq (G^+)'$ is part of the Euler-embedding. Finally undo any contraction of edges on the outer-cycle~$S_s$ that were part of said cuts of~$\SSS$ that have been performed to obtain~$\SSS'$ from~$\SSS$ by sub-dividing edges accordingly. Then~$G' \coloneqq (G^+)' \cup R)-D$ does the trick where we `glue' both graphs back together along their common cycle~$S_s$.

\end{proof}

Note that since the final flat swirl can be Euler-embedded, the concentric swirl-cycles have concentric outlines bounding discs, where the discs form a nested sequence, i.e., every swirl-cycle is bounded between its two neighbouring swirl-cycles. 

\begin{observation}\label{obs:euler_embedded_swirl_in_disc_is_insulation}
    Let~$G$ be Eulerian and let~$\SSS \subset G$ be an Euler-embedded~$t$-swirl in a disc~$\Delta$ with~$\SSS=\bigcup_{i=1}^tS_i$ for swirl cycles~$S_i$ with~$1\leq i \leq t$ and some~$t \in \N$ as in \cref{def:swirl}. Then for every~$1 \leq i \leq t$ the outline~$S_i^\star$ of the cycle~$S_i$ bounds a disc~$\Delta(S_i^\star) \subset \Delta$ such that~$\Delta(S_i) \subseteq \Delta(S_j)$ for all~$1 \leq i\leq j \leq t$ and~$S_i \subset G[\Delta(S_i)\setminus \Delta(S_{i-1})]$.
\end{observation}
\begin{proof}
    This is obvious using the fact that the swirl-cycles are concentric and the graph is Euler-embedded in a disc.
\end{proof}

This is exactly how one would imagine a swirl in the first place; we will come back to this once we talk about \emph{insulations} starting in \cref{subsec:swirl_ins_min_counterex}.

\section{Routing with Routers}
\label{sec:Routing}

In this section we deal with the case that we have found a large router~$\RRR$ in~$G$ when trying to entangle a swirl from a directed cylindrical wall in \cref{subsec:finding_the_swirl}. We leverage a result due to Frank \cite{Frank1988} to prove that, given an instance~$G+D$, we can decide in~$fpt$-time whether there exists a solution routing~$p \coloneqq \Abs{E(D)}$ edge-disjoint paths through the wall~$\WWW$ using the router. We will use this to prove that given a large router in~$G$ there exists an \emph{irrelevant cycle} in the router, that is, there exists a solution that misses said cycle of the router. 

   \paragraph{The Setup:} Throughout this section we assume an instance~$G+D$ of the Eulerian edge-disjoint paths problem to be given. As it turns out to be of importance which terminals of~$V(D)$ are sources and which are sinks we rewrite the instance as~$(G,S,T;D)$ (or simply~$(G,S,T)$) where~$D$ is still the demand graph and~$S,T\subseteq V(D)$ such that~$s \in S$ and~$t \in T$ if there exists an edge~$(t,s) \in D$. Note that by assumption (see \cref{lem:reduce-to-degree-4}) we can omit the case that~$S$ and~$T$ are multi-sets, that is, we can assume that~$\Abs{S}=\Abs{T} = \Abs{E(D)}$ and the vertices are pairwise disjoint. In particular we can enumerate the sets~$S=\{s_1,\ldots,s_p\}$ and~$T=\{t_1,\ldots,t_p\}$ such that~$E(D)=\{(t_i,s_i)\mid 1\leq i \leq p\}$.

\medskip

Recall that, given a graph~$G$, a vertex set~$X \subset V(G)$ and a subgraph~$H\subset G$, we write~$\delta_H^+(X)$ (respectively~$\delta_H^-(X)$) to mean the number of edges with a tail (respectively head) in~$X$ and a head (respectively tail) in~$\bar{X}$. The following definitions and theorem are taken from \cite{Frank1988}.

\begin{definition}[Directed Cut Criterion]
    Let~$G+D$ be an Eulerian digraph. 
    If~$\delta^{-}_G(X) \geq \delta^{+}_H(X)$ for every~$X \subseteq V(G+D)$, then we say that~$G+D$ satisfies the \emph{directed cut criterion}.
    \label{def:directed_cut_criterion}
\end{definition}
\begin{remark}
    The directed cut criterion is satisfied if and only if~$\delta_G^+(\bar{X}) \geq \delta_H^-(\bar{X})$.
\end{remark}

\begin{definition}[Star]
    A star~$H$ (centered at~$p \in V(H)$) is a digraph~$H$ such that either all edges have~$p$ as head or all edges have~$p$ as tail.
\end{definition}

Frank \cite{Frank1988} proved that if the demand graph~$D$ is the union of two \emph{stars}, then the directed cut  criterion is satisfied (which can be verified in polynomial time on the instance) if and only if the instance is a \emph{YES}-instance.

\begin{theorem}[Theorem 2.3 in \cite{Frank1988}]
    If~$G + D$ is an Eulerian digraph, and~$D$ is the union of two stars, then the directed cut criterion is necessary and sufficient for the solvability of the directed edge-disjoint paths problem. (In particular it can be checked in polynomial time on the instance).
    \label{thm:Frank_Algorithm}
\end{theorem}
    
We start by giving a construction that we need in order to use \cref{thm:Frank_Algorithm}.

\paragraph{Constructing~$G^v$:} We describe a reduction that given an instance~$(G,S,T;D)$ constructs a new instance~$(G^v,S,T;D^v)$ such that the demand graph~$D^v$ is the union of two stars.

Throughout this section when we are given some branching set~$B_{\RRR}$ of some router~$\RRR\subset G$ we tacitly assume that any router-cycle~$C \in \RRR$ contains at most one vertex in~$B_{\RRR}$ unless stated otherwise.

\begin{reduction}
    Let~$(G,S,T)$ be an instance of Eulerian edge-disjoint paths with~$\Abs{S} = \Abs{T} =p  \in \N$. Further let~$\RRR$ be a~$t$-Router with a given branching set~$B_\RRR$. We define~$G^v_\RRR$ to be the graph obtained from~$G$ by adding a vertex~$v$ to~$G$ and connecting~$v$ to all the branch vertices~$b \in B_\RRR$ in both directions. We define the demand graph~$D^v=(V,E)$ via
    \begin{align*}
        V &\coloneqq S \cup T \cup \{v\} \\
        E &\coloneqq \{(t_i,v),(v,s_i) \mid s_i \in S,\: t_i \in T\}.
    \end{align*}
    Then~$G^v_\RRR + D^v$ is Eulerian and~$(G^v_\RRR,S,T;D^v)$ is an instance of Eulerian edge-disjoint paths. If~$\RRR$ is clear from the context we omit the subscript and write~$G^v$.
    \label{red:adding_v}
\end{reduction}
\begin{remark}
    Note that~$D^v$ is independent of~$\RRR$ and solely depends on~$S,T$ and~$v$. In particular~$D^v$ forgets about the intrinsic matching between~$t_i$ and~$s_i$. Similarly~$G_\RRR^v$ is independent of~$S \cup T$ and solely depends on~$\RRR, B_\RRR$ and~$v$. See \cref{fig:Gv} for an exemplary representation.
\end{remark}

\begin{figure}
    \centering
    \includegraphics[width=.4\linewidth]{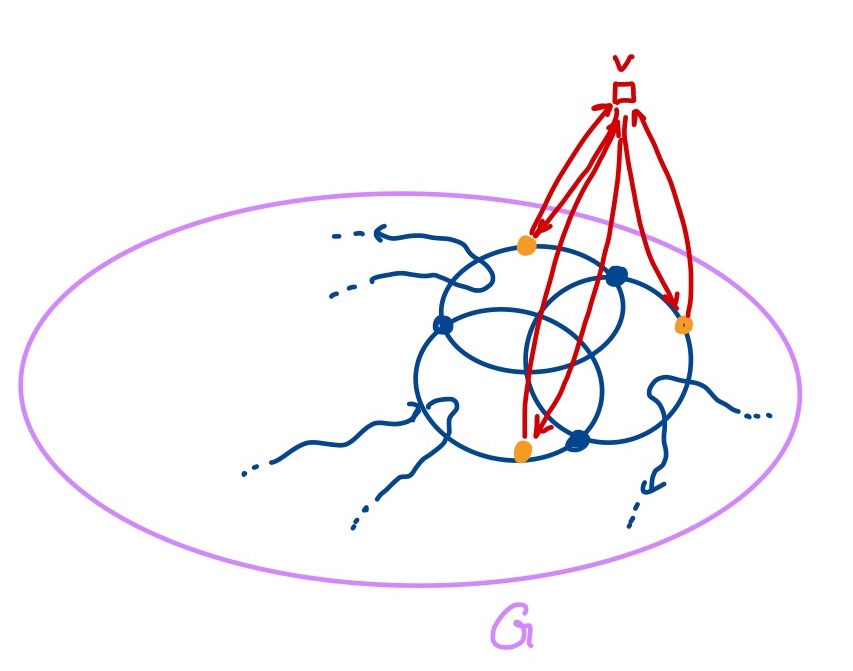}
    \caption{A representation of~$G^v$ by connecting~$v$ to a~$3$-router in~$G$}
    \label{fig:Gv}
\end{figure}

Using \cref{thm:Frank_Algorithm}, we can decide the instance~$(G^v_\RRR,S,T;D^v)$ in~$fpt$-time by verifying whether the directed cut criterion is satisfied (see \cref{def:directed_cut_criterion}). Suppose now that~$(G^v_\RRR,S,T;D^v)$ is a \emph{YES}-instance. If we were to delete the edges of the router~$\RRR$ in~$G^v_\RRR$ and still get a \emph{YES}-instance, then one can show that the solution can be modified by rerouting the paths between branch vertices that use the vertex~$v$ via paths in our router in total then omitting the vertex~$v$. Thus we can find a solution to the original problem, proving that~$(G,S,T;D)$ is in turn a \emph{YES}-instance. We will make this argument more precise in what is to follow.

\smallskip

Let~$H\subset G$ be a subgraph, then we write~$\operatorname{deg}_H(v)=k~$ for~$v \in V(H)$ if the degree of~$v$ in~$H$ is~$k\in \N$. The first key observation of this section reads as follows.

\begin{observation}
    Let~$(G,S,T;D)$ be an instance of the Eulerian edge-disjoint paths problem with~$\Abs{E(D)} = p \in\N$. Let~$\RRR=C_1 \cup \ldots \cup C_{2p} \subset G$ be a~$2p$-router and let~$X \subseteq V(G)$ with~$S\cup T \subseteq X$ induce a cut of order~$\leq 2p$ such that there exists~$1 \leq i \leq 2p$ with~$C_i \subset G[\bar{X}]$. Then there exists no~$(p+1)$-router~$\RRR' \subset \RRR$ such that~$\RRR' \subset G[X]$. More precisely, there exist at most~$p$ cycles of~$\RRR$ that are not disjoint from~$G[X]$.
    \label{obs:most_of_router_on_one_side}
\end{observation}
\begin{proof}
    This follows at once from the router-structure. If~$(p+1)$ distinct router-cycles have a vertex in~$G[X]$ then we find~$(p+1)$ edge-disjoint paths from these vertices along their router-cycles to~$C_i \subset G[\bar{X}]$ contradicting the assumption on the cut induced by~$X$.
\end{proof}

Given \cref{obs:most_of_router_on_one_side} we define \emph{pushed router-cuts}.

\begin{definition}
     Let~$G+D$ be Eulerian with~$\Abs{E(D)} = p\in \N$ and let~$\RRR = C_1 \cup \ldots \cup C_{2p+1}$ be a~$(2p+1)$-router. Let~$X \subseteq V(G)$ satisfy the following:
    \begin{itemize}
        \item[1.] It holds~$S\cup T \subseteq X$ and~$X$ induces a cut in~$G$ such that there exists~$1 \leq i \leq 2p+1$ with~$C_i \subset G[\bar{X}]$,
        \item[2.]~$X$ is of minimal order subject to 1.,
        \item[3.]~$X$ is inclusion-maximal with respect to 2.
    \end{itemize}
Then we call~$X$ a \emph{pushed router-cut for~$(G,S,T,\RRR)$}.
\label{def:pushed_cut}
\end{definition}

A helpful result in the direction of proving what we mentioned above reads as follows.
\begin{lemma}
    Let~$G+D$ be Eulerian with~$\Abs{E(D)} = p\in \N$ and let~$\RRR = C_1 \cup \ldots \cup C_{2p+1}$ be a~$(2p+1)$-router. Let~$X \subseteq V(G)$ be a pushed router-cut for~$(G,S,T,\RRR)$.
    Then the following holds true:
    \begin{itemize}
        \item[a)] The order of the cut induced by~$X$ is~$\rho \leq 2p$, and
        \item[b)] For every~$1 \leq j \leq 2p+1$ such that~$C_j \subset G[\bar{X}]$ a pushed router-cut~$Y$ for~$(G-C_j,S,T,\RRR-C_j)$ satisfies~$X\subseteq Y$.
    \end{itemize}
\label{lem:pushed_X_contained_in_Y}
\end{lemma}
\begin{proof}
    Part a) of the lemma follows from the fact that~$S\cup T$ induces a cut of order~$2p$ satisfying 1. of \cref{def:pushed_cut}. Thus we are only left to prove part b).
    
    Let~$Y$ be a pushed router-cut for~$(G-C_j,S,T,\RRR-C_j)$ of order~$\rho' \leq \rho$ (clearly it cannot be bigger) where~$C_j \subset G[\bar{X}]$. Let~$A \coloneqq X\cap Y$,~$B\coloneqq X\setminus Y$,~$D = Y \setminus X$ and~$C=V(G)\setminus (A\cup B \cup D)$, then they all induce cuts in~$G$. Using \cref{obs:most_of_router_on_one_side} we deduce that there exists some router-cycle~$C' \in \RRR$ with~$C' \subset G[C]$.
    For two sets~$Z_1,Z_2 \subset V(G)$ we write~$\delta(Z_1,Z_2)$ to mean the number of edges in~$G$ with their heads (or tails) in~$Z_1$ and their tails (or heads) in~$Z_2$ respectively.
    \begin{Claim}
        It holds~$\delta(A, B) =  \delta(B, D\cup C)$
    \end{Claim}
    \begin{ClaimProof}
        Suppose that~$\delta(A, B) <  \delta(B, D\cup C)$, then~$A$ induces a cut of lower order than~$X$ in~$G$ while still satisfying 1. of \cref{def:pushed_cut} as a contradiction to assumption 2. of \cref{def:pushed_cut} for~$X$. To see this first note that~$S \cup T \subset A$ by definition and~$A \cap C = \emptyset$. The rest easily follows from a straightforward inspection of~$\delta(X)=\delta(A,C) + \delta(A,D) + \delta(B,C) + \delta(B,D)$. 
        
\smallskip

        Suppose that~$\delta(A, B) >  \delta(B, D\cup C)$, then~$A\cup B \cup D$ induces a cut of lower order than~$Y$ in~$G-C_j$ for~$B \cap C_j = \emptyset$ since~$X \cap V(C_j) = \emptyset$. But~$A\cup B \cup D$ still satisfies 1. of \cref{def:pushed_cut} as a contradiction to assumption 2. of \cref{def:pushed_cut} for~$Y$. Again the rest follows by an easy inspection of~$\delta(A\cup B \cup D)$ and~$\delta(Y)=\delta(A,B) + \delta(A,C) +\delta(D,B) + \delta(D,C)$.
    \end{ClaimProof}
    Thus we know that~$\delta(A, B) =  \delta(B, D\cup C)$, which in turn implies that~$B \subset Y$ by assumption 3. of \cref{def:pushed_cut} on~$Y$ and the fact that~$B \cap V(C_j) = \emptyset$. Again this follows from noting that~$\delta(Y)= \delta(A,B) +\delta(A,C) + \delta(D,B) + \delta(D,C)$ and~$\delta(A\cup B \cup D)=\delta(A,C) + \delta(B,C) + \delta(D,C)$). 
\end{proof}

We are now ready to prove the first main theorem of this section which will allow us to delete a large router of the graph while keeping the instance equivalent.

\begin{theorem}
     There exists a function~$g(p)$ satisfying the following. Let~$G$ be a graph and let~$S,T \subset V(G)$ with~$\Abs{S} = \Abs{T} = p \in \N$ such that any matching~$D$ of edges connecting vertices in~$T$ to vertices in~$S$ results in an Eulerian graph~$G+D$. Let~$\RRR=C_1\cup\ldots C_{g(p)}$ be a~$g(p)$-router in~$G$. Further assume that any pushed router-cut for~$(G,S,T,\RRR)$ is of order~$2p$.
    
   Then there exists a cycle~$C \in \RRR$ such that any pushed router-cut for~$(G-C,S,T,\RRR-C)$ is of order~$2p$. In particular we can determine~$C$ in~$fpt$-time on~$p$.
    \label{thm:Frank_router_reduction} 
\end{theorem}
\begin{proof}
Let~$g(p) = f(p)^2+p+2$ for some~$f(p) \geq 2p$; we claim that~$g(p)$ satisfies the theorem. To this extent let~$X$ be some router-pushed cut for~$(G,S,T,\RRR)$ of order~$2p$ (which can easily be computed given~$S,T$ and~$\RRR$ using several iterations of Menger). Then by \cref{obs:most_of_router_on_one_side} we deduce that there exists an~$(f(p)^2+2)$-router~$\RRR_{f^2} \subset \RRR$ with~$\RRR_{f^2} \subset G[\bar{X}]$. For each router-cycle~$C_i \in \RRR_{f^2}$ subdivide an edge~$e \in E(C)$ by introducing a vertex~$b_i$ and let~$B_{\RRR_{f^2}} = \{b_0,\ldots, b_{f(p)^2+1}\}$. (Note that subdividing edges does not alter the instance in question). Let~$\RRR_f \subset \RRR_{f^2}$ be any~$f(p)$-router and let~$B_{\RRR_f} \subset B_{\RRR_{f^2}}$ be its respective branching set.
   \begin{claim}
 ~$(G^v_{\RRR_f},S,T;D^v)$ is a \emph{YES}-instance.
       \label{thm:Frank_router_reduction_claim1}
   \end{claim}
   \begin{ClaimProof}
       For suppose otherwise. Then there exists some cut~$Y \subset V(G^v)$ violating the directed cut criterion (see \cref{def:directed_cut_criterion}) by \cref{thm:Frank_Algorithm}. In particular there exists~$Y \subset V(G^v)$ such that~$\delta_H^+(Y)> \delta_G^-(Y)$. Now first it is easy to see that~$v \notin Y$, for otherwise look at~$\bar{Y}$ where the argument follows by symmetry using~$\delta_H^-(\bar{Y})> \delta_G^+(\bar{Y})$. Thus~$v \in \bar{Y}$ and without loss of generality~$S \cup T \subset Y$ for every vertex in~$S\cup T$ is of degree one in~$G$ and thus moving them to~$Y$ does not change the difference~$\delta_H^+(Y) - \delta_G^-(Y)$. Finally assume that some vertex~$b_i \in Y$ for some~$b_i \in B_{\RRR_f}$. Then adding~$b_i$ to~$\bar{Y}$ would not change the order of the cut, for there is only one incoming and one outgoing edge to~$b_i$ in~$G$. Thus~$Y$ induces a cut in~$G$ of order~$\leq 2p-2$ such that~$Y$ contains no branch-vertex of any router-cycle of~$\RRR_f$. Since~$X$ was a pushed router-cut for~$(G^v,S,T,\RRR)$ however, \cref{def:pushed_cut} implies that no cycle of~$\RRR_f$ is contained in~$G[\bar{Y}]$, for otherwise~$Y$ is a smaller cut contradicting that~$X$ was pushed. But then we get~$f(p)\geq 2p$ paths from~$G[Y]$ to~$B_{\RRR_f} \subset G[\bar{Y}]$; a contradiction.
   \end{ClaimProof}
   Recall that to define~$D^v$ we did not need~$D$. We inductively define edge-disjoint~$p$-routers~$\RRR^1_p,\ldots,\RRR^{f(p)}_p \subset \RRR_{f^2}\eqqcolon \RRR^0$ with respective branching sets~$B_{\RRR_i} \subset B_{\RRR_{f^2}}$ (note that the routers, and thus the branching sets, are pairwise disjoint) as follows. Let~$\CCC^0 =\{ C_1 ,\ldots ,C_{f(p)^2+2}\}$ be the set of router-cycles in~$\RRR^0$. 
   
   \smallskip 
   
   Let~${\RRR_1}' \subset \RRR^0$ be any~$f(p)$-router. Let~$\LLL^1=\{L^1_1,\ldots,L^1_p\}$ be a solution to~$(G^v_{{\RRR_1}'},S,T;D^v)$ which exists by \cref{thm:Frank_router_reduction_claim1}. Then the following holds.

   \begin{Claim}
       For every~$1 \leq j \leq p$ there exists a sub-path~${L^1_j}' \subset L^1_j$ such that~${L^1_j}' \subset G$, they are pairwise edge-disjoint and they each end in different cycles of~${\RRR_1}'$. \label{thm:Frank_router_reduction_claim2}
   \end{Claim}
   \begin{ClaimProof}
       For every path~$L^1_j$ choose~${L^1_j}'$ such that~${L^1_j}'$ starts in~$s_j$ and ends in some~$b_j \in B_{{\RRR_1}'}$ such that it does not visit any other~$b \in B_{{\RRR_i}'}$. These do the trick (for there is only one possible edge in~$G$ ending in~$b_j$, every other edge must come from~$v$).
   \end{ClaimProof}

   For~$\LLL^1$ let~$C^1_1,\ldots, C^1_p \in \CCC^0$ be the first~$p$ distinct cycles visited by~$\LLL^1$ (which exists for~$\LLL^1$ visits at least~$p$ distinct cycles by the above claim). Define~$\RRR^1_p \coloneqq C^1_1 \cup \ldots \cup C^1_p~$. Truncate~$\LLL^1$ so that every path in~$\LLL^1$ is an~$S$ to~$\RRR^1_p$ path. Finally let~$\CCC^1 \coloneqq \CCC^0 \setminus \{C^1_1,\ldots, C^1_p \}$ and~$\RRR^1 \coloneqq \RRR^0-\RRR^1_p$.
    
    Now inductively define~$\RRR^i_p$ for~$1 \leq i \leq f(p)$ given~$\RRR^{i-1}$ and~$\CCC^{i-1}$ by maintaining the following:
    \begin{itemize}
        \item[1.]~$\RRR^{i-1}$ is an~$\big(f(p)^2-p(i-1) +2\big)$-router, and
        \item[2.]~$\RRR^1_p,\ldots,\RRR^{i-1}_p$ are pairwise edge-disjoint.
    \end{itemize}

    Let~${\RRR_i}' \subset \RRR^{i-1}$ be any~$f(p)$-router which exists for~$f(p)^2- p(i-1) +2 \geq f(p)^2 - p\cdot f(p) \geq 2p^2 > f(p)$. Let~$\LLL^i=\{L^i_1,\ldots,L^i_p\}$ be a solution to~$(G^v_{{\RRR_i}'},S,T;D^v)$ which exists by \cref{thm:Frank_router_reduction_claim1}. For~$\LLL^i$ let~$C^i_1,\ldots, C^i_p \in \CCC^{i-1}$ be the first~$p$ distinct cycles visited by~$\LLL^i$ (which exist for~$\LLL^i$ visits at least~$p$ distinct cycles by \cref{thm:Frank_router_reduction_claim2}, replacing~$L^1_j$ by~$L^i_j$ and~$\RRR_1'$ by~$\RRR_i'$). Define~$\RRR^i_p \coloneqq C^i_1 \cup \ldots \cup C^i_p~$.
    Finally let~$\CCC^i \coloneqq \CCC^{i-1} \setminus \{C^i_1,\ldots, C^i_p \}$, and~$\RRR^i \coloneqq \RRR^{i-1}-\RRR^i_p$. Then clearly Hypothesis 1. holds by definition, and so does Hypothesis 2 by construction. Truncate~$\LLL^i$ so that every path in~$\LLL^i$ is a path from~$S$ to~$\RRR^i_p$.

    Finally let~$\RRR_{\text{rest}} \coloneqq \RRR_{f^2} - \bigcup_{i=1}^{f(p)}R^i_p$, then clearly~$\RRR_{\text{rest}}$ is a router of order~$>1$. Let~$C_* \in \RRR_{\text{rest}}$. Let~$Y \subset V(G)$ be a pushed router-cut for~$(G-C_*,S,T,\RRR-C_*)$. Then~$S\cup T \subseteq Y$ by \cref{def:pushed_cut}. Now by assumption on~$X$---namely 2. of \cref{def:pushed_cut}---we know that~$Y$ is a cut of order~$\geq 2p$ in~$G$. 
    \begin{Claim}
        It holds true that~$\Abs{\delta_{G-C_*}(Y)} \geq 2p$.
    \end{Claim}
    \begin{ClaimProof}
        Suppose not and let~$D \coloneqq Y \setminus X$ and~$C \coloneqq V(G-C_*)-Y$. Since~$X$ and~$Y$ are pushed router-cuts \cref{lem:pushed_X_contained_in_Y} (b) implies that~$X \subset Y$ (note that~$Y\neq X$, else we are done). In particular now there exists~$C^* \in \RRR$ such that~$C^* \in G[C]$ by assumption 1. of \cref{def:pushed_cut} for~$Y$. Furthermore~$Y$ induces a~$<2p$ cut in~$G-C_*$ by the above, but then again by assumption on~$X$ it induces a~$\geq 2p$-cut in~$G$ for~$X$ is a pushed router-cut for~$(G,S,T,\RRR)$. In particular~$V(C_*) \cap Y \neq \emptyset$. By the above construction it holds~$\RRR^i_{p} \subset G[\bar{X}]$ for every~$1 \leq i \leq f(p)$. 

       Note that for every~$1 \leq i \leq f(p)$, it holds that~$\LLL^i$ is a~$p$-linkage from~$S$ to~$\RRR^i_p$ and by construction it omits~$C_*$.
        Let~$\RRR^i_C \subset \RRR^i_{p}$ be maximal with the property that every cycle contains an endpoint of some path in~$\LLL^i$ (which by construction omits~$C_*$) that is in~$C$. Let~$\RRR^i_D~$ be maximal with the property that every cycle contains an endpoint of some path in~$\LLL^i$ that is in~$D$. Then~$\RRR^i_D$ is a router of order at least~$1$ for else all the paths in~$\LLL^i$ end in~$\RRR^i_C$ and thus we have found~$p$ edge-disjoint~$S$ to~$C$ paths contradicting the assumption that~$Y$ is a cut of order~$<2p$ using Eulerianness; let~$R_i \in \RRR^i_D$ be a router-cycle. Again by construction there is some path in~$\LLL^i$ that visits~$R_i$ before possibly visiting~$C_*$. But then this means that there exist disjoint cycles~$R_1,\ldots,R_{f(p)} \subset G[\bar{X}]$ each having at least one vertex in~$D$. They are pairwise distinct by construction and they each intersect~$D$. Thus together with~$C^* \in \RRR$ for which it holds~$C^* \in G[C]$ we deduce that there exist~$f(p) \geq p$ edge-disjoint paths from~$D$ to~$C^*$; using Euerianness this yields a contradiction to~$Y$ inducing a~$<2p$ -cut.
    \end{ClaimProof}
    
    The above yields an algorithm to find~$C_*$ using the fact that \cref{thm:Frank_Algorithm} can be tested in polynomial time on the instance, and the respective paths can be found in polynomial time time too as shown in \cite{Frank1988}. In particular note that the cycle~$C_*$ was \emph{independent} of the demand graph~$D$.
\end{proof}

Next we show how to use this \cref{thm:Frank_router_reduction} to reroute solutions that use large routers.

\begin{lemma}
     There exists a function~$g(p)$ satisfying the following. Let~$G$ be a graph with designated vertices~$S,T \subset V(G)$ such that~$\Abs{S} = \Abs{T} =p \in \N$ such that any matching~$D$ of edges connecting vertices in~$T$ to vertices in~$S$ results in an Eulerian graph~$G+D$.  Let~$\RRR \subset G$ be a~$g(p)$-router in~$G$. Further assume that every pushed router-cut~$X$ of~$(G,S,T,\RRR)$ is of order~$2p$. Then there exists some cycle~$C \in \RRR$ such that~$(G,S,T;D)$ is a \emph{YES}-instance if and only if~$(G-C,S,T;D)$ is a \emph{YES}-instance for every possible demand graph~$D$. In particular we can find~$C$ in~$fpt$-time parameterized by~$p$.

    \label{thm:rerouting_with_routers}
\end{lemma}
\begin{proof}
    Let~$g(p)\geq (6p+1)\cdot g_{\ref{thm:Frank_router_reduction}}(p)$ and let~$\RRR$ be a~$g(p)$-router. Let~$\RRR_1,\ldots,\RRR_{g_{\ref{thm:Frank_router_reduction}}} \subset \RRR$ be pairwise edge-disjoint~$(6p+1)$-routers. Clearly then for each~$1 \leq i \leq g_{\ref{thm:Frank_router_reduction}}(p)$ we may view~$\RRR_i$ as a single cycle (for it is an Eulerian graph) of an~$g_{\ref{thm:Frank_router_reduction}}(p)$-router. By \cref{thm:Frank_router_reduction} there exists~$1 \leq j \leq g_{\ref{thm:Frank_router_reduction}}(p)$ such that any pushed router-cut in~$(G-\RRR_j,S,T,\RRR-\RRR_j)$ is of order~$2p$.

    Without loss of generality assume~$j > 2p$. Now we define~$2p$ disjoint triples each consisting of~$3$ router-cycles of~$\RRR_j$ as follows. Choose~$C^1,\ldots, C^{2p} \in \RRR_j$ pairwise distinct. For~$1 \leq i \leq 2p$ let~$b_i \in V(C^i)\cap V(R_i)$ for distinct cycles~$R_i \in \RRR^i$. Let~$\RRR' \coloneqq R_1 \cup \ldots \cup R_{2p}$ and let~$B_{\RRR'} = \{b_1,\ldots,b_{2p}\}$. For each of the cycles~$C^i$ we choose two more cycles in~$\RRR_j$, namely let~$C^1_{l},C^1_{r},\ldots,C^{2p}_{l},C^{2p}_{r}$ be distinct router-cycles of~$\RRR_j - (C^1\cup \ldots \cup C^{2p})$ such that~$C^i$ visits~$c^i_l,b_i,c^i_r$ in that order for~$c^i_l \in V(C^i_l) \cap V(C^i)$ and~$c^i_r \in V(C^i_r) \cap V(C^i)$ (which obviously exists for this is a cyclic permutation of three elements) for every~$1\leq i \leq 2p$. Since~$\RRR_j$ was a~$(6p+1)$-router, there is a cycle~$C^* \in \RRR_j$ left that has not been matched up.
    \begin{Claim}
        If~$(G^v_{\RRR'}-\RRR_j,S,T;D^v)$ is a \emph{YES}-instance, then so is~$(G-C^*,S,T;D)$ for any demand graph~$D$.
    \end{Claim}
    \begin{ClaimProof}
    Let~$\LLL=\{L^s_1,\ldots,L^s_p,L^t_1,\ldots,L^t_p\}$ be a solution to~$(G^v_{\RRR'}-\RRR_j,S,T;D^v)$ (recall that~$D^v$ is defined without the knowledge of any demand graph). Now fix any demand graph~$D$. Let~$L_\ell' \coloneqq L_{\ell}^s \cup L_{\ell}^t$ be a path from~$s_\ell$ to~$ t_\ell$ for~$1\leq \ell\leq p$, then~$L_i'$ and~$L_j'$ are pairwise edge-disjoint by construction for any distinct~$1 \leq i,j \leq p$. Suppose it uses the path~$(b_i,v,b_j)$ for some~$1\leq i \neq j \leq f(p)$ (note that every path uses at most one such path for else we could reroute at~$v$). We proceed by rerouting~$L_\ell'$ to a path~$L_\ell$ from~$s_\ell$ to~$t_\ell$ in~$G$ using~$\RRR_j$ that omits~$v$. For every~$1\leq i,j \leq 2p$, let~$c_{i,j} \in V(C^i_r)\cap V(C^j_l)$, which exists given their router structure. To reroute the path now simply use~$b_iC^ic^i_r\cup c^i_rC^i_rc_{i,j} \cup c_{i,j}C^j_lc^j_l \cup c^j_lb_j$. All of these paths---and thus~$L_1,\ldots,L_p$---are clearly pairwise edge-disjoint by construction for any~$(b_i,v,b_j)$-paths with~$1\leq i \neq j \leq f(p)$. To see this note that there exists at most one path entering each~$b_i$ from~$\RRR$ since~$\operatorname{deg}_{G-\RRR_j}(b_i) =2$.

    Since none of the paths used~$C^*$, this proves the claim.
    \end{ClaimProof}
    Given the claim it is enough to prove that~$(G^v_{\RRR'}-\RRR_j,S,T;D^v)$ is a \emph{YES}-instance. But since~$\operatorname{deg}_{\RRR'-\RRR_j}(b) = 2$ for very~$b \in B_{\RRR'}$, the claim follows exactly as in Claim 1. of \cref{thm:Frank_router_reduction} using that~$X$ was a pushed router-cut.
\end{proof}

Armed with \cref{thm:rerouting_with_routers} we are ready to prove the \emph{Irrelevant Cycle Theorem for Routers}.

\begin{theorem}[Irrelevant Cycle Theorem for Routers]
   For every function~$f(p)$ there exists a function~$t(p)\coloneqq t(p;f)$ satisfying the following. Let~$(G,S,T;D)$ be an instance of the Eulerian edge-disjoint paths problem with~$\Abs{E(D)} = p \in \N$. Let~$\RRR\coloneqq C_1\cup \ldots \cup C_{t(p)} \subset G$ be a~$t(p)$-router in~$G$. Then there exists a sub-router~$\RRR_{f} \subseteq \RRR$ of size~$f(p)\leq t(p)$ which can be found in~$fpt$-time such that each router-cycle~$C \subset \RRR_{f}$ is irrelevant to the instance, i.e.,~$(G-C,S,T;D)$ is an equivalent instance.
    \label{thm:irrelevant_cycle_in_router_general}
\end{theorem}
\begin{proof}
    Let~$X$ be a pushed router-cut of~$(G,S,T,\RRR)$ of order~$2\rho \leq 2p$; it cannot be larger by (a) of \cref{lem:pushed_X_contained_in_Y}. Then there exists a~$(t(p)-p)$-router~$\RRR' \subset \RRR$ such that~$\RRR \subset G[\bar{X}]$. Look at~$G[[[\bar{X}]]]$ and let~$\delta(X)$ be the edges in the cut. Let~$S' \coloneqq\{s_1,\ldots,s_{\rho}\} \subset X$ and~$T' \coloneqq \{t_1,\ldots,t_{\rho}\} \subset X$ such that~$s_i$ is the tail of an edge in~$\delta(X)$ and~$t_i$ is the head of an edge in~$\delta(X)$ for every~$1 \leq i \leq \rho$. (Note that some~$s_i$ may equal some~$t_j$ which is no problem, one could simply split the vertex into two for example similar to the \cref{def:Euler-restriction} of Euler-restrictions). By construction there is no pushed router-cut for~$(G[[[\bar{X}]]],S,T,\RRR')$ of smaller order than~$2\rho$. Thus~$G[[[\bar{X}]]]$ together with~$S'$ and~$T'$ satisfies \cref{thm:rerouting_with_routers} implying the existence of an irrelevant cycle~$C \in \RRR'$. That is, for any possible demand graph~$D'$ matching~$T'$ to~$S'$ it holds that the instance~$(G[[[\bar{X}]]]-C,S',T';D')$ is equivalent to the instance~$(G,S',T';D')$.

    We claim that~$C$ is irrelevant to the instance~$(G,S,T;D)$. For if it is a \emph{NO}-instance it is trivial, thus assume it to be \emph{YES}-instance and let~$\LLL$ be a solution. Let~$\LLL_X$ be the restriction of~$\LLL$ to~$G[[[\bar{X}]]]$, that is the linkage given by~$\LLL \cap G[[[\bar{X}]]]$ in the obvious way (see \cref{lem:switching_linkages_at_cuts_prelims}) . Since~$\delta(X) = 2\rho$ then~$\LLL_X$ is at most a~$\rho$-linkage (note that~$S\cup T \subset G[{X}]$) connecting part of~$S'$ to part of~$T'$. Let~$D'$ be any demand graph modelling the connection of~$\LLL_X$. Let~$\LLL^X$ be a solution to the instance~$(G[[[\bar{X}]]]-C,S',T';D')$. Then using \cref{lem:switching_linkages_at_cuts_prelims} we can exchange both linkages to get a linkage in~$G$ with pattern~$D$ (thus equivalent to~$\LLL$) omitting~$C$ and thus completing the proof.
\end{proof}

This concludes our analysis of routers and the proof of the irrelevant-cycle theorem given large routers in~$G$. We continue with the case that we are given a graph~$G$ of high tree-width without large routers. In particular said graph admits a large flat swirl by \cref{thm:flat_swirl_theorem} for which we will prove that its most deeply nested cycle is irrelevant to the edge-disjoint path instance. The proof of this will take up all of \cref{sec:charting_eulerian_digraphs,sec:structure_of_min_examples,sec:shippings,sec:structure_thms}.

\section{Charting Eulerian Digraphs}
\label{sec:charting_eulerian_digraphs}

In this section we consider Eulerian digraphs that are 'quasi-Euler-embedded' in some surface. The precise meaning of quasi-Euler-embeddings will be given later in the form of \emph{coastal maps}. In a nutshell, a graph~$G+D$ is quasi-Euler-embeddable or admits a \emph{weak coastal map} if there exist two graphs~$\Gamma,I$ (likely non-eulerian) with~$G = \Gamma \cup I$ such that~$\Gamma$ can be Euler-embedded into some surface~$\Sigma$ with boundary,~$I$ is the disjoint union of a bounded number of components each of which is connected to~$\Gamma$ exactly at the vertices of~$\Gamma$ which are embedded on one cuff of~$\Sigma$ and the components themselves do not contain large linkages between the vertices they share with~$\Gamma$. That is~$G$ is Euler-embeddable up-to a bounded number of \emph{bounded-depth islands}, the components of~$I$ (see \cref{def:strong_island,def:weak_island} and more generally \cref{subsec:coastlines,subsec:charting_an_island}). The islands in our setting correspond to what is called \emph{vortices} in the graph minor structure theorem for undirected graphs (see e.g.~\cite{GMXVI,GMXVII,KawarabayashiTW2021}) and they will play a similar role in our proof. However, as we work on directed graphs, the structure of islands and vortices differ in some aspects relevant to our proof. 

Our focus in this section is to understand the structure of solutions~$\LLL$ for~$G+D$ which are \emph{unique}, i.e., there is no other linkage in~$G$ with the same pattern that is \emph{exhaustive} in the sense that~$\LLL$ uses every edge of~$G$. We will call such linkages \emph{rigid} (see \cref{def:rigid_linkage}). The main result of this section is to prove that if~$G$ is quasi-Euler-embedded as above, then~$G+D$ does not admit such a rigid solution unless the tree-width of~$G$ is bounded by some function in~$p = \Abs{E(D)}$. The proof strategy is closely related to, and inspired by the results in~\cite{GMXXI} roughly following the same line of arguments. However, many of the proofs we provide differ significantly; it turns out that the fact that we work with Eulerian digraphs and edge-disjoint paths allows us to simplify some of the arguments and to provide shorter and sometimes smoother proofs; while at some points we will need a lot more and different analysis that is of interest on its own.  

\smallskip

The section is organised as follows: As outlined above, we work with Eulerian digraphs that are partitioned into a part~$\Gamma$ that is Euler-embedded in a surface~$\Sigma$ and a bounded number of parts~$\III$---the islands---that cannot be embedded alongside~$\Gamma$ into the same surface. The main part of this section deals with analysing how a linkage can `cross' between the embedded and non-embedded part. As we work with edge-disjoint paths, we will focus on the set of edges of the linkage that build the bridge between the embedded and non-embedded parts of~$G$. It turns out that it is easier to capture this formally if we `upgrade' edges to first-class citizens, that is, elements that can exist independent of their endpoints rather than treating edges merely as connections between pairs of vertices. One common way of achieving this is to work on the \emph{line-graph} of the instance, and we will use this argument at some point as well. However, if we were to switch to the line graph of~$G+D$, then our linkages would no longer use all the edges of the resulting graph, but rather all its vertices. The reasoning in this setting turns out to be way more cumbersome, for deleting vertices, splitting vertices and contracting edges are all operations resulting in possibly non-equivalent instances, and worse, in the original graph they may result in the instance becoming non-Eulerian. For this reason we work with incidence structures in this section. That is, our graph consists of edges and vertices together with two incidence relations specifying the tail and the head of an edge. However, we will allow that the subgraphs of the graphs in question contain edges without their tail or head or even without both. We formalise the notion of incidence structures and the corresponding subgraph relation in~\cref{subsec:embedded_notation_revisited}. In~\cref{subsec:rigid-linkages} we prove some fundamental facts about rigid solutions to the edge-disjoint paths problem. 

\smallskip

In~\cref{subsec:coastlines} we introduce the concept of \emph{coast lines} and coastal maps which capture the properties of the islands discussed above that we need for our results. We will then show that with some additional assumptions on the embedding we can modify the embedding to guarantee certain `linkedness' conditions of the coast lines and the respective islands. This leads to the \cref{def:strong_coastal_map} of \emph{strong coastal maps}. That this is possible is the main result, \cref{thm:structure-thm-to-coastal-maps},  of~\cref{subsec:charting_an_island}.

It is however not until~\cref{sec:structure_of_min_examples,sec:shippings} that the definitions and results gathered in this section will come to fruition. And it is only in \cref{sec:irrelevant_cycle_theorem} that we prove the main result of this journey, namely that rigid linkages can only exist in instances~$G+D$ of tree-width bounded by a function in the number of terminals, which in turn will guarantee the existence of an irrelevant cycle if the tree-width is too high.

\subsection{Incidence structures and their drawings}
\label{subsec:embedded_notation_revisited}
As mentioned above we switch our setting and view graphs as incidence structures.

Throughout this section we define a graph as a structure~$G=(V,E,\tail, \head),$ where~$V$ and~$E$ are finite disjoint  sets and for every~$e \in E$ there exist as most one~$u \in V$ such that~$(e,u) \in \tail$ and at most one~$v \in V$ with~$(e,v) \in \head$. If they exist, we call~$u$ the \emph{tail} and~$v$ the \emph{head} of~$e$. An edge with only head or only tail is called a \emph{half-edge}. An edge with neither head nor tail is called \emph{empty}.
We continue to write~$V(G)$ for the set of vertices and~$E(G)$ for the set of edges. 

$G$ is \emph{partial} if it contains an edge without head or without tail. Otherwise it is called \emph{proper}. 
Unless specified otherwise, all our graphs are assumed to be proper. In particular, we always assume that~$G$ and~$D$ are proper digraphs. But we will work with partial sub-digraphs of~$G + D$.

Clearly, every directed graph can be translated into an incidence structure where~$e = (u,v) \in E$ is equivalent to~$(e,u) \in \tail$ and~$(e,v)\in \head$. To ease notation we will write graphs as~$G=(V,E)$ meaning the obvious. Further, whenever we mean an edge together with its endpoints we will write it as~$e=(u,v) \in E(G)$ whereas~$e \in E(G)$ denotes the edge as an element without its endpoints unless specified otherwise. That is,~$e = (u, v)$ denotes the proper subgraph of~$G$ with vertex set~$\{u, v\}$ and edge set~$\{e\}$ whereas by~$e \in E(G)$ we mean the vertex-less partial subgraph of~$G$ containing the edge~$e$ as its single element.

The next definition introduces some non-standard notation that will be very convenient to work with once we talk about cuts in incidence digraphs.

\begin{definition}\label{def:induced-partial-graphs}
    Let~$G \coloneqq (V, E, \tail,\head)$ be a digraph and let~$X \subseteq V(G)$ be a set of vertices. 
    \begin{itemize}
    \item By~$G[X]$ we denote the subgraph of~$G$ induced by~$X$ in the usual sense, i.e.~the graph with vertex set~$X$ and all edges of~$G$ whose head and tail are both in~$X$.
    \item By~$G[[X]]$ we denote the (possibly partial) subgraph of~$G$ with vertex set~$X$ and all edges of~$G$ which are incident to at least one vertex of~$X$. 
    \item Finally, by~$G[[[X]]]$ we denote the (proper) subgraph of~$G$ obtained from~$G[[X]]$ by adding for each half-edge of~$G[[X]]$ the missing endpoint of~$G$.  
    \end{itemize}
\end{definition}
This is in accordance with our definition of~$G[[[X]]]$ and~$G[X]$ as in \cref{sec:prelims}.

\subsection{Linkages in incidence digraphs}\label{subsec:linkages-in-incidence-digraphs}

As we work with edge-disjoint paths and focus on incidence graphs, it will often be more convenient to assume that paths start and end in edges rather than vertices. We therefore adjust the concepts of edge-disjoint paths and~$p$-linkages to accommodate for this. The following definitions do not change significantly from the definition given in \cref{sec:prelims}, but we provide them for completion.

\begin{definition}[Paths, Linkages, and their directed Patterns]
    Let~$G+D$ be an Eulerian graph. A path of length~$t$, or just~$t$-\emph{path}, is a sequence~$L=(e_1,\ldots,e_t)$ of~$t \geq 1$ disjoint edges~$e_i \in E(G)$, for~$1\leq i \leq t$, such that, for every~$1<j\leq t$,~$\tail(e_j)$ exists in~$G$ and~$\tail(e_j) = \head(e_{j-1})$. 
    Any subsequence~$L' \coloneqq (e_i,\ldots,e_{i+k})$ of~$L$, where~$1 \leq i \leq i+k \leq t$, is called a \emph{subpath} of~$L$ and its length is~$k+1$.
    
    The \emph{directed pattern of~$L$} is defined as~$\pi(L) \coloneqq  (e_1,e_t)$.
    A collection~$\LLL \coloneqq \{L_1,\ldots,L_p\}$ of~$p \geq 1$ pairwise edge-disjoint paths is called a~$p$-linkage. We define the \emph{directed pattern of~$\LLL$} via~$\pi(\LLL) \coloneqq \bigcup_{1\leq i \leq p}\{\pi(L_i)\}$.
\end{definition}

Note that the directed pattern of a linkage is a set of pairs of edges. 

\begin{remark}
    Given a path~$L=(e_1,\ldots,e_t)$ we may write~$(e_1,u_1,e_2,u_2,\ldots,u_t,e_t) \subset L$ (or possibly omitting the parentheses) to specify the vertices~$u_1,\ldots,u_k$ visited by the path.
\end{remark}
 When talking about a linkage~$\LLL$ we may abuse notation and write~$E(\LLL) = \bigcup_{L \in \LLL}E(L)$, where~$E(L) = \{e_1,\ldots,e_t\}$ for~$L=(e_1,\ldots,e_t)$ for some~$t \geq 1$. Similarly we write~$E(\pi(\LLL)) \coloneqq \{e_1,e_2 \mid(e_1,e_2) \in \pi(\LLL)\}$ to be the set of edges that appear in the directed patterns of~$\LLL$.

 \begin{definition}[The Pattern of a linkage]
     Let~$G+D$ be Eulerian and let~$\LLL=\{L_1,\ldots,L_p\}$ be a~$p$-linkage solving the edge-disjoint paths problem encoded by~$G+D$ for~$p\coloneqq \Abs{E(D)}$. We say that~$\LLL$ is a linkage \emph{with pattern~$D$}. The endpoints of the edges in~$D$ are called \emph{terminal-vertices} or simply \emph{terminals}.
 \end{definition}
\begin{remark}
    Note that~$\bigcup\LLL \cup D$ form~$\leq p$ edge-disjoint cycles (which may intersect in vertices), i.e., they form edge-disjoint Eulerian subgraphs.
\end{remark}

In a sense the directed patterns can be seen as `terminals' of edge-disjoint paths when seen as sequences of edges. Note that since all terminal-vertices in~$G+D$ have degree two, and thus degree one in~$G$, every linkage~$\LLL$ solving the edge-disjoint paths problem encoded by~$G+D$---i.e., every linkage with pattern~$D$---must have the same directed pattern. Thus, the edges in the directed pattern for any linkage with pattern~$D$ are uniquely defined by~$D$ which means that, given the set of terminal pairs~$\{(s_i,t_i) \sth (t_i,s_i) \in D\}$, we get a unique corresponding directed pattern~$\{(e_i^s,e_i^t) \sth \tail(e_i^s) = s_i$ and~$\head(e_i^t) = t_i \}$ which we denote by~$\pi(G+D, D)$. The following is obvious.

\begin{observation}\label{obs:pattern_agrees_with_directed_pattern_of_solution}
Let~$\LLL_1,\LLL_2$ be two~$p$-linkages with pattern~$D$. Then~$\pi(\LLL_1) = \pi(G+D,D) = \pi(\LLL_2)$.
\end{observation}

We continue with defining the restriction of a linkage to a subgraph.

\begin{definition}[Restricting linkages]
   Let $G$ be a (proper) graph and~$\LLL$ a an edge-disjoint linkage in~$G$. Let~$H \subseteq G$ be a (partial) subgraph. We define the \emph{restriction of~$\LLL$ to~$H$} via~$\restr{\LLL}{H}$ to mean the linkage consisting of the set of maximal sub-paths~$P\subset L \cap H$ for~$L\in \LLL$.
   \label{def:restricting_linkages_to_subgraphs}
\end{definition}
In particular if a path~$L \in \LLL$ intersects itself in some vertex~$v$ and~$v \in H$, it may happen that~$\restr{\LLL}{H}$ contains two edge-disjoint sub-paths of~$L$ both containing~$v$. Further if a path~$L \in \LLL$ contains a vertex~$v \notin H$, then the path is split at~$v$ into two paths of~$\restr{\LLL}{H}$.
\smallskip
Recall the definition of \emph{splitting off at vertices}. The definition transfers verbatim to incidence graphs. We will use it extensively in the following.

\subsection{Euler-embedded incidence digraphs}
\label{subsec:embedded-incidence-digraphs}
In this section we adapt the concept of embeddings to incidence digraphs and fix the setting we work with in the remainder of this section.

\begin{definition}[Drawing of an incidence digraph]
    Let~$G=(V,E,\tail,\head)$ be a graph, then a \emph{drawing} of~$G$ or an \emph{embedding} of~$G$ on a surface~$\Sigma$ is a subset~$U(G) \subset \Sigma$ together with an injective map~$\nu: V \cup E \to U(G)$ such that the components of~$U(G) \setminus \nu(V\cup E)$ are pairwise disjoint lines---injective images of open intervals on a surface---and if~$(e,u) \in \tail$ and~$(e',v) \in \head$, then there exist disjoint components~$e_1,e_2 \in U(\Gamma)\setminus \nu(V\cup E)$ such that~$e_i = \gamma_i(]0,1[)$ for some injective map~$\gamma_i$ and the closure of~$e$ and~$e'$ in~$U(\Gamma)$ contains~$\nu(u)$ and~$\nu(v)$ respectively. We may add the respective directions of the edges by drawing the lines as arrows in the obvious way. If a graph can be drawn on a surface in this way, and given some drawing~$(U,\nu)$, we say that the graph is \emph{embedded in~$\Sigma$}. 
    
\smallskip

    Given an embedding~$(U,\nu)$ of a graph~$G$ in~$\Sigma$ and some subset~$\Delta \subseteq \Sigma$, let~$X \coloneqq V(\nu^{-1}(\Delta))$ denote the set of vertices that are drawn in~$\Delta$. We let~$G[\Delta]\coloneqq G[X]$,~$G[[\Delta]] \coloneqq G[[X]]$ and~$G[[[\Delta]]]\coloneqq G[[[X]]]$.
    \label{def:embedding_incidence_graphs}
\end{definition}
\begin{remark}
    Note that~$G[\Delta]$ contains all the vertices drawn in~$\Delta$ and all the edges having both end-points in~$\Delta$, in particular it is a proper subgraph. Similarly~$G[[\Delta]]$ is the partial subgraph of~$G$ obtained from~$G[\Delta]$ by extending~$G[\Delta]$ by all the edges~$e \in E(G)$ with~$\nu(e) \in \Delta$ such that~$e$ has an incidence~$v \in V(G)$ with~$v \in G[\Delta]$. Finally~$G[[[\Delta]]]$ is the proper subgraph obtained from~$G[[\Delta]]$ by adding all the incidences of edges in~$G[[\Delta]]$.
\end{remark}
The \cref{def:faces_and_2cell} of faces and~$2$-cell embeddings transfers verbatim. Similarly the \cref{def:Euler-embedding} of Euler-embedding transfers almost verbatim to this setting, where we are only interested in the order of the edges locally at vertices. However we provide another specification for degree two vertices on the boundary.

\begin{definition}\label{def:embedding_boundary_vertices}
    Let~$\Gamma$ be drawn in a surface~$\Sigma$ with boundary given some embedding~$(U,\nu)$. Let~$C_\Sigma \in c(\Sigma)$ be a cuff with some fixed orientation given by some curve~$\gamma_c$ tracing it and let~$v \in V(\Gamma)$ be a degree two vertex drawn on the cuff. We call such a vertex \emph{boundary-vertex}. Let~$\Gamma^+ \coloneqq \Gamma \cup C$ where~$C$ is a cycle obtained from connecting the vertices drawn on the curve in the order given by~$\gamma_c$, and let~$(U^+,\nu^+)$ be the obvious extension of the embedding drawing~$C$ onto the cuff~$C_{\Sigma}$ (or close to it) such that the cycle keeps the orientation as given by~$C_\Sigma$. Then we say that a boundary-vertex~$v \in V(\Gamma)$ is \emph{Euler-embedded} if~$v$ is Euler-embedded in the embedding of~$\Gamma^+$. 

    We say that~$\Gamma$ is \emph{Euler-embedded} if all its degree-four vertices are Euler-embedded (in the usual sense) as well as all its boundary-vertices are Euler-embedded as just defined.
\end{definition}

Essentially, the drawing of~$G$ as an incidence structure is the same as taking a drawing of~$G$ viewed as a standard graph and subdividing edges with new vertices that correspond to the edges, while the two arising lines represent both incidences. This highlights that a drawing of a graph~$G$ in the standard model can be easily transferred to a drawing in the incidence-model and vice-versa. From the definition it is clear how to draw partial graphs as an obvious generalisation of the above. 

Finally note that given an embedding~$(U,\nu)$ of a graph~$G$ on a surface~$\Sigma$ we may refer to~$\nu(\chi)$ for some element~$\chi \in V(G) \cup E(G)$ simply as~$\chi$ when it is clear from the context; for example saying that~$e_1,u,e_2$ are visited in this order on a cuff~$C \in c(\sigma)$ for some~$e_1,e_2 \in E(G)$ and~$u \in V(G)$.

\paragraph{The setting in Section~\ref{sec:charting_eulerian_digraphs}.} From here on we will tacitly assume the following to be given unless stated otherwise. Let~$G + D$ be
Eulerian of maximum degree four where the vertices in~$V(D)$ are of degree two in~$G+D$ and all other vertices are of degree four. Let~$\Sigma$ be a surface (possibly with boundary) and fix an orientation for every cuff~$C \in c(\Sigma)$, i.e., fix a closed
curve~$\gamma_c:[0,1] \to C$ traversing the cuff that is injective
on~$]0,1[$ and satisfies~$\gamma_c(0) = \gamma_c(1)$. 

\smallskip

Let~$\Gamma$ and~$I$ be proper graphs (not necessarily Eulerian) such that~$G = \Gamma \cup I$ and~$\Gamma \cap I \subset V(G)\setminus V(D)$ where every vertex of~$\Gamma \cap I$ is of degree two in~$\Gamma$ and~$\Gamma$ is Euler-embedded in some
surface~$\Sigma$ (with boundary)
where~$\nu(I \cap \Gamma) \subseteq \bd(\Sigma)$; in particular every non-terminal vertex has degree two or four in~$\Gamma$ and every terminal vertex~$v \in V(\Gamma)\cap V(D)$ is drawn in~$\Sigma$ away from the boundary. Further for every~$u \in V(\Gamma) \cap V(I)$ let~$C \in c(\Sigma)$ be the cuff with~$\nu(u) \in C$ and let~$e_1,e_2 \in E(\Gamma)$ be the two edges incident to~$u$. Then we assume that~$\nu(e_1),\nu(e_2) \in C$ are drawn on the same cuff such that either~$(\nu(e_1),\nu(u),\nu(e_2))$ or~$(\nu(e_2),\nu(u),\nu(e_1))$ are visited on~$C$ in that order (with respect to the orientation~$\gamma_c$)---say in the order~$(\nu(e_1),\nu(u),\nu(e_2))$---and such that there is no other~$\chi \in E(G)\cup V(G)$ such that~$\nu(\chi)$ is visited between~$(\nu(e_1),\nu(u),\nu(e_2))$. Note that we can easily assume this once we know that~$\nu(u) \in C$ by placing~$\nu(e_1)$ and~$\nu(e_2)$ close to~$\nu(u)$ in the drawing and then pulling~$\nu(e_i)$ onto the cuff for~$i=1,2$.

\smallskip

Finally, for convenience, we assume that no edge~$e=(u,v) \in E(G)$ has
both its endpoints in the same cuff (otherwise we subdivide it, which does not change whether or not~$G+D$ has a solution and neither its uniqueness).

\smallskip

The terminal vertices that are part of~$\Gamma$ will play a special role and thus lead to the following definition.

\begin{definition}\label{def:trap}
    Let~$v \in V(\Gamma) \cap V(D)$. Then we call~$u$ a \emph{trap} of~$\Gamma$.
\end{definition} 

Given traps we define trapped routes.
\begin{definition}\label{def:trapped_route}
    Let~$L$ be a path in~$\Gamma$ such that both its ends are traps, then we call~$L$ a \emph{trapped route}. 
\end{definition}

Finally combining all of the above we define \emph{pseudo-Eulerian graphs}.
\begin{definition}[pseudo-Eulerian graphs]
    Let~$\Gamma$ be a graph such that every vertex of~$\Gamma$ is either of degree one, two or four. Let~$2p \in 2\N$ be the number of degree-one vertices in~$\Gamma$. Then we call~$\Gamma$ \emph{pseudo-Eulerian (of order~$p$)}.
\end{definition}

Note that using the degree conditions (the only odd degree is one) and double counting, it is clear that pseudo-Eulerian graphs have evenly many vertices of degree one and thus the above definition is well-defined. The following observation justifies the name.
\begin{observation}\label{obs:pseudo_Euler_can_be_capped}
    Let~$\Gamma$ be pseudo-Eulerian and let~$X\subset V(\Gamma)$ be the set of vertices of degree one. Then there exists a demand-graph~$D$ with~$V(D) = X$ such that~$\Gamma+D$ is Eulerian.
\end{observation}
\begin{proof}
    Let~$u \in X \neq \emptyset$ be a vertex of degree one in~$\Gamma$. Start and extend a path~$L$ starting at~$u$ (or ending at~$u$ and then extending it backwards) as long as possible using the fact that all the vertices are either of degree two, four or one. Then said path must end in another vertex~$u' \in X$ and~$L$ contains exactly two vertices of degree one; without loss of generality assume that the path starts in~$u$ and ends in~$u'$. Then we can add~$(u',u)$ to~$D$ and delete the cycle~$L + (u',u)$ from the graph and dissolve any arising degree zero vertices. The claim follows by induction.
\end{proof}

Thus in particular for our graph~$G+D$ the graph~$G$ is pseudo-Eulerian. We further get the following obvious observation.

\begin{observation}\label{obs:terminals_are_traps}
    Let~$\Gamma$ be pseudo-Eulerian and let~$D$ be any demand graph such that~$\Gamma +D$ is Eulerian. Assume further that~$\Gamma$ is Euler-embedded in a surface~$\Sigma$ such that all the degree one vertices are not drawn on the boundary. Then the terminals~$V(D)$ are exactly the traps of~$\Gamma$.
\end{observation}
 In regards of the above we define pseudo-Eulerian embeddings.

\begin{figure}
    \centering

    \includegraphics[width=.4\textwidth]{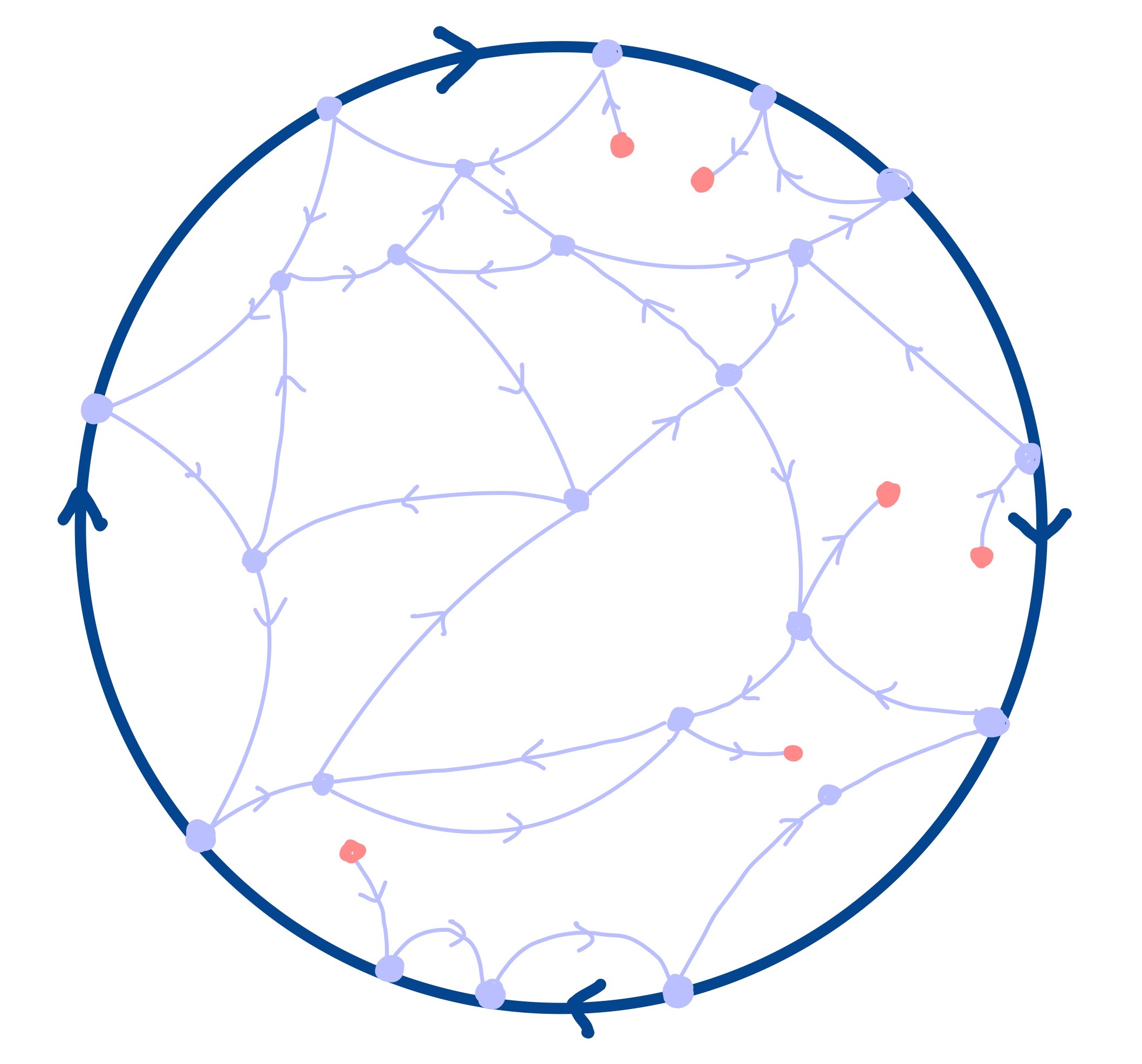}
    \caption{A pseudo Euler-embedding of a pseudo-Eulerian graph in a disc where the traps are marked in red.}
    \label{fig:pseudo-Euler-embedd}
\end{figure}

\begin{definition}\label{def:pseudo_euler_embedding}
    Let~$\Gamma$ be a pseudo-Eulerian graph and let~$D\subseteq V(\Gamma)$ be its set of degree one vertices. Let~$(U,\nu)$ be an Euler-embedding of~$\Gamma$ in a surface~$\Sigma$ such that all the degree one vertices are traps and all the vertices in~$\nu^{-1}(\bd(\Sigma))$ are of degree two. Further assume that given the orientation of the cuff boundary-vertices are Euler-embedded. Then we say that~$\Gamma$ is pseudo Euler-embedded and call~$(U,\nu)$ a pseudo Euler-embedding. 
\end{definition}
See \cref{fig:pseudo-Euler-embedd} for an example of a pseudo Euler-embedding of a graph in a disc.

Thus summarising our setting introduced above we get that~$G=\Gamma \cup I$ where~$\Gamma$ and~$I$ are pseudo-Eulerian graphs and~$\Gamma$ can be pseudo Euler-embedded in~$\Sigma$. The notion of traps and trapped routes will play its first small role in \cref{subsec:charting_an_island} but then unfold its true, unfortunately a little troublesome, nature in 
\cref{subsec:bounding_jumps_in_strong_maps}. In light of what will follow we continue by gathering some easy results related to traps and trapped routes in pseudo Euler-embedded graphs.

\begin{lemma}\label{lem:paths_in_pseudo_Euler_cut_disc_into_pseudo_Discs}
    Let~$\Gamma$ be a pseudo-Eulerian graph that is pseudo Euler-embedded in some disc~$\Delta$. Let~$L,L^* \subset \Gamma$ be two edge-disjoint paths each itself a vertex-disjoint path and both having their ends in the boundary of the disc. Let~$\Delta_1,\Delta_2 \subset \Delta$ be the two discs obtained when cutting the disc~$\Delta$ along~$L$ where we add the boundary traced by~$L$ to both discs, i.e.,~$\Delta_1 \cap \Delta_2 = \ell$ where~$\ell = \nu(L)$ is the embedding of~$L$. Then~$L^* \subset \Gamma[\Delta^*]$ for some~$\Delta^* \in \{\Delta_1,\Delta_2\}$ and~$L^*$ can either be extended to a cycle or a trapped route in~$\Gamma[\Delta^*]$.
\end{lemma}
\begin{proof}
    Since~$\Gamma$ is pseudo Euler-embedded in the disc~$\Delta$, it is clear that~$L^* \subset G[\Delta^*]$ for some~$\Delta^* \in \{\Delta_1,\Delta_2\}$ for~$L^*$ cannot cross~$L$ and thus must stay on one side of it using any orientation of the disc~$\Delta$---otherwise we would get a strongly planar vertex or a contradiction to the embedding being planar. Since~$G[\Delta^*]$ contains~$L$ one easily sees that~$\Gamma[\Delta^*]$ is pseudo-Eulerian: every vertex on~$L$ will be part of the new cuff and is of degree either four, two or zero inside~$\Gamma[\Delta^*]$ for no vertex on~$L$ is a trap: if it is of degree four we can slightly change the drawing to move the vertex away from the boundary by slightly pulling it in. If there is a vertex~$v$ on the boundary of~$\Delta^*$ of degree zero (for example it has degree four in the other disc) then we dissolve it. In particular~$\Gamma[\Delta^*]$ inherits the pseudo Euler-embedding of~$\Gamma$ concluding the claim. Then~$L^*$ can be extended in~$\Gamma[\Delta^*]$ along degree two or degree four vertices until we either close a cycle or end in a trap. If we end in a trap, then, extending~$L^*$ backwards via incoming edges, we necessarily will end in another trap for there are no odd degree vertices left apart from degree one vertices. This concludes the proof.
\end{proof}
\begin{remark}
    Note that although technically~$G[\Delta_1] \cap G[\Delta_2] \neq \emptyset$, we could easily introduce a tie-breaker so that~$V(G[\Delta_1]) \cap V(G[\Delta_2]) = V(L)$ but~$E(L) \subset E(G[\Delta_1])$ while~$E(L) \cap E(G[\Delta_2]) = \emptyset$ which would be more in accordance with what we want to achieve. In this specific scenario this distinction is however irrelevant.
\end{remark}

Although the lemma is rather trivial using the fact that it is easily seen to be forced by the restrictions imposed on the embedding, it will have huge consequences in the sections to come; a first direct consequence being the following.

\begin{corollary}\label{cor:linkage_in_pseud_Ebedding_in_disc_yields_traps_or_cycle}
    Let~$\Gamma$ be a pseudo-Eulerian graph that is pseudo Euler-embedded in some disc~$\Delta$. Let~$\LLL$ be a~$p$-linkage in $\Gamma$ where each path~$L \in \LLL$ has its ends in the boundary of the disc and is itself a vertex-disjoint path. 
    Then either~$\LLL$ can be extended to a~$p$-linkage~$\LLL'$ where each path~$L' \in \LLL'$ is a trapped route, or one of the paths~$L \in \LLL'$ can be extended to a cycle in~$\Gamma$.
\end{corollary}
\begin{proof}
    Choose one path~$L \in \LLL$ and using~$L$ define~$\Delta_1,\Delta_2\subset \Delta$ as in \cref{lem:paths_in_pseudo_Euler_cut_disc_into_pseudo_Discs}. Continue inductively by choosing paths~$L_1\in \Gamma[\Delta_1]$ and~$L_2 \in \Gamma[\Delta_2]$ until a disc satisfies \cref{lem:paths_in_pseudo_Euler_cut_disc_into_pseudo_Discs}. Then the claim follows by induction and the fact that at each step we cut discs into two new discs.
\end{proof}

 We leave the realm of (pseudo)-drawings for now and come back to them starting with \cref{subsec:coastlines}.

\subsection{Rigid, unique, and exhaustive linkages}
\label{subsec:rigid-linkages}
The following definition is due to Robertson and Seymour \cite{GMXXI}, transferred to the Eulerian setting.

\begin{definition}
    Let~$G+D$ be an Eulerian graph and let~$\LLL$ a~$p$-linkage in~$G$ with pattern~$D$ for some~$p \coloneqq \Abs{E(D)}$. Then we call~$\LLL$ \emph{spanning} if~$V(G) = V(\LLL)$, and we call it \emph{exhaustive} if~$G+D = \bigcup\LLL+D$. Further we say that~$\LLL$ is \emph{unique} if there is no other linkage~$\LLL'$ in~$G$ with the same pattern as~$\LLL$. A unique spanning linkage is called  \emph{vital} and a unique exhaustive linkage is called \emph{rigid}.
\label{def:rigid_linkage}
\end{definition}
\begin{remark}
    One can easily extent the definitions of vital, exhaustive and rigid to linkages in arbitrary graphs.
\end{remark}
Given a linkage~$\LLL$ in~$G+D$ with pattern~$D$ note that since~$\LLL+ D$ is a set of pairwise edge-disjoint cycles, the graph~$G - \bigcup \LLL$ is  Eulerian.
\begin{observation}
    Let~$G+D$ be Eulerian and let~$\LLL$ be a linkage with pattern~$D$, then~$G - \bigcup \LLL$ is Eulerian.
\end{observation}

We show next that in our setting every vital linkage is necessarily also exhaustive.

\begin{lemma}\label{lem:vital_implies_rigid}
    Let~$G+D$ be Eulerian and let~$\LLL$ be a vital linkage in~$G$ with pattern~$D$. Then~$\LLL$ is rigid.
\end{lemma}
\begin{proof}
    Assume the contrary. Then~$\LLL$ is not exhaustive and thus there 
    exists (an edge and hence) a cycle~$C \subseteq G - \bigcup \LLL$.
    Let~$v \in V(C)$ be any vertex of~$C$. Then, since~$\LLL$ is vital, there exists some~$L \in \LLL$ with~$v \in V(L)$. Thus~$L$ and~$L' = L+C$ have the same directed pattern and we can replace~$L$ by~$L'$ in~$\LLL$ to get a second linkage with the same pattern.  But then~$\LLL$ was not unique, a contradiction to it being vital.
\end{proof}

A direct consequence which may come in handy is the following.

\begin{lemma}
Let~$G+D$ be an Eulerian graph and let~$\LLL$ be a rigid linkage in~$G$ with pattern~$D$. Then~$L \in \LLL$ is a vertex-disjoint path, i.e., no vertex in~$G +D$ is visited twice by a single path.
\label{lem:rigid_implies_no_intersection}
\end{lemma}
\begin{proof}
   Assume the contrary, then there is~$L\in \LLL$ with a self-intersection. But then we can reroute~$L$ at this intersection to get a new linkage with the same pattern; contradiction.
\end{proof}

Next we highlight a noteworthy relation between exhaustiveness and Eulerianness.
\begin{lemma}\label{lem:exhaustive_yields_pseudo-eulerian}
    Let~$G$ be a graph and let~$\LLL$ be an exhaustive~$p$-linkage in~$G$. Then there exists a demand graph~$D$ with~$V(D) \subseteq V(G)$ such that~$\Abs{V(D)} \leq 2p$ and~$G+D$ is Eulerian; in particular if~$G$ has maximum degree four and no degree three vertices, then~$G$ is pseudo-Eulerian of order~$p$.
\end{lemma}
\begin{proof}
    Since~$\LLL$ is exhaustive it holds~$G = \bigcup_{L\in \LLL} L$ and in particular all the vertices of odd degree have to be terminals. Thus there are~$\leq 2p$ vertices of odd degree. Let~$D$ be the set of back edges we get by closing every~$L \in \LLL$ to a cycle. Then~$G+D$ is Eulerian: this follow at once from the fact that every non-terminal vertex of~$L \in \LLL$ is visited an even number of times---possibly~$0$---by~$L$ and every terminal vertex is visited an odd number of times by~$L$ and thus closing~$L$ to a cycle~$C_L = L+e$ where~$e=(t,s)$ with~$L$ being an~$s$ to~$t$ path, then every vertex is visited an even number of times by~$C_L$. In particular then every vertex in~$G+D$ is of even degree.
\end{proof}
\begin{remark}
    \cref{lem:exhaustive_yields_pseudo-eulerian} has huge consequences: as a matter of fact it proves that almost all of the tools we will consequently develop apply in a much larger scale: If one can prove the existence of a rigid linkage, then the graph is inherently pseudo-Eulerian and thus the results we develop next can be applied.

As a matter of fact most of the results we provide do not need to know the demand graph~$D$ and could be restated simply as being given~$G$ and a rigid linkage~$\LLL$ where the linkage implicitly encodes a demand graph~$D$.
\end{remark}

In what is to follow we will have to cut the graph~$G$ extensively, resulting in new linkages. Thus, for iterative use, we need to redefine the demand graphs to encode these new linkages we seek in a way that the resulting graphs remain Eulerian.

\begin{definition}[Cuts and Induced Cuts]    \label{Def:induced_cuts}
    Let~$G$ be a directed graph and~$X \subset V(G)$. We define~$\delta(X) \coloneqq \{ e \in E(G) \sth e$ is incident to a vertex in~$X$ and to a vertex in~$V(G) \setminus X \}$ and refer to it as the \emph{cut induced by~$X$}. 

    We define the \emph{closed cut~$\delta[X]$ induced by~$X$} as the 
    (proper) subgraph of~$G$ with vertex set~$\{ v \in V(G) \sth v$ is 
    incident to an edge in~$\delta(X) \}$. The \emph{order} of the 
    (closed) cut is given by~$|\delta(X)|$. 

    We write~$\delta^+[X] \coloneqq \{(u,v) \in \delta[X] \sth u \in X\}$ 
    and~$\delta^-[X] \coloneqq \delta[X] \setminus \delta^{+}[X]$. 
    Similarly, we define~$\delta^+(X) \coloneqq E(\delta^+[X])$ as well as~$\delta^-(X) 
    \coloneqq E(\delta^-[X])$.

    Let~$H \subset G$ be a subgraph and~$X\subset V(H)$. Then we define~$\delta_H(X)$---and analogously for the other definitions---to mean the respective cut in~$G[H]$.
\end{definition}
\begin{remark}
    Note that~$G[[[X]]] \coloneqq G[X] \cup \delta_G[X]$ and~$G[[X]] \coloneqq G[X] \cup \delta_G(X)$.

    Also, given~$G+D$ we will differentiate between cuts in~$G+D$ and cuts in the pseudo-Eulerian graph~$G$. In particular given~$G+D$ when we talk about~$G[X]$,~$G[[X]]$ and~$G[[[X]]]$ we talk about the respective subgraph of~$G$ viewing~$X$ as inducing a cut in~$G$.
\end{remark}

By definition, given any subset~$X \subseteq V(G)$, then~$\delta(X)$ is a set of edges or, equivalently, a \emph{vertex-less partial subgraph of~$G$} containing only the edges that cross between~$X$ and~$\overline{X}$. 
Similarly,~$\delta[X]$ then is the subgraph of~$G$ that contains the edges of~$\delta(X)$ together with their endpoints. 

$G[[X]]$ is again a partial subgraph of~$G$ which consists of the subgraph of~$G$ induced by~$X$ in the usual sense together with the edges in~$\delta(X)$ and their incidence relation with their endpoints in~$X$ (but not their endpoint in~$\overline{X})$.
Finally, to obtain~$G[[[X]]]$ we add to~$G[[X]]$ the endpoints of the edges in~$\delta(X)$ that were missing in~$G[[X]]$ together with their incidence relation. Thus,~$G[[[X]]]$ is again a (non-partial) subgraph of~$G$.

\begin{observation}
    Let~$G+D$ be Eulerian and let~$X \subset V(G)$. The following hold by definition:
    \begin{itemize}
        \item~$G[[[X]]] \cap G[[[\overline{X}]]] = \delta_G[X]$, and
        \item~$G[[X]] \cap G[[\overline{X}]] = \delta_G(X)$,
        \item~$E(G[X]) \cap E(D) = \emptyset$,~$E(G[[X]]) \cap E(D) = \emptyset$, and~$E(G[[[X]]]) \cap E(D) = \emptyset$.
    \end{itemize}
   In particular~$G[[X]] \cap G[[\overline{X}]] \subseteq E(G)$. Note that for~$X\subseteq V(G)$ we have~$G[[X]] - G[X] = \delta_G(X)$, since there is no edge in~$\delta_G(X)$ with head and tail in~$X$. 
   \label{obs:Euler-extension_is_a_cut_and_disj_from_D}
\end{observation}

We next prove a simple fact about cuts in Eulerian digraphs we need in order to define \emph{Euler-extended cuts} which in turn are a version of Euler-restrictions (see \cref{def:Euler-restriction}) tailored to the setting of demand graphs and rigid linkages.

\begin{lemma}\label{lem:euler-cut-matchings}
    Let~$\LLL$ be a linkage in~$G+D$ with pattern~$D$. Let~$X \subseteq V(G)$ and let~$F^+ \coloneqq \delta_G^+(X) \setminus E(\bigcup \LLL)$ and~$F^- \coloneqq \delta_G^-(X) \setminus E(\bigcup \LLL)$. 
    Then there is an edge-disjoint linkage~$\LLL_M$ in~$G[[\overline{X}]]$ that is edge-disjoint from~$\LLL$ and~$E(D)$ such that each path~$P \in \LLL_M$ starts in an edge~$e \in F^+$ and ends in an edge~$e' \in F^-$ and each edge of~$F^+ \cup F^-$ appears on exactly one path in~$\LLL_M$. 

    Furthermore, splitting off along each path in~$\LLL_M$ yields a perfect matching~$M(\LLL_M)$ of~$F^+$ to~$F^-$.
\end{lemma}
\begin{proof}
    Recall that~$G - \bigcup \LLL$ is Eulerian and note that~$F^+ = \delta^+(X) \setminus E(D \cup \bigcup \LLL)$ and analogously for~$F^-$ with~$X$ inducing a~$2k$-cut in~$G+D$ for some~$k \in \N$.
    
    We will construct linkages~$\LLL_i$, for~$0 \leq i \leq k$, and graphs~$G_i \subseteq G$ inductively as follows, maintaining the property that~$G - (\bigcup\LLL \cup G_i)$ is Eulerian---in particular~$G_i$ is Eulerian---and~$\LLL_i$ is edge-disjoint from~$\LLL$.
    If an edge in~$F^+$ appears on a path~$P_e$ in~$\LLL_i$ we say that it is \emph{matched} (to the other end of~$P_e$). Otherwise we call~$e$ \emph{unmatched}.
    
    We start the construction with~$\LLL_0 \coloneqq \emptyset$ and~$G_0 \coloneqq \emptyset$.
    Now suppose we have already constructed~$\LLL_i$ and~$G_i$ with the desired property. 

    If there is an unmatched edge~$e \in F^+ \setminus E(G_i)$, choose such an edge~$e$. As~$G - (\bigcup\LLL \cup G_i)$ is Eulerian, there is a cycle~$C_e \subseteq G-(\bigcup\LLL \cup G_i)$ containing~$e$.
    Let~$e = e_1, \dots, e_s$ be the edges of~$C_e \cap F^+$ in the order in which they appear on~$C_e$ starting at~$e_1 = e$. 
    Clearly, for all~$1 \leq j \leq s$, if we traverse~$C_e$ starting at~$e_j$ then we must eventually reach an edge~$e'_j \in F^-$ such that the sub-path~$P_j$ of~$C_e$ between~$e_j$ and~$e'_j$ only contains edges of~$G[\overline{X}]$. 
    We define~$\LLL_{i+1} \coloneqq \LLL_i \cup \{ P_j \sth 1\leq j \leq s\}$ and set~$G_{i+1} \coloneqq G_i + C_e$. It is readily verified that~$G - (\bigcup \LLL \cup G_{i+1})$ is Eulerian. 

    Now suppose that all edges in~$F^+$ 
    have already been matched. We set~$\LLL_M \coloneqq \LLL_i$ which satisfies 
    the first claim of the lemma. Splitting off along each path in~$\LLL_M$ 
    yields the matching required for the second claim of the lemma.
\end{proof}
\begin{remark}
    Note that by construction~$\LLL_M$ is edge-disjoint (and vertex-disjoint) from~$D$ and note further that~$E(D)$ was irrelevant to the construction apart from~$D$ defining the pattern of~$\LLL$.
\end{remark}

We define \emph{Euler-extended cuts} as follows.

\begin{definition}    \label{def:euler_extended_cut}
    Let~$p,k \geq 0$ and~$\LLL$ be a~$p$-linkage in~$G+D$ with pattern~$D$. 
    Let~$X \subseteq V(G)$. By \cref{lem:euler-cut-matchings} there is a matching~$M$ between the edges of~$\delta_G^+(X) \setminus E(\bigcup \LLL)$ and the edges of~$\delta_G^-(X) \setminus E(\bigcup \LLL)$. 
    
    We define~$X^{\LLL} \coloneqq G[[X]] \cup M$ (including the incidences) and~$D_X^{\LLL}$ to be the graph containing the set of back-edges for each path in~$\restr{\LLL}{ X^{\LLL}}$ defined as follows. Let~$L \subseteq \restr{\LLL}{X^{\LLL}}$ be a path in the restricted linkage and let~$\pi(L) \coloneqq (e_1,e_2)$ be its pattern. Then if~$e_1 \in \delta(X)$ we add a vertex~$v_1^L$ and the incidence~$\tail(e_1, v_1^L)$ to~$D_X^\LLL$ (which represents its tail that is not in~$G[[X]]$) and if~$e_2 \in \delta(X)$ we add a vertex~$v_2^L$ and the incidence~$\head(e_2, v_2^L)$ to~$D_X^\LLL$ (which represents its head). Otherwise we add the tail (or the head) of~$e_1$ (or~$e_2$) to~$D_X^\LLL$. Finally add the back-edge~$e_L$ from the respective vertex representing the head of~$e_2$ to the vertex representing the tail of~$e_1$ to~$D_X^\LLL$. 
    We refer to~$X^{\LLL}+D_X^{\LLL}$ as an \emph{Euler-extension of~$X$ with respect to~$\LLL$}. We call~$(\overline{X}^{\LLL},X^{\LLL})$ an \emph{Euler-extended cut}.
\end{definition}

\begin{remark}
    Note that in the above construction every path~$L \in \restr{\LLL}{X^{\LLL}}$ with pattern~$\pi(L) = (e_1,e_2)$ was assigned an edge~$e_L \in D_X^{\LLL}$ such that~$L + e_L$ forms a cycle in~$X^\LLL + D_X^\LLL$. Note further that some edges~$e \in E(D)$ may neither be part of~$E(D_X^{\LLL})$ nor~$E(D_{\overline{X}}^\LLL)$.
    
\end{remark}

Note here too, that the only dependence of the Euler-extended cut on~$D$ is actually the pattern of~$\LLL$.

The following justifies the name.
\begin{lemma}
\label{lem:Eulerian_extension_properties_gen}
    Let~$G+D$ be Eulerian, let~$\LLL$ be a linkage in~$G$ and let~$X \subseteq V(G)$. Let~$(\overline{X}^{\LLL},X^{\LLL})$ be an Euler-extended cut. Then
    \begin{enumerate}
        \item~$X^{\LLL} + D_X^{\LLL}$ and~$\overline{X}^{\LLL} + D_{\overline{X}}^{\LLL}$ are Eulerian of maximum degree four where every terminal vertex (with respect to the demand graphs~$D_X^{\LLL},D_{\overline{X}}^\LLL$) has degree two, and
        \item~$\overline{X}^{\LLL} \cap X^{\LLL} = \delta_G(X)$,
    \end{enumerate}
    where in turn~$\delta_G(X) = \delta(X) \setminus E(D)$.
\end{lemma}
\begin{proof}
    Eulerianness is clear from construction as every vertex remains Eulerian (has equal in- and out-degree): for each edge~$e\in\delta(X)$ in the cut that is not part of~$\LLL \cup D$ we extend the partial graph via a matching edge stemming from splitting off some path~$P_e$ as in \cref{lem:euler-cut-matchings}; in particular the edge~$e$ has a head and a tail in~$X^\LLL$, both Eulerian vertices. Clearly none of the newly introduced edges are the same in both graphs, for they connect opposing ends of the cut-edges that are extended via disjoint paths in the two graphs (once tail once head). Every component of~$\restr{\LLL}{X^{\LLL}}$ gets a matching back-edge in~$D_X^{\LLL}$ resulting again in the respective new terminal-vertices to be Eulerian. Maximum degree four follows from the fact that any edge in~$D_X^{\LLL}$ was either already an edge in~$D$ or it connects two terminals that are newly introduced vertices (representing terminals lying on different sides of the cut) and thus having degree~$1$ in~$X^{\LLL}+D_X^\LLL \setminus E(D_X^\LLL)$. In particular the new terminals are again degree two. The second claim follows since~$\delta_G(X) = \delta(X) \setminus E(D) \subset \overline{X}^{\LLL} \cap X^{\LLL}$ is obvious by construction.
\end{proof}
\begin{remark}
    Note that if a vertex~$v \in V(\overline{X})$ is an end of two edges~$e_1,e_2 \in \delta(X)$, then~$X^{\LLL}$ gets \emph{two} copies of~$v$ by construction, as we add new vertices to represent the possibly missing heads or tails of~$\delta_G(X)$ in~$G[[X]]$.
\end{remark}

In general, given a~$k$-cut induced by~$X$ the Euler-extended cut is not unique, for we have a choice when splitting off. In our setting however, that is~$G+D$ has maximum degree four and~$\LLL$ is rigid, it is.

\begin{lemma}
    Let~$G+D$ be Eulerian and~$\LLL$ a rigid~$p$-linkage in~$G+D$ with pattern~$D$. Let~$X$ induce a~$k$-cut in~$G$ for some~$k\geq 0$. Then
\begin{itemize}
    \item there exists a unique Euler-extended cut~$(\overline{X}^{\LLL},X^{\LLL})$ with~$X^{\LLL} = G[[X]]$,
    \item~$X^{\LLL} + D_X^{\LLL}$ is Eulerian of maximum degree four and the terminal vertices are of degree two, and
    \item~$\overline{X}^{\LLL} \cap X^{\LLL} = \delta_G(X)$,
\end{itemize}
where~$\delta_G(X) = \delta(X) \setminus E(D)$.
\label{lem:Eulerian_extension_properties_rigid}
\end{lemma}
\begin{proof}
Since~$\LLL$ is rigid it is exhaustive by definition. Thus~$\LLL$ visits every edge in the cut~$\delta_G(X)$ and there is no choice of paths~$P_e$ in the definition of~$X^{\LLL}$ (see \cref{lem:euler-cut-matchings}). The rest follows from \cref{lem:Eulerian_extension_properties_gen}.
\end{proof}

\begin{remark}
    Thus, whenever~$\LLL$ is rigid we may write~$G[[X]]$ for~$X^{\LLL}$, where both are partial subgraphs with each edge~$e \in \delta(X)$---in particular also edges~$e \in E(D)\cap \delta(X)$---having exactly one end in~$G[[X]]$ and the other end in~$D_X^\LLL$ (representing the vertex in~$G[[\overline{X}]]$). 
\end{remark}

We have the following obvious observation relating Euler-extended cuts to Euler-restrictions. Note that---in contrast to Euler-extended cuts---Euler-restrictions depend on~$V(D)$, i.e., we need~$V(D)$ to lie on a single side of the cut.
\begin{observation}
    Let~$G+D$ be an Eulerian graph and let~$X \subset V(G)\setminus V(D)$. Let~$\LLL$ be a rigid linkage in~$G$ with pattern~$D$ and let~$G_X \coloneqq \restr{G}{X}$. Then~$\pi(\restr{\LLL}{G[[X]]}) = \pi(\restr{\LLL}{G_X})$.
    \label{obs:rigid_linkage_in_euler_restrictions}
\end{observation}

We extend the notion of induced cuts to the the more general setting of not necessarily induced cuts, that is \emph{separations} (or \emph{edge-cuts} or \emph{cuts}) which will come in handy later.
\begin{definition}[Separation (or Cut)]
    Let~$G$ be an Eulerian graph. Let~$A,B \subseteq G$ be partial subgraphs such that~$A\cup B = G$ and~$A\cap B \subset E(G)$. Let~$d = \Abs{A\cap B}$, then we call~$(A,B)$ a \emph{separation of~$G$} or a \emph{(edge)-cut of~$G$} where~$d$ is its \emph{order}. 
    \label{def:edge_separation}
\end{definition}
\begin{remark}
    By a rule of thumb we will talk of separations when we are interested in the sets~$A,B$ and of edge-cuts if we are interested in~$A \cap B$.
\end{remark}

We continue with a few basic lemmas that we will need henceforth.

\begin{lemma}
    Let~$\LLL$ be a~$p$-linkage in~$G$ with pattern~$D$. Let~$X\subset V(G)$ induce a~$k$-cut in~$G$ for some~$k \in \N$ and let~$(\overline{X}^{\LLL},X^{\LLL})$ be some Euler-extended cut. Let~$\LLL_X \coloneqq \restr{\LLL}{G[[X]]}$ and let~$Z = \pi(\LLL)$ and~$Z_X = \pi(\LLL_X)$. Then the following are true:

    \begin{enumerate}
        \item~$\LLL_X$ is a~$\leq p+ k$-linkage in~$X^{\LLL}$ with pattern~$D_X^{\LLL}$,
        \item~$Z_X = \pi(G[[X]] + D_X^{\LLL}, D_X^{\LLL})$
        \item~$E(Z_X) \subseteq (E(Z)\cap G[[X]]) \cup \left(\delta(X)\cap E(\LLL)\right)$,
        \item if~$(e,e) \in Z_X$ then~$(e,e) \in Z \cap G[[X]]$.
    \end{enumerate}
    \label{lem:cuts_and_linkages_in_euler_extension} 
\end{lemma}
\begin{proof}
    The proof that~$\LLL_X$ is a linkage is clear by \cref{def:euler_extended_cut,def:restricting_linkages_to_subgraphs} and  applying \cref{lem:Eulerian_extension_properties_gen}. The order of the linkage can obviously not exceed~$p+k$.  

    Using \cref{def:euler_extended_cut,def:restricting_linkages_to_subgraphs} and \cref{lem:Eulerian_extension_properties_gen} it is also clear that~$Z_X = \pi(G[[X]] + D_X^{\LLL},D_X^\LLL)$ noting that every newly introduced terminal vertex has degree two; thus~$2.$ follows. Further~$(E(Z)\cap G[[X]]) \subseteq E(Z_X)$ is clear and every newly introduced path in the linkage either starts or ends in an edge of~$\delta(X) \cap E(\LLL)\subseteq \delta_G(X)$ and clearly each edge of~$\delta(X) \cap E(\LLL)$ gives rise to one edge of some directed pattern in~$\LLL_X$ by definition. 
    
    Finally, if~$(e,e) \in Z_X$, then either~$e \in \delta_G(X)$ or~$e \in Z$ (and clearly~$e\notin E(D)$). If~$e \notin Z$ then there is some path in~$\LLL$ containing~$e$ as an internal edge, i.e.,~$e$ is neither the first nor the last edge of the path. The sub-path starting at~$e$ is either contained in~$G[[X]]$, then it ends in~$e' \in Z$, thus~$e' \neq e$, or it must use another edge in~$\delta(X)$; in both cases the pattern would not be~$(e,e)$; a contradiction.
\end{proof}

A direct consequence of the above is the following.

\begin{lemma}
    Let~$\LLL$ be a rigid~$p$-linkage in~$G$ with pattern~$D$. Let~$X\subset V(G)$ induce a~$k$-cut in~$G$ for some~$k \geq 0$. Let~$\LLL_X \coloneqq \restr{\LLL}{G[[X]]}$ and let~$Z \coloneqq \pi(\LLL)$ and~$Z_X \coloneqq \pi(\LLL_X)$. Then
    \begin{enumerate}
        \item~$\LLL_X$ is a rigid~$\leq p+k$-linkage in~$G[[X]] = X^{\LLL}$ with pattern~$D_X^{\LLL}$,
        \item~$Z_X = \pi(G[[X]] + D_X^{\LLL},D_X^\LLL)$
        \item~$E(Z_X) = (E(Z)\cap G[[X]]) \cup (\delta_G(X)\cap \bigcup \LLL)$,
        \item if~$(e,e) \in Z_X$ then~$(e,e) \in Z \cap G[[X]],$
        \item if~$(e_1,e_2) = \pi(L_X) \in Z$ for some~$L_X \in \LLL_X$, then~$L_X \in \LLL$.
    \end{enumerate}
    \label{lem:cuts_and_rigid_linkages} 
\end{lemma}
\begin{proof}
    Since~$\LLL$ is rigid it follows with \cref{lem:Eulerian_extension_properties_rigid} that~$X^{\LLL} = G[[X]]$, in particular~$E(X^{\LLL}) \subset E(G)$ and thus, since~$\LLL$ is rigid, clearly~$\LLL_X$ is a rigid~$\leq p+k$-linkage using \cref{lem:cuts_and_linkages_in_euler_extension}. Let~$(e_1,e_2) = \pi(L_X) \in Z$ for some~$L_X \in Z_X$, then there exists~$L \in Z$ with~$\pi(L) = (e_1,e_2)$ and~$L_X \subseteq L$ by definition of~$L_X$. Thus, using \cref{lem:rigid_implies_no_intersection} we deduce that~$L = L_X$ for~$L_X$ is a sub-path of~$L$. 
    
    For~$3.$ note that each edge~$e \in \delta_G(X)$ either misses a head or a tail in~$X^\LLL$ and thus, if~$L \in \LLL$ contains~$e$, then there is a path~$L' \in \restr{\LLL}{X^\LLL}$ starting or ending in~$e$ by \cref{def:restricting_linkages_to_subgraphs}. Thus~$e \in E(\pi(\LLL_X)) = E(Z_X)$. The remainder follows at once from \cref{lem:cuts_and_linkages_in_euler_extension}.
\end{proof}

We continue with a lemma that allows us to switch linkages using cuts in the spirit of \cref{lem:switching_linkages_at_cuts_prelims}.

\begin{lemma}
    Let~$G+D$ be Eulerian and~$\LLL$ a linkage in~$G$ with pattern~$D$. Let~$X\subset V(G)$ induce a~$k$-cut in~$G$ with some Euler-extended cut~$(\overline{X}^{\LLL},X^{\LLL})$. Let~$\LLL'$ be a linkage in~$X^{\LLL}$ with pattern~$D_X^{\LLL}$. Then the following are true:
    \begin{enumerate}
        \item Let~$\LLL_X \coloneqq \restr{\LLL}{G[[X]]}$, then~$\pi(\LLL_X) = \pi(\LLL')$, and
        \item there exists a linkage~$\LLL''$ in~$G$ with pattern~$D$ such that~$\restr{\LLL''}{G[[\overline{X}]]} = \restr{\LLL}{G[[\overline{X}]]}$ and~$\restr{\LLL''}{G[[X]]} = \restr{\LLL'}{G[[X]]}$.
    \end{enumerate}
  In particular, if~$\LLL$ is a rigid linkage in~$G$ then~$\LLL'' = \LLL$.  \label{lem:switching_linkages_at_cuts}
\end{lemma}
\begin{proof}
    The first claim follows at once from \cref{lem:cuts_and_linkages_in_euler_extension} since~$\pi(\LLL') = \pi(X^{\LLL} + D_X^{\LLL}, D_X^{\LLL}) = \pi(\LLL_X)$. 
    
    Let~$\phi:X^{\LLL} \to G$ be an immersion mapping~$G[[X]] \cap X^{\LLL} \subset X^{\LLL}$---they are not necessarily equal for~$\LLL$ may not be rigid---to itself and~$e\in E(X^{\LLL}) \setminus E(G[[X]])$ to the path~$P_e \subset G[\overline{X}]$ from the proof of \cref{lem:euler-cut-matchings} and in turn the path behind the \cref{def:euler_extended_cut} of Euler-extension. Then~$\phi(\LLL')$ is an edge-disjoint linkage with pattern~$D_X^{\LLL}$ in~$G$ that by definition of~$X^{\LLL}$ is edge-disjoint from~$\restr{\LLL}{G[[\overline{X}]]}$ up-to edges in~$E(\bigcup \LLL) \cap \delta(X)$. 
    
    Using~$\pi(\LLL_X) = \pi(\LLL')$ we deduce that~$E(\bigcup \LLL) \cap \delta(X) = E(\pi(\LLL_X))\cap \delta(X)$ (note here that~$E(\bigcup \LLL) \cap \delta(X) = E(\bigcup \LLL) \cap \delta_G(X)$, that is it does not contain edges in~$E(D)$). Thus we can `glue' the linkages back together at~$\delta(X) \cap E(\bigcup \LLL)$; note that here we implicitly identify the vertices in~$V(D_X^\LLL)$, that were added to the `half-edges' of~$\delta(X)$, with the respective vertex of the edge in~$\delta[X]$. One now easily verifies that~$\LLL''$ is a linkage in~$G$ with pattern~$D$ (for the head and tail of edges in~$\delta(X) \cap E(\bigcup \LLL)$ lie on different sides of the cut and their incidences do not change due to the gluing procedure). The last assertion then follows from the definition of rigidity.
\end{proof}
\begin{remark}
    \cref{lem:switching_linkages_at_cuts} is the key observation needed to reduce our instances at cuts, keeping the new instance Eulerian.
\end{remark}

As mentioned earlier, we will make heavy use of \emph{splitting off at vertices} (recall \cref{def:splitting_off_vertex}) in this section to reduce the problems in question to smaller graphs if possible. To this extent we prove the following lemma.

\begin{lemma}
    Let~$\LLL$ be a~$p$-linkage in~$G+D$ with pattern~$D$ and let~$(e_1,e_2) \subset L \in \LLL$ be a sub-path of length~$2$. Let~$G'$ be the graph obtained by splitting off along~$e_1,e_2$ and let~$\LLL'$ be the respective linkage consistent with the splitting off. Then~$G'+D$ is Eulerian and~$\LLL'$ is a~$p$-linkage in~$G'+D$. Furthermore~$\LLL$ is rigid if and only if~$\LLL'$ is rigid.
\label{lem:splitting_off_linkages_remains_rigid}
\end{lemma} 
\begin{proof}
    The proof is obvious using the fact that splitting off is a reversible operation (note that isolated vertices are deleted in the process).
\end{proof}

Next we define the notion of \emph{splitting edges}. In contrast to the undirected case \cite{GMXXI} we will not split at vertices but edges, which under the assumption that~$\LLL$ is rigid (and thus exhaustive) yields a natural way to turn the resulting graph Eulerian again so that the resulting linkages can in turn be naturally glued back together. 
\begin{definition}[Splitting edges]
    Let~$G+D$ be Eulerian and let~$e=(u,v) \in E(G)$ be some edge. We say that the graph~$G' = (V',E')$ with~$V' = V(G)\cup\{u_e,v_e\}$ and~$E' = (E(G) \setminus e) \cup \{(u,u_e),(v_e,v)\}$ is obtained by \emph{splitting~$e$}. 
\end{definition}

\begin{lemma}
    Let~$G'$ be obtained from~$G$ by splitting an edge~$e=(u,v) \in E(G)$, i.e.~$V(G') = (V(G)\cup\{u_e,v_e\})$. Suppose further that~$G+D$ has a rigid~$p$-linkage~$\LLL$ with pattern~$D$. Then there exists~$D'$ such that~$G' + D'$ is Eulerian and~$G'$ has a rigid~$(p+1)$-linkage with pattern~$D'$. In particular every new terminal vertex is of degree one in~$G'$.
    \label{lem:rigid_linkages_after_cutting_edges}
\end{lemma}
\begin{proof}
    Since~$\LLL$ is rigid it follows that~$e \in L$ for some~$L \in \LLL$ with pattern~$\pi(L) = (s,t)$. One easily sees defining~$D' \coloneqq (D\setminus \{(t,s)\}) \cup \{(u_e,s),(t,v_e)\}$ yields that~$G' + D'$ is Eulerian and the obvious linkage~$\LLL'$ obtained by splitting~$L \in \LLL$ at~$e$ is clearly exhaustive for~$G'$ with pattern~$D'$. It is rigid, for let~$\LLL''$ be another linkage with the same pattern as~$\LLL'$, then gluing the two paths with terminals~$(s,u_e)$ and~$(v_e,t)$ back together by reintroducing~$e$ (and reverting the splitting of the edge) results in a~$p$-linkage with pattern~$D$ that is different from~$\LLL$; contradiction to the rigidity of~$\LLL$.
\end{proof}

We now prove the main theorem of this section (compare to \cite[Theorem 2.5]{GMXIII}).

\begin{theorem}
    Let~$p,k \geq 1$. Then there exists a function~$f(p,k)$ such that the following holds. Let~$\LLL$ be a rigid~$p$-linkage in~$G+D$ with pattern~$D$. Let~$n \geq f(p,k)$ and for~$ 1 \leq i \leq n$ let~$X_i \subseteq V(G)$ induce a~$k$-cut in~$G$ such that
    \begin{enumerate}
        \item[1.]~$X_i \subseteq X_j$ for all~$1 \leq i < j \leq n$, and
        \item[2.] for~$1 \leq i <n$ there is an edge-disjoint~$\delta_G(X_i){-}\delta_G(X_{i+1})$-linkage~$M_i$ of order~$k$ in~$(G[[X_{i+1}]] \cap G[[\overline{X}_{i}]])$.
    \end{enumerate}
Let~$X^i \coloneqq \overline{X}_i \cap X_{i+1}$ and let~$\LLL^i \coloneqq \restr{\LLL}{G[[X^i]]}$ for every~$1 \leq i \leq n$. Then there exists~$i$ with~$1 \leq i < n$ such that~$\LLL^i = M_i$.
\label{thm:rigid_linkage_in_laminar_cuts_is_Menger}
\end{theorem}

\begin{proof} We claim that~$n \geq(p+k+1)^{2(p+k)}+1$ satisfies the theorem.

First note that we will only care about the~$k$-cut induced by~$X_i$ in~$G$ and not the respective~$2\rho$-cut~$G+D$ for some~$\rho \in \N$. Hence the cuts we are interested in do not necessarily induce cuts of even order in~$G$, but they definitely are cuts of even order---Eulerian cuts---in~$G+D$. The fact that~$X_i$ need only be a cut in~$G$ rather than~$G+D$ makes arguments later on easier, and in fact Eulerianness is not needed in this proof\footnote{Note here that in the definition of Euler-extended cuts the edges~$E(D)$ are left out from the cuts.}. Note here too that~$(G[[X_{i+1}]] \cap G[[\overline{X}_{i}]]) \cap E(D) = \emptyset$ using \cref{obs:Euler-extension_is_a_cut_and_disj_from_D} and thus for every~$1 \leq i \leq n$ the linkage~$M_i$ is disjoint from~$E(D)$ by assumption. The following are obvious but crucial observations.

\begin{claim}
Let~$1 \leq i \leq n$ for some~$n\in \N$ and let~$P \in M_i$. 
    \begin{itemize}
        \item[(i)] Let~$\pi(P) = (e_1,e_2)$ be the directed pattern of~$P$. Then~$e_1 \in \delta_G^+(X_i) \iff e_2 \in \delta_G^+(X_{i+1})$ and~$e_1 \in \delta_G^-(X_{i+1}) \iff e_2 \in \delta_G^-(X_{i})$.
        \item[(ii)] If~$e \in \delta_G(X_i) \cap \delta_G(X_{i+1})$ then~$e$ is a one-edge path in~$M_i$.
    \end{itemize}
    \label{thm:rigid_linkage_in_laminar_cuts_is_Menger_claimOneEdge}
\end{claim}
\begin{ClaimProof}
    The first claim (i) follows for otherwise~$M_i$ could not induce a~$k$ matching on the cut-edges and the second claim (ii) follows at once by assumption~$2.$ of the theorem and the fact that~$X_i$ and~$X_{i+1}$ both induce~$k$-cuts.
\end{ClaimProof}

\begin{claim}
    Let~$1 \leq i,j \leq n$, then~$\Abs{\delta_G^-(X_i)} = \Abs{\delta_G^-(X_j)}$ and~$\Abs{\delta_G^+(X_i)} = \Abs{\delta_G^+(X_\ell)}$. In particular~$k_- \coloneqq \Abs{\delta_G^-(X_\ell)}$ and~$k_+ \coloneqq \Abs{\delta_G^+(X_\ell)}$ are independent of~$1 \leq \ell \leq n$. 
    \label{thm:rigid_linkage_in_laminar_cuts_is_Menger_claimSignedMatchingCuts}
\end{claim}
\begin{ClaimProof}
    This follows inductively for~$1 \leq i \leq n$ using (i) of \cref{thm:rigid_linkage_in_laminar_cuts_is_Menger_claimOneEdge} together with assumption~$2.$ of the theorem.
\end{ClaimProof}

With both observations at hand we are ready to prove the theorem. To this extent let~$Z \coloneqq E(\pi(\LLL))$ be the set of edges in the directed pattern of~$\LLL = \{L_1,\ldots,L_p\}$, i.e.,~$Z = \{e_{s_1},\ldots,e_{s_p}, e_{t_1},\ldots,e_{t_p}\}$ where~$(e_{s_i},e_{t_i}) = \pi(L_i)$, that is they are the starting and ending edge of~$L_i$ respectively for~$1 \leq i \leq p$. 
 
 For each set~$X_i$ we enumerate~$\delta_G(X_i)^+ = \{e^{i,+}_1,\ldots,e^{i,+}_{k_+}\}$ and~$\delta_G(X_i)^-=\{e^{i,-}_1,\ldots,e^{i,-}_{k_-}\}$ analogously for some~$k_+,k_- \in \N$ satisfying~$ k_- + k_+ = k$. Let~$\LLL_i \coloneqq \restr{\LLL}{G[[\overline{X}_i]]}$, then~$\LLL_i$ is a rigid linkage in~$G[[\overline{X}_i]]$ by \cref{lem:cuts_and_rigid_linkages} for every~$1 \leq i \leq n$, and it is disjoint from~$E(D)$ by definition. Intuitively~$\LLL_i$ is the linkage \emph{away from} the graph induced by~$X_i$. Let~$$Z_i \coloneqq E(\pi(\LLL_i)), \qquad \text{and} \qquad \pi_i \coloneqq \pi(\LLL_j),$$ that is~$\pi_i$ denotes the directed pattern of~$\LLL_i$ for every~$1 \leq i \leq n$. By~$2.$ of \cref{lem:cuts_and_rigid_linkages} we know that~$Z_i = (Z \cap G[[\overline{X}_i]]) \cup (\delta_G(X_i)\cap E(\bigcup \LLL))$ (again noting that~$\delta(X_i)\cap E(\bigcup\LLL) = \delta_G(X_i)\cap E(\bigcup\LLL)$) for every~$1 \leq i \leq n$. Therefore we may define the following: let~$\phi_i : Z_i \rightarrow Z_1$ via~$\phi_i(z) \coloneqq z$ if~$z \in Z$ and otherwise~$\phi_i(e^{i,x}_j) = e^{1,x}_j$ for~$x \in \{+,-\}$ and any~$1 \leq j \leq k_x$ where~$e^{i,x}_j \notin Z$; we refer to the image of~$\phi_i$ as a \emph{pull-back}. Then~$\phi_i$ is an injection with image domain~$Z_1$ of size at most~$k+p$ and since~$\pi_i \in Z_i^2$,~$\phi_i$ pulls back the directed patterns to a partition of a subset of~$Z_1$ where each block has size~$1$ if the pattern was a single edge, i.e.,~$(e,e)$, or size~$2$ otherwise. Since~$\Abs{Z_1} \leq p+k$ there are at most~$(p+k+1)^{2(p+k)}$ such partitions. Thus for~$n$ large enough, i.e.,~$n \geq(p+k+1)^{2(p+k)}+1$, we find distinct~$\pi_i$ and~$\pi_j$ that are pulled back to the same blocks, i.e,~$\phi_i(\pi_i) = \phi_j(\pi_j)$ for~$1 \leq i< j \leq n$. By definition then~$G[[\overline{X}_i]] \cap Z = G[[\overline{X}_j]] \cap Z$ and in turn~$G[[{X}_i]] \cap Z = G[[X_j]] \cap Z$. Thus we conclude the following:
\begin{equation}
    \label{lem:Menger_eq1}
Z \cap \left(G[[X_{j}]]\cap G[[\overline{X}_{i}]]\right) \subseteq  \delta_G(X_i) \cap \delta_G(X_j).
\end{equation}

\begin{claim}
    Let~$e \in \delta_G(X_i) \cap \delta_G(X_j)$, then~$e= e^{i,x}_t = e^{j,x}_t$ for~$x \in \{+,-\}$ and some~$1 \leq t \leq k_x$. 
    \label{lem:Menger_claim1}
\end{claim}
\begin{ClaimProof}
    Since~$e \in \delta_G(X_i) \cap \delta_G(X_j)$ it follows that~$e = e^{i,x}_{t_i}$ and~$e = e^{j,y}_{t_j}$ for some~$x,y \in \{-,+\}$ and some~$1 \leq t_i \leq k_x$ and~$1 \leq t_j \leq k_y$. Since~$X_i \subseteq X_{i'}$ for every~$i' \geq i$ it follows that~$e \in M_i \cap \ldots \cap M_{j-1}$ by repeatedly using (ii) of \cref{thm:rigid_linkage_in_laminar_cuts_is_Menger_claimOneEdge} and in particular it follows that~$e$ is a one-edge path in each of these linkages. This immediately implies the claim using (i) of \cref{thm:rigid_linkage_in_laminar_cuts_is_Menger_claimOneEdge}. 
\end{ClaimProof}

\begin{claim}
    Let~$e^{i,x}_t \in \delta(X_i) \cap Z$ or~$e^{j,y}_t \in \delta(X_j) \cap Z$ for some~$x,y \in \{-,+\}$, then~$e^{i,x}_t = e^{j,y}_t$ and~$x = y$.
\label{lem:Menger_claim2}
\end{claim}
\begin{ClaimProof}
    Note again that in the first case~$e^{i,x}_t \in \delta_G(X_i)$ and otherwise~$e^{j,y}_t \in \delta_G(X_j)$. The first part follows at once from the definition of~$\phi_i$ and~$\phi_j$ since for any~$z \in Z$ where~$\phi_i$ is defined, then~$\phi_i(z)= z = \phi_j(z)$ since they agree on the directed patterns; in particular~$e^{i,x}_t = e^{j,y}_t$. The fact that~$x = y$ then follows from \cref{lem:Menger_claim1} (for until now they may be the same edges but not have the same `signs'~$x,y$ in their linkage).
\end{ClaimProof}

Let~$\MMM \coloneqq (M_i\cup\ldots\cup M_{j-1})$ be the linkage obtained by gluing the paths of each~$k$-linkage~$M_\ell$ to the paths in~$M_{\ell+1}$ in the only possible (and obvious) way for~$i \leq \ell < j-1$ using \cref{thm:rigid_linkage_in_laminar_cuts_is_Menger_claimSignedMatchingCuts}. In particular~$\MMM$ is a~$k$-linkage with~$E(\pi(\MMM)) \subset \delta_G(X)$ for~$X \coloneqq \overline{X_i} \cap X_j$ (recall that each~$M_i$ is disjoint from~$E(D)$ by assumption~$2.$ of the theorem). Clearly~$\delta(X) \subseteq \delta(X_i) \cup \delta(X_j)$, and clearly~$\delta_G(X) \subseteq \delta_G(X_i) \cup \delta_G(X_j)$.
\begin{claim}
The following hold true:
    \begin{enumerate}
    \item~$\delta(X_i) \cup \delta(X_j) = \delta(X) \cup (\delta(X_i)\cap \delta(X_j))$ (and analogously for~$\delta_G$),
    \item~$\delta(X) \cap (\delta(X_i)\cap \delta(X_j)) = \emptyset$ (and analogously for~$\delta_G$), and
    \item each edge~$e \in \delta_G(X_i) \cap \delta_G(X_j)$ is a one-edge path in each of~$M_i,\ldots,M_{j-1}$.
    \end{enumerate}
\label{lem:Menger_claim3}
\end{claim}
\begin{ClaimProof}
The first and second claim are easily derivable using that~$X_i \subseteq X_j$ and thus~$X = \overline{X}_i \cap X_j = X_j \setminus (X_j \cap X_i) = X_j \setminus X_i$.

For the first claim let~$e \in \delta(X_i) \cup \delta(X_j)$. If~$e \in \delta(X_i) \cap \delta(X_j)$ there is nothing to show. Thus assume~$e=(u,v) \in \delta(X_i)$, then one endpoint of~$e$, say~$u$ is in~$X_i$ and the other, say~$v$ in~$\overline{X}_i$. Since~$e \notin \delta(X_j)$ we know that~$v \in \overline{X}_j$. In particular~$ u \in X$ and~$v \in \overline{X}$. Thus let~$e=(u,v) \in \delta(X_j)$ with~$u \in X_j$ and~$v \in \overline{X}_j$; using~$e \notin \delta(X_i) \cap \delta(X_j)$ we derive~$v \in \overline{X}_j \cap \overline{X}_i$. Then again~$e \in \delta(X)$

For the second claim recall from the above that~$\delta(X) = \delta(X_j \setminus X_i)$; let~$e \in \delta(X_j \setminus X_i)$. Then either~$e$ has one endpoint in~$X_j$ and the other in~$X_i = X_{j}\cap X_i$ in which case both endpoints are in~$X_j$ and thus~$e \notin \delta(X_j)$, or~$e$ has one end-point in~$X_j$ and one in~$\overline{X}_i$ in which case~$e \notin \delta(X_i)$; this proves the second claim.  

The third claim is an easy extension of \cref{thm:rigid_linkage_in_laminar_cuts_is_Menger_claimOneEdge}.
\end{ClaimProof}

 Let~$(\overline{X}^{\LLL},X^{\LLL})$ be the unique Euler-extended cut as given by \cref{def:euler_extended_cut} and \cref{lem:Eulerian_extension_properties_rigid}, in particular~$X^{\LLL} = G[[X]]$ and~$\overline{X}^{\LLL} = G[[\overline{X}]]$. By 2. of \cref{lem:Menger_claim3} it follows that~$G[[X]] \cap (\delta(X_i)\cap \delta(X_j)) = \emptyset$. Now define~$\LLL_i' \subseteq G[[\overline{X}_i]]$ as follows. For each~$L \in \LLL_j$ add~$L$ to~$\LLL_i'$. Further let~$\pi_L=(e_1,e_2) \in \pi_j$ be the respective directed pattern of~$L$, possibly with~$e_1 = e_2$. Then, for~$\alpha =1,2$, either~$e_\alpha \in Z$ or~$e_\alpha \in \delta_G(X_j)$ (or both) by~$3.$ of \cref{lem:cuts_and_rigid_linkages}. If~$e_1 = e^{j,x}_t \in \delta_G(X_j)$ for respective~$x$ and~$t$, then we add the unique path in~$\MMM$ containing~$e^{j,x}_t$ to~$\LLL_i'$ and glue both paths to a single path at that edge; do the same with~$e_2$ and note that the resulting walk is a well-defined path, i.e., it does not visit any edge twice for~$E(\bigcup \MMM) \cap E(\bigcup \LLL_j) \subset \delta_G(X_j)$ and the paths in~$\MMM$ and~$\LLL_j$ are otherwise edge-disjoint. We claim that the resulting collection of paths~$\LLL_i'$ is an edge-disjoint linkage in~$G[[\overline{X}_i]]$. A first observation in this direction is that for any path~$M \in \MMM$, at most one path in~$\LLL_j$ is glued to it in this construction. To see this let~$M \in \MMM$ with~$\pi(M) = (f_1,f_2)$ for~$f_1,f_2 \in \delta_G(X)$; by construction and without loss of generality~$f_1 \in \delta_G(X_i)$ and~$f_2 \in \delta_G(X_j)$ (otherwise renumber the edges). If~$M$ is glued at~$f_1$ in the process, then~$f_1 \in \delta_G(X_j)$ and thus~$f_1 \in \delta_G(X_i)\cap \delta_G(X_j)$ is itself a one-edge path in~$\MMM$ by \cref{lem:Menger_claim3}, in particular~$f_1 = f_2$. In that case again clearly only one path~$L \in \LLL_j$ has~$f_1$ in its pattern (for~$\LLL_j$ is an edge-disjoint linkage) and thus at most two paths are glued together in the process as claimed. The same argument can be used to prove that no cycle arises in the construction, i.e., if~$M$ and~$L$ are glued together they form a path and if the path were to close, then some~$L$ would need to be glued to an edge~$f_1 \in \delta_G(X_i)$ but as above this edge would already itself be a path in~$\LLL_j$. Thus this concludes the first part of the following.

\begin{claim}
    It holds that~$\LLL_i'$ is an edge-disjoint linkage and~$\pi(\LLL_i) \subseteq \pi(\LLL_i')$.
    \label{lem:Menger_claim4}
\end{claim}
\begin{ClaimProof}
The first part of the claim follows by the above discussion.
The proof of the second part follows by case distinction as shown next. Let~$L_i \in \LLL_i$ with directed pattern~$\pi(L_i)=(e_1,e_2)$, then there exists~$L_j \in \LLL_j$ with directed pattern~$\pi(L_j)=(e_1',e_2')$ such that~$\phi_i((e_1,e_2)) =\phi_j((e_1',e_2'))$. We deal with the following cases:
\begin{itemize}
    \item[1:]~$e_1, e_2 \in Z \setminus \delta_G(X_i).$ Since~$e_1,e_2 \in Z$ we deduce~$e_1 = e_1'$ and~$e_2 = e_2'$. Using 5. of \cref{lem:cuts_and_rigid_linkages} we deduce that~$L_i \in \LLL$ and~$L_j \in \LLL$, and since~$\LLL$ is rigid we get~$L_i = L_j$. Further, by \cref{lem:Menger_claim2} we deduce that~$e_1,e_2 \notin \delta_G(X_j)$; but then~$L_j \in \LLL_i'$ by construction and hence~$L_i \in \LLL_i'$ concluding the proof of the claim in this case.

    \item[2:]~$\Abs{\{e_1,e_2\} \cap \delta_G(X_i)} = 1$, say~$e_1 \in \delta_G(X_i)$. We have two sub-cases: the first is~$e_1 = e_2$, which implies~$e_1' = e_2'$. But then 4. of \cref{lem:cuts_and_rigid_linkages} implies~$e_1 \in Z$. This then implies~$e_1 = e_2 = e_1' = e_2'$ and thus~$e_1 \in \delta_G(X_i) \cap \delta_G(X_j)$. Finally~$e_1 \in \MMM$ is a path in the linkage itself by~$3.$ of \cref{lem:Menger_claim3}.

    The second case is~$e_1 \neq e_2$. Since~$e_2 \notin \delta_G(X_i)$,~$3.$ of \cref{lem:cuts_and_rigid_linkages} implies that~$e_2 \in E(Z)\cap G[[\overline{X}_i]]$. Finally using \cref{lem:Menger_eq1} we deduce~$e_2 \notin \delta_G(X_j)$, i.e.,~$e_2 \in Z\setminus \delta_G(X_j)$. In particular since the pattern are equivalent we deduce~$e_2 = e_2' \neq e_1'$. The path~$M_{e_1}\in \MMM$ with directed pattern~$(e_1,e_1')$ is glued to the path~$L_j$ with pattern~$(e_1',e_2)$ in the above construction of~$\LLL_{i}'$ and thus we conclude the proof.

    \item[3:]~$\Abs{\{e_1,e_2\} \cap \delta_G(X_i)} = 2$, in particular~$e_1 \neq e_2$. Then~$e_1' \neq e_2'$ again. Let~$e_1 = e^{i,x}_s$ and~$e_2 = e^{i,y}_t$ for some~$x,y\in \{+,-\}$ and~$1\leq s \leq k_x$ and~$1 \leq t \leq k_y$. Clearly~$y \neq x$ for both edges are cut edges and the path~$L_i$ must hence once enter and once leave~$G[X_i]$. Then~$e_1' = e^{j,x}_s$ and~$e_2' = e^{j,y}_t$ using that~$\phi_i$ and~$\phi_j$ agree. Finally the paths in~$\MMM$ with directed pattern~$(e_1,e_1')$ and~$(e_2',e_2)$ together with~$(e_1',e_2')$ glued together do the trick. Said path is by construction in~$\LLL_i'$.
\end{itemize}
\end{ClaimProof}
Let~$\tilde{\LLL}_i \subseteq \LLL_i'$ be the set of edge-disjoint paths with~$\pi(\tilde{\LLL}_i) = \pi(\LLL_i)$ which exists by \cref{lem:Menger_claim4}. Then~$\tilde{\LLL}_i$ is a linkage in~$G[[\overline{X}_i]]$ with the same pattern as~$\LLL_i$ and thus~$\tilde{\LLL}_i = \LLL_i$ since~$\LLL_i$ is rigid by 1. of \cref{lem:cuts_and_rigid_linkages}; in particular~$\tilde{\LLL}_i$ is rigid. This in turn implies~$\tilde{\LLL}_i = \LLL_i'$ for~$E(\tilde{\LLL}_i) = E(G[[\overline{X}_i]])$ by rigidity and the fact that~$\tilde{\LLL}_i \subseteq \LLL_i'$. We derive the following.

\begin{claim}
    It holds that~$\LLL_i' = \LLL_i$ and~$\MMM \subseteq \restr{\LLL_i'}{G[[X]]}$ is a sub-linkage.
    \label{lem:Menger_claim5}
\end{claim}
\begin{ClaimProof}
    The first assertion follows at once from the discussion above. By construction of~$\LLL_i'$ it holds for every path~$M \in \MMM$ that either~$M\subseteq L \in \LLL_i'$ or~$E(M) \cap E(\LLL_i') = \emptyset$, and since~$M \subset G[[\overline{X}_i]]$ and~$\LLL_i'$ is rigid for~$G[[\overline{X}_i]]$, the latter is impossible. Finally~$M \subseteq \restr{L}{G[[X]]}$ for that is how~$L \in \LLL_i'$ was defined. Now assume that~$ \pi(M) \notin \pi(\restr{L}{G[[X]]})$ (recall that~$\restr{L}{G[[X]]}$ may result in a linkage of order higher than~$1$, adding a path for each internal edge of~$L$ that is part of~$\delta_G(X)$). Since~$M \subset L$ this then implies that there exists~$M' \in \MMM$ with~$M' \neq M$ such that~$M' \subset \restr{L}{G[[X]]}$ by construction of~$L$---it is glued using paths in~$\MMM$ and~$\LLL_j$---and~$M\cup M'$ forms a sub-path of~$\restr{L}{G[[X]]}$ (see \cref{def:restricting_linkages_to_subgraphs}). That is, let~$\pi(M) = (m_1,m_2)$ and~$\pi(M') = (m_1',m_2')$ then (up to a relabelling)~$(m_2,m_1') \subset G[[X_j]]$ is a~$2$-path with~$m_1,m_2' \in \delta_G(X_j)$ which is impossible by the definition of~$G[[X_j]]$ for any edge in~$\delta_G(X_j)$ has exactly one of its ends in~$G[[X_j]]$.
\end{ClaimProof}

Now \cref{lem:Menger_claim5} implies the theorem, for~$\LLL_i = \LLL_i'$ and thus~$\MMM \subseteq \restr{\LLL_i}{G[[X]]}$ but also~$E(\pi(\restr{\LLL_i'}{G[[X]]})) \subseteq E(\pi(\MMM))$ by definition. Thus~$\MMM = \restr{\LLL_i}{G[[X]]}$ and one easily verifies that~$\restr{\LLL_i}{G[[X]]} = \restr{\LLL^i}{G[[X]]}$ by rigidity for both have the same pattern.
\end{proof}

\begin{remark}
    At the start of the proof we mentioned that Eulerianness was not directly needed in the proof; recall that~$G$ is pseudo-Eulerian. One of the reasons being that most of the `Eulerianness' we need comes inherently with the exhaustiveness of the linkage as we have seen in \cref{lem:exhaustive_yields_pseudo-eulerian}. Further note that the Euler-extended cut \cref{def:euler_extended_cut} when given a rigid linkage can be defined without using the Eulerianness of the graph to begin with, for the rigid linkage yields a matching on the cut-edges by rigidity (see \cref{lem:cuts_and_rigid_linkages}); again compare this to \cref{lem:exhaustive_yields_pseudo-eulerian} which comes with all the needed benefits of Eulerianness. In particular note that~$X$ was any cut in~$G$ and any~$G$ can be made Eulerian given a rigid linkage~$\LLL$ by adding back-edges connecting the terminal pairs of every~$L \in \LLL$ in the obvious way. 
\end{remark}

We generalise the above as follows.

\begin{corollary}
    Let~$p,k \geq 0$, then there exists a function~$f$ such that the following holds. Let~$n \geq f(p,k)$. Let~$\LLL$ be a rigid~$p$-linkage in~$G+D$ with pattern~$D$ and for~$ 1 \leq i \leq n$ let~$X_i \subset V(G)$ induce a~$\leq k$-cut in~$G$ such that
    \begin{enumerate}
        \item~$X_i \subseteq X_j$ for all~$1 \leq i < j \leq n$, and
        \item for~$1 \leq i < i' \leq n$ with~$\Abs{\delta(X_i)} = \Abs{\delta_G(X_{i'})} = k'~$ and~$\Abs{\delta_G(X_j)} > k'$ for some~$k'$ and every~$i < j < i'$, there exists an edge-disjoint~$\delta_G(X_i){-}\delta_G(X_{i'})$-linkage~$M_{ii'}$ in~$G[[X_{i+1}]] \cap G[[\overline{X}_{i'}]]$.
    \end{enumerate}
Then there exists~$1\leq i < i' \leq n$ as above such that for~$X^{ii'} \coloneqq \overline{X}_i \cap X_{i'}$ and~$\LLL^{ii'} \coloneqq \restr{\LLL}{G[[X^{ii'}]]}$ it holds that~$\LLL^{ii'} = M_{ii'}$.
\label{thm:rigid_linkage_in_laminar_cuts_is_Menger_general}
\end{corollary}
\begin{proof}
For~$n\geq f(k,p)$ for some~$f$ large enough---the function~$f(k,p) \coloneqq \Pi_{i=1}^{k}\big((p+i+1)^{2(p+i)}+1\big)$ does the trick---we find a sequence~$i_1<\ldots<i_{n'}$ where~$i_j,i_{j+1}$ satisfy the assumptions in the theorem for some~$n' \geq f_{\ref{thm:rigid_linkage_in_laminar_cuts_is_Menger}}(p,k)$. Now the claim follows at once by looking at the family of cuts given by~$X_{i_1} \subseteq \ldots \subseteq X_{i_{n'}}$.
The details of the proof can be found in the proof of \cite[Theorem 2.6]{GMXXI} and are of purely combinatorial flavor. 
\end{proof}

As promised we next turn back to drawings.
\subsection{Coast Lines}

\label{subsec:coastlines}
Recall the setting as outlined in~\cref{subsec:embedded-incidence-digraphs}, i.e.~$G= \Gamma \cup I$ for pseudo-Eulerian graphs where~$\Gamma$ is pseudo-Euler embedded. We start by defining the \emph{island zone of a cuff~$C \in c(\Sigma)$}. Intuitively one can think of
this as extending the cuff a little bit so that the
vertices drawn on the cuff~$C \subset \bd(\Sigma)$ `fall'
into~$I$ while their two adjacent edges in~$\Gamma$ stay put on the enlarged cuff. So the resulting~$\Gamma'$ and~$I'$ are partial graphs agreeing exactly on the edges drawn on the cuff. That is~$(\Gamma',I')$ is a separation in the sense of \cref{def:edge_separation}.

\begin{definition}[Island Zone, Shores, and Ports]
  Let~$C \in c(\Sigma)$ be some cuff with orientation fixed
  by a curve~$\gamma_c$. Let~$v_1,\ldots,v_n$ be an enumeration of the
  vertices~$V(\nu^{-1}(C))$ with respect to~$\gamma_c$ and
  let~$(l(v_1),r(v_1),\ldots,l(v_n),r(v_n))$ be an ordering (with respect to~$\gamma_c$) of the edges in~$\Gamma$ adjacent to the vertices~$v_1,\ldots,v_n$, where~$l(v_i),r(v_i)$ are the two edges adjacent to~$v_i$ for every~$1 \leq i \leq n$; recall that~$\Gamma$ is pseudo Euler-embedded. The \emph{island-zone of~$C$} is a
  tuple~$\Zone(C) \coloneqq (l(v_1),s_1,r(v_1),\ldots,s_{t-1},l(v_t),s_t,r(v_t))$ where the~$s_i$ are distinct new elements for~$1\leq i \leq j \leq n$.

    We call the edges~$e_i$ \emph{the ports of~$C$} and~$s_i$ \emph{the shores of~$C$} for~$1\leq i \leq t$.
    We further define~$\Shore(C) = \Shore(\Zone(C)) \coloneqq \{s_i \mid 1\leq i\leq t\}$ and~$\Port(C) = \Port(\Zone(C)) \coloneqq \{l(v_i),r(v_i) \mid 1\leq i \leq t\}$. Similarly we write~$\Port(s)$ to mean the adjacent ports to the shore~$s$ and~$\Shore(p)$ to mean the adjacent shores to the port~$p$. Finally we define~$\Port(G) = \Port(G;\Sigma) \coloneqq \bigcup_{C \in c(\Sigma)}\Port(C)$ and~$\Shore(G) = \Shore(G;\Sigma) \coloneqq \bigcup_{C \in c(\Sigma)}\Shore(C)$.
\label{def:ports_shores_zones}
\end{definition}

\begin{remark}
    Note that since there is no edge having both its endpoints in two cuffs by our assumptions, we can draw the island-zone as a closed curve `parallel' to the respective cuff such that it intersects the drawing only in ports using that~$\Gamma$ is embedded in~$\Sigma$. This intuition turns out rather convenient later on, as one can think of the island-zone as pushing the vortex a little further into the surface, so it contains exactly the relevant edges. We will make this more precise once we introduce \emph{cut-lines} (see \cref{def:cut-line}).
    \label{rem:island-zone_as_a_cuff}
\end{remark}

A direct observation reads as follows.
\begin{observation}\label{obs:evenly_many_ports}
    The number of ports of~$G$ given by~$\Abs{\Port(G)}$ is even.
\end{observation}

Using the island zones we can, as hinted at above, define a new separation of~$G$ capturing the importance of ports, namely~$(\Gamma_\ZZZ,I_\ZZZ)$ with~$G= \Gamma_{\ZZZ} \cup I_\ZZZ$.

\begin{definition}[A separation of~$G$ tailored to~$\Zone$]
    Let~$G=\Gamma \cup I$ as above. Let~$V' \coloneqq V(\Gamma) \setminus V\big(\nu^{-1}(\bd(\Sigma))\big)$. Let~$\Gamma_\ZZZ \coloneqq \Gamma \setminus V'$ be the partial graph (with an induced Eulerian embedding) obtained from~$\Gamma$ by deleting the vertices on~$\bd(\Sigma)$. Let~$I_\ZZZ \coloneqq I \cup \Port(G)$. We call~$(\Gamma_\ZZZ,I_\ZZZ)$ a \emph{separation of~$G$ tailored to~$\Zone$}.
    \label{def:separation_tailored_to_zone}
\end{definition}

In a first instance we prove that the name bears significance.
\begin{observation}\label{lem:separation_tailored_to_zone}
    The separation tailored to~$\Zone$ given by~$(\Gamma_\ZZZ,I_\ZZZ)$ is a separation satisfying~$G=\Gamma_\ZZZ \cup I_\ZZZ$ and~$\Gamma_\ZZZ \cap I_\ZZZ = \Port(G)$.
\end{observation}
\begin{proof}
    The claim~$G=\Gamma_\ZZZ \cup I_\ZZZ$ follows at once noting that the incidences of vertices drawn on the cuffs with their edges in~$\Gamma$ are now taken care of in~$I_\ZZZ$ since~$\Port(G) \subset I_\ZZZ$. Then~$\Gamma_\ZZZ \cap I_\ZZZ = \Port(G)$ follows at once form the fact that~$E(\Gamma_\ZZZ \cap I_\ZZZ) = \Port(G)$ by definition and~$V(\Gamma_\ZZZ \cap I_\ZZZ) =\emptyset$ by definition since~$V(\Gamma \cap I) = \nu^{-1}(\bd(\Sigma))$ by assumption.
\end{proof}

Since we fixed the orientations for each cuff in the beginning, and since we assume~$G$ to be given as above, the island-zone of every cuff is unambiguously defined. Note that, technically, the ports and shores depend on the embedding of~$\Gamma$ and thus~$\Port(G)$ and~$\Shore(G)$ implicitly assume the embedding of~$\Gamma$ for~$G = \Gamma \cup I$ to be given as described in out setting outlined in~\cref{subsec:embedded-incidence-digraphs}.

\smallskip

We next define \emph{coast lines}, \emph{coastal maps}, and \emph{islands} of~$G$; concepts that will be heavily used throughout the remaining sections.

\begin{definition}[Coast lines]\label{def:coast line}
    Let~$C \in c(\Sigma)$ be a cuff with island-zone~$\Zone(C)$. A \emph{coast line of~$\Zone(C)$} or simply a \emph{coast (line) of~$C$} is a consecutive sequence of ports and shores, starting and ending in distinct ports that is a sub-sequence of~$\Zone(C)$. That is, a coast line~$\CCC$ is given by~$\CCC \coloneqq (p_1, s_1, \dots, s_{l-1}, p_{l})$. We call~$p_1, p_l$ the \emph{ends} of~$\CCC$ and we call any other port~$p_i$ with~$1<i<l$ an \emph{interior port (to~$\CCC$)}. We write~$\Shore(\CCC)$ and~$\Port(\CCC)$ to mean the obvious.
\end{definition}

\begin{definition}[Weak Coastal Map]\label{def:weak_coastal_map}
Let~$\Zone$ be defined for~$c(\Sigma)$ and let~$(\Gamma_\ZZZ,I_\ZZZ)$ be a separation of~$G$ tailored to~$\Zone$. Let~$\CCC_1, \dots, \CCC_r$ be mutually disjoint coast lines where~$\CCC_i=(p_1,s_1,\ldots,s_{\ell_i-1},p_{\ell_i})$ for some~$\ell_i \geq 1$ and every~$1 \leq i \leq r$ such that~$\Port(G) \subseteq \bigcup_{1\leq i \leq r} \Port(\CCC_i)$. Let~$\mu:\Port(G)\cup \Shore(G) \to \{I' \mid I' \subseteq I\}$ be a map with~$\mu(s) = I_\rho$ for some \textit{partial} subgraph~$I_\rho \subseteq I$ and~$\mu(p) = E_\sigma$ for some \textit{vertex-less} graph~$E_\sigma \subseteq E(I)$ for every~$p \in \Port(G)$ and every~$s \in \Shore(G)$. Further let~$d_i \geq 0$ be integers for every~$1\leq i \leq r$ and define~$d \coloneqq \max \{ d_i \mid 1 \leq i \leq r\}$.

We call the tuple~$\CC \coloneqq (\Gamma, \Zone, \CCC_1, \dots, \CCC_r, \mu)$ a \emph{weak coastal map (of~$G$ in~$\Sigma)$} with~$r$ \emph{sights} and \emph{depth}~$d$ if the following conditions are met.
\begin{enumerate}
    \item [\WMI]\label{def:coastal-map:1}~$G = \Gamma_\ZZZ \cup \bigcup_{i=1}^{r}\mu(\CCC_i)$ where all of the graphs~$\mu(s)$ for~$s\in \Shore(G)$ are mutually vertex disjoint and if~$\chi_1,\chi_2 \in \Port(G) \cup \Shore(G)$ lie in zones of different cuffs, then~$\mu(\chi_1) \cap \mu(\chi_2) = \emptyset$.

    \item[\WMII]\label{def:coastal-map:2} Let~$s \in \Shore(\CCC_i)$ for some~$1 \leq i \leq r$ be a shore with~$\mu(s) \cap V(G) \neq \emptyset$ and let~$p_l,p_r \in \Port(s)$ be its adjacent ports. Then~$\mu(s) \cap \Gamma_\ZZZ = \{p_l,p_r\}$,~$\mu(p_l) \cap E(\Gamma_\ZZZ) = \emptyset$,~$\mu(p_r) \cap E(\Gamma_\ZZZ) = \emptyset$ and~$\mu(p_l),\mu(p_r) \subseteq \mu(s).$

    \item[\WMIII] For each cuff~$C \in c(\Sigma)$ and~$\chi_1,\chi_2 \in \Port(C) \cup \Shore(C)$ and~$p_1,p_2 \in \Port(C)$ such that~$\chi_1,p_1,\chi_2,p_2$ occur in this order on~$\Zone(C)$, then~$\mu(\chi_1)\cap\mu(\chi_2) \subseteq \mu(p_1)\cup \mu(p_2)$. 

   \item[\WMIV] For every~$C \in c(\Sigma)$ and every~$e \in \mu(p)$ for some~$p \in \Port(C) \subset \Zone(C)$, there exists~$p^\ast \in \Port(G) \cap \Zone(C)$ with~$e \notin \mu(p^\ast)$. 
   
   Further, let~$\chi \in \Port(C) \cup \big(\Shore(C)\setminus \bigcup_{i=1}^{r}\Shore(\CCC_i) \big)$ such that either~$\chi$ is a (\emph{deserted}) shore or an interior port of~$\CCC_i$ for some~$1 \leq i \leq r$ and let~$e \in \mu(\chi)$ with ends~$u,v \in V(G)$. Then there exist~$s_l,s_r \in \Shore(C)$ with~$u,v \in \mu(s_l)\cup \mu(s_r)$ and $u \in \mu(s_l) \iff v \in \mu(s_r)$ such that~$s_l,\chi,s_r$ appear in that order on~$\Zone(C)$ and further for every~$\chi' \in \Port(C) \cup \Shore(C)$ such that~$s_l,\chi',s_r$ appear in this order on~$\Zone(C)$ it holds~$e \in \mu(\chi')$.

 \item[ \WMV] For each~$1 \leq i \leq r$ it holds true that~$|\mu(p)| = d_i$ for all ports~$p\in \Port(\CCC_i)$. Let~$s \in \Shore(\CCC_i)$ be a shore and~$p_l,p_r \in \Port(s)$ be its adjacent ports. Then there exist~$d_i$ edge-disjoint~$\mu(p_l){-}\mu(p_r)$-paths in~$\mu(s)$.

 \item[\WMVI] Let~$C \in c(\Sigma)$ and~$s \in \Shore(C)\setminus \bigcup_{i=1}^r \Shore(\CCC_i)$ be a (deserted) shore with adjacent ports~$p_l,p_r \in \mu(s)$. Then~$\mu(p_l)\cap \mu(p_r) \subseteq \mu(s)$.
\end{enumerate}
We say that~$\CCC_i$ has \emph{depth}~$d_i$ for every~$1 \leq i \leq r$ in accordance with the above.
\end{definition}
\begin{remark}
First it is crucial to note that \WMI implies that all of the relations~$\head,\tail$ must be covered by the map, for~$G = \Gamma_\ZZZ \cup \bigcup_{i=1}^r \mu(\CCC_i)$. This is important in the sense that otherwise we would lose information on the graph and its connectivity and paths. This together with the second part of \WMIV then implies that, using the notation of \WMIV, it holds~$e \in\mu(s_l)\cap \mu(s_r)$.

 Note further that the first part of \WMIV implies that for each~$C \in c(\Sigma)$ it holds~$\bigcap_{p \in \Port(C)}\mu(p) = \emptyset$, i.e., there is no port that is central to an \emph{island}, a notion we define next.

 Finally we want to highlight that we would not actually need \WMVI in what is to follow, for the most crucial part behind \WMVI follows from \WMI-\WMV; in a sense \WMVI closes the gap left open in the second part of \WMIV for ports that form ends of coasts. But for the sake of slickness and easier argumentation we included it; also it makes the intentions behind weak-coastal maps clearer.
\end{remark}

We want to highlight that in \WMII the shore~$s \in \Shore(\CCC_i)$ is part of a coast line, while the shores in \WMIII may not. This, as already hinted at in \WMVI, leads to the definition of \emph{deserted shores}. 

\begin{definition} \label{def:deserted_shore}
    Let~$\CC =(\Gamma, \Zone, \CCC_1, \dots, \CCC_r, \mu) $ be a weak coastal map for some graph~$G$ in some surface~$\Sigma$. Let~$s \in \Shore(G)$ be such that~$s \notin \bigcup_{i=1}^{r}\Shore(\CCC_i)$, then we call~$s$ \emph{deserted}.
\end{definition}

Note that for weak coastal maps deserted shores are vertex-less. This follows at once from \WMVI, we will however provide a different proof to highlight that \WMVI is not needed for this observation.
\begin{observation}\label{obs:deserted_shore_vertexless}
    Let~$\CC =(\Gamma, \Zone, \CCC_1, \dots, \CCC_r, \mu) $ be a weak coastal map for some graph~$G$ in some surface~$\Sigma$. Let~$s \in \Shore(G)$ be deserted. Then~$\mu(s) \subset E(G)$.
\end{observation}
\begin{proof}
    This follows at once from~$G = \Gamma_\ZZZ \cup \bigcup_{i=1}^r \mu(\CCC_i)$ imposed by \WMI and the fact that all the~$\mu(s),\mu(s')$ are mutually vertex-disjoint for distinct~$s,s' \in \Shore(G)$.
\end{proof}

Similarly \WMII and \WMIV imply that if two ports are adjacent in~$G$, then the shore between them cannot be deserted and further this implies that the number of ports per coast line of a coastal map is even (compare this to \cref{obs:evenly_many_ports}).
\begin{observation}\label{obs:adj_ports_give_nondeserted_shore_and_are_even}
    Let~$\CC =(\Gamma, \Zone, \CCC_1, \dots, \CCC_r, \mu) $ be a weak coastal map for some graph~$G$ in some surface~$\Sigma$. Let~$p_l,p_r \in \Port(G)$ be two consecutive ports and let~$s \in \Shore(p_l) \cap \Shore(p_r)$. Assume further that the ports are adjacent in~$I_\ZZZ$, i.e., there is~$v \in \nu^{-1}(\bd(\Sigma))$ such that~$p_l$ and~$p_r$ are adjacent to~$v$. Then \begin{enumerate}
        \item[(i)] $p_i \in \mu(\chi)$ if and only if~$\chi = s$ for any~$\chi \in \Port(G) \cup \Shore(G)$ and~$i\in\{l,r\}$, and
        \item[(ii)] $v \in \mu(s)$.
    \end{enumerate}
    
    Further let~$u \in V(\Gamma) \cap V(I)$. Then there are distinct adjacent edges~$p_1,p_2 \in \Port(G)$ and in particular~$\Abs{\Shore(\CCC_i)}$ is even for all~$1 \leq i \leq r$.
\end{observation}
\begin{proof}
    By \WMIV the shore~$s$ cannot be deserted for~$p_l,p_r$ have only one end in~$I_\ZZZ$ and thus combining this with \WMIII they can only be part of one shore, that is~$\mu(s)$ by \WMII---this implies (i). So~$s$ is part of~$\CCC_i$ for some~$1 \leq i \leq r$. Using \WMII we have~$p_l,p_r \in \mu(s)$. Now either~$v \in \mu(s)$ or~$v,p_l,p_r \in \mu(s')$ for some other shore~$s' \in \Shore(G)$ by \WMI. The latter is impossible by (i).
    
\smallskip

    The second part follows at once from the fact that~$\Gamma$ is pseudo Euler-embedded and the fact that by \WMI every~$v \in V(\Gamma)\cap V(I)$ is part of some~$\mu(s)$ for some non-deserted shore~$s \in \Shore(G)$; the claim then follows by the first part.
\end{proof}
This is a slightly sneaky but rather impactful observation; note that the observation seems like a natural requirement that could have been included as another assumption on the coastal maps.
Similarly we get the following observation which is not important for the following but it shows how strong our assumptions are.

\begin{observation}\label{obs:every_second_shore_is_vertexless}
   Let~$\CC =(\Gamma, \Zone, \CCC_1, \dots, \CCC_r, \mu) $ be a weak coastal map for some graph~$G$ in some surface~$\Sigma$. Let~$p_l,p_r \in \Port(G)$ be two consecutive ports and let~$s \in \Shore(p_l) \cap \Shore(p_r)$. Assume further that the ports are \emph{not} adjacent in~$I_\ZZZ$. Then~$\mu(s) \subset E(G)$ is vertex-less.
\end{observation}
\begin{proof}
    This follows at once from \WMII and \WMIII for~$p_l \in \mu(s^*)$ for some~$s^* \in \Shore(p_l)$ but also~$p_l \notin \mu(p_l)$ and thus if~$\mu(s) \cap V(G) \neq \emptyset$ we have~$p_l \in \mu(s)$ which together with~$p_l \in \mu(s^*)$ and \WMIII implies~$p_l \in \mu(p_l)$ as a contradiction.
\end{proof}

We continue with the definition of \emph{weak islands}.
\begin{definition}[Weak Island]
    Let~$(\Gamma,\Zone, \CCC_1, \dots, \CCC_r, \mu)$ be a weak coastal map of~$G$ in~$\Sigma$ with~$r \geq 1$ sights. Let~$C \in c(\Sigma)$ be a cuff and let~$\CCC_{i_1},\ldots,\CCC_{i_s}$ for some~$s \leq r$ be such that~$\Port(\bigcup_{t=1}^s \CCC_{i_t})=\Port(C)$, where~$1 \leq i_t \leq r$. Then we refer to~$(C,\Zone(C),\CCC_1,\ldots,\CCC_s,\mu)$ as the \emph{weak island of~$C$}.
    \label{def:weak_island}
\end{definition}
Note that by \WMI all of the weak islands are disjoint.

\begin{observation}
    The weak-islands of a weak coastal map are pairwise disjoint.
    \label{obs:weak_islands_are_disjoint}
\end{observation}

Another observation following from the \cref{def:weak_coastal_map} is that, the shores and ports containing a fixed edge form a coast line. This is where we need \WMVI, and the following lemma will help us shorten some proofs in what is to follow. 

\begin{lemma}\label{lem:edges_define_coast lines}
    Let~$(C,\Zone(C),\CCC_1,\ldots,\CCC_s,\mu)$ be a weak island of~$C \in c(\Sigma)$. Let~$e \in \bigcup_{i=1}^{s} \mu(\CCC_i)$ be an edge in the island. Let~$P \coloneqq \{p \in \Port(G) \mid e \in \mu(p)\}$ and~$S \coloneqq \{s \in \Shore(G) \mid e \in \mu(s)\}$ be the sets of ports and shores containing~$e$. Then there exists a sub-sequence~$C^\ast\subset \Zone(C)$ such that~$\Port(C^\ast) = P$ and~$\Shore(C^\ast) = S$.
\end{lemma}
\begin{proof}
    By \WMI it is clear that~$P,S \subset \Port(C)\cup \Shore(C)$. By the first part of \WMIV there is~$p^\ast \in \Port(C)$ with~$e \notin \mu(p^\ast)$, in particular~$P\cup S$ does not cover all of~$\Zone(C)$. Towards a contradiction and without loss of generality assume that~$P\cup S$ forms exactly two disjoint sub-sequences~$C_1,C_2$ of~$\Zone(C)$ and choose both to be maximal, i.e., none can be extended by some~$\chi \in P\cup S$. If~$P\cup S$ would form more than~$2$ sub-sequences we could choose one to be~$C_1$ and combine all the others to form~$C_2$ which is not a sub-sequence per se but the proof works analogously. 
    
    By \WMII both~$C_1$ and~$C_2$ contain at least one shore. Let~$\chi_1 \in \Shore(C_1)$ and~$\chi_2 \in \Shore(C_2)$ abusing notation in the obvious way and, without loss of generality, assume~$\chi_1,p^\ast,\chi_2$ appear in this order (else rename~$\chi_1,\chi_2$). Then using \WMIII and the fact that~$e \notin \mu(p^\ast)$ we deduce that for all~$p \in \Port(G)$ such that~$\chi_2,p,\chi_1,p^\ast$ appear in this order on~$\Zone(C)$ it follows~$e \in \mu(p)$ and hence~$p \in P$. But then using the second part of \WMVI we deduce that~$e \in \mu(s)$---and thus~$s \in S$---for all the shores~$s \in \Shore(C)$ such that~$\chi_2,s,\chi_1,p^\ast$ appear in said order. In particular~$C_1$ and~$C_2$ have either not been maximal or not been disjoint; contradiction.
\end{proof}
\begin{remark}
    As mentioned above this lemma is the main reason why we impose \WMVI. Without \WMVI we could get several sub-sequences~$C_1,C_2,\ldots,C_l$ appearing in this order on~$\Zone(C)$ such that~$C_i,C_{i+1}$ are adjacent to some common deserted shore for every~$1 \leq i <l$ for without \WMVI we could set~$\mu(s) = \emptyset$ for deserted shores without violating the other conditions.
\end{remark}
It is important to note that~$C^\ast$ in \cref{lem:edges_define_coast lines} may not be a coast line in the sense of \cref{def:coast line} for it may not end in ports but in deserted shores for example. The following definition comes naturally.

\begin{definition}[The sub-sequence~$\CCC(e)$]\label{def:coastline_defined_by_edge}
    Let~$(C,\Zone(C),\CCC_1,\ldots,\CCC_s,\mu)$ be a weak island given some weak coastal map~$\CC$. Let~$e \in \bigcup_{i=1}^{s} \mu(\CCC_i)$. Then we define~$\CCC(e) \subset \Zone(C)$ via~$\CCC(e) \coloneqq \CCC^\ast$ where~$\CCC^\ast$ is defined as in \cref{lem:edges_define_coast lines}.
\end{definition}

We continue with a strengthening of \cref{def:weak_coastal_map}, namely the definition of \emph{strong coastal maps}.

\begin{definition}[Strong Coastal Map]
\label{def:strong_coastal_map}
Let~$\Gamma, \Zone, \CCC_1, \dots, \CCC_r,\mu$ be defined as in \cref{def:weak_coastal_map} (with~$\mu$ not necessarily satisfying \WMI to \WMV yet). Further let~$d_i\geq 0$ be integers for every~$1 \leq i \leq r$ and let~$d \coloneqq \max\{d_i \mid 1 \leq i \leq r\}$. Define~$H_i \coloneqq \mu(\CCC_i) = \bigcup_{1\leq i \leq \ell_i} \mu(p_i) \cup \bigcup_{1\leq i \leq \ell_i-1} \mu(s_i)$. We say that~$\CC \coloneqq (\Gamma, \Zone, \CCC_1, \dots, \CCC_r, \mu)$ is a \emph{strong coastal map (of~$G$ in~$\Sigma$)} with~$r$ \emph{sights} and \emph{depth}~$d$ if the following conditions are met.

\begin{enumerate}
    \item[{\small(SM1)}]~$G = \Gamma_\ZZZ \cup \bigcup_{i=1}^{r}H_i$ where all of the graphs~$\mu(s)$ for~$s\in \Shore(G)$ are mutually vertex disjoint and~$H_i \cap H_j = \emptyset$ for distinct~$1\leq i,j \leq r$.
    \item[{\small(SM2)}] The map satisfies \WMII. 
    \item[{\small(SM3)}] For~$1 \leq j \leq r$ if~$s,s' \in \Shore(\CCC_i)$ are shores and~$p \in \Port(\CCC_i)$ lies between them, then~$\mu(s)\cap \mu(s')\subseteq \mu(p)$ (and consequently~$V(\mu(s)) \cap V(\mu(s')) = \emptyset$).
    \item[{\small(SM4)}] Let~$e \in \mu(p)$ for some interior~$p \in \CCC_i$. Then there exist~$s_l,s_r \in \Shore(\CCC_i)$ with~$s_l,p,s_r$ appearing on~$\CCC_i$ in that order such that~$u,v \in \mu(s_l)\cup \mu(s_r)$ and~$u \in \mu(s_l) \iff v \in \mu(s_r)$.

    \item[{\small(SM5)}] The map satisfies \WMV.
    
    \item[{\small(SM6)}] Let~$C \in c(\Sigma)$ and~$s \in \Shore(C)\setminus \bigcup_{i=1}^r \Shore(\CCC_i)$ be a (deserted) shore. Then~$\mu(s) = \emptyset$.
\end{enumerate}
We say that~$\CCC_i$ has \emph{depth}~$d_i$ for every~$1 \leq i \leq r$ in accordance with the above.
\end{definition}
\begin{remark}
    As for weak coastal maps the condition {\small{(SM6)}} is not actually needed to prove the main theorem, but it helps with the intuition highlighting the difference between strong and weak coastal maps.
\end{remark}

Using the above we define strong islands as follows.
\begin{definition}[Strong Island]
    Let~$(\Gamma, \Zone,\CCC_1, \dots, \CCC_r, \mu)$ be a strong coastal map of~$G$ in~$\Sigma$ with~$r \geq 1$ sights. Let~$C \in c(\Sigma)$ be a cuff and let~$\CCC_{i_1},\ldots,\CCC_{i_s}$ for some~$s \leq r$ be such that~$\Port(\bigcup_{t=1}^s \CCC_{i_t})=\Port(C)$, where~$1 \leq i_t \leq r$. Then we refer to~$(C,\Zone(C),\CCC_1,\ldots,\CCC_s,\mu)$ as the \emph{strong island of~$C$}.
    \label{def:strong_island}
\end{definition}

The following is immediate from the \cref{def:weak_coastal_map} and \cref{def:strong_coastal_map}.
\begin{observation}\label{obs:disjointness_of_strong_coastlines}
    Let~$(\Gamma,\Zone,\CCC_1,\ldots,\CCC_r,\mu)$ be a weak or strong coastal map. Then for every~$C \in c(\Sigma)$ there exists at least one~$s \in \Shore(C)$ such that~$s \notin \Shore(\CCC_i)$ for all~$1\leq i \leq r$, i.e.,~$s$ is \emph{deserted}.
\end{observation}
\begin{remark}
    The notion of deserted shores for strong coastal maps is defined analogously to the deserted shores for weak coastal maps; see \cref{def:deserted_shore}.

    Note however that for strong coastal maps {\small(SM6)} yields~$\mu(s) =\emptyset$ for any deserted shore~$s$ which does not (necessarily) hold for weak coastal maps---compare \small{(SM3)} and \WMIII. 
\end{remark}

In particular we may naturally define the `left' and `right' sub-coasts of a coast line with respect to a port for strong (weak) coastal maps; we will only need this notion for strong coastal maps in what is to follow.
\begin{definition}
    \label{def:left_right_coastlines}
    Let~$(\Gamma,\Zone,\CCC_1,\ldots,\CCC_r,\mu)$ be a strong (or weak) coastal map and let~$p \in \Port(G) \cap \CCC_i$ for some~$1\leq i \leq r$. Then there exist two coastlines~$\CCC_i^l(p) \subset \CCC_i$ and~$\CCC_i^r(p)\subset \CCC_i$ that are internally disjoint and agree exactly at~$p$ such that~$\CCC_i^l(p) \cup \CCC_i^r(p) = \CCC_i$. Assume further that~$\CCC_i^l(p)$ is visited prior to~$\CCC_i^r(p)$ with respect to the ordering on~$\CCC_i$ induced by~$\Zone$. Then we call~$\CCC_i^l(p)$ the \emph{left coast of~$\CCC_i$ at~$p$} and~$\CCC_i^r(p)$ the \emph{right coast of~$\CCC_i$ at~$p$}.
\end{definition}

Further we get the following relation for the ends of ports.

\begin{observation}\label{obs:ports_are_divided_by_separation_tailored_to_zone}
    Let~$\CC$ be a weak or strong coastal map for~$G$ in~$\Sigma$ where~$G=\Gamma_\ZZZ \cup I_\ZZZ$ as above. Then for every port~$p \in \Port(G)$ with ends~$u,v$, say,~$p=(u,v)$, it holds that~$u \in \Gamma_\ZZZ \iff v \in I_\ZZZ$ and in particular they are not both contained in the same part of the separation~$(\Gamma_\ZZZ,I_\ZZZ)$.
\end{observation}
\begin{proof}
    This follows at once from \cref{lem:separation_tailored_to_zone} together with \WMI and similarly (SM1) for weak and strong coastal maps respectively.
\end{proof}

Finally note that weak and strong coastal maps naturally divide the graph into~$\Gamma_\ZZZ$ and~$I_\ZZZ$; in a sense they `map' the separation tailored to~$\Zone$.
\begin{observation}\label{obs:maps_yield_separation_tailored_to_zone}
    Let~$\CC=(\Gamma,\Zone, \CCC_1, \ldots, \CCC_r, \mu)$ be a weak (or strong) coastal map of~$G$ in~$\Sigma$ and let~$(\Gamma_\ZZZ,I_\ZZZ)$ be the separation tailored to~$\Zone$. Then~$I_\ZZZ \coloneqq \bigcup_{i=1}^r\mu(\CCC_i)$ and in particular~$\Gamma_\ZZZ \cap \bigcup_{i=1}^r\mu(\CCC_i) = \Port(G)$.
\end{observation}
\begin{proof}
    The second part follows at once from \cref{def:separation_tailored_to_zone}. Now by \WMI (and (SM1) respectively) we deduce~$I_\ZZZ \subseteq \bigcup_{i=1}^r\mu(\CCC_i)$. Using \WMII (and (SM2) respectively) we further deduce that~$\Gamma_\ZZZ \cap \bigcup_{i=1}^r\mu(\CCC_i) \subseteq \Port(G)$. This concludes the proof.
\end{proof}

As discussed above, since the zones are unambiguously defined given a fixed orientation of the cuffs---which we assume---we may simply write~$(\Gamma, \CCC_1, \ldots, \CCC_r, \mu)$ for coastal maps ignoring the zones.

the remainder of this subsection is devoted to definitions concerning \emph{strong} coastal maps---note however that we could define them analogously for weak coastal maps, but we do not need to. We start by defining the linkages of interest for this section, that is \emph{$\Sigma$-routes} and \emph{shippings}. 

\begin{definition}[$\Sigma$-route and~$\Sigma_\ZZZ$-route]
    Let~$G+D$ be Eulerian, let~$\LLL$ be a~$p$-linkage in~$G$ for some~$p \in \N$, and let~$(\Gamma, \Zone, \CCC_1, \dots, \CCC_r, \mu)$ be a strong coastal map of~$G$ in~$\Sigma$. Then any edge-minimal sub-path~$L' \subset \restr{L}{\Gamma}$ of~$L \in \LLL$ contained in~$\Gamma$ with both endpoints in~$V\big(\nu^{-1}(\bd(\Sigma))\big)$ that is internally disjoint from~$\bd(\Sigma)$ is called a \emph{~$\Sigma$-route in~$\Gamma$}. 

    Let~$L' \subset \restr{L}{\Gamma}$ be an edge-minimal sub-path (not necessarily ending in vertices) with~$E(\pi(L)) \subseteq \Port(G)$ such that~$L'$ is internally edge-disjoint from~$\Port(G)$, then we call~$L'$ a~\emph{$\Sigma_\ZZZ$-route}.

    Let~$L'' \subset \restr{L}{\Gamma}$ be an edge-maximal sub-path (not necessarily ending in vertices) with~$E(\pi(L)) \subseteq \Port(G)$ such that~$L''$ is internally edge-disjoint from~$I$. Then we call~$L''$ a~$\Sigma_\ZZZ$-bounce.
    \label{def:sigma_routes}
\end{definition}

\begin{remark}
    A~$\Sigma$-route is thus a shortest path with both `ends'---which are vertices---on the boundary of~$\Sigma$, that is otherwise disjoint from~$\nu^{-1}(\bd(\Sigma))$.

    A~$\Sigma_\ZZZ$-route is a shortest path in~$\Gamma$ with both `ends'---which are edges---in island zones with respect to~$\Zone$; in particular they are ports. Note that the~$\Sigma_\ZZZ$-routes are exactly given by the linkage~$\restr{\LLL}{\Gamma_\ZZZ}$.

    A~$\Sigma_\ZZZ$-bounce is a longest path in~$\Gamma$ with both ends in island-zones. The difference to routes coming from the fact that routes do not know whether the sub-path will enter the island after its end or whether it \emph{bounces back} and stays in~$\Gamma$.
\end{remark}

We immediately get the following.
\begin{observation}
    The~$\Sigma$-routes are in one-to-one correspondence with~$\Sigma_\ZZZ$-routes. Further, let~$L' \subseteq L \in \LLL$ be a~$\Sigma_\ZZZ$-route then~$L' \in \restr{\LLL}{\Gamma_\ZZZ}$.
    \label{obs:sigma_routes_are_port_paths}
\end{observation}
\begin{proof}
    Let~$L$ be a~$\Sigma$-route, then since it is contained in~$\Gamma$ and ends in the boundary-vertices of~$\Sigma$ we immediately derive~$E(\pi(L)) \subset \Port(G)$ and thus it yields a~$\Sigma_\ZZZ$-route by the edge-minimality condition imposed in \cref{def:sigma_routes}. Similarly adding the terminal vertices to~$\Sigma_\ZZZ$-routes, they end in boundary-vertices of~$\Sigma$ and give a~$\Sigma$-route.
    
    The second claim is obvious by \cref{def:restricting_linkages_to_subgraphs} and \cref{def:sigma_routes}
\end{proof}

We will use~$\Sigma$-bounces and routes to analyse how the linkage~$\LLL$ enters and leaves~$I_\ZZZ$, i.e., to what extend it uses the ports. Note however that simply using ports is not telling us whether the linkage will leave~$\Gamma$ for it may `bounce back'. This leads to the following definition.
\begin{definition}[Shipping]
    Let~$(\Gamma, \Zone, \CCC_1, \dots, \CCC_r, \mu)$ be a coastal map of~$G$ in some surface~$\Sigma$. Let~$\LLL$ be a rigid~$p$-linkage with pattern~$D$. Then we call
    \begin{itemize}
        \item $\LLL$ a \emph{$p$-shipping} if every~$L \in \LLL$ that is a trapped route in~$\Gamma$ satisfies~$V(L) \cap V\big(\nu^{-1}(\bd(\Sigma)\big) \neq \emptyset$,
        \item $\LLL$ a \emph{piercing}~$p$-shipping if every~$\Sigma_\ZZZ$-bounce is a~$\Sigma_\ZZZ$-route. 
    \end{itemize}
    If~$G$ is completely Euler-embedded in a surface~$\Sigma$, then we call any linkage~$\LLL$ with pattern~$D$ a~\emph{(piercing) $p$-shipping}.
    \label{def:shipping}
\end{definition}

In order to manipulate~$G$ based on the drawing of~$\Gamma$ (and~$\Gamma_\ZZZ$) we continue by defining certain \emph{lines} in~$\Sigma$ along which we will split edges of~$G$ and cut open the surface for inductive reasoning.

\begin{definition}[Cut-lines]
     A \emph{cut-line}~$\ell = \gamma([0,1])$ is the image of a map~$\gamma:[0,1] \to \Sigma$ such that~$\restr{\gamma}{]0,1[}$ is injective and~$\gamma(]0,1[) \cap U(\Gamma) \subset \nu(E)$, i.e., it is a curve in~$\Sigma$ that internally only passes through edges (drawn as points) on~$\Sigma$. We call~$\gamma(0)$ the \emph{starting point} and~$\gamma(1)$ the \emph{endpoint} (or the \emph{ends}) of~$\ell$ and we say that the cut-line~$\ell$ is \emph{closed} if~$\gamma(0) = \gamma(1)$.
     \label{def:cut-line}
\end{definition}
\begin{remark}
Given this definition we can, in the spirit of \cref{rem:island-zone_as_a_cuff}, draw the island-zone of a cuff as a closed cut-line such that it only intersects ports of the respective cuff in the order prescribed by the orientation of the cuff. Viewed that way one can think of~$\Sigma_\ZZZ$ as being the surface resulting from enlarging the cuffs by cutting along the island zones (and enlarging~$I$ to~$I_\ZZZ$) keeping the ports drawn on the boundary of the new cuffs and throwing all of the vertices initially drawn on the cuff into~$I_\ZZZ$, such that the new embedded graph is an embedding of~$\Gamma_\ZZZ$. This reveals the idea behind the definition of~$\Sigma_\ZZZ$-routes.
\end{remark}

Next we refine the notion of cut-lines further.

\begin{definition}[Route-tracing Cut-lines]
    Let~$(\Gamma, \Zone, \CCC_1,\ldots,\CCC_r,\mu)$ be a coastal map. Let~$\ell$ be a cut-line in~$\Gamma$. Then~$\ell$ is \emph{route-tracing} if its ends lie in~$\Port(G)$ which are drawn as points on~$\bd(\Sigma)$. That is~$\ell$ has both its ends in island-zones.
\label{def:route-tracing_cut-line}
\end{definition}

This finishes our setup for the following sections. Note that we will only use (route-tracing) cut-lines in \cref{sec:shippings}. In a next step we prove how to \emph{chart islands}, that is, we prove how to obtain a weak coastal-map given a separation~$(\Gamma,I)$ of~$G$ as assumed.

\subsection{Charting an island}
\label{subsec:charting_an_island}

The main result of this section is the following theorem yielding a weak coastal map for graphs adhering to the general setting outlined in \cref{subsec:embedded-incidence-digraphs}; recall said setting. That is, we are given some Eulerian graph~$G+D$ and some proper pseudo-Eulerian subgraphs~$\Gamma,I$ such that~$G= \Gamma \cup I$ where~$\Gamma$ is pseudo Euler-embeddable in some surface~$\Sigma$ with~$I$ attached to~$\Gamma$ exactly on the boundary of~$\Sigma$ and each component of~$I$ being connected to vertices of at most one cuff. That is,~$I$ is the disjoint union of (possibly not connected and possibly empty) components~$I_1,\ldots,I_c$ for~$c \coloneqq \Abs{c(\Sigma)}$ where each component has vertices in common with at most one of~$\nu^{-1}(C^i)$ for~$1 \leq i \leq c$ and~$C^i \in c(\Sigma)$. Note that by the assumptions of our setting we assume that~$V(D) \cap V(\nu^{-1}(\bd(\Sigma)) = \emptyset$, in particular~$V(D) \subset V(G) \setminus \big(V(\Gamma) \cap V(I)\big)$.

Also recall from \cref{par:surfaces} the definition of the surface~$\hat{\Sigma}$ obtained from a surface~$\Sigma$ with boundary by \emph{capping} each hole by a disc.  

We are ready to state and prove the main theorem, which will only be used in \cref{sec:structure_thms}; the reader may hence skip this part for now and come back to it once needed.  

\begin{theorem}\label{thm:structure-thm-to-coastal-maps}
  Suppose that~$G+D$ is Eulerian with~$\Abs{V(D)} = 2p \in \N$ and let~$\Gamma$ and~$I$ be pseudo-Eulerian digraphs such that~$G = \Gamma \cup I$ and~$\Gamma \cap I \subset V(G)\setminus V(D)$. In addition, let~$(U, \nu)$ be a pseudo Euler-embedding of~$\Gamma$ into some surface~$\Sigma$ (with boundary) where~$\nu(I \cap \Gamma) \subseteq \bd(\Sigma)$. 
  
  Let~$C^1, \dots, C^\sigma \in c(\Sigma)$ be the cuffs of~$\Sigma$ for~$\sigma \coloneqq \Abs{c(\Sigma)}$. 
  Suppose that there are pairwise disjoint digraphs~$I_1 \cup \dots \cup I_\sigma$ such that~$I = I_1 \cup \dots \cup I_\sigma$ and~$V(\Gamma \cap \nu^{-1}(C^i)) = V(\Gamma \cap I_i)$, for all~$1\leq i \leq \sigma$. 

  For~$1 \leq i \leq \sigma$ let~$v^i_1, \dots, v^i_{l_i}$ be the cyclic ordering of~$V(\Gamma \cap \nu^{-1}(C^i))$ induced by the orientation of~$C^i$. Suppose that for all~$1 \leq i \leq \sigma$ there is~$d_i \geq 0$ such that for all~$1 \leq j \leq j' \leq l_i$ there is no edge-disjoint linkage of order~$> d_i$ between~$\{v^i_j, \dots, v^i_{j'}\}$ and~$\{v^i_1, \dots, v^i_{j-1}, v^i_{j'+1}, \dots, v^i_{l_i}\}$.

  Further, assume that for each~$1 \leq i \leq \sigma$ there is a set~$\CCC_i$ of~$\kappa \coloneqq (p+1)(2(p+d_i)+1)$ pairwise edge-disjoint cycles~$C^i_0, \dots, C^i_{\kappa-1}$ in~$\Gamma$ such that, for all~$0 \leq j \leq \kappa-1$,~$C^i_j$ is itself a vertex-disjoint cycle and bounds a disc~$\Delta^i_j$ in~$\hat{\Sigma}$ such hat~$\Delta^i_{\kappa}$ contains~$C^i$ and~$\Delta^i_j$ contains~$\Delta^i_{j+1}$, for all~$0\leq j <\kappa-1$.

  Then there are~$r \geq 1$ and~$\hat{d_i} \geq 0$, for all~$1 \leq i \leq r$, and a subdivision~$G'$ of~$G$ and an embedding~$(U',\nu')$ of a subgraph~$\Gamma' \subseteq G'$ which agrees with~$\Gamma$ (up to subdividing and a homeomorphism) such that~$G'$ has a weak coastal map 
~$\MMM \coloneqq (\Gamma', \Zone, \CCC'_1, \dots, \CCC'_r, \mu)$ with~$r$ sights and depth~$d \coloneqq \max\{\hat{d_i} \mid 1 \leq i \leq r\}$ and each~$\CCC_i$ has depth~$d_i$ for~$1 \leq i \leq r$.
\end{theorem}
 The remainder of this chapter is devoted to a proof of \cref{thm:structure-thm-to-coastal-maps} which, for the ease of readability, is provided via several lemmas. The main ideas of this chapter are adaptations of the proof of \cite[Lemma 15]{DiestelKMW2012}; note however that our proof is slightly more involved for we cannot simply delete part of the vertices and edges of the graph (by inluding them to some apex-set) and thus need to argue how to put them back into the graph when constructing the linear decomposition as we will explain next.

 \smallskip
 
  We will apply the following construction independently to each of the cuffs~$C^i \in c(\Sigma)$ to obtain the required outcome of the theorem. 
  The construction does only manipulate the part of~$\Gamma$ contained in the part bounded by the cycles around the cuff~$C^i$, and is thus locally not dependent on the other cuffs. As by assumption the (possibly non-connected and possibly empty) components~$I_1,\ldots,I_\sigma$ making up~$I$ are pairwise vertex disjoint, it will follow directly from the construction that we can do this independently for each cuff and at the end combine the coast lines obtained for distinct cuffs. So it suffices to describe the procedure for a single cuff~$C^i$. Henceforth we proceed as follows.
  
  \smallskip
  
  Let~$C_{\Sigma} \coloneqq C^i$, let~$d \coloneqq d_i$, let~$H \coloneqq I_i$, and let~$O \coloneqq V(\Gamma \cap C_{\Sigma})=(v_1,\ldots,v_n)$ be an ordering of the vertices on the cuff~$C_\Sigma$ for some~$1\leq i \leq \sigma$. Let~$C_0, \ldots, C_{\kappa-1}$ be the set of~$\kappa \coloneqq (p+1)(2(p+d)+1)$ concentric cycles surrounding~$C_\Sigma$, such that~$C_i$ bounds a disc in~$\hat{\Sigma}$ denoted by~$\Delta(C_i)$, which contains~$C_{i+1}$, for all~$0 \leq i \leq \kappa -1$ and such that~$C_{\kappa-1}$ is as `close as possible' to~$C_{\Sigma}$ with respect to the containment relation of the disc bounded by the cycle. Recall that by \cref{obs:disc_embedd_is_2cell} the restriction of the embedding~$(U,\nu)$ to the disc~$\Delta(C_0)$ is~$2$-cell (under the assumption that~$G[\Delta(C_0)]$ is connected, otherwise faces may be homeomorphic to cylinders).  
  
  \paragraph{Simplifying the assumptions.} 
  We claim that we can slightly change the cuff~$C_\Sigma$ so that~$V(\nu^{-1}(C_\Sigma)) = V(C_{j})$ for some~$0 \leq j < \kappa-2(p+d)+1$ such that there are~$2(p+d)+1$ cycles~$C_{j+1},\ldots,C_{j+2d+1}$ surrounding~$C_j$ where the cylinder bounded by~$C_j$ and~$C_{j+2(p+d)+1}$ does not contain any traps. In particular~$C_{j}$ can be assumed to cover the vertices in~$O$; i.e., we may assume that there is a cycle in~$\Gamma$ concentric to~$C_\Sigma$ that goes through all the vertices on the cuff by enlarging the cuff appropriately. The proof is given in two major steps.

 \begin{lemma}\label{lem:charting_a_cycle_to_the_cuff}
     There exists a vertex-disjoint cycle~$C^+=(u_0,\ldots,u_t) \subset \Gamma$ for some~$u_1,\ldots,u_t \in V(\Gamma)$ and some~$t\geq 1$ that is concentric with~$C_\Sigma$ such that for the disc~$\Delta(C^+) \subset \hat{\Sigma}$ bounded by~$C^+$ and containing~$C_\Sigma$ it holds that there exists no linkage~$\LLL$ of order~$t(C^+)+d$ in~$G[\Delta(C^+)]-C^+$ between~$\{u_j, \dots, u_{j'}\}$ and~$\{u_1, \dots, u_{j-1}, u_{j'+1}, \dots, u_{t}\}$ for respective~$1 \leq j,j' \leq t$ where~$2t(C^+) \leq 2p$ is the number of traps in~$\Gamma[\Delta(C^+)]$. Moreover~$C^+ = C_{\kappa-1}$ and hence there exist at least~$\kappa-1$ cycles concentric with~$C^+$ that are drawn outside of~$\Delta(C^+)$.
 \end{lemma} 
 \begin{proof}
    Choose~$C^+$ vertex-disjoint such that~$\Delta(C^+)$ is minimal, i.e., there exists no other vertex-disjoint cycle~$C' \in \Gamma$ concentric with~$C_\Sigma$ such that~$\Delta(C') \subseteq \Delta(C^+)$. In particular then either~$C^+ = C_{\kappa-1}$ or~$C^+ \subset \Delta(C_{\kappa-1})$ and thus the second part of the lemma is satisfied. (Note that by assumption on~$C_{\kappa-1}$ it follows that~$C^+ = C_{\kappa-1}$). We claim that~$C^+$ does the trick using the fact that~$\Gamma$ is pseudo Euler-embedded, that is every vertex in~$V(\Gamma) \setminus V(D)$ is of degree two or four and~$\Gamma$ is Euler-embedded; in particular note that~$V(C^+) \cap V(D) = \emptyset$. To see this, first let~$H' \coloneqq \Gamma[\Delta(C^+)] - C^+$, i.e., we delete all the edges of the cycle (compare this to \cref{lem:paths_in_pseudo_Euler_cut_disc_into_pseudo_Discs} and the subsequent remark). By the assumptions of our setting, every vertex~$v \in V(C^+)$ is either of degree zero or degree two in~$H'$; dissolve the degree zero vertices.

    Finally let~$\hat{H'} \coloneqq H' \cap \Gamma$. 
    \begin{claim}\label{lem:charting_a_cycle_to_the_cuff_claim1}
        $\hat{H'}$ is a pseudo-Eulerian graph that is pseudo Euler-embedded in the disc~$\Delta(C^+)$ and the number of traps in~$\hat{H'}$ equals the number of traps in~$\Gamma[\Delta(C^+)]$.
    \end{claim}
    \begin{ClaimProof}
       To see that~$\hat{H'}$ is pseudo Euler-embedded note that every vertex inside the cylinder bounded by~$C^+$ and~$C_\Sigma$ is either a trap or of degree two or four where the latter vertices are Euler-embedded using the fact that~$\Gamma$ and~$\hat{H'}$ agree away from the boundary of the cylinder. Finally every vertex of~$C^+$ has degree two in~$\hat{H'}$ by construction---we dissolved degree zero vertices---while every vertex of~$\nu^{-1}(C)$ is of degree two in~$\hat{H'}$ using the fact that~$\Gamma$ is pseudo Euler-embedded. Thus no new traps arise in~$\hat{H'}$ altogether concluding the proof.
    \end{ClaimProof}

    In particular then the number of traps in~$\hat{H'}$ is even (see also \cref{obs:pseudo_Euler_can_be_capped}); let~$2t(C^+)$ denote the number of traps in~$\hat{H'}$.
    
    \smallskip

    For the sake of contradiction assume that there is a linkage~$\LLL$ of order~$t(C^+)+d$ with paths disjoint from~$E(C^+)$ satisfying the conditions of the lemma. 

    Next let~$\LLL' \subset \LLL$ be the maximal sub-linkage such that every path~$L \in \LLL'$ is contained in~$H'$, i.e., the paths of~$\LLL'$ that are drawn in~$G[\Delta(C^+)]$ edge-disjoint from~$H$. 
    With \cref{lem:charting_a_cycle_to_the_cuff_claim1} at hand we can deduce the following.
    \begin{claim}\label{lem:charting_a_cycle_to_the_cuff_claim2}
       Either~$\LLL'$ is of order~$\leq t(C^+)$ or there is~$L \in \LLL'$ that can be extended to a cycle~$C' \subset \hat{H'}$.
    \end{claim}
    \begin{ClaimProof}
        The proof follows at once by applying \cref{lem:charting_a_cycle_to_the_cuff_claim1} and \cref{cor:linkage_in_pseud_Ebedding_in_disc_yields_traps_or_cycle} noting that there are at most~$p$ trapped routes in~$\hat{H'}$. 
    \end{ClaimProof}

We are left with deriving a contradiction for both outcomes of \cref{lem:charting_a_cycle_to_the_cuff_claim2}. First assume that there exists a vertex disjoint cycle~$C' \subset \hat{H'}$ (simply reroute at intersections to make it vertex-disjoint)---in particular it is a cycle in~$\Gamma[\Delta(C^+)]$ distinct from~$C^+$. If this cycle is concentric with~$C_\Sigma$ we get a contradiction to the minimality of~$C^+$, thus it is not. Then let~$C' = L \cup L'$ where~$L \in \LLL$ and~$L'$ is the extended part closing it to a cycle. Then both have their ends on~$C^+$ and thus we can reroute~$C^+$ either along~$L$ or~$L'$ staying a vertex-disjoint cycle concentric with~$C^+$ contradicting minimality once more. So there cannot be such a cycle and hence~$\LLL'$ is of order~$\leq p$. Hence~$\LLL \setminus \LLL'$ is a linkage of order~$\geq d$. In particular then~$\LLL''\coloneqq \restr{\LLL}{H}$ is a linkage of order~$\geq d$ that is contained in~$H$ and, using the fact that~$\Gamma$ was Euler-embedded it is straight forward to verify that the ends of the linkage contradict the assumption of the \cref{thm:structure-thm-to-coastal-maps}.
\end{proof}

 Using \cref{lem:charting_a_cycle_to_the_cuff} and setting~$d' \coloneqq d+t(C^+)$ we can thus enlarge the cuff to~$C^+$ in the obvious way. In particular we may without loss of generality assume that there are~$C_0,\ldots,C_{\kappa-1}$ edge-disjoint cycles concentric to~$C_\Sigma$ such that each of them is vertex-disjoint and~$C_{\kappa-1}=C^+$ bounds the new enlarged cuff, i.e,~$C_{\kappa-1}$ is a cycle using all of the vertices and edges drawn on the cuff. 

 By applying \cref{lem:charting_a_cycle_to_the_cuff} inductively we get the following.

 \begin{corollary}\label{cor:charting_the_cuff_away_from_traps}
     There exist edge-disjoint concentric cycles~$C_0',\ldots,C_{2(p+d)+1}' \subset \Delta(C_{0})$ each of them being vertex-disjoint with~$\Delta(C_i)' \subseteq \Delta(C_j')$ for~$0 \leq j \leq i \leq 2(p+d)+1$ such that the cylinder bounded by~$C_0'$ and~$C_{2(p+d)+1}'$ does not contain any traps. Further there is no linkage~$\LLL$ in~$G[\Delta(C_{0}')]-C_{0}'$ of order~$p+d$ with its ends on~$C_0'$ in the same spirit as \cref{thm:structure-thm-to-coastal-maps} and \cref{lem:charting_a_cycle_to_the_cuff}.
 \end{corollary}
 \begin{proof} 
    Iteratively (starting with~$C^+$ and going out-wards) apply \cref{lem:charting_a_cycle_to_the_cuff} by enlarging~$H$ and adding all of~$G[\Delta(C^+)]$ to it and taking~$C^+$ as the new cuff until we find a sequence of~$2(p+d)+1$ concentric cycles~$C_i,\ldots,C_{i+2(p+d)+1}$ such that the cylinder bounded by~$C_i$ and~$C_{i+2(p+d)+1}$ in~$\Sigma$ does not contain any traps. Since there are~$\geq (p+1)(2(p+d)+1)$ concentric cycles~$C_j$ the sequence must exist using a simple counting argument---use the pigeon-hole principle. Note here that an even number of traps is bounded between any two consecutive cycles using the Euler-embedding and the induced pseudo Euler-embedding we get on that cylinder similar to the proof of \cref{lem:paths_in_pseudo_Euler_cut_disc_into_pseudo_Discs}.
    The iteration stops once we have found said cycles~$C_0',\ldots,C_{2(p+d)+1}'$ where there is no linkage~$\LLL$ of order~$t(C_{2(p+d)+1}') +d \leq p+d$ using inductive reasoning on every new arising cuff in the process (for the maximal linkage goes only up by the new amount of traps we add to the disc by switching to new cycles).

 \end{proof}

 Finally combining \cref{lem:charting_a_cycle_to_the_cuff} and \cref{cor:charting_the_cuff_away_from_traps}, enlarging the cuff~$C_\Sigma$ in every iteration step as mentioned above and setting~$\hat{d} \coloneqq p+d$ we get the following.

 \begin{observation}\label{obs:rearranging_the_vortex}
     There exist edge-disjoint cycles~$C_0',\ldots,C_{2\hat{d}+1}'$ each of them vertex-disjoint and concentric to the cuff~$C_\Sigma$ such that~$\Delta(C_{2\hat{d}+1}') \subset \ldots \subset \Delta(C_0')$ and~$C_{2\hat{d}+1}$ visits every vertex on~$\nu^{-1}(C)$. Further there exists no linkage~$\LLL$ contained in~$G[\Delta(C_{0})]-C_{0}$ with ends in~$C_{0}'$ of order~$\hat{d}$ as in \cref{thm:structure-thm-to-coastal-maps}. Finally the cylinder bounded by~$C_{2\hat{d}+1}'$ and~$C_0'$ contains no traps.
 \end{observation}

Using \cref{obs:rearranging_the_vortex} we can extend the cuff~$C_\Sigma$ to~$C_0'$ and redefine~$\Gamma$ and~$H$ accordingly. This is why we may without loss of generality, and for simplicity in order to not further bloat the rest of the proof assume the following henceforth.

\smallskip

Let~$C_\Sigma$ be the cuff,~$O=(v_1,\ldots,v_n)$ an ordering of the vertices on the cuff and let~$H$ and~$d$ as above. Then we may assume that there exist~$2d+2$ edge-disjoint cycles~$C_0,\ldots,C_{2d+1}$ each of them vertex-disjoint and concentric to~$C_\Sigma$ with~$\Delta(C_{2d+1}) \subset \ldots \subset \Delta(C_0)$ such that~$C_0$ and~$C_{2d+1}$ bound a cylinder in~$\Sigma$ and such that~$\Gamma[\Delta(C^0)]$ contains no traps. 
 Further we assume that~$O \subseteq V(C_{2d+1})$ where the vertices are visited by~$C_{2d+1}$ in the order prescribed by~$O$ and thus for simplicity one may think of~$C_{2d+1}$ as a tracing of the cuff~$C_{\Sigma}$---note again that we can (and thus assume that we do) simply enlarge the cuff~$C_\Sigma$ so that~$\nu^{-1}(C_\Sigma) = E(C_{2d+1}) \cup V(C_{2d+1})$ where they satisfy the assumptions of our setting from \cref{subsec:embedded-incidence-digraphs}. Thus we assume that~$v \in V(C_{2d+1})$ with incident edges~$e_1,e_2 \in E(C_{2d+1})$ are either visited in the order~$(e_1,v,e_2)$ or~$(e_2,v,e_2)$ by the cuff~$C_\Sigma$. We will refer to the cylinder bounded by~$C_0$ and~$C_{2d+1}$ as~$\Phi \subset \Sigma$ which we have and will use extensively; assume an orientation of the cylinder~$\Phi$ to be given (we can derive it from the orientation of~$C_\Sigma$ for example). Finally for every~$0 \leq i \leq 2d$ we assume that for every cycle~$C^*$ such that~$\Delta(C_i) \subseteq \Delta(C^*) \subseteq \Delta(C_{i+1})$ it holds that either~$C^*$ is not edge-disjoint from at least one of both cycles~$C_i,C_{i+1}$ or~$C^* \in \{C_i,C_{i+1}\}$, i.e., the (sequence of) cycles are \emph{tight}.

 This yields the following immediate observations.

 \begin{observation}\label{obs:Gamma_is_Euler_between_cycles}
     Let~$0\leq i <j \leq 2d+1$. Let~$\Phi_i^j \coloneqq \Delta(C_{i}) \setminus \Delta(C_{j}) \cup \nu(C_j)$, i.e., the `cylinder' bounded by~$C_{i}$ and~$C_{j}$ including the boundary cycles. Then~$\Gamma[\Phi_i^j]$ is Eulerian and Euler-embedded in~$\Phi_i^j$.
 \end{observation}
 \begin{proof}
     This follows at once using the fact that~$\Gamma[\Phi]$ is Eulerian and Euler-embedded together with the fact that every vertex on the boundary of the cylinders~$\Phi_i^j$ is of degree at least two in~$\Gamma[\Phi_i^j]$ for~$C_i,C_j \subset \Gamma[\Phi_i^j]$ and thus they cannot be traps.
 \end{proof}

 \begin{observation}\label{obs:cycles_away_from_Ci_not_concentric}
     Let~$e \in E(\Gamma[\Phi_i^{i+1}])\setminus E(C_i\cup C_{i+1})$ be an edge that is neither part of~$C_i$ nor~$C_{i+1}$ for some~$0 \leq i \leq 2d$. Then~$e \in C_e \subset \Gamma[\Phi_i^{i+1}]-(C_i\cup C_{i+1})$ for some cycle~$C_e$ that is not concentric to~$C_\Sigma$ (and thus the the cuff~$C_\Sigma$).
 \end{observation}
 \begin{proof}
     Assume the contrary. Then~$e$ can be extended in~$\Gamma[\Phi_i^{i+1}]-(C_i\cup C_{i+1})$---note that this is technically not a cylinder in the case that~$V(C_i)\cap V(C_{i+1}) \neq \emptyset$ but this does not change the arguments---using the Eulerianness given by \cref{obs:Gamma_is_Euler_between_cycles}. In particular it can be closed to a cycle~$C_e$. If the cycle were concentric with~$C_\Sigma$ we get a contradiction to our assumptions on the minimality---we sometimes call this the tightness---of the cycles. 
 \end{proof}

We will use the above later.
 \smallskip

 In a nutshell we will next---after some more massaging of the graph---extend~$H$ by the cylinder~$\Gamma[\Phi]$ and implicitly define~$C_0$ to be the new cuff (or rather the model for the island-zone), in a sense pushing all of the drawing of~$\Gamma[\Phi]$ into the cuff~$C_{\Sigma}$ and enlarging the cuff to~$C_0$. This will help us to get a grip on the paths in~$H$ between vertices in the cuff, that is, it allows us to \emph{balance out} the `thickness' of the linkages inside~$H$ having their ends on the cuff.  
 
\smallskip

With the above simplifications, definitions and discussions in mind we continue with the proof of \cref{thm:structure-thm-to-coastal-maps}. It is easily seen that the claim is trivial if~$n \leq 2$, so we may assume~$n>2$.

\medskip

\paragraph{Laminar cuts~$(A_i',B_i')$ of the island and~$\mu'$.} For~$1 \leq i < n-1$ let~$(A_i', B_i')$ be a cut of~$H$ of minimal order such that~$v_1, \dots, v_i \in A_i'$ and~$v_{i+1}, \dots, v_n \in B_i'$ and, subject to that,~$A_i'$ is minimal.
Finally, let~$(A_{n-1}', B_{n-1}')$ be the cut with~$B_{n-1}' \coloneqq \{ v_n\}$. We also define~$A_n' = B_0' = V(H)$.

\begin{lemma}\label{lem:vortex-laminar-cuts}
    For all~$1 \leq i < j \leq n$, it holds~$A_i' \subseteq A_j'$ and~$B_i' \supseteq B_j'$, i.e., the cuts are laminar. 
\end{lemma}
\begin{proof}
    We show that for all~$1 \leq i < n$,~$A_i' \subseteq A_{i+1}'$ and~$B_i' \supseteq B_{i+1}'$, from which the claim follows. 

    So let~$1 \leq i < n$. By \cref{lem:submodularity},~$|\delta(A_i')| + |\delta(A_{i+1}')| \geq |\delta(A_i' \cap A_{i+1}')| + |\delta(A_i' \cup A_{i+1}')|$. 
    Let~$k_1 \coloneqq |\delta(A_i')|$ and~$k_2 \coloneqq |\delta(A_{i+1}')|$. If~$|\delta(A_i' \cap A_{i+1}')| < k_1$ then this contradicts our choice of~$A_i'$ as a minimum order cut. So~$|\delta(A_i' \cap A_{i+1}')| \geq k_1$. But~$|\delta(A_i' \cap A_{i+1}')|$ cannot be strictly larger than~$k_1$ as otherwise~$|\delta(A_i' \cup A_{i+1}')| < k_2$ contradicting the choice of~$A_{i+1}'$ as a minimum order cut. So~$|\delta(A_i' \cap A_{i+1}')| = k_1$ and, by minimality of~$A_i'$,~$A_i' = A_i' \cap A_{i+1}'$ and thus~$A_i' \subseteq A_{i+1}'$. As~$B_i' = V(H) \setminus A_i'$ and~$B_{i+1}' = V(H) \setminus A_{i+1}'$ the claim follows.
\end{proof}
For every~$1 \leq i \leq n$ and~$1\leq j <n$, let
\begin{align*}
    \mu'(v_i) &\coloneqq H[A_i' \cap B_{i-1}'] \cup \delta(A_i') \cup \delta(B_{i-1}'), \text{ and}\\
    S'_j &\coloneqq \mu'(v_j) \cap \mu'(v_{j+1}),
\end{align*} 
then~$S'_j \subseteq \delta(A_j') \subset E(H)$ holds by definition. Note that by construction~$S'_j$ is a set of cut-edges in~$H$, that is, there is no path between
$\{v_1, \dots, v_i\}$ and~$\{ v_{i+1}, \dots, v_n\}$ in~$H - S'_j$.
Furthermore,~$\rho \coloneqq |S'_j| \leq 2d$ by construction again. As noted already, every~$e \in S'_j$ must have one end in~$A_j'$ and the other in~$B_j'$. We say that the end of~$e$ in~$A_j'$ is the \emph{left} end of~$e$ and the end in~$B_j'$ is the \emph{right} end of~$e$. This will be needed later.

\medskip

\paragraph{Cutting the cylinder via a cut-line~$F$.}
Now take a cut-line~$F$ (see \cref{def:cut-line}) in~$\Sigma$ with its starting point in an edge~$e_0 \in E(C_0)$ and its end-point in an edge~$e=(v_j,v_{j+1}) \in E(C)$ for some~$1 \leq j \leq n$ (again using cyclic addition in the indices). Subject to this, choose~$F$ so that it minimises the number of edges crossed by~$F$ (in particular the choice of the starting and end-points of~$F$ are chosen this way); let one end again be~$e=(v_j,v_{j+1}) \in E(C)$ without loss of generality where the index~$j \in \{1,\ldots,n\}$ will be fixed from here on. 

Observe that the curve~$F$ \emph{cuts} the cylinder (which is synonymous with annulus here)~$\Phi \subset \Sigma$ bounded by~$C_{2d+1}$ and~$C_0$ as it has one end in an edge of each of its two boundary cycles respectively. 

The edges crossed by~$F$---including its ends which are seen as being crosses by~$F$---and the set of separating edges~$S'_j$ in~$H$ forms a complete edge cut which cuts any of the cycles in~$\Gamma[\Phi]\cup H$ concentric with the cycle bounding the cuff~$C_\Sigma$. To make this more precise define
\begin{align*}
    Z &\coloneqq S'_j, \\
    F' &\coloneqq \{e \in E(\Gamma[\Phi]) \mid \nu(e) \text{ is intersected by } F\}.
\end{align*} 
Then the following follows as discussed above.

\begin{observation}
    The set~$Z \cup F'$ forms an edge-cut in~$\Gamma[\Phi]\cup H$ such that for every cycle~$C^\ast \subset \Gamma[\Phi]\cup H$ that is concentric with~$C_\Sigma$ it holds~$E(C^\ast) \cap (Z\cup F') \neq \emptyset$.
    \label{obs:Z_and_F_cut_the_cylinder}
\end{observation}

 Intuitively, as a remaining step to obtain the desired weak-coastal map locally at~$C_\Sigma$ (rather at the new cuff~$C_0$) for the (to-be-)island~$H \cup \Gamma[\Phi]$ (rather a subgraph thereof as we will `tighten' and re-choose~$C_0$ later), we need to provide a linear decomposition of our island and its surrounding cycles---this will be~$\mu$. To achieve this, we would like to delete the bounded number of edges in the cut~$Z \cup F'$ which in turn would allow us to start a linear decomposition at the edges `on the right' of the curve~$F$ (after choosing an orientation of the cylinder~$\Phi$) and from there on take a sequence of laminar cuts---similar to \cref{lem:vortex-laminar-cuts}---going round the cylinder between~$C_{2d+1}$ and~$C_0$ until we reach the edges on the other side of~$F$. Then, using these laminar cuts we can easily chop up the new cuff~$C_0$ into coastal lines and define a map~$\mu$ satisfying the needed constraints for weak-coastal maps as given in \cref{def:weak_coastal_map}. Essentially this is what will happen in the remainder of the proof. The problem as opposed to \cite[Lemma 15]{Diestel2012} is that, as we want to maintain that our graph is Eulerian, we cannot simply delete edges or vertices. Instead, we have to simulate the deletion of edges by enforcing certain linkages in the island.

To this extent we define an ordering on~$F'$ via~$F' = \{ e_0, \dots, e_h\}$, in the order in which they appear on~$F$ when traversing~$F$ from its starting point on~$e_0 \in E(C_0)$ to its end~$e=(v_j,v_{j+1})$ on~$C_{2d+1}$. That is,~$e_0$ is the unique edge on~$C_0$ contained in~$F$ (note that by minimality of~$F$ it contains only one edge of~$C_0$).

\medskip

\paragraph{Defining the right~$F_1$ and left~$F_2$.}
We start by subdividing  each edge~$e_i \in F'$ twice, i.e., replace 
\[
  e_i = (u_i, v_i) \quad\text{by} \quad(u_i, e_i^1, l_i, e_i^2, r_i, e_i^3, v_i).,
\]
for every~$0 \leq i \leq h$.
 Then, using the orientation of the cylinder~$\Phi$, we can locally define a \emph{right} and \emph{left} side of~$F$, i.e., each edge~$e_i \in F'$ is crossed by~$F$ which 
allows us to distinguish between the end of~$e_i$ to the left of~$F$ and the other end to the right of~$F$ (when traversing~$F$ starting at its end on~$C_0$) and thus, presuming that~$F$ passes through~$\nu(e_i^2)$ after the sub-division, exactly one of~$\{e_i^1,e_i^3\}$ is on the right side of~$F$ and one is on the left side of~$F$ for every~$0 \leq i \leq h$. Note that it may well be that~$e_1^1$ is on the left side of~$F$ and~$e_2^1$ is on the right side of~$F$. Using this we define~$F_1$ and~$F_2$ as follows. 

\begin{definition}[The edge-sets $F_1$ and~$F_2$]\label{def:F1_and_F2}
    Let~$F_1 \subseteq \{ e_i^3, e_i^1 \mid 0 \leq i \leq h\}$ be the set of edges on the right side of~$F$ and let~$F_2$ be the respective set of edges on the left side of~$F$.
\end{definition}

\paragraph{Defining the final cut~$S$ of the cylinder.}
Having defined~$F_1$ and~$F_2$ by subdividing the edges in~$F'$ we will continue by sub-dividing the edges of~$Z$ for smoother reasoning in what is to follow. To this extent recall the definition of~$Z=S_j'$ from the paragraph `\emph{laminar cuts of the island and~$\mu'$}'. Let~$e_{Z, 1},  \dots, e_{Z, \rho} \in E(Z)$ be the edges of~$Z$. As discussed in the same paragraph, each edge~$e_{Z, i}$ has a \emph{left} end and a \emph{right} end in~$A_j'$ and~$B_j'$ respectively for~$1 \leq i \leq \rho$. We subdivide each edge~$e_{Z, i} = (u^Z_i, v^Z_i)$ three times, that is, replace the edge by a~$4$-path 
\[
   (u^Z_i, e_i^{Z,l}, l^Z_i,e_i^{Z,m,l},m^Z_i,e_i^{Z,m,r},r^Z_i, e_i^{Z,r}, v^Z_i).
\] 

Then we define~$S,S_F$ and~$S_Z$ as follows via
\begin{align*}
    S &\coloneqq \bigcup S_F \cup S_Z, \text{ where}\\
    S_F &\coloneqq \{e_i^2 \mid 0 \leq i \leq 2d+1\}, \text{ and,}\\
    S_Z &\coloneqq \bigcup\{e_i^{Z,l},e_i^{Z,m,l},e_i^{Z,m,r},e_i^{Z,r}, \mid 1 \leq i \leq \rho\}.
\end{align*}
Again~$S$ cuts the cylinder in the spirit of \cref{obs:Z_and_F_cut_the_cylinder}, for we only subdivided the edges in~$F'$ and~$Z$ and for each edge in~$F'$ and~$Z$ we added (at least) the middle edge(s) of the subdivision to~$S$, namely~$S_F$ and~$\{e_i^{Z,m,l},e_i^{Z,m,r} \mid 1 \leq i \leq \rho\}$.

\paragraph{Taming the linkage using tightness.}
Let~$H' \coloneqq (\Gamma[\Phi] \cup H) - S$, where~$\Gamma$ and~$H$ were implicitly altered to contain all the aforementioned sub-divisions. Further dissolve the degree zero vertices~$l_i^Z,m_i^Z,r_i^Z$ for every~$1 \leq i \leq \rho$ as they are of no direct importance.

Then~$H'$ contains (at least)~$2d+2$ pairwise edge-disjoint directed paths between~$F_1$ and~$F_2$, namely given by the~$2d+2$ edge-disjoint concentric cycles~$C_0,\ldots,C_{2d+1}$ using \cref{obs:Z_and_F_cut_the_cylinder}. Note that those paths together with~$S_F$ form cycles concentric with~$C_{2d+1}$ (and~$C_\Sigma$). Recall that this simply says that there are such paths which are pairwise edge-disjoint and each have one end (which is an edge) in~$F_1$ and the other end (which also is an edge) in~$F_2$. In particular, this does not mean that these paths, being directed, all go from~$F_1$ to~$F_2$. There can be paths that start with an edge in~$F_2$ and end in~$F_1$. This is summarised as follows.

\begin{observation}\label{obs:F1_F2_at_least_2d}
    There exists an~$F_1{-}F_2$-linkage in~$H'$ of order~$2d+2$.
\end{observation}

Let~$e_1 \in \{e_0^1, e_0^3\} \cap F_1$ be the edge in~$F_1$ on (the subdivision of)~$C_0$ and~$e_2 \in \{e_0^1, e_0^3 \} \cap F_2$  be the edge in~$F_2$ on (the subdivision of)~$C_0$.

\begin{lemma}\label{lem:vortex-coastal-map:3}
    In every maximal~$F_1{-}F_2$-linkage~$\LLL$ in~$H'$ the path~$P\in \LLL$ that contains~$e_1$ is completely contained in~$H' \cap \Gamma$. 
\end{lemma}
\begin{proof}
    For otherwise~$P$ contains a vertex on the boundary of~$H$. Let~$v$ be the first vertex on~$P$ after~$e_1$ on the boundary. Then the sub-path~$P'$ of~$P$ between~$e_1$ and~$v$ lies entirely in the Euler-embedded part~$H' \cap \Sigma$ and therefore~$E(P') \cup \mu'(v)$ separates~$F_1$ from~$F_2$. Therefore, all paths in~$\LLL$ from~$F_1 - e_1$ to~$F_2$ must cross~$\mu'(v)$, a contradiction as~$|\mu'(v)| \leq 2d$ and~$\Abs{\LLL \setminus P} \geq 2d+1$ by \cref{obs:F1_F2_at_least_2d}.
\end{proof}

It turns out we can derive something stronger for the~$F_1{-}F_2$-linkages. To this extent note that~$\Gamma^* \coloneqq \Gamma[\Phi] - S_F$ is pseudo-Eulerian where~$S_F$ is its demand graph, i.e.,~$\Gamma^* + S_F$ is Eulerian (as given by \cref{obs:Gamma_is_Euler_between_cycles}. In particular using the Euler-embedding of~$\Gamma$ in~$\Phi$ and the fact that there are no traps of~$\Gamma$ in~$\Phi$ we derive the following.

\begin{observation}\label{obs:Gamma*_is_pseudo_Euler}
    The graph~$\Gamma^*$ is pseudo-Eulerian and it is pseudo Euler-embedded in~$\Phi$ and~$\Delta(C_0)$
\end{observation}

We use \cref{obs:Gamma*_is_pseudo_Euler} to prove the following.

\begin{lemma}\label{lem:F1_F2_fully_linked}
    It holds~$\Abs{F_1} = 2d+2$ and there exists an~${F_1}{-}F_2$-linkage of order~$\Abs{F_1}$ in~$H'$.
\end{lemma}
\begin{proof}
    Let~$\LLL$ be a maximal~$F_1{-}F_2$-linkage in~$\Gamma^* \coloneqq \Gamma[\Phi] - S_F$ such that~$\LLL$ contains the~$2d+2$ paths coming from the edge-disjoint cycles~$C_0,\ldots,C_{2d+1}$. Assume that there is~$f_1 \in F_1$ that is not matched by the linkage; as mentioned above~$\Gamma^* + S_F$ is Eulerian. Then we can extend said edge via a path~$L$ (possibly going backwards) in~$\Gamma^*$ that is edge-disjoint from~$\LLL$ until it reaches~$f_1' \in F_1$ using the Eulerianness of the graph~$\Gamma[\Phi]$. Note that it cannot reach~$F_2$ prior to passing through~$S_F$ for else~$\LLL$ was not maximal. Let~$\Phi_i^j \coloneqq \Delta(C_{i}) \setminus \Delta(C_{j}) \cup \nu(C_i)$ for some~$0 \leq i < j \leq 2d+1$ be the cylinder bounded by the two cycles. Then~$G[\Phi_i^{j}]$ is Eulerian using \cref{obs:Gamma_is_Euler_between_cycles}. 

    \begin{claim}\label{lem:F1_F2_fully_linked_claim1}
        Let~$0\leq i \leq 2d$ and let~$C_L \subset G[\Phi_i^{i+1}]$ be a cycle that is edge-disjoint from~$C_i\cup C_{i+1}$. Then~$\Abs{V(C_L) \cap V(C_i)} \leq 1$. 
    \end{claim}
    \begin{ClaimProof}
        Assume the contrary and let~$v_1,v_2 \in V(C_L)\cap V(C_i)$ be two distinct vertices. Further assume that~$C_L$ is maximal, i.e., it is an Euler-cycle of the component~$K \subset G[\Phi_i^{i+1}]\setminus (C_i \cup C_{i+1})$ containing~$L$ which by \cref{obs:Gamma_is_Euler_between_cycles} is indeed Eulerian.
        
        First we know that~$C_L$ is not concentric---and contains no concentric sub-cycle---with~$C_\Sigma$ using \cref{obs:cycles_away_from_Ci_not_concentric}; in particular it has winding-number zero inside the cylinder~$\Phi_i^{i+1}$, i.e., it is null-homotopic. Let~$P_1^2$ and~$P_2^1$ be the sub-paths given by~$v_1C_Lv_2$ and~$v_2C_Lv_1$ respectively; in particular~$P_1^2P_2^1$ form the cycle~$C_L$. Since~$\Gamma[\Phi_i^{i+1}]$ is Euler-embedded by \cref{obs:Gamma_is_Euler_between_cycles}, we conclude that either~$P_1^2$ or~$P_2^1$ contains a vertex-disjoint sub-path~$P$ with ends~$v_1,v_2$---reroute at self-intersection points to make it vertex-disjoint---such that we can reroute~$C_i$ along~$P$ to get a new cycle~$C_i'$ that is itself vertex-disjoint, concentric to~$C_\Sigma$ and satisfies~$\Delta(C_i') \subset \Delta(C_i)$ by construction; contradiction to our assumption that the cycle sequence~$C_0,\ldots,C_{2d+1}$ was tight. To see this assume without loss of generality that~$v_1$ is visited prior to~$v_2$ by~$C_i$. There are two cycles to consider now:~$C_1^2$ is the cycle obtained from gluing~$v_1C_iv_2$ and~$P_2^1$ and~$C_2^1$ is the cycle contained from gluing~$v_2C_iv_1$ and~$P_1^2$ in the obvious ways. If neither of the two cycles is concentric to~$C_\Sigma$ then both are contractible in~$\Phi_{i}^{i+1}$ and so is the cycle arising when gluing~$C_1^2$ and~$C_2^1$ at~$v_1$. But gluing~$C_1$ and~$C_L$ at~$v_1$ results in a cycle that has winding number greater than zero for~$C_L$ is contractible in~$\Phi_i^{i+1}$ by assumption but~$C_i$ winds around the cuff. Since both constructions result in the same curve in~$\Phi_i^{i+1}$ the latter must be contractible since the first is by assumption (this follows from the fact that the fundamental group of closed curves is a group where the `gluing' procedure is formally a composition of curves); contradiction. 

    \end{ClaimProof}
    \cref{lem:F1_F2_fully_linked_claim1} is crucial to the following proof.
    
    \begin{claim}\label{lem:F1_F2_fully_linked_claim2}
        It holds~$\Abs{F_1} = 2d+2$.
    \end{claim}
     \begin{ClaimProof}
        We will prove this by providing a way to construct a curve~$F$ for which~$\Abs{S_F} = 2d+2$ inductively.

        To this extent let~$\Phi_i^{i+1} \coloneqq \big(\Delta(C_{i}) \setminus \Delta(C_{i+1})\big) \cup \nu(C_{i+1})$ for every~$0 \leq i \leq 2d$ and recall that~$\Gamma[\Phi_{i}^{i+1}]$ is Eulerian for every~$0 \leq i \leq 2d$ by \cref{obs:Gamma_is_Euler_between_cycles}. Let~$G_i \coloneqq \Gamma[\Phi_i^{i+1}]\setminus \big(V(C_i) \cup V(C_{i+1})\big)$ for every~$0 \leq i \leq 2d$.
        
    \begin{itemize}
        \item \textbf{Induction start:} Let~$F_0$ be a curve in~$\Phi_0^1$ such that~$F_0 \cap \nu(G_0) = \emptyset$. Such a curve must exist for else there is a cycle in~$\Gamma[\Phi_0^1]$ disjoint from~$C_0$ and~$C_1$ as a contradiction to \cref{obs:cycles_away_from_Ci_not_concentric}. This is fairly obvious using the  Euler-embedding of~$G_i$, but a more involved proof uses the same argument as the proof of the induction step (simply assume that~$F_{-1}$ was the curve consisting of a single point and use the induction step). 
    
        \item \textbf{Hypothesis:} Let~$0\leq i \leq 2d-1$ and let~$F_i$ be a curve in~$\Phi^i \coloneqq \big(\Delta(C_{0}) \setminus \Delta(C_{i+1})\big) \cup \nu(C_{i+1})$ such that~$F_i$ intersects~$\nu(\Gamma[\Phi^i])$ exactly in the cycles~$C_0,\ldots,C_i$ and it intersects every cycle exactly once and in an edge. 
    
        \item \textbf{Induction step:} We extend~$F_i$ to a curve~$F_{i+1}$ by extending it in~$\Phi_i^{i+1}$ such that it satisfies the hypothesis above (note that~$\Phi^i \cup \Phi_i^{i+1} = \Phi^{i+1}$). If~$G_{i+1} = \emptyset$ this is trivially possible, thus assume that~$G_{i+1} \neq \emptyset$. Let~$e^i \in E(C_i)$ be the endpoint of~$F_i$. 
            Let~$F^*$ be some arbitrary curve---for example a straight-line---from~$e^i$ to some edge in~$E(C_{i+1})$ in~$\Phi_i^{i+1}$ that id internally disjoint from the cycles~$C_i,C_{i+1}$. Let~$e^* \in E(G_{i+1})$ be the first edge intersected by~$F^*$; note that if there is no such edge we can extend~$F_i$ via~$F^*$ satisfying the Hypothesis. Let~$C^* \subseteq G_{i+1}$ be the maximal component containing~$e^*$, then~$C^*$ is Eulerian (since~$G_{i+1}$ is Eulerian by \cref{obs:Gamma_is_Euler_between_cycles}) and thus we can assume~$C^*$ to be an Eulerian cycle of the component. 
    
            By \cref{lem:F1_F2_fully_linked_claim1} we know that~$\Abs{V(C^*) \cap V(C_i)} \leq 1$. We define~$F_{i+1}$ by taking the curve~$F^*$ until shortly before it intersects~$e^*$; until here the curve is internally disjoint from~$\nu(G_{i+1})$ and thus a good candidate. We refer to the outline of~$C^*$ as the vertex-disjoint cycle we get by drawing a curve following in parallel to~$C^*$ that keeps~$C^*$ on the same side (after defining left and right) probably rerouting at self-intersections of~$C^*$; recall \cref{def:outline_of_cycle}. Then we extend~$F^*$ via two different curves~$F^*_1$ and~$F^*_2$ where~$F^*_1$ follows~$C^{\star}$---rerouting at self-intersections using the Euler-embedding of~$C^*$---in the direction of~$e$ keeping the outline on its left side say, until it is \emph{trapped}---that is it has to intersect either itself or some edge to continue in parallel to the~$C^\star$---while~$F_2$ follows the outline~$C^{\star}$ in the other direction keeping~$C^\star$ on its right side until it is trapped. Note that left and right are well-defined for~$\Phi_i^{i+1}$ can unambiguously be oriented. In a first instance note that both curves will be trapped for either~$V(C^*)\cap V(C_i) \neq \emptyset$ or~$V(C^*)\cap V(C_{i+1}) \neq \emptyset$ since~$\Gamma$ and~$\Gamma[\Phi_i^{i+1}]$ are connected. 
    
            There is several cases left to consider. First~$F^*_1$ can be trapped because either following the outline~$C^\star$ of the cycle~$C^*$ closes a cycle and thus the curve~$F^*_1$ needs to intersect itself to continue the cycle, or it ends because the outline~$C^\star$ uses a vertex in~$V(C_i) \cup V(C_{i+1})$; note that~$C^*$ is vertex-disjoint from all other components of~$G_{i+1}$ by choice. 
    
            The first case---$F_1^*$ needs to self-intersect---means that the outline~$C^\star$ of~$C^*$ is a cycle in~$G_{i+1}$ concentric to~$C_\Sigma$ since every edge in~$C^\star$ must be incident to some common face containing~$F^*$, but the outline is also not disjoint from the boundary of the cylinder~$\Phi_i^{i+1}$ as discussed above. This is however in contradiction to \cref{obs:cycles_away_from_Ci_not_concentric}. Note that if it were not concentric to~$C_\Sigma$ then following the outline~$C^\star$ we need to pass through a vertex of~$V(C_i)\cup V(C_{i+1})$---again since~$C^*$ is not disjoint from both cycles---changing faces of the embedding and thus the curve would be trapped earlier and~$F_1^*$ would not self-intersect. 
            
            Thus both~$F_1^*$ and~$F_2^*$ are trapped because the outline~$C^\star$ of~$C^*$ switches faces in the drawing, in particular they end in a face adjacent to a vertex in~$V(C_i)$ or~$V(C_{i+1})$. If~$F_1^*$(or~$F_2^*$) ends in a face adjacent to~$V(C_{i+1})$ then we can extend~$F_1^*$ (or~$F_2^*$) in said face to an edge of~$C_{i+1}$. Then extending~$F_i$ along this curve satisfies the hypothesis.
    
            Thus both~$F_1^*$ and~$F_2^*$ are trapped in faces that are adjacent to~$V(C_i)$ but not~$V(C_{i+1})$ and in particular then~$V(C_i) \cap V(C^*) \neq \emptyset$. Thus using the fact that~$\Abs{V(C_i) \cap V(C^*)} \leq 1$ we deduce that both curves are trapped in faces adjacent to~$v \in V(C_i) \cap V(C^*)$ being the single such vertex. Again using the Euler-embedding (rather the fact that the embedding is planar in the cylinder) they cannot both be trapped in the same face: otherwise the edge~$e^i$ and~$e^*$ are both part of the face that both ends are trapped in. But then either~$C^\star$ must intersect~$C_{i}$ in two different vertices---a contradiction---or~$C^\star$ is contained in that same face and thus it is either again a cycle concentric to~$C_\Sigma$---a contradiction---or it is a cycle contractible in~$\Phi_i^{i+1}$ and disjoint from~$C_{i+1}$ in which case both curves are trapped in a face adjacent to~$V(C_{i+1})$ a case we already discussed.
    
            Finally then the outline~$C^\star$ of~$C^*$ is a cycle that must contain~$v$ and since the cycle does not intersect~$V(C_i)$ in any other vertex and since~$v$ is of degree four, then~$v$ is adjacent to exactly two faces in the embedding of~$\Gamma[\Phi_i^{i+1}]$ in the cylinder.  Since neither of~$F_1^*$ and~$F_2^*$ intersected~$C^*$ by construction, they are both contained in the same face adjacent to~$v$ and the outline~$C^\star$; a contradiction to the above. 
    
            Altogether this proves that we can extend~$F_i$ via one of~$F_1^*$ or~$F_2^*$ in the desired way such that~$F_{i+1}$ satisfies the hypothesis.
    \end{itemize}

    The claim follows by induction.
\end{ClaimProof}

Finally \cref{lem:F1_F2_fully_linked_claim2} immediately implies the lemma recalling that the edge-disjoint cycles~$C_0,\ldots,C_{2d+1}$ provide the respective linkage.

\end{proof}
\begin{remark}
    Using \cref{lem:vortex-coastal-map:3}, \cref{lem:F1_F2_fully_linked} and the fact that~$\Gamma$ is pseudo Euler-embedded, the path~$P$ in~$\LLL$ starting in~$e_1$ cannot be crossed by any other path in~$\LLL$ and thus~$P$ must end in~$e_2$. Note that this observation really needs the fully-linkedness guaranteed by \cref{lem:F1_F2_fully_linked} for else~$\Abs{F_2}$ could be larger than~$2d+2$ and thus the path starting in~$e_1$ could a priori end in another edge of~$F_2$. Also note that said path must not necessarily (and most likely will not) agree with the path between~$e_1$ and~$e_2$ given by~$C_0$.
\end{remark} 

Using \cref{lem:vortex-coastal-map:3} and the previous remark, the following is well-defined.

\begin{definition}[$P(\LLL)$ and~$\Delta(\LLL)$]
    Let~$\LLL$ be an~$F_1{-}F_2$-linkage of order~$\Abs{F_1}$ in~$H'$. We denote by~$P(\LLL)$ the path starting in~$e_1$ and ending in~$e_2$ as given by \cref{lem:vortex-coastal-map:3}; we call it the \emph{circumference of~$\LLL$}. Together with the edge~$e_0^2$ in~$F'$ incident to~$e_1$ and~$e_2$ the path~$P(\LLL)$ forms a cycle bounding a disc~$\Delta(\LLL)$ in~$\hat{\Sigma}$ containing~$C_\Sigma$.
    \label{def:P(L)_and_Delta(L)}
\end{definition}

Note that~$P(\LLL)$ may be (and likely is) different from~$C_0$ so in a sense we re-tighten the outer-cycle, that is we choose a circumference for the cuff as tight as possible keeping a~$2d+2$-linkage that surrounds it.

\paragraph{Constructing the cut-open island~$I(\LLL)$ by tightening the circumference and balancing~$H'$.} 
Let~$\LLL$ be an~$F_1{-}F_2$-linkage in~$H'$ of order~$2d+2$ which exists by \cref{lem:F1_F2_fully_linked}. We continue by defining the final graph needed for the definition of~$\mu$, namely the \emph{cut-open island}~$I(\LLL)$.
$$I(\LLL) \coloneqq \big(\Gamma[\Delta(\LLL)] \cup H \big) - S,$$
and dissolve the degree zero vertices~$l^Z_i,m^Z_i,r^Z_i$ for every~$1 \leq i \leq \rho$.

By construction,~$\LLL \subseteq I(\LLL)$. Without loss of generality assume that~$\LLL$ is an~$F_1{-}F_2$-linkage in~$H'$ minimising~$I(\LLL)$, i.e., minimising~$\Abs{E(I(\LLL))} + \Abs{V(I(\LLL))}$; this is a `tightened' and `balanced' version of the previous~$H'$, an intuition we will make rigorous in the next step. Let~$f_1, \dots, f_l \in E(P(\LLL))$ be the edges of the circumference~$P(\LLL)$, i.e.,~$P(\LLL)=(f_1,\ldots,f_l)$. For convenience we set~$f_0 \coloneqq e_0^2$.

The following lemma now makes precise what we mean by saying that~$I(\LLL)$ is tight and balanced.

\begin{definition}\label{def:cut_seperating_F1_F2}
    Let~$G$ be a graph and let~$F_1,F_2 \subseteq E(G)$ with~$F_1 \cap F_2 = \emptyset$. Let~$A \subset V(G)$ and~$B=\bar{A}$. We say that the cut~$(A,B)$ separates~$F_1$ from~$F_2$ if~$F_1 \subseteq G[[A]]$ and~$F_2 \subseteq G[[B]]$ and either~$F_1 \cap \delta(A) = \emptyset$ or~$F_2 \cap \delta(A) = \emptyset$ (or both).
\end{definition}

\begin{lemma} \label{lem:minimal_cuts_in_I}
    Let~$\LLL$ be an~$F_1{-}F_2$-linkage of order~$2d+2$ chosen as above and let~$P(\LLL) = (f_1,\ldots,f_l)$ be its circumference. Let~$I(\LLL)$ be defined as above and let~$1 \leq i \leq l$ be arbitrary. Then, any minimal cut in~$I(\LLL) - f_i$ separating~${F_1}$ from~${F_2}$ is of order~$2d +1$ and all the cut-edges come from different paths in~$\LLL$.
\end{lemma}
\begin{proof}
    Let~$A \subset V(I(\LLL))$ induce a minimal cut in~$I(\LLL)-f_i$ separating~$F_1$ from~$F_2$; we refer to the cut as~$(A,B)$.  Then~$F_1 \subseteq G[[A]]$ and~$F_2 \subseteq G[[B]]$ (possibly deleting~$f_1$ or~$f_l$ in the case that~$i=1$ or~$i=l$). Since~$\LLL\setminus \{P(\LLL)\}$ is an~$F_1{-}F_2$-linkage of order~$2d+1$ we deduce that~$\Abs{\delta(A)} \geq 2d+1$ and each path in~$\LLL \setminus \{P(\LLL)\}$ must have at least one edge in~$A\cap B$. 

    But clearly~$\Abs{\delta(A)} \leq 2d+1$ for otherwise the well-known theorem of Menger for edge-disjoint paths and cuts implies the existence of an~$F_1{-}F_2$-linkage~$\LLL'$ of order~$2d+2$ in~$I(\LLL)-f_i$ a contradiction to the minimality of~$I(\LLL)$.

\end{proof}

We continue with defining laminar cuts for~$I(\LLL)$ which in turn will yield the desired linear decomposition of~$I(\LLL)$.
\medskip
\paragraph{Laminar cuts in~$I(\LLL)$.}
We finally start with constructing a linear decomposition of~$I(\LLL)$ which we will in turn use to define the weak coastal map~$\mu$. To this extent, for each~$1 \leq i \leq l$ choose a cut~$(A_i, B_i)$ of~$I(\LLL)$ with
$f_i \in \delta(A_i) = \delta(B_i)$ separating~${F_1}$ and~${F_2}$ in~$I(\LLL)$ and minimising~$B_i$ where~$B_i = \bar{A_i}$.

\begin{lemma}\label{lem:laminar-cuts-new-vortex}
It holds~$B_j \subsetneq B_i$ and~$\Abs{\delta(A_i)} = 2d+2$ for all~$1 \leq i < j \leq l$.
\end{lemma}
\begin{proof}
  By \cref{lem:minimal_cuts_in_I}~$\delta(A_i)$ must contain exactly one edge from each path in~$\LLL-P(\LLL)$, and thus, since~$f_i \in \delta(A_i)\cap E(P(\LLL))$ we deduce that each path in~$\LLL$ has exactly one edge in the cut; this proves~$\Abs{\delta(A_i)}=2d+2$. In particular this implies~$f_i \nin \delta(A_j)$ for any~$1 \leq j\neq i \leq l$ and therefore it suffices to show
  that~$B_j \subseteq B_i$. Towards this aim, observe that~$(A_i \cap A_j, B_i \cup
  B_j)$ and~$(A_i \cup A_j, B_i \cap  B_j)$ are also cuts separating~$F_1$
  and~$F_2$. The claim then follows from \cref{lem:submodularity} and the minimality of~$B_i$ and~$B_j$. 
\end{proof}
\begin{remark}
    The laminar family of cuts~$\big((A_i,B_i)\big)_{1 \leq i \leq l}$ yields the desired balanced linear decomposition where each cut is of order~$2d+2$.
\end{remark}

A direct consequence of the proof of \cref{lem:laminar-cuts-new-vortex} reads as follows.
\begin{observation}\label{obs:I(l)_cuts_are_tight}
    Let~$1 \leq i \leq l$ then~$\Abs{\delta(A_i)} = 2d+2$ and~$\delta(A_i)$ contains one edge per path in~$\LLL$.
\end{observation}

We are finally ready to get to the last part of this construction: the definition of the weak coastal map for the cuff~$C_{\Sigma}$. We start with the definition of the island-zone~$\Zone(C_\Sigma)$.

\medskip

\paragraph{Defining the island-zone for~$C_{\Sigma}$.}

Let~$\LLL$ be chosen as above let~$P(\LLL)=(f_1,\ldots,f_l)$ be its circumference. Let~$C(\LLL)\coloneqq (v_l,f_0,v_0,f_1,v_1,\ldots,v_{l-1},f_l,v_l)$ be the cycle obtained from~$P(\LLL)$ by adding the edge~$f_0$ as discussed above. Recall the laminar cut-family~$\big((A_i,B_i)\big)_{1 \leq i \leq l}$; see \cref{lem:laminar-cuts-new-vortex}.

The following result highlighting that the laminar cuts `segregate' the linkage~$\LLL$ comes in handy later.
\begin{lemma} \label{lem:segregation_of_the_path_in_I(L)}
    It holds~$v_0 \in V(A_1)$,~$v_l \in V(B_l)$ where both vertices are of degree one in~$I(\LLL)$ and further~$v_i \in A_{i+1}\cap B_{i}$ for every~$1 \leq i \leq l-1$.
\end{lemma}
\begin{proof}
Clearly~$v_0$ and~$v_l$ are both of degree two in~$G$ for they arise from splitting an edge in~$F'$, and thus since~$f_0 \notin I(\LLL)$ they are both of degree one in~$I(\LLL)$. To see that~$v_0 \in V(A_1)$ recall that~$f_1 \in F_1 \subseteq G[[A_1]]$ using \cref{def:cut_seperating_F1_F2}, and the cuts are chosen so that~$B_i$ is minimal, thus, using the fact that~$v_0$ is of degree one in~$I(\LLL)$ the claim follows. 

Next let~$1 \leq i \leq l-1$. Then, by definition of the cut-family, we have~$f_{i+1} \in \delta(A_{i+1})$. The claim now follows by inductive reasoning using the fact that~$v_0 \in V(A_1)$ and~$v_1 \in V(B_0)$ since~$f_1 \in \delta(A_1)$ by construction. The laminarity of the cut-family given by \cref{lem:laminar-cuts-new-vortex} implies~$f_j \in G[[A_{i+1}]]$ for all~$1\leq j\leq i+1$. This together with the fact that~$\delta(A_{i+1})=2d+2$ and the fact that there is exactly one edge per path in~$\LLL\setminus \{P(\LLL)\}$ in the cut~$\delta(A_i)$ by \cref{obs:I(l)_cuts_are_tight} implies that~$v_i \in V(A_{i+1})$ for the sub-path~$(v_0,f_1,v_1,\ldots,v_i,f_{i+1},v_{i+1}) \subset P(\LLL)$ can only have the edge~$f_{i+1}$ in the cut, thus since~$v_0 \in V(A_{i+1})$ then all~$v_0,\ldots,v_i \in V(A_{i+1})$. But then for~$f_{i+1}$ to be in the cut~$\delta(A_{i+1})$ (which holds by definition) we deduce that~$v_{i+1} \in V(B_{i+1})$. Thus in particular by induction we know~$v_i \in V(B_i)$ concluding this part of the claim.

 We are left to prove~$v_l \in V(B_l)$. But this follows from the inductive reasoning above showing that~$v_{i+1}\in B_{i+1}$ for~$i= l-1$. This concludes the proof.
\end{proof}

 We continue by tightening the cylinder~$\Phi$ further by choosing a new cylinder~$\hat{\Phi} \subset \Phi$ bounded by~$C_{\Sigma}$ and~$C(\LLL)$ in the obvious way; recall that~$C(\LLL)$ is completely contained in~$\Gamma$ by \cref{lem:vortex-coastal-map:3} and thus this is indeed a well-defined cylinder in the surface~$\Sigma$. Then, for every~$v \in V(P(\LLL))$, using the fact that~$\deg(v)=4$ (up-to~$v_0$ and~$v_l$) in~$G$ , the fact that~$\Gamma$ is pseudo Euler-embedded and the fact that~$\nu(v) \notin \bd(\Sigma)$ by construction (and the fact that each cuff comes with enough cycles surrounding it), we deduce that there are two edges~$l(v),r(v) \in E(\Gamma)$ incident to~$v$ with~$(l(v),r(v))$ being a two-path disjoint from~$C(\LLL)$. Using the fact that~$\Gamma$ is pseudo Euler-embedded we deduce that either~$\nu(l(v)),\nu(r(v)) \in \hat{\Phi}$ or~$\nu(l(v)),\nu(r(v)) \in \Sigma \setminus \hat{\Phi}$, i.e., both edges are either drawn inside the cylinder---and are thus are part of~$I(\LLL)$---or outside of it in which case they are \emph{not} part of~$I(\LLL)$. 
 
 \smallskip
 
 With all of the previous work (especially \cref{obs:rearranging_the_vortex} and \cref{obs:I(l)_cuts_are_tight}) we can now reshape our surface~$\Sigma$ to~$\Sigma^*$ by pushing the cuff~$C_\Sigma$ into the surface and enlarging~$H$ to contain all of~$I(\LLL)$ except~$C(\LLL)$ where we draw~$C(\LLL)$ onto the new cuff. Technically this results in graphs~$\Gamma' \coloneqq \Gamma[\Sigma\setminus \hat{\Phi}]$---which contains~$C(\LLL)$---and~$I' \coloneqq (I \cup I(\LLL))-C(\LLL)$ where~$\Gamma' \cap I' \subset \nu^{-1}(\bd(\Sigma^*))$ and on the new enlarged cuff they agree exactly in~$V(C(\LLL))$. Note that this surface still satisfies the assumptions of the theorem (in particular there exists no large linkage~$\LLL'$ in~$I'$ with all its ends on the same cuff in the desired way). For simplicity and readability we will \emph{not} change the surface and graphs explicitly but simply keep in mind that we can---for we would need to do it in order for the island-zone to be a $\Zone$ as in the definition of \cref{def:ports_shores_zones}, i.e., we need the ports and shores to lie \emph{on} the cuff. With this in mind we are ready to define the island-zone of~$C_\Sigma$ (or rather the version enlarged up-to~$C_0$).

\begin{definition}[$\Zone(C_\Sigma)$]\label{def:zone_of_I(L)}
    Let~$\{i_0,\ldots,i_\xi\} \subset \{1,\ldots,l-1\}$ be maximal such that~$\nu(l(v_{i_j})),\nu(r(v_{i_j})) \in \Sigma \setminus \hat{\Phi}$ for all~$j \in \{0,\ldots,\xi\}$. Let~$u_j \coloneqq v_{i_j}$ for all~$j \in \{0,\ldots,\xi\}$. We define \emph{the island-zone of~$C_\Sigma$} via
    $$\Zone(C_\Sigma) \coloneqq (u_0',l(u_0),u_0,r(u_0),u_1',l(u_1),u_1,r(u_1),u_2',\ldots,u_\xi',l(u_\xi),u_\xi,r(u_\xi),u_0').$$
   Further we define~$\Port(C_\Sigma) \coloneqq \{l(u_j),r(u_j) \mid 0 \leq j \leq \xi\}$ and~$\Shore(C_\Sigma)\coloneqq \{u_j',u_j \mid 0 \leq j \leq \xi\}$.
\end{definition}
\begin{remark}
     Intuitively we get the island-zone by drawing a closed cut-line~$\gamma$ outside of~$\hat{\Phi}$, concentric and very close to~$C(\LLL)$ such that it cuts exactly all the edges~$e \in E(\Gamma)$ where~$e \in \{l(v),r(v)\}$ for some~$v \in \{v_0,\ldots,v_l\}$. This mimics pushing the cuff to the cycle~$C_0$ as discussed earlier. We then label the shores to be the spaces between the respective ports, where every second shore is marked by a vertex, namely the vertex~$u$ incident to~$l(u)$and~$r(u)$.
\end{remark}
This makes precise what we mentioned in the beginning, namely that we push the zone of the cuff outside the cylinder~$\hat{\Phi}$ (the tighter version of~$\Phi$), keeping in mind that formally we put~$C(\LLL)$ (the tighter version of~$C_0$) on the new cuff.

\smallskip

Note here that we will silently assume that~$\xi \geq 1$, i.e., there is at least two vertices on the circumference with edges outside of the cylinder. The reason being that if this were not the case, then using the Euler-embedding we deduce that the single vertex would have a loop---namely~$l(u_0)=r(u_0)$---enclosing the cuff and thus since all the vertices are degree four the graph~$G$ would either not be connected, or completely contained inside~$\hat{\Phi}$. The latter case is impossible assuming we have at least~$2d+3$ edge-disjoint cycles concentric to the cuff for then at least one of them is outside of~$\hat{\Phi}$.

\smallskip

Finally we redefine the family of laminar cuts~$\big((A_i,B_i)\big)_{1 \leq i \leq l}$ to match the new zone. That is we define

\begin{align*}
    A^\ZZZ_j &\coloneqq \bigcup_{1 \leq i \leq i_{j-1}}A_i, \text{ for every } 1 \leq j \leq \xi
\end{align*}
providing a new family of laminar cuts in~$I(\LLL)$ given by~$\big((A^\ZZZ_i,B^\ZZZ_i)\big)_{1 \leq i \leq \xi}$ where~$B^\ZZZ_i \coloneqq \bar{A^\ZZZ_i}$ in~$I(\LLL)$ for every~$1 \leq i \leq \xi$ as usual.

\begin{observation}\label{obs:laminar_cuts_for_zone_I(L)}
    
    The family of cuts~$\big((A^\ZZZ_i,B^\ZZZ_i)\big)_{1 \leq i \leq \xi}$ satisfies the following.
    \begin{enumerate}
        \item[(i)] It is laminar, i.e.,~$A^\ZZZ_i \subseteq A^\ZZZ_{i+1}$ and it holds~$\Abs{\delta(A^\ZZZ_i)} = \Abs{\delta(A^\ZZZ_{i+1})}$ for all~$1 \leq i < \xi$,
        \item[(ii)] for every~$1 \leq i \leq \xi$~$\delta(A_i^\ZZZ)$ contains one edge per path in~$\LLL$, and
        \item[(iii)] it segregates the linkage~$\LLL$, i.e.,~$u_0 \in V(A_1^\ZZZ)$,~$u_\xi \in V(B_\xi^\ZZZ)$ and~$u_i \in V(A_{i+1}^\ZZZ \cap B_i^\ZZZ)$ for all~$1 \leq i \leq \xi-1.$
    \end{enumerate}
\end{observation}
\begin{proof}
    This follows at once by definition using \cref{lem:laminar-cuts-new-vortex} for the first part, \cref{obs:I(l)_cuts_are_tight} for the second and \cref{lem:segregation_of_the_path_in_I(L)} for the third part.
\end{proof}

\smallskip

With the island-zone at hand we get a natural separation tailored to the zone as in \cref{def:separation_tailored_to_zone}.

\begin{observation}\label{obs:sep_tailored_to_zone_I(L)}
    The separation tailored to~$\Zone$ is given by~$(\Gamma_\ZZZ,I_\ZZZ)$ where~$I_\ZZZ \coloneqq I(\LLL) \cup \Port(C_\Sigma)$ and~$\Gamma_\ZZZ \coloneqq \Gamma \setminus \big(I(\LLL)\cup \Port(C_{\Sigma})\big)$.
\end{observation}
We are finally ready to complete the construction of the coastal map.

\medskip

\paragraph{Defining~$\mu$ and chopping the island-one into coasts.} 

 Recall the definition of~$S,S_F$ and~$S_Z$. For convenience we split~$S_Z$ into two sets namely~$S_Z^m \coloneqq \{e_i^{Z,m,l},e_i^{Z,m,r} \mid 1 \leq i \leq \rho\}$ and~$S_Z^{l,r} \coloneqq \{e_i^{Z,l},e_i^{Z,r} \mid 1 \leq i \leq \rho\}$. For each edge~$e_i^2 \in S_F$, let~$u^{f,1}_i$ be the endpoint of~$e^2_i$ incident to an edge in~$F_1$ and let~$u^{f,2}_i$ be the other endpoint incident to an edge in~$F_2$; note that~$\{l_i,r_i\} = \{u^{f_1}_i,u^{f,2}_i\}$ but we do not know which is which a priori. We define~$\mu$ as follows.

\begin{definition}[The map~$\mu$]\label{def:mu_for_I(L)}
Let~$\Zone(C_\Sigma)=(u_0',l(u_0),u_0,r(u_0),u_1',\ldots,u_\xi',l(u_\xi),u_\xi,r(u_\xi),u_0')$ as in \cref{def:zone_of_I(L)}. We define~$\mu:\Port(C_{\Sigma}) \cup \Shore(C_\Sigma) \to \{I \mid I \subseteq I(\LLL)\}$ via
\begin{itemize}
\item $\mu(u_0') \coloneqq S_F$
\item $\mu(l(u_0)) \coloneqq S_F \cup S_Z^{m}$
\item~$\mu(u_0) \coloneqq  G[[A^\ZZZ_1]] \cup \{l_i^Z,m_i^Z,r_i^Z \mid 1 \leq i \leq \rho\} \cup S \cup \{l(u_0),r(u_0)\}$, 
\item~$\mu(r(u_0)) \coloneqq \delta(A^\ZZZ_1) \cup S_Z^{l,r}$
\item For~$1 \leq i < \xi$ we set
\begin{align*}
     &\mu(u_i'),\mu(l(u_{i})) \coloneqq \delta(A^\ZZZ_i)\cup \{ e_j^{Z,l} \mid u^Z_j \in B^\ZZZ_{i} \} \cup \{ e_j^{Z,r} \mid v^Z_j \in B^\ZZZ_{i} \}\\
    &\mu(u_i) \coloneqq G[[A^\ZZZ_{i+1}\cap B^\ZZZ_{i}]] \cup \delta(A^\ZZZ_{i+1}) \cup \delta(B^\ZZZ_i) \cup \{ e_j^{Z,l} \mid u^Z_j \in B^\ZZZ_{i} \} \cup \{ e_j^{Z,r} \mid v^Z_j \in B^\ZZZ_{i} \} \cup \{l(u_i),r(u_i)\} \\
    &\mu(r(u_i)) \coloneqq \delta(A^\ZZZ_{i+1})\cup \{ e_j^{Z,l} \mid u^Z_j \in B^\ZZZ_{i} \} \cup \{ e_j^{Z,r} \mid v^Z_j \in B^\ZZZ_{i} \}
\end{align*}
\item Finally, for~$i=\xi$ we define
\begin{align*}
 &\mu(u_\xi'),\mu(l(u_{\xi})) \coloneqq \delta(A^\ZZZ_\xi) \cup \{ e_j^{Z,l} \mid u^S_j \in B^\ZZZ_{\xi} \} \cup \{ e_j^{Z,r} \mid v^S_j \in B^\ZZZ_{\xi} \} \\
 &\mu(u_\xi) \coloneqq G[[B_\xi]] \cup S_F \cup \{ e_j^{Z,l} \mid u^S_j \in B^\ZZZ_{\xi} \} \cup \{ e_j^{Z,r} \mid v^S_j \in B^\ZZZ_{\xi} \} \cup \{l(u_\xi),r(u_\xi)\}\\
 &\mu(r(u_\xi))\coloneqq S_F \cup \{ e_j^{Z,l} \mid u^S_j \in B^\ZZZ_{\xi} \} \cup \{ e_j^{Z,r} \mid v^S_j \in B^\ZZZ_{\xi} \},
 \end{align*}
\end{itemize}
\end{definition}
\begin{remark}
    Note that by definition~$\{u_i^{f,1} \mid 1 \leq i \leq 2d+2 \} \cup F_1 \subset G[[A^\ZZZ_1]] \subset \mu(u_1)$ as well as $\{u_i^{f,2} \mid 1 \leq i \leq 2d+2 \} \cup F_2 \subset G[[B^\ZZZ_\xi]] \subset \mu(u_\xi)$.

    Similarly note that~$u_i \in V(\mu(u_i))$ by construction using (iii) of \cref{obs:laminar_cuts_for_zone_I(L)} and further~$l(u_i),r(u_i) \in \mu(u_0)$ by definition for every~$0 \leq i \leq \xi$.
\end{remark}

With~$\mu$ at hand we next chop up the~$\Zone(C_\Sigma)$ into coast lines. To this extent note that~$\Abs{\delta(A^\ZZZ_i)} = \Abs{\delta(A^\ZZZ_j)}$ for every~$1 \leq i,j \leq \xi$ using \cref{obs:laminar_cuts_for_zone_I(L)}. This together with the \cref{def:mu_for_I(L)} of~$\mu$ and the laminarity of the cut family~$\big((A^\ZZZ_i,B^\ZZZ_i)\big)_{1 \leq i \leq \xi}$---which again follows from \cref{obs:laminar_cuts_for_zone_I(L)}---we derive the following.

\begin{observation}\label{obs:laminarity_of_depths}
    Let~$d^x_i \coloneqq \Abs{\mu(x(u_i))}$ for every~$0 \leq i \leq \xi$ and~$x \in \{l,r\}$. Then~$4d+2 \geq {d^l_0} = {d^r_0}\geq d^l_1 \geq \ldots \geq {d^l_{\xi}} = d^r_\xi \geq 2d+2$.
\end{observation}
\begin{proof}
    The claim~$4d+2 \geq d^l_0 $ follows from~$\Abs{S_F \cup S_Z^{m}} \leq \Abs{S_F} + 2\rho \leq (2d+2) + 2d$. The claim that~$d^r_{i} \leq d^l_{i}$ and~$d^l_j\leq d^r_{j-1}$ for every~$1 \leq i < \xi$ and~$1<j<\xi$ follows at once from the laminarity of the cut-family and the cuts all having order~$2d+2$; see \cref{obs:laminar_cuts_for_zone_I(L)}. By construction it holds~$d^l_0 = d^r_0$. And similar we have~$d^r_\xi = d^l_\xi$ since~$\Abs{\delta(A_\xi)} = \Abs{S_F} =2d+2$ and~$d^l_\xi \leq d^r_{\xi-1}$ holds again using the laminarity of the cut-fmaily.
\end{proof}

Throughout the remainder of this section we fix~$d^x_i \coloneqq \Abs{\mu(x(u_i))}$ for every~$0 \leq i \leq \xi$ and~$x \in \{l,r\}$ as in \cref{obs:laminarity_of_depths}. It turns out we can say even more about the respective depths.

\begin{observation}\label{obs:depths_left_equal_right}
    It holds~$d^l_i = d^r_i$ for every~$0 \leq i \leq \xi$.
\end{observation}
\begin{proof}
    This is a direct consequence of the definition of~$\mu'$ using \cref{obs:laminar_cuts_for_zone_I(L)} for~$1 \leq i \leq \xi-1$ and \cref{obs:laminarity_of_depths} for~$i \in \{0,\xi\}$.
\end{proof}

Using \cref{obs:laminarity_of_depths} and \cref{obs:depths_left_equal_right} we can define unique depths for every pair~$d^l_i,d^r_i$.

\begin{definition}[The depths~$d_i'$] \label{def:final_depths_I(L)}
    We define~$d_i' \coloneqq d^l_i$ for every~$0 \leq i \leq \xi$.
\end{definition}

With the above at hand we can chop up the zone into~$r \leq 2d$ coast lines as follows.

\begin{definition}\label{def:coastlines_for_I(L)}
     Let~$r \coloneqq \Abs{\{d_0',d_1',\ldots,d_\xi'\}}$ with~$r \leq 2d$ by \cref{obs:laminarity_of_depths} and \cref{obs:depths_left_equal_right}. Let~$0 \eqqcolon i_0 < i_1<\ldots < i_r \coloneqq \xi$ (with~$r \geq 1$) be such that~$d_{i_j}' > d_{i_{j+1}}'$ and~$d_i'=d_k'$ for every~$0 \leq j < r$ and every~$i_j \leq i,k \leq i_{j+1}$ . Let~$d_j \coloneqq d_{i_j}'$ for all~$0 \leq j \leq r$.

    Finally we define the coast lines~$\CCC_1,\ldots,\CCC_r \subset \Zone(C_{\Sigma})$ via
    \begin{align*}
        \CCC_j &\coloneqq (l(u_{i_{j-1}}),u_{i_{j-1}},\ldots,u_{i_j},r(u_{i_j})), \text{ for every } 1\leq j \leq r.
    \end{align*}
\end{definition}

A direct observation to the above \cref{def:coastlines_for_I(L)} reads as follows.

\begin{observation}
The shore~$u_0'\in \Shore(C_\Sigma)$ is deserted and let~$s \in \Shore(C_\Sigma)$ be a deserted shore of~$\Zone(C)$, then~$s \in \{u_i' \mid 0 \leq i \leq \xi\}$. Further~$\mu(u_i') \cap V(G) = \emptyset$ for every~$0\leq i \leq \xi$.
    \label{obs:deserted_shores_for_I(L)}
\end{observation}

Note here that \cref{obs:every_second_shore_is_vertexless} and \cref{obs:adj_ports_give_nondeserted_shore_and_are_even} are satisfied by construction. 
Thus, having chopped up the~$\Zone(C_\Sigma)$ into promising coast lines, we are almost ready for the proof of \cref{thm:structure-thm-to-coastal-maps}. In order to fully prove \cref{thm:structure-thm-to-coastal-maps} we still need to explain how to iteratively repeat the above construction for each cuff in order to get the desired coastal map for the whole of~$G$---up until now we only defined sort of a local version around a single cuff, i.e., we defined a single island as in \cref{def:weak_island} (of course we still need to prove that the map satisfies the needed constraints).

\medskip

\paragraph{Iteratively mapping every cuff of~$\Sigma$}
The only purpose of this paragraph is to cover \WMI in the \cref{def:weak_coastal_map} of weak coastal-maps; note that \WMII-\WMV are all of local nature, i.e., they are to be checked for each single cuff locally. In a sense we show how to extend the above to all the cuffs, which, as we have mentioned in the beginning of the proof, is more of a book-keeping exercise. The proof of \WMI is rather straightforward once one has thought about how to define~$\mu$ in general when we have defined it locally for each cuff. The reader may hence skip this paragraph and read the proof of the axioms \WMII-\WMV first as they use the notation introduced above, then come back if they are interested in the definition of the `global' coastal-map needed in the proof of \WMI. This is what we do next.

\smallskip

Repeat the whole construction above iteratively for each cuff~$C_\Sigma \in c(\Sigma)$; denote the cuffs by~$C_\Sigma^1,\ldots,C_\Sigma^\sigma$ for~$\sigma = \Abs{c(\Sigma})$. 
That is, for every~$1 \leq i \leq \sigma$, repeating the above construction around~$C^i_\Sigma$ and making sure to choose pairwise different cycles~$C_1,\ldots,C_{2d+2}$ for each cuff, we get the following tuple for the cuff~$C^i_\Sigma$

\begin{equation}(I(\LLL)^i, P(\LLL)^i, S^i,\Zone(C^i_\Sigma), \Port(C_\Sigma^i), \Shore(C_\Sigma^i),r^i, \CCC_1^i,\ldots,\CCC_{r^i}^i,\mu^i,d_1^i,\ldots,d_{\xi^i}^i),
\label{eq:tuple_for_each_cuff_I(L)}
\end{equation}

where~$r^i \leq 2d$ and~$d_j^i \leq 4d+2$ using \cref{obs:laminarity_of_depths} for every~$1 \leq j \leq r^i$ and some~$r^i \leq \xi^i \leq l^i$ where~$l^i$ was the length of~$P(\LLL)^i$.

\smallskip

We define~$\mu$ as the union of all~$\mu^i$ in the obvious way, and similar we define~$\Zone$ for each of the cuffs as given by~$\Zone(C_\Sigma^i)$. Let~$r \coloneqq r^1 + \ldots + r^\sigma$ and let~$\hat{d} \coloneqq \max\{d_j^i \mid 1 \leq i \leq \sigma, 1 \leq j \leq r^i\}$. Let~$\tilde{\Gamma} \coloneqq \Gamma \setminus \bigcup_{1 \leq i \leq \sigma} (I(\LLL)^i \cup S_F^i)$. Finally let

$$\CCC \coloneqq (\tilde{\Gamma},\Zone,\CCC_1^1,\ldots,\CCC_{r^1}^1,\ldots,\CCC_1^\sigma,\ldots,\CCC_{r^\sigma}^\sigma,\mu).$$

We are left to prove that~$\CC$ is indeed a weak-coastal map of~$G$ in~$\Sigma$.

\medskip

\paragraph{Proving that~$\CC$ is a coastal map.}

Having concluded the construction and definitions of~$\CC$ we are ready to prove \cref{thm:structure-thm-to-coastal-maps}, which we restate as follows.

\begin{theorem}
    The map~$\CC$ is a weak-coastal map of~$G$ in~$\Sigma$ with~$r \leq \sigma \cdot 2d$ sights and depth~$\hat{d} \leq 4d+2$.
    \label{thm:coastal_map_for_I(L)}
\end{theorem}
\begin{proof}
    Before dedicating ourselves to a proof of the axioms \WMI-\WMVI, we verify the needed conditions for the coast lines.

    \begin{claim}\label{thm:coastal_map_for_I(L)_ports_covered_by_coasts}
        The coast lines~$\CCC_1^i,\ldots,\CCC_{r^i}^i$ are pairwise disjoint and~$\Port(G) \subseteq \bigcup_{1 \leq i \leq \sigma}\bigcup_{1 \leq j \leq r^i} \Port(\CCC_j^i)$.
    \end{claim}
    \begin{ClaimProof}
        The first part of the claim is obvious from the \cref{def:coastlines_for_I(L)} of~$\CCC_1^i,\ldots,\CCC_{r^i}^i$, as is the second claim for each~$l(u_j)$ and~$r(u_j)$ is part of some coastline for every~$1 \leq j \leq \xi^i$ by construction.
    \end{ClaimProof}
    Further $r \leq \sigma \cdot 2d$ and~$\hat{d} \leq 4d+2$ are clear using \cref{def:coastlines_for_I(L)} and \cref{obs:laminarity_of_depths}.

    We continue with verifying the five axioms of weak-coastal maps.

    \begin{claim} \label{thm:coastal_map_for_I(L)_claim_WM1}
        The map~$\CC$ satisfies \WMI.
    \end{claim}
    \begin{ClaimProof}
        It follows from the construction that~$G = \tilde{\Gamma}_\ZZZ \cup \bigcup_{1 \leq i \leq \sigma}\bigcup_{1 \leq j \leq r^i} \mu(\CCC_j^i)$. To see this, note that~$\tilde{\Gamma}_\ZZZ = \tilde{\Gamma} \setminus \bigcup_{1 \leq i \leq \sigma}\bigcup_{1 \leq j \leq r^i} \Port(\CCC_j^i)$, the rest then follows from \cref{obs:sep_tailored_to_zone_I(L)} together with the fact that~$I(\LLL)^i \subseteq \bigcup_{1\leq j \leq r^i}\mu(\CCC_j^i)$ for every~$1 \leq i \leq \sigma$ by \cref{def:mu_for_I(L)};again note here that the incidences of~$l(u_i)$ and~$r(u_i)$ with~$u_i$ are covered by~$\mu(u_i)$ using (iii) of \cref{obs:laminar_cuts_for_zone_I(L)} and the \cref{def:mu_for_I(L)}. Since the~$I(\LLL)^i$ and~$I(\LLL)^j$ are pairwise disjoint for~$i \neq j$, then again using the definition of~$\mu$, we easily derive that~$\mu(s)$ are pairwise disjoint for~$s \in \Shore(G) = \bigcup_{1 \leq i \leq \sigma}\Shore(C_\Sigma^i)$. Finally, again using that the different islands~$I(\LLL)^i$ and~$I(\LLL)^j$ are pairwise disjoint we derive that~$\mu(\chi_1) \cap \mu(\chi_2) = \emptyset$ for~$\chi_1,\chi_2 \in \Port(G) \cup \Shore(G)$ lying in different zones with respect to~$\Zone$. This concludes the proof of the claim. 
    \end{ClaimProof}
    
    Note that the remaining axioms \WMII-\WMVI are of local nature, that is, they depend on the choice of a cuff and are local with respect to that cuff, more precisely with respect to that island. To this extent let~$C_\Sigma \coloneqq C_\Sigma^i$ be arbitrary, and omit the super-script~$i$ in the tuple given as in \cref{eq:tuple_for_each_cuff_I(L)} to avoid clutter, i.e., proceed with the notation as in the single-cuff case prior to the paragraph `Iteratively mapping every cuff of~$\Sigma$'.
    
    \begin{claim}
        The map~$\CC$ satisfies \WMII.
         \label{thm:coastal_map_for_I(L)_claim_WM2}
    \end{claim}
    \begin{ClaimProof}
        Let~$s \in \Shore(\CCC_j)$ for some~$1 \leq j \leq r$ be a shore of a coast line and let~$p_l,p_r \in \Port(s)$ be its adjacent ports. In a first instance assume that~$\mu(s) \cap V(G) \neq \emptyset$, then by \cref{obs:deserted_shores_for_I(L)} we deduce that~$s \in \{u_0,\ldots,u_\xi\}$. Hence, by \cref{def:mu_for_I(L)} of~$\mu$ it is clear that~$\mu(s) \cap \Gamma_\ZZZ = \{p_l,p_r\}$. Further~$\mu(p) \cap E(\Gamma_\ZZZ) = \emptyset$ for every port~$p \in \Port(G)$ by definition of~$\mu$. It remains to prove that~$\mu(p_l)\cup \mu(p_r) \subseteq \mu(s)$. 
        There is several cases to consider for~$s$.
        \begin{itemize}
            \item[1.] If~$s = u_i'$, then~$i \in \{1,\ldots,\xi\}$ and in particular~$i\neq 0$ for~$u_0'$ is a deserted shore as seen in \cref{obs:deserted_shores_for_I(L)}. Then by \cref{def:mu_for_I(L)} of~$\mu$ it holds~$\mu(u_i') = \mu(l(u_{i}))$ and in particular~$\mu(l(u_i) \subset \mu(u_i'))$. Further by definition it holds that~$\mu(l(u_{i+1})) \subseteq \mu(r(u_{i}))$ for every~$1 \leq i \leq \xi-1$. If~$r(u_i) = l(u_{i+1})$ we are done, thus towards a contradiction assume~$ r(u_{i}) \neq l(u_{i+1})$. Then~$d^r_{i+1}>d^l_i$ and in particular~$d_{i+1}'>d_i'$ using \cref{def:final_depths_I(L)} and \cref{obs:depths_left_equal_right}. Finally using \cref{def:coastlines_for_I(L)} we deduce that~$u_i' \in \Shore(C_\Sigma)$ with~$r(u_i),l(u_{i+1}) \in \Port(u_i')$ is a deserted shore as a contradiction to~$u_i' \in \Shore(\CCC_j)$.

            \item[2.] Let~$s=u_i$ for some~$0 \leq i \leq \xi$. Then the claim holds trivially by definition of~$\mu$noting that~$\delta(A_1^\ZZZ) \subset G[[A_1^\ZZZ]]$ and~$\delta(A_\xi^\ZZZ) = \delta(B_\xi^\ZZZ) \subset G[[B_\xi^\ZZZ]]$.
        \end{itemize}
    This covers all possible cases and thus concludes the proof of this claim.
    \end{ClaimProof}

    \begin{claim}
        The map~$\CC$ satisfies \WMIII. \label{thm:coastal_map_for_I(L)_claim_WM3}
    \end{claim}
    \begin{ClaimProof}
        Let~$\chi_1,\chi_2 \in \Port(C_\Sigma) \cup \Shore(C_\Sigma)$ and~$p_1,p_2 \in \Port(C_\Sigma)$ such that~$\chi_1,p_1,\chi_2,p_2$ occur in this order on~$\Zone(C_\Sigma)$. Note here that the~$\chi_1,\chi_2$ may well be deserted shores opposed to the assumption in \WMII. We ought to show that~$\mu(\chi_1)\cap \mu(\chi_2) \subseteq \mu(p_1)\cup \mu(p_2)$. By definition of~$\mu$ we know that~$\mu(\chi_1) \cap \mu(\chi_2) \subset E(G)$, for~$\mu(\chi_1), \mu(\chi_2)$ are vertex-disjoint by \WMI. So let~$e \in \mu(\chi_1) \cap \mu(\chi_2)$ and without loss generality let~$\chi_1,\chi_2$ appear in this order on~$(u_0',l(u_0),u_0,\ldots,u_\xi',l(u_\xi),u_\xi,r(u_\xi))$; else rename them and use the symmetry in the claim. For ease of notation we define an order relation~$\prec$ on~$(u_0',l(u_0),u_0,\ldots,u_\xi',l(u_\xi),u_\xi,r(u_\xi))$ given exactly by the ordering of the tuple, i.e.,~$u_0'\prec l(u_0) \prec u_0 \prec \ldots \prec u_\xi \prec r(u_\xi)$. 

        We assume that~$\chi_1,p_1,\chi_2,p_2$ are pair-wise distinct for else the claim is obvious. We start by verifying the claim for the possible cases. 
        
        \begin{itemize}
            \item[1.] Let~$r(u_0) \preceq \chi_1 \prec \chi_2 \preceq u_\xi$. Then, using the \cref{def:mu_for_I(L)} following the laminarity of the cut-family~$\big((A_i^\ZZZ,B_i^\ZZZ)\big)_{1\leq i \leq \xi}$ given by \cref{obs:laminar_cuts_for_zone_I(L)} one easily deduces that
            $$\mu(\chi_1') \cap \mu(\chi_2') \subset \mu(\chi')$$ for very~$r(u_0) \preceq \chi_1' \preceq \chi' \preceq \chi_2' \preceq u_e$. Thus the claim follows in this case.
            
           \item[2.] The remaining cases to consider after 1. are~$(\chi_1,\chi_2)$  with~$ \chi_1 \in \{u_0',l(u_0),u_0\}$ and~$\chi_2$ arbitrary or~$\chi_2 = r(u_\xi)$ and~$\chi_1$ arbitrary. 
           
           We start with the case~$\chi_1=u_0'$. The case~$\chi_1=u_0'$ and~$\chi_2=u_0$ is trivial for then~$p_1 = l(u_0)$. Similarly the cases~$(\chi_1,\chi_2)=(u_0',u_\xi)$ and~$(\chi_1,\chi_2)=(u_0',r(u_\xi))$ are trivial, for the latter note that~$p_1\in\{u_0',r(u_\xi)\}$. Note further that for~$\chi_1 = u_0'$ it holds~$\mu(u_0') \cap \mu(\chi) = \emptyset$ for all~$r_0 \preceq \chi \preceq l(u_\xi)$ by definition; thus we have covered all the cases for~$\chi_1 = u_0'$. 
           
           \smallskip
           
           Next we look at~$\chi_1 = u_0$, where the only remaining cases are~$(\chi_1,\chi_2)=(u_0,\chi)$ for~$r(u_0)\preceq \chi \preceq r(u_\xi)$. The case~$\chi_2 = r(u_\xi)$ is easy for then~$p_1 \in \{u_0',l(u_0)\}$ and~$\mu(u_0) \cap \mu(r(u_\xi)) \subseteq \mu(u_0') \subset \mu(l(u_0))$. Thus let~$\chi_1 = u_0$ and~$r_0 \preceq \chi_2 \preceq u_\xi$. For~$\chi \neq u_\xi$ the claim follows again by laminarity of the cut-family noting that~$\mu(\chi_2) \cap S_Z^m = \emptyset$. 
           
           Thus let~$\chi_2 = u_\xi$. First note that~$S_F \subset \mu(r(u_\xi)) \cap \mu(u_0') \cap \mu(l(u_0))$ and thus~$S_F \subset \mu(p_2)$ for the only relevant choices~$p_2 \in \{r(u_\xi),u_0',l(u_0)\}$. Secondly note that~$\{e_j^{Z,l} \mid u_j^S \in B_\xi^\ZZZ\} \cup \{e_j^{Z,r} \mid v_j^S \in B_\xi^\ZZZ\} \subset \mu(\chi')$ for all~$r(u_0)\preceq \chi' \preceq u_\xi$ again using the laminarity of the cut-family, and thus in particular~$\{e_j^{Z,l} \mid u_j^S \in B_\xi^\ZZZ\} \cup \{e_j^{Z,r} \mid v_j^S \in B_\xi^\ZZZ\} \subset \mu(p_1)$ for~$p_1 \in \{r(u_0),\ldots,u_\xi\}$. Both observations combined imply that~$\mu(u_0) \cap \mu(u_\xi) = S_F \cup \{e_j^{Z,l} \mid u_j^S \in B_\xi^\ZZZ\} \cup \{e_j^{Z,r} \mid v_j^S \in B_\xi^\ZZZ\} \subseteq \mu(p_2)\cup\mu(p_1)$ and thus the proof.

            \smallskip

           Finally note that the cases for~$\chi_1=l(u_0)$ follow easily from the cases above using the fact that~$\mu(l(u_0)) \subset \mu(u_0)$ for then~$\mu(l(u_0)) \cap \mu(\chi_2) \subset \mu(u_0) \cap \mu(\chi_2)$ for which we verified the cases. Similarly the cases for~$\chi_2 = r(u_\xi)$ follow easily from the above using the fact that~$\mu(r(u_\xi)) \subset \mu(u_{\xi})$.
        \end{itemize}

        Having verified all of the cases the claim follows.

    \end{ClaimProof}
    
    \begin{claim}
        The map~$\CC$ satisfies \WMIV.\label{thm:coastal_map_for_I(L)_claim_WM4}
    \end{claim}
    \begin{ClaimProof}
        Let~$p \in \Port(C_\Sigma)$ and~$e \in \mu(p)$. In a first instance we prove that there exists~$p^\ast \in \Port(C_\Sigma)$ with~$e \notin \mu(p^\ast)$. To this extent note that for any~$r(u_0) \preceq \chi \preceq l(u_\xi)$ it holds that~$\mu(\chi) \cap \mu(l(u_0)) = \emptyset$ by definition; thus we may choose~$p^\ast = l(u_0)$ for all ports~$p$ with~$r(u_0) \preceq p \preceq l(u_\xi)$. Let~$p = l(u_0)$ then we choose~$p^\ast = r(u_0)$ since~$\mu(r(u_0)) \cap (S_F\cup S_Z^m) = \emptyset$. Finally let~$p = r(u_\xi)$ and let~$e \in \mu(r(u_\xi))$. If~$e \in S_F$ we may choose~$p^\ast = r(u_0)$ as in the case of~$l(u_0)$. If~$e \in S_Z^{l,r}$ we may choose~$p^\ast = l(u_0)$ for~$\mu(l(u_0)) \cap S_Z^{l,r} = \emptyset$. This concludes the proof of the first part of this claim.
        
        \smallskip
        
        For the second part of \WMIV let~$\chi \in \Port(C\Sigma) \cup \big(\Shore(C_\Sigma)\setminus \bigcup_{i=1}^{r}\Shore(\CCC_i) \big)$ such that either~$\chi$ is a deserted shore or an interior port of~$\CCC_i$ for some~$1 \leq i \leq r$ and let~$e \in \mu(\chi)$ with ends~$u,v \in V(G)$. We ought to prove that then there exist~$s_l,s_r \in \Shore(C_\Sigma)$ with~$u,v \in \mu(s_l)\cup \mu(s_r)$ and $u \in \mu(s_l) \iff v \in \mu(s_r)$ such that~$s_l,\chi,s_r$ appear in that order on~$\Zone(C)$ and further for every~$\chi' \in \Port(C) \cup \Shore(C)$ such that~$s_l,\chi',s_r$ appear in this order on~$\Zone(C)$ it holds~$e \in \mu(\chi')$. To prove this note that it suffices to prove the following.
        \begin{itemize}
            \item[{\small (WM+)}] Let~$e \in E(G)$ with ends~$u,v$ be such that~$u \in \mu(s_l)$ and~$v \in \mu(s_r)$ for some~$s_l,s_r\in \Shore(C_\Sigma)$. Let~$\chi \in \Shore(C_\Sigma)$. Then~$e \in \mu(\chi)$ if and only if~$s_l,\chi,s_r$ appear in this order on~$\Zone(C)$. 
        \end{itemize}

    First we show that (WM+) implies \WMIV---this is proven as~$(\star)$---and then we prove that (WM+) holds for~$\mu$---this is proven as~$(\star\star)$.
    
\smallskip
\begin{itemize}
    \item[$(\star)$] \textbf{The condition (WM+) implies the second part of \WMIV.} To see this let~$e \in \mu(\chi)$ with ends~$u,v \in V(G)$. Then using \WMI proved in \cref{thm:coastal_map_for_I(L)_claim_WM1} there exist~$s_l,s_r \in \Shore(C_\Sigma)$ (probably equal) such that~$u \in \mu(s_l)$ and~$v \in \mu(s_r)$. By (WM+) then for all~$\chi \in \Shore(C_\Sigma)$ it holds that~$e \in \mu(\chi')$ if and only if~$s_l,\chi',s_r$ appear in this order on~$\Zone(C)$. Using \WMII proved in \cref{thm:coastal_map_for_I(L)_claim_WM2} then the same holds for~$\chi' \in \Port(C_\Sigma)$ (not necessarily interior) noting that every coast line contains a shore and by \cref{thm:coastal_map_for_I(L)_ports_covered_by_coasts} all ports are covered by the coast lines. Thus the second part of \WMIV follows, altogether proving~$(\star)$. 
\end{itemize}
\begin{itemize}
    \item[$(\star\star)$] \textbf{The map~$\mu$ satisfies (WM+).} The proof is by a tedious but straightforward case-distinction.
    
        \begin{itemize}
            \item[1.] Let~$e \in S_Z^m$, then~$e \in \mu(\chi)$ if and only if~$\chi \in \{l(u_0),u_0\}$ by \cref{def:mu_for_I(L)} of~$\mu$. In particular~$s_l,s_r = u_0$ do the trick.
            \item[2.] Let~$e \in S_F$, then~$e=e_i^2$ for some~$1 \leq i \leq 2d+2$ with ends~$u_i^{f,1},u_i^{f_2}$. By construction of~$\mu$ then~$e \in \mu(\chi)$ if and only if~$\chi\in \{u_\xi,r(u_\xi),u_0',l(u_0),u_0\}$. Further it holds~$u_i^{f,1} \in G[[A_1^\ZZZ]]$ and~$u_i^{f,2} \in G[[B_xi^\ZZZ]]$ and thus~$s_l = u_\xi$ and~$s_r = u_0$ do the trick.
            \item[3.] Let~$e \in S_Z^{l,r}$, then~$e \in \{e_j^{Z,l},e_j^{Z,r}\}$ for some~$1 \leq j \leq \rho$, using the symmetric definition of~$\mu$ regarding edges in~$S_Z^{l,r}$ we may assume without loss of generality that~$e = e_j^{Z,l}$; the other case is analogous. Let~$1 \leq i \leq \xi$ be maximal such that~$e_j^{Z,l}$ has an end in~$B_i^\ZZZ$. Then~$e_j^{Z,l}$ has one end in~$\mu(u_0)$---namely~$l_j^Z$---and the other end in~$\mu(u_i)$ by \cref{def:mu_for_I(L)}. If~$i = \xi$, then~$s_l = u_0$ and~$s_r= u_\xi$ do the trick, for~$\mu(u_0') \cap S_Z^{l,r} = \emptyset$ and by the laminarity of the cuts it holds~$e_j^{Z,l} \in \mu(\chi)$ for every~$u_0 \preceq \chi \preceq u_\xi$. Thus assume that~$i< \xi$. We claim that~$s_l = u_0$ and~$s_r = u_i$ do the trick. To see this note that~$e_j^{Z,l} \notin \mu(u_k)$ for all~$i<k\leq \xi$ for~$i$ was chosen maximal and thus~$e_j^{Z,l}$ has no end in~$B_k^\ZZZ$ which in particular implies by \cref{def:mu_for_I(L)} of~$\mu$ that~$e_j^{Z,l} \notin \mu(u_k)$. This concludes the case that~$e \in S_Z^{l,r}$.

            \item[4.] Let~$e \in E(I(\LLL))$ with ends~$u,v$. Then there are three major cases to consider. 
            First assume that~$u,v \in A_1^\ZZZ$. Then~$e \in \mu(u_0')$ and~$e \notin \mu(u_i)\cup \mu(u_i')$ for any other~$1 \leq i \leq \xi$ by \cref{def:mu_for_I(L)} of~$\mu$ using the laminarity of the cut-family as given by \cref{obs:laminar_cuts_for_zone_I(L)} and the fact that if~$u \in A_1^\ZZZ$ then~$u\notin A_{i+1}^\ZZZ \cap B_i^\ZZZ$ for any~$1 \leq i \leq \xi-1$ and clearly~$u \notin B_\xi^\ZZZ$. The case that~$u,v\in B_\xi^\ZZZ$ is analogous to the first case by symmetry of cuts. 

            \smallskip
            
            For the second case assume that~$\{u,v\} \cap B_\xi^\ZZZ \neq \emptyset$ and~$\{u,v\} \cap A_\xi^\ZZZ \neq \emptyset$; without loss of generality let~$u \in A_\xi^\ZZZ$ and~$v \in B_\xi^\ZZZ$. Let~$1 \leq i \leq \xi$ be minimal such that~$u \in A_i^\ZZZ$. If~$i=1$ then~$e \in \delta(A_j^\ZZZ)$ for every~$1 \leq j \leq \xi$ using the laminarity of the cut-family given by \cref{obs:laminar_cuts_for_zone_I(L)}. Hence~$u \in \mu(u_0)$,~$v \in \mu(u_\xi)$ and~$e \in \mu(\chi)$ for every~$u_0 \preceq \chi \preceq u_\xi$ using that~$\delta(A_i^\ZZZ) \subset \mu(u_i) \cap \mu(l(u_i))$ and~$\delta(A_k^\ZZZ) \subset \mu(r(u_{k-1}))$ for every~$1 \leq i \leq \xi$ and~$1<k\leq \xi$. Also by definition~$e \notin \mu(\chi')$ for~$\chi' = u_0'$. Thus~$s_l=u_0$ and~$s_r=u_\xi$ do the trick.

            Thus assume that~$i > 1$. Then~$e$ has one end in~$\mu(u_{i-1})$ since~$G[[A_i^\ZZZ \cap B_{i-1}^\ZZZ]] \subset \mu(u_{i-1})$. Further it is easy t see that~$e \notin \mu(u_k)$ for all~$1 \leq k < i-1$ by \cref{def:mu_for_I(L)} of~$\mu(u_k)$. Again using laminarity of the cuts then~$e \in \delta(A_j^\ZZZ)$ for all~$ i\leq j \leq \xi $ and thus~$e \in \mu(u_j)$ for all~$ i-1 \leq j \leq \xi$. Choosing~$s_l=u_{i-1}$ and~$s_r = u_\xi$ does the trick analogously to the case~$i=1$ above. 

            Finally note that the case that~$\{u,v\} \cap A_1^\ZZZ \neq \empty$ and~$\{u,v\} \cap B_1^\xi \neq \empty$ is analogous to the above.

            \smallskip

            For the third and finally case let~$1 \leq i_u,i_v < \xi$ be maximal such that~$u \in B_{i_u}^\ZZZ$ and~$v \in B_{i_v}^\ZZZ$.; note that we covered all cases of~$\{u,v\} \cap A_1^\ZZZ\cup B_\xi^\ZZZ$ in the two previous cases. There is several sub-cases to consider. 
            
            Suppose first that~$i_u = i_v$. Then~$u,v \in \mu(u_i)$ for~$u,v \in A_{i+1}^\ZZZ \cap B_i^\ZZZ$ using the maximality of~$i_u,i_v$ and the laminarity of the cut-family given by \cref{obs:laminar_cuts_for_zone_I(L)}. By \WMI as proved in \cref{thm:coastal_map_for_I(L)_claim_WM1} we deduce that~$u,v \notin \mu(\chi)$ for any~$\chi \in \Shore(C_\Sigma)$ with~$\chi \neq u_i$. Then~$s_l,s_r = u_i$ does the trick.

            Suppose next that~$i_u \neq i_v$ and without loss of generality assume that~$i_u < i_v$. Then~$u \in \mu(u_{i_u})$ and~$v \in \mu(u_{i_v})$ by \cref{def:mu_for_I(L)} of~$\mu$. Again by laminarity of the cuts and definition of~$\mu$ we derive---analogously to the previous two cases---that~$e \in \mu(\chi)$ for all~$u_{i_u} \preceq \chi \preceq u_{i_v}$. Also analogously to the previous cases we derive that~$e \notin \mu(\chi)$ for any~$\chi \in \Shore(C_\Sigma)$ such that~$u_{i_v} \preceq \chi \preceq u_{i_u}$. This is trivial for~$\chi= u_0'$ and it follows for all other cases from the fact that~$u,v \notin A_1^\ZZZ \cup B_\xi^\ZZZ$ and~$u,v\notin G[[A_{i+1}\cap B_i]]$ for all~$i< i_u$ (which is irrelevant if~$i_u =1$) and for all~$i>i_v$. This covers all the possible sub-cases and thus finishes this case.
        \end{itemize}
        This proves~$(\star\star)$.
    \end{itemize}
        
    Having discussed all the possible cases, this concludes the proof of the claim.

    \end{ClaimProof}
    
    \begin{claim}
        The map~$\CC$ satisfies \WMV.
    \end{claim}
    \begin{ClaimProof}
        First note that for every~$1 \leq i \leq r$ it holds that~$d_i = \Abs{\mu(p)}$ for every~$p \in \Port(\CCC_i)$ by \cref{def:final_depths_I(L)} of the~$d_i$ and the \cref{def:coastlines_for_I(L)} of the coast lines. Hence the first part of \WMV follows by construction.

        We continue with the second part and analyse it via several sub-cases; to this extent note that~$u_0'$ is deserted using \cref{obs:deserted_shores_for_I(L)} and thus irrelevant to \WMV.

        \begin{itemize}
            \item[1.] Let~$s= u_0$ and~$p_l = l(u_0)$ as well as~$p_r=r(u_0)$. Then there exists an~$S_F$-$\delta(A_1^\ZZZ)$-linkage of order~$2d+2$ by first connecting each edge~$e_i^2 \in S_F$ to~$e_i^{f,1}$ where~$e_i^{f_1}$ is the unique edge in~$\{e_i^1,e_i^3\} \cap F_1$---recall the definition of~$F_1$---via the path~$(e_i^2,u_i^{f,1},e_i^{f_1}) \subset \mu(u_0)$ for every~$1 \leq i \leq 2d+2$, and then concatenating said paths with the respective paths in~$\LLL \cap G[[A_1^\ZZZ]]$ starting at~$e_i^{f,1}$ and ending in an edge of~$\delta(A_1^\ZZZ)$ using the fact that~$\delta(A_1^\ZZZ)$ contains exactly one edge per path in~$\LLL$ by \cref{obs:laminar_cuts_for_zone_I(L)}. Note that this linkage is disjoint from~$S_Z$ using the fact that~$\LLL$ is a linkage in~$I(\LLL)$.

            Finally there exists an~$S_Z^m{-}S_Z^{l,r}$-linkage of order~$2\rho$ in~$\mu(u_0)$ given by the paths~$(e_i^{Z,l},l_i^Z,e_i^{Z,m,l})$ and~$(e_i^{Z,r},r_i^Z,e_i^{Z,m,r})$ disjoint form the above. Thus we get a~$\mu(l(u_0)){-}\mu(r(u_0))$-linkage in~$\mu(u_0)$ of order~$2d+2+2\rho = d_1$ concluding the proof of this case.

            \item[2.] Let~$1 \leq i < \xi$ and let~$s = u_i$ and~$p_l = l(u_i)$ as well as~$p_r = r(u_i)$. Using \cref{obs:laminar_cuts_for_zone_I(L)} we immediately get a~$\delta(A_i^\ZZZ){-}\delta(A_i^\ZZZ)$-linkage~$\LLL_1$ of order~$2d+2$ in~$\mu(u_i)$ induced by~$\LLL \cap G[[A_{i+1}\cap B_i]]$. By definition this linkage is fully contained in~$I(\LLL)$ and thus independent of~$S_Z$. Further, since~$\{e_j^{Z,l} \mid u_j^Z \in B_i^\ZZZ\} \cup \{e_j^{Z,r} \mid v_j^Z \in B_i^\ZZZ\} \subset \mu(l(u_i)) \cap \mu(r(u_i))$ we get second linkage~$\LLL_2$ of order~$\Abs{\{e_j^{Z,l} \mid u_j^Z \in B_i^\ZZZ\} \cup \{e_j^{Z,r} \mid v_j^Z \in B_i^\ZZZ\}}$ where every path is exactly one of these edges. As mentioned above then~$\LLL_1 \cup \LLL_2$ forms an edge-disjoint linkage. By \cref{def:coastlines_for_I(L)} of the coast lines and by \cref{def:final_depths_I(L)} of the depths~$d_t$ we deduce that the order of the linkage is exactly~$d_t$ where~$1\leq t \leq r$ is chosen such that~$u_i \in \Shore(\CCC_t)$; recall that~$u_i$ is non-deserted as highlighted in \cref{obs:deserted_shores_for_I(L)}. This concludes the proof of this case.

            \item[3.] Let~$s= u_\xi$ and~$p_l = l(u_\xi)$ as well as~$p_r=r(u_\xi)$. Then there is a~$\delta(A_\xi^\ZZZ){-}S_F$ linkage of order~$2d+2$ analogously to case 1. where now the paths use~$F_2$ and~$u_i^{f_2} \in G[[B_\xi]]$ for every~$1 \leq i \leq 2d+2$. Again that linkage is edge-disjoint from~$S_Z$ and thus we can extend the linkage via the linkage given by the single edges in~$\{e_j^{Z,l} \mid u_j^Z \in B_\xi^\ZZZ\} \cup \{e_j^{Z,r} \mid v_j^Z \in B_\xi^\ZZZ\} \subset \mu(l(u_\xi)) \cap \mu(r(u_\xi))$ similar to case 2. The arguments are analogous as in the respective cases concluding the proof of this case.

            \item[4.] Let~$s=u_i'$ for some~$1 \leq i \leq \xi$ such that~$u_i'$ is non-deserted, i.e,~$u_i' \in \Shore(\CCC_j)$ for some~$1 \leq j \leq r$. By \cref{def:coastlines_for_I(L)} of the coast lines this then implies that~$d^r_{i-1}=d^l_i$ where~$d^r_{i-1} = \Abs{\mu(r(u_{i-1})}$ and~$d^l_i = \Abs{\mu(l(u_i)}$---recall \cref{obs:laminarity_of_depths}, \cref{obs:depths_left_equal_right} and \cref{def:final_depths_I(L)}. But then this in turn implies that
            \begin{equation}
                \Abs{\{e_j^{Z,l} \mid u_j^Z \in B_{i-1}^\ZZZ\} \cup \{e_j^{Z,r} \mid v_j^Z \in B_{i-1}^\ZZZ\}} = \Abs{\{e_j^{Z,l} \mid u_j^Z \in B_{i}^\ZZZ\} \cup \{e_j^{Z,r} \mid v_j^Z \in B_{i}^\ZZZ\}}
                \label{eq:Claim_WM5_for_ui'}
            \end{equation}
            since by \cref{def:mu_for_I(L)}~$\mu(l(u_i))$ and~$\mu(r(u_i))$ may only differ in their subset of~$S_Z^{l,r}$. Using the laminarity of the cut-family given by \cref{obs:laminar_cuts_for_zone_I(L)}, the \cref{eq:Claim_WM5_for_ui'} implies that~
            $$e_j^{Z,l} \mid u_j^Z \in B_{i-1}^\ZZZ\} \cup \{e_j^{Z,r} \mid v_j^Z \in B_{i-1}^\ZZZ\} = \{e_j^{Z,l} \mid u_j^Z \in B_{i}^\ZZZ\} \cup \{e_j^{Z,r} \mid v_j^Z \in B_{i}^\ZZZ\},$$
            and in turn~$\mu(r(u_{i-1}) = \mu(l(u_i))$; the claim follows trivially.
        \end{itemize}

        Having discussed all possible cases we have concluded the proof of the claim.
    \end{ClaimProof}

     \begin{claim}
        The map~$\CC$ satisfies \WMVI.
    \end{claim}
    \begin{ClaimProof}
        This follows at once from \cref{obs:deserted_shores_for_I(L)} and the \cref{def:mu_for_I(L)} of~$\mu$ for the shores~$\{u_0',\ldots,u_\xi'\}$.
    \end{ClaimProof}

Combining all of the above claims we have verified all the constraints in the \cref{def:weak_coastal_map} for weak coastal maps, concluding the proof.
\end{proof}

This concludes the proof of \cref{thm:structure-thm-to-coastal-maps} and with it this sub-section.

We will leave the realm of coastal maps for now and get back to them in \cref{sec:shippings}---in said section we will make heavy use of the key result of the following section.

\section{The Structure of Minimal Counterexamples}
\label{sec:structure_of_min_examples}

This section is devoted to the structural analysis of minimal counter-examples to the \emph{irrelevant cycle theorem}, see \cref{thm:irrelevant_cycle} in \cref{sec:irrelevant_cycle_theorem}. The proof of the irrelevant cycle \cref{thm:irrelevant_cycle} uses the \emph{minimal counterexample technique}: we will show that a minimal counterexample to the theorem either contains a large router or admits a coastal map with a large Euler-embedded flat swirl, both of which contain irrelevant cycles, contradicting the counter-example. To this extent it is crucial to understand the structural properties of such a minimal counterexample.

Before stating and proving the main \cref{thm:irrelevant_cycle_minimal_counterexample} of this section, we give some definitions and observations which allow us to reduce part of its proof (as well as the base-case of \cref{thm:shipping_in_open_sea} proved in the next \cref{sec:shippings}) to an already known \cref{thm:directed_irr_vertex_Marx}.

\begin{definition}[Line-graph]
    Let~$G$ be an Eulerian graph. The \emph{line-graph~$\LL_G$ of~$G$} is defined as follows:~$V(\LL_G) \coloneqq E(G)$ and there exists an edge~$(e_1,e_2) \in E(\LL_G)$ if and only if~$(e_1,e_2)$ is a~$2$-path in~$G$.
    \label{def:linegraph}
\end{definition}
The following is an easy observation.

\begin{observation}
    Let~$\mathcal{S}=\bigcup_{i=1}^s S_i$ be an~$s$-swirl for some~$s \in \N$ Euler-embedded in some disc~$\Delta$. Then~$\LL_{\mathcal{S}}$ is again an~$s$-swirl that can be planar embedded in the same disc, where the resulting swirl-cycles are pairwise vertex-disjoint. Moreover, for every~$S_i$ the respective cycle in~$\LL_{\mathcal{S}}$ is itself a vertex-disjoint cycle.
    \label{obs:linegraph_of_swirl_is_swirl}
\end{observation}
\begin{proof}
    Since~$\mathcal{S}$ is Euler-embedded, the drawing of~$\LL_\SSS$ is easily seen to be planar embeddable too. Finally the line-graph of each cycle~$S_i$ is again a cycle~$S_i'$ which can be drawn with the same orientation in~$\Delta$; the edges of the cycle become vertices and they have the same cyclic connection, i.e., the~$k$-cycle~$(e_1,e_2,\ldots,e_k,e_1)$ in~$G$ remains a cycle in~$\LL_G$ where~$e_i$ and~$e_{i+1}$ (as vertices in~$V(\LL_G)$) are connected by an edge by definition for every~$1 \leq i \leq k \in \N$. Since the edges in~$\SSS$ become vertices in~$\LL_{\SSS}$, the pairwise edge-disjoint cycles~$S_i,\:S_j$ become pairwise vertex-disjoint in~$\LL_{\SSS}$ for every~$1 \leq i\neq j \leq s$. Clearly, since~$S_i$ viewed as a cycle is itself edge-disjoint (the cycle has no repeating edges), the respective cycle in~$\LL_{\SSS}$ is a vertex-disjoint cycle by definition.
\end{proof} 

\subsection{On swirls in minimal counterexamples}\label{subsec:swirl_ins_min_counterex}
In this subsection we prove the main result of this section; a result extensively used in the proof of the irrelevant cycle \cref{thm:irrelevant_cycle} given in \cref{sec:irrelevant_cycle_theorem}. Most notably it will allow us to reduce the general case of a linkage~$\LLL$ in~$G+D$ to the case that~$\LLL$ is rigid, which is a very strong assumption for which we will deploy machinery in the upcoming \cref{sec:shippings,sec:structure_thms} relying on the results provided in \cref{subsec:rigid-linkages}. Further we use it to prove \cref{thm:shipping_in_open_sea} which is the base-case for \cref{thm:no_rigid_linkages_on_strong_maps} which in turn is needed for the proof of \cref{thm:no_rigid_linkages_on_weak_maps}. To this extent we define a notion closely related to swirls: \emph{insulations}.

\begin{definition}[Insulation]
    Let~$\Gamma+D$ be an Eulerian graph such that~$\Gamma$ is Euler-embedded in some surface~$\Sigma$ and let~$S=\bigcup_{i=1}^{2h}C_i \subseteq \Gamma$ be a collection of concentric edge-disjoint cycles embedded in a disc~$\Delta \subset \Sigma$ such that~$C_i$ and~$C_{i+1}$ have alternating orientation (with respect to the orientation of~$\Delta$) for every~$1 \leq i < 2h$ and such that~$\Delta_1 \subseteq \Delta_2 \ldots \subseteq \Delta_{2h}$ where~$\Delta_i \subset \Delta$ is a disc bounded by the outline of~$C_i$ in~$\Sigma$ containing~$C_i$ but no edge of~$C_{i+1}$. Let~$D' \subset \Gamma[\Sigma \setminus \Delta_{2h}]$. Then an Eulerian subgraph~$G' \subseteq \Gamma[\Delta_1]$ is called~\emph{$h$-insulated} from~$V(D')$. We call~$S$ an~\emph{$h$-insulation}.
    \label{def:insulation}
\end{definition}
\begin{remark}
    Note that a priori the curves traced by the cycles~$C_i$ do not necessarily bound single discs themselves, for they are only edge-disjoint cycles and thus may still have vertices of degree four---a nuisance that we will get rid of shortly. But the outline of each cycle does using \cref{obs:faces_in_Euler_embeddings_bounded_by_cycle} since the insulation is Euler-embedded in a disc and thus the embedding is~$2$-cell.
\end{remark}

 Note that~$h$-insulations are special cases of embedded swirls (compare to \cref{obs:euler_embedded_swirl_in_disc_is_insulation}); we will still refer to them as \emph{$h$-insulations} for it matches the idea of~$h$-insulated graphs and sticks to the terminology of the respective results in the literature \cite{GMXXI, CyganMPP2013}. The following is an easy observation, following from the degree being bounded by four and the embedding restriction imposed by the discs~$\Delta_1\subset \ldots \subset \Delta_{2h}$.
\begin{observation}
    Let~$S = \bigcup_{i=1}^{2h}C_i$ be an~$h$-insulation drawn in some disc~$\Delta$ such that every vertex~$v \in V(S)$ is of degree at most four. Then~$V(C_i) \cap V(C_{i+2}) = \emptyset$ for~$1\leq i \leq 2h-2$.
    \label{obs:insulations_dont_have_same_orient_touchpoints}
\end{observation}

We are ready to prove the main theorem of this section, the proof of which follows a similar line of interest as the proof of \cite[Theorem 3.1]{GMXXII}, although the arguments made are very different and of their own interest for they are tailored to edge-disjoint paths and Eulerian graphs. 

\begin{theorem}
    For every integer~$p \geq 0$ and every function~$\xi:\N \to \N$ there exist functions~$\chi,\,h:\N \to \N$ with the following property. Let~$G+D$ be an Eulerian graph of maximum degree four such that~$|E(D)| \leq p$ and every vertex in~$V(D)$ has degree two in~$G+D$. Let~$\Gamma,K$ be subgraphs of~$G$ with~$G = \Gamma \cup K$ where~$\Gamma$ is Eulerian with~$V(D) \cap V(\Gamma) = \emptyset$ and~$\Gamma$ is Euler-embedded in some surface~$\Sigma$. Let~$C_0 \subset \Gamma$ be an Eulerian subgraph~$h$-insulated from~$V(\Gamma \cap K)$. Denote the~$h(p)$-insulation by~$S = \bigcup_{i=1}^{2h}C_i$. Assume further that subject to the above~$G+D$ is minimal satisfying the following: 
    \makeatletter
    \tagsleft@false
    \makeatother
   \begin{align*}
        \parbox{0.9\textwidth}{
      Let~$\LLL$ be a~$p$-linkage in~$G$ with pattern~$D$ such that there exists no other linkage~$\LLL'$ in~$G- C_0$ with the same pattern.} \tag{$\star$}
   \end{align*}
     \makeatletter
    \tagsleft@true
    \makeatother
     Then, writing~$h = h(p)$, the following hold true:
        \begin{enumerate}
            \item[1.]~$V(G) = V(D) \cup V(S) \cup V(C_0)$, where~$S \cup C_0$ is a flat~$(2h+1)$-swirl and every cycle of~$C_0,\ldots,C_{2h}$ is vertex-disjoint.
            \item[2.]~$G +D =\mathcal{I} \cup S \cup C_0$, where each of~$\mathcal{I},\:S$ and~$C_0$ are edge-disjoint Eulerian graphs. Further~$\mathcal{I}$ consists exactly of the edges with both endpoints in~$V(C_{2h})$ disjoint from~$E(C_{2h})$, the edges of~$D$ and edges with one endpoint in~$V(D)$ and one end-point in~$V(C_{2h})$; in particular~$E(\SSS) \cap E(D) = \emptyset$. 
            \item[3.]~$\LLL$ is rigid, and does not contain any path~$L\in \LLL$ such that~$L$ has a sub-path of length two in~$\mathcal{I}$ or any~$C' \in \{C_0,C_1,\ldots,C_{2h}\}$.
            \item[4.] Let~$S' \coloneqq C_{2h}\cup \ldots \cup C_{2h+3 -\chi(p)}$, then~$\tw(S') \geq \chi(p)$ and~$S'$ contains a flat-swirl grasped (and induced) by a tile of size~$\xi(p)$.
        \end{enumerate}
    \label{thm:irrelevant_cycle_minimal_counterexample}
\end{theorem}
\begin{proof}
    Let~$\chi(p) \coloneqq 6\xi(p) + 3$ and let~$h(p) \geq d_{\ref{thm:directed_irr_vertex_Marx}}(\chi(p)) + \chi(p)$. We claim that~$h(p)$ satisfies the theorem. For suppose it does not and let~$G+D$ be a minimal counterexample with respect to~$\Abs{E(G)}+\Abs{V(G)}$. That is, there exist~$\LLL,\Gamma,K,C_0,S$ as in the theorem where~$S = \bigcup_{i=1}^{2h}C_i$ is the respective collection of cycles in the~$h$-insulation and~$\LLL$ satisfies~$(\star)$. Let~$\Delta_i$ be the discs defined as in \cref{def:insulation} and recall that~$G - \bigcup \LLL$ is Eulerian. 
    \begin{claim}
   $V(G) = V(D) \cup V(S) \cup V(C_0)$.
        \label{thm:irrelevant_cycle_claim1}
    \end{claim}
    \begin{ClaimProof}
        Assume there is~$v \in V(G) \setminus \big(V(S)\cup V(D)\cup V(C_0)\big)$. Let~$e \in E(G)$ be an edge using~$\tail(e)=v$; clearly~$e \notin E(D) \cup E(S) \cup E(C_0)$. We show how to split off at~$v$ along~$e$ by respecting the linkage~$\LLL$: If~$e \in E(L)$ for~$L \in \LLL$ then there exists~$e' \in E(L)$ with~$\head(e')=v$---by assumption~$v \notin V(D)$ and thus said edge exists---and thus we can split off along~$(e',e)\subset L$, again noting that~$e' \notin (E(D)\cup E(S) \cup E(C_0))$.
        Otherwise, since~$G - \bigcup \LLL$ is Eulerian, there exists~$e' \in E(G)\setminus \big(E(D)\cup E(S) \cup E(C_0)\big)$ with~$\head(e') =v$ such that~$e,e'$ are both not part of~$E(\bigcup\LLL)$ and thus we can split off along~$(e',e)$. Call the resulting graph~$G'$ and let~$\LLL'$ be the resulting linkage in~$G'$. Then~$\LLL'$ is still a linkage with pattern~$D$ by \cref{lem:splitting_off_linkages_remains_rigid}. Further let~$\Gamma' = G[\Delta_{2h}]$ and~$K' = G' - \Gamma'$,~$C_0' = C_0$, and~$S'=S$. Then this yields a smaller counter-example since~$S' \cup C_0' \subseteq \Gamma'$ where~$S'$ is an~$h$-insulation and~$\Gamma'$ is still Euler-embedded for splitting off along two edges in an Euler-embedded graph~$\Gamma$ results in a new Euler-embedded graph; contradiction to the minimality.
    \end{ClaimProof}
    Often the argument that the remaining instance is still a counterexample to the theorem is analogous to the above, thus we will omit the presentation of~$\LLL',\Gamma',K',C_0',S'$ if they are clear and easily derivable from the context.
    \begin{claim}
        The linkage~$\LLL$ uses an edge of every cycle~$C_i$ of~$S$ for~$1 \leq i \leq 2h$, and it uses every vertex~$v \in V(G)$, i.e., the linkage is spanning.
\label{thm:irrelevant_cycle_claim2}
    \end{claim}
    \begin{ClaimProof}
    If the linkage~$\LLL$ does not use any edge at all of~$E(C_i)$ for some~$1 \leq i \leq 2h$, then we can delete~$C_i$ from the graph~$G$ without changing the pattern of the linkage. But now, using the fact that~$\Gamma$ is Euler-embedded which in turn implies \cref{obs:insulations_dont_have_same_orient_touchpoints},~$V(D)$ is not connected to~$G[\Delta_1]$ anymore since every vertex is part of at most two swirl-cycles using the fact that the maximum degree is four and thus there is no path left from~$C_{i+1}$ to~$C_{i-1}$ (or~$C_0$ if~$i=1$); a contradiction for there is no linkage with the same pattern as~$\LLL$ not visiting~$C_0$ by~$(\star)$. Next, assume towards a contradiction that some vertex~$v \in V(G)$ is not visited by~$\LLL$. If~$v \notin V(S)$ then we can split off along any two-path at~$v$, and if~$v \in V(S)$ then we can split off along a two-path~$(u,v,w)$ that is a sub-path of some cycle~$C_i$ for some~$1\leq i \leq 2h$ such that~$(u,v),(v,w) \notin E(\LLL)$ (which must exist for all edges incident to~$v$ are not visited by~$\LLL$ by assumption); a contradiction to the minimality, for splitting off along a cycle of~$S$ does not destroy the insulation, and splitting off away from all cycles of~$S$ and~$\LLL$ (and thus~$V(D)$) does not violate any of the assumptions on the counter-example.
    \end{ClaimProof}
    
    \begin{claim}
        There is no~$L \in \LLL$ such that~$L \cap C$ contains a sub-path of length~$\geq 2$ for any cycle~$C \in \{C_0,C_1,\ldots,C_{2h}\}$ or any cycle~$C \subset G - (S\cup C_0)$.
\label{thm:irrelevant_cycle_claim3}
    \end{claim}
    \begin{ClaimProof}
        Suppose there exists~$L \in \LLL$ such that~$(e_1,e_2) \subset L$ is a sub-path consisting of two edges~$e_1,e_2 \in E(C_i)$ for some~$i \in \{1,\ldots, 2h\}$. Then we can split off along~$(e_1,e_2)$ keeping the resulting~$C_i'$ a cycle and~$\Gamma'$ Euler-embedded and thus guaranteeing the hypothesis of the counter-example to remain valid. In particular the counterexample was not minimal with respect to~$\Abs{E(G)} + \Abs{V(G)}$. The other case~$C=C_0$ or~$C\subset G-(S\cup C_0)$ is analogous for no cycle of~$S$ was destroyed and the remaining assumptions on the counter-example stay valid.
    \end{ClaimProof}

    Using this we derive the following.
    \begin{claim}
        Let~$C \in\{C_1,\ldots,C_{2h-1}\}$ (in particular~$C \neq C_{2h})$ and let~$(v_1,\ldots,v_k)$ be an enumeration of the vertices visited in order by~$C$ (where~$v_i = v_j$ may happen for~$i\neq j$ since~$C$ is not necessarily vertex-disjoint). Then there exists no edge~$e=(v_i,v_j) \in E(G) \setminus E(C)$, with both endpoints in~$V(C)$ such that~$i<j$ (after a possible cyclic shift) and such that~$e + v_iCv_j$ is \emph{no} cycle.
        \label{thm:irrelevant_cycle_claim4}
    \end{claim}
    \begin{ClaimProof}
        Assume towards a contradiction that such an edge~$e=(v_i,v_j)$ exists; in particular~$e + v_iCv_j$ is \emph{no} cycle. First, since~$S$ (and~$G[\Delta_{2h}]$) is Euler-embedded in the disc~$\Delta\subset \Sigma$ it follows that~$e$ is drawn in the same disc~$\Delta$. We deduce that~$e \notin E(S)$, for else it must be part of a cycle~$C' \in \{C_1,\ldots,C_{2h}\}$ of different orientation than~$C$ by \cref{obs:insulations_dont_have_same_orient_touchpoints}, and thus, since every face in an Euler-embedded graph is bounded by a cycle,~$C'$ must lie completely in the region bounded by~$e$ together with~$C$, a contradiction to it being concentric to~$C$ by the \cref{def:insulation} of insulation. Hence, using Eulerianness,~$e$ must lie on some cycle~$\tilde{C}$ edge-disjoint from~$S$. If~$e \notin E(L)$ for any~$L \in \LLL$ then~$(v_{i},v_{i+1}) \in E(L)$ for some~$L \in \LLL$ by \cref{thm:irrelevant_cycle_claim2}, since the degree at~$v_i$ is four, and~$(v_i,v_{i+1})$ is the only other possible outgoing edge incident to~$v_i$. This implies that~$(v_{i-1},v_{i})\notin E(L)$ either for otherwise~$(v_{i-1},v_{i},v_{i+1}) \subset L \cap C$; contradiction to \cref{thm:irrelevant_cycle_claim3}. But then we get an~$h$-insulation~$S'$ where~$C$ is replaced by the cycle~$e + v_jCv_i$ containing~$((v_{i-1},v_i),e)$ as a sub-path along which we can split off contradicting minimality---note that we can slightly enlarge the disc~$\Delta_C$ bounding~$C$ to contain~$e$. Thus~$e \in E(L)$ for some~$L \in \LLL$. Then either~$((v_{i-1},v_i),e) \subset L$ or~$(e',e) \subset~L$ with~$(e',e) \subseteq \tilde{C}$. In either case we can split off along both edges in~$L$---the first case in analogous to the just described case and the latter case is \cref{thm:irrelevant_cycle_claim3}---without destroying any cycles in~$S$ (or~$S'$) and keeping the graph Euler-embedded (for we split off along an Euler-embedded path) and thus contradicting the minimality of the counterexample once again.
    \end{ClaimProof}

    \begin{claim}
    Let~$e \in E(G)$ such that~$e$ is drawn in~$\Delta_{2h}\setminus \Delta_1$. Then~$e \in E(S)$.
\label{thm:irrelevant_cycle_claim5}
    \end{claim}
    \begin{ClaimProof}
        Assume otherwise, then since~$G$ is Euler-embedded it holds true that~$e$ is drawn inside some~$A_i \coloneqq \Delta_{i+1}\setminus \Delta_{i}$ bounding the respective cycles~$C_{i+1}$ and~$C_i$ as in \cref{def:insulation}. Since~$G[A_i]$ is Euler-embedded it follows that~$e$ is part of a cycle~$C$ in~$G[A_i] - (C_{i+1} \cup C_{i})$; let~$C \subset G[A_i]$ be an edge-maximal cycle containing~$e$. By \cref{thm:irrelevant_cycle_claim1} the endpoints of~$e$ lie in~$V(C_i)\cup V(C_{i+1})$. If~$C$ is not contractible in the annulus~$A_i$---it has winding number larger than~$0$---then it has the same orientation as one of the two cycles~$C_i,\,C_{i+1}$. Note that~$C$ `bounds' a disc~$\Delta_C$ with~$\Delta_i \subseteq \Delta_C \subseteq \Delta_{i+1}$ (by following its outline, see \cref{def:outline_of_cycle}, for it is itself not necessarily vertex-disjoint). But in both cases either~$G[\Delta_C]$ is a component disjoint from~$C_{i+1}$ or~$G[\Delta_i]$ is a component disjoint from~$C$. To see this suppose that~$C$ and~$C_i$ have the same orientation. Then using the Euler-embedding we deduce that~$V(C) \cap V(C_i) = \emptyset$ for else we find a strongly planar vertex (compare also \cref{thm:irrelevant_cycle_claim4}). Note that there cannot be any other edge with one endpoint in~$G[\Delta_C]$ and one endpoint in~$V(C_i)$ for else~$C$ was not maximal and could have been extended. 
        
        In either case then the linkage~$\LLL$ does not use edges of every cycle in~$S$---we just proved that the graph is disconnected---contradicting \cref{thm:irrelevant_cycle_claim2} (and the assumption of the theorem for~$\LLL$ cannot visit~$C_0$ as a contradiction to~$(\star)$). Thus~$C$ is contractible in~$A_i$. Then, using the Euler-embedding, \cref{thm:irrelevant_cycle_claim1} together with \cref{thm:irrelevant_cycle_claim4} and the fact that~$C$ is edge-maximal imply that~$C$ has vertices in both~$V(C_i)$ and~$V(C_{i+1})$. Finally we can reroute~$C_i$ at a common vertex with~$C$ so that it contains~$C$ (this results in an edge-disjoint cycle that may not be vertex-disjoint, i.e., the cycle contains a vertex of degree four). Since~$C$ is contained in the annulus~$A_i$ we can enlarge~$\Delta_{i}$ respectively to contain~$C$ and still be contained in~$\Delta_{i+1}$ as in \cref{def:insulation}; just follow the outline of the cycle given by \cref{def:outline_of_cycle} in the obvious way using the fact that the graph is Euler-embedded.
    \end{ClaimProof}
    
Finally we get rid of the slightly annoying degree four vertices.
    \begin{claim}
        There exists no vertex in~$C_i$ that is of degree four in~$C_i$ for any~$0 \leq i \leq 2h$.
        \label{thm:irrelevant_cycle_nodeg4}
    \end{claim}
    \begin{ClaimProof}
         Let~$v \in V(C_i)$ such that it is of degree four in~$C_i$. By \cref{thm:irrelevant_cycle_claim2} it follows that~$v \in V(L)$ for some~$L \in \LLL$. But then we immediately get that~$\LLL$ contains a two-path in~$C_i$; a contradiction to \cref{thm:irrelevant_cycle_claim3}.
    \end{ClaimProof}

    Similarly we get the following.

    \begin{claim}
        Let~$v \in V(C_0 \cup \ldots \cup C_{2h-1})$. Then~$v$ is of degree four in~$G$ and~$v$ is part of exactly two cycles of~$C_0,\ldots,C_{2h}$.
        \label{thm:irrelevant_cycle_nodeg2}
    \end{claim}
    \begin{ClaimProof}
        Assume that~$v$ is of degree two in~$G$. Then either both edges incident to~$v$ are not part of~$\bigcup \LLL$ or both edges are part of the same~$L \in \LLL$. The former is a contradiction to \cref{thm:irrelevant_cycle_claim2} and the latter is a contradiction to \cref{thm:irrelevant_cycle_claim3}. The rest follows from \cref{thm:irrelevant_cycle_nodeg4}.
    \end{ClaimProof}

    Thus combining \cref{thm:irrelevant_cycle_claim5}, \cref{thm:irrelevant_cycle_nodeg4} and \cref{thm:irrelevant_cycle_nodeg2} we deduce that~$S=C_1\cup\ldots\cup C_{2h}$ where all the~$C_i$ are \emph{vertex-disjoint} cycles, i.e., connected subgraphs where all the vertices are of degree two. Further, since~$G[\Delta_{2h}]$ is Euler-embedded, it follows that~$C_0$ has different orientation than~$C_1$ and since it is itself again vertex-disjoint by \cref{thm:irrelevant_cycle_claim5} and concentric with~$C_i$ for all~$1 \leq i \leq 2h$ we deduce that~$S \cup C_0$ is a~$(2h+1)$-swirl with vertex-disjoint cycles~$C_0,\ldots,C_{2h}$. Moreover~$S\cup C_0$ is flat which follows at once from the embedding being Eulerian and plane in~$\Delta$.
    
    Note that there may be edges left away from~$E(S \cup C_0)$ with one or both ends in~$V(C_{2h})$. Since~$S\cup C_0$ is Eulerian, we derive that there is an Eulerian subgraph~$\mathcal{I}$ disjoint from~$S\cup C_0$ with edge-set~$E_{\mathcal{I}}=E(G+D) \setminus E(S \cup C_0)$. Combining both observations we derive the following.
    \begin{claim}
        The embedded graph~$S\cup C_0$ is a flat~$(2h+1)$-swirl where~$C_0,\ldots,C_{2h}$ are each vertex-disjoint cycles that are pairwise edges-disjoint. Further~$G+D = S\cup C_0 \cup \III$.\label{thm:irrelevant_cycle_claim6}
    \end{claim}
\begin{ClaimProof}
    This follows at once from the above together with the fact that~$E(C_0) = E(G[\Delta_1])$ for~$G[\Delta_1]$ is Eulerian and we have a minimal counterexample.
\end{ClaimProof}
    We continue with proving that~$\LLL$ is exhaustive.
    \begin{claim}
 ~$\LLL$ is exhaustive, i.e.,~$E(G) = E(\LLL)$.
    \label{thm:irrelevant_cycle_claim7}
    \end{claim}
    \begin{ClaimProof}
        Assume towards a contradiction that there exists some edge~$e \in E(G) \setminus E(\LLL)$. Using \cref{thm:irrelevant_cycle_claim1} and \cref{thm:irrelevant_cycle_claim6} we deduce that~$e\in E(S)\cup E(C_0) \cup E(\mathcal{I})$. In particular~$e \in E(C)$ for some cycle~$C$---either~$C = C_i$ for some~$0 \leq i \leq 2h$ or~$C \subset \mathcal{I}$. Clearly~$e$ has no end-point in~$V(D)$ for every vertex in~$V(D)$ is of degree one in~$G$ and part of~$\LLL$. Let~$(e_1,e,e_2)\subseteq C$ be the sub-path of the cycle~$C$ containing~$e$, where a priori each of the equalities~$e_1 = e_2 = e$ are possible.
        
        First we see that either~$e_1 \neq e$ or~$e_2 \neq e$, for otherwise~$e$ is a loop and thus~$C$ consists of a single edge. If~$C=C_0$ then~$\LLL$ would not use~$C_0$; a contradiction to the assumption~$(\star)$ of the theorem. If~$C \subset \mathcal{I}$, then, since it is a loop, it must be disjoint from~$E(D)$ and~$E(\pi(\LLL))$ and thus we could delete it to get a smaller counter-example. Thus~$C = C_i$ for some~$1\leq i \leq 2h$. Let~$v$ be the single vertex of~$e$, then~$v$ must have degree four (otherwise~$C$ is a component disjoint from~$V(D)$) and since we are Euler-embedded, \cref{obs:insulations_dont_have_same_orient_touchpoints} implies that either~$\LLL \cap G[\Delta_{i}] = \emptyset$ or~$\LLL\cap G[\Sigma\setminus \Delta_i] = \emptyset$, for~$V(C_i) = \{v\}$ and~$v$ is part of at most two cycles of~$S$; the first implies~$\LLL\cap C_0 = \emptyset$ which is a contradiction to the assumption~$(\star)$ of the theorem on~$\LLL$, the latter is impossible for then~$\LLL$ would be disjoint from~$V(D)$. We conclude that~$|\{e_1,e,e_2\}| \geq 2$; without loss of generality~$e_1 \neq e$, else rename them so that~$e=e_2$. 

        Assume now that~$e_1 \notin E(\LLL)$. Then, in the spirit of \cref{thm:irrelevant_cycle_claim3}, we can split off along~$(e_1,e)$---a two-path that is disjoint from~$E(\LLL)$---and get a smaller counterexample as the hypothesis to our counterexample is still valid: no insulation-cycle~$C_1,\ldots,C_{2h}$ was destroyed for~$e_1,e$ are consecutive edges on some cycle~$C$ either in~$\{C_0,\ldots,C_{2h}\}$ or a subset of~$\mathcal{I}$, and again we remain Euler-embedded by spitting off along a two-path in an Euler-embedding. We deduce that~$e_1 \in E(L)$ for some~$L \in \LLL$. But then, since~$G - \bigcup \LLL$ is Eulerian, we may extend~$L$ using the cycle~$C_e \in G-\bigcup \LLL$ containing~$e$ (which is disjoint from~$E(D)$). 

        The new linkage~$\LLL'$ is again a witness to the counterexample, thus satisfying all of the above claims. But this is a contradiction similar to the one in \cref{thm:irrelevant_cycle_nodeg4}: By \cref{thm:irrelevant_cycle_claim1} there is~$v \in V(S \cup C_0)$ (since~$V(D)$ has only degree-two vertices) such that~$v$ is of degree four in~$G$ and~$v \in L$ is of degree four in~$L$ (coming from the cycle~$C_e$). Then in either case~$L$ contains some two-path (using~$v$) that is a sub-path of some~$C \in \{C_0,\ldots,C_{2h}\}$ as a contradiction to \cref{thm:irrelevant_cycle_claim3}.
    \end{ClaimProof}
    
    Using the above we get to the main claim.
    \begin{claim}
  ~$\LLL$ is rigid.
    \label{thm:irrelevant_cycle_claim8}
    \end{claim}
    \begin{ClaimProof}
        Assume the contrary and let~$\LLL'$ be another linkage with the same pattern. By \cref{thm:irrelevant_cycle_claim7} we deduce that~$E(\LLL') = E(G) = E(\LLL)$, for~$\LLL'$ satisfies all of the above claims too, for otherwise we find a smaller counter-example. Let~$L' \in \LLL'$ and~$L \in \LLL$ with pattern~$\pi(L') = \pi(L) = (e_1,e_2)$ such that the paths do not agree, i.e., they have some edge that they do not share. By \cref{thm:irrelevant_cycle_claim1} \cref{thm:irrelevant_cycle_claim3} and \cref{thm:irrelevant_cycle_claim6} we deduce that~$e_1 = (s,v)$ and~$e_2 = (v',t)$ for some~$v,v' \in C_{2h}$ and~$(t,s) \in E(D)$. Let~$e=(u,w) \in L$ be the first edge with respect to the path~$L$ such that~$e \notin L'$, in particular~$e \notin \{e_1,e_2\}$ and similar~$e \notin E(\pi(\LLL))$ in general. Let~$e'$ be the preceding edge in~$L$ thus~$e' \in L\cap L'$ By \cref{thm:irrelevant_cycle_claim6} we deduce~$e' \in C'$ for some cycle~$C' \in \{C_0,C_1,\ldots,C_{2h}\}$ or~$C' \subset \mathcal{I}$. Let~$u$ be the vertex incident to~$e'$ and~$e$. Then there exists only one more edge~$e_o\in E(G)$ which has~$u$ as a tail. Now~$e\notin E(C')$ for else we could split~$\LLL$ along~$(e',e)$ contradicting \cref{thm:irrelevant_cycle_claim3}. This implies however that~$e_o \in E(C')$ and since~$L$ and~$L'$ do not agree on~$e$ and~$u$ has degree four it follows that~$e_o \in L'$ and we may split off~$L'$ along~$(e',e_o)$; a contradiction to \cref{thm:irrelevant_cycle_claim3}.
    \end{ClaimProof}

    Using that~$\LLL$ is rigid we could continue the proof along the same lines as in the undirected case, i.e., the proof of \cite[Theorem 3.1]{GMXXII}. The following argument is different and a little shorter, though relying on a rather strong \cref{thm:directed_irr_vertex_Marx} proved in \cite{CyganMPP2013}.

    \begin{claim}
        There exist at least~$\chi(p)$ components of~$\LLL$ in~$\LLL \cap G[\Delta_{h(p)-\chi(p)}]$, i.e.,~$\restr{\LLL}{G[\Delta_{h(p)-\chi(p)}]}$ is a~$\geq \chi(p)$-linkage.
\label{thm:irrelevant_cycle_claim9}
    \end{claim}
   \begin{ClaimProof}
        If not then this means that~$\restr{\LLL}{G'}$ is a rigid~$\chi(p)$-linkage in~$G' \coloneqq G[[\Delta_{h(p)-\chi(p)}]]$ using \cref{lem:cuts_and_rigid_linkages}, where~$G[\Delta_{h(p)-\chi(p)}] \subset G'$ and thus~$G'$ contains an~$h'(p)$-insulation for~$h'(p) \coloneqq h(p)-\chi(p)$. 
        
        Let~$\LL_{G'}$ be the line-graph of~$G'$ as in \cref{def:linegraph}. Then, using that~$G'$ has no strongly planar vertex---for~$G' \subset S\cup C_0$ where~$S\cup C_0$ was Euler-embedded---the line-graph~$\LL_{G'}$ can be planar embedded maintaining~$h'(p)$ cycles of alternating orientation, i.e.,~$\LL_{G'}$ contains an~$h'(p)$-insulation as in \cref{obs:linegraph_of_swirl_is_swirl}; recall that every cycle in~$G'$ is a cycle in the line-graph. One easily sees that the line-graph~$\LL_{C_0}$ of~$C_0 \subset G'$ is~$h'(p)$-insulated from the terminals of~$\LLL'$ being the unique vertex-disjoint linkage induced by~$\restr{\LLL}{G'}$ in~$G'$. 
        
        Finally using~$h'(p)\geq f_{\ref{thm:directed_irr_vertex_Marx}}(\chi(p))$, \cref{thm:directed_irr_vertex_Marx} implies that there exists an irrelevant vertex~$v_e \in V(\LL_{C_0})$ to the linkage~$\LLL'$, i.e., going back to~$G'$ this means that there exists an irrelevant edge~$e \in E(C_0)$ to~$\restr{\LLL}{G'}$ (note that the terminals of the vertex-disjoint paths in~$\LLL'$ are exactly the directed patterns of the respective paths in~$\restr{\LLL}{G'}$). Thus we can reroute the linkage inside~$G'$ omitting this edge (or even the most deeply nested cycle). Switching back to the standard graph we have found a new linkage with the same pattern contradicting the rigidity of~$\restr{\LLL}{G'}$ (using \cref{lem:switching_linkages_at_cuts}). 
   \end{ClaimProof}
   This immediately implies the theorem for the first~$\chi(p)$ cycles are pairwise edge-disjoint and the~$\chi(p)$ components of~$\LLL$ are pairwise edge-disjoint, witnessing that~$G+D$ has undirected tree-width at least~$\chi(p)-3$ (see \cite[Theorem 3.2]{GMXXII}) and thus directed tree-width at least~$\frac{1}{6}(\chi(p)-3)=\xi$ by \cref{thm:undirected_vs_directed_tw_in_Eulerian_graphs}. Further, since every vertex is of degree four in~$S$, the disjoint cycles and components immediately provide a witness for the tile that induces a swirl similar to the proof of \cref{obs:subtiles_of_grasping_tiles_induce_swirls}. 
\end{proof}

The most important part of the above theorem needed in the proof of \cref{thm:irrelevant_cycle} is that, taking a minimal counter-example, the linkage~$\LLL$ can be assumed to be \emph{rigid}.

However, it turns out that we can say even more about the graph~$G$ of the minimal counter-example analysed above. Throughout the remainder of this section let~$G$ be as in \cref{thm:irrelevant_cycle_minimal_counterexample}. Then the following lemmas hold true. Note that all the claims referenced in the following proofs are the claims in the proof of \cref{thm:irrelevant_cycle_minimal_counterexample}---we outsourced the lemmas for ease of readability and to mark that the above theorem is the most crucial part to \cref{thm:irrelevant_cycle}.

\begin{lemma}
  Assume that~$\Abs{V(C_{2h})}\geq 4$ and let~$\{u_1,\ldots,u_n\} = V(C_{2h})$ be an enumeration of~$C_{2h}$ such that~$u_1,\ldots,u_n$ are visited in that order for~$n \coloneqq \Abs{V(C_{2h})}$. Then there exists no edge~$e \in E(\mathcal{I})$ with~$e=(u_i,u_{i+1})$ or~$e=(u_{i+1},u_i)$ for some~$1 \leq i \leq n$ (where addition is to be taken cyclically). \label{lem:minimal_counterexample_one}

\end{lemma}
\begin{proof}
    Assume the contrary; there is two cases to consider. First, without loss of generality, assume that~$e=(u_2,u_3)$ (up to a cylic rotation of the indices). Then we can draw~$e$ such that~$e+u_2C_{2h}u_3$ bounds a face after extending the drawing of~$\Gamma$ by~$e$, and~$e+u_2C_{2h}u_3$ is \emph{not} a cycle. In particular~$e$ is homotopic to the path~$u_2C_{2h}u_3$ (it can be continuously deformed to it when seen as a curve on the surface). Now since~$\LLL$ is rigid there exists~$L \in \LLL$ with~$e\in L$. Further using \cref{thm:irrelevant_cycle_claim3} we deduce that the edge~$(u_1,u_2) \in C_{2h}$ satisfies~$(u_1,u_2) \in E(L)$, noting that~$((u_1,u_2),e)$ is a two-path. We can now split off along~$((u_1,u_2),e)$ maintaining an~$h$-insulation and keeping the graph Euler-embedded as a contradiction to the minimality of~$G$.

    Thus assume that~$e=(u_2,u_1)$ (up to a cyclic shift). Then~$e+ u_1C_{2h}u_2$ forms a vertex-disjoint cycle bounding a face after extending the embedding of~$\Gamma$ by~$e$. But this is impossible for by \cref{thm:irrelevant_cycle_claim3} we know that~$((u_2,u_1),e) \subset L \in \LLL$ and thus~$L$ contains a loop and in particular it contains all the edges incident to~$u_1$ (and~$u_2$); a contradiction to \cref{thm:irrelevant_cycle_nodeg4} or more directly to \cref{lem:rigid_implies_no_intersection} by \cref{thm:irrelevant_cycle_claim9}, for rigid linkages are vertex-disjoint.
\end{proof}

\begin{remark}
    The assumption that~$\Abs{V(C_{2h})} \geq 4$ is not restrictive since every terminal vertex~$v\in V(D)$ has an edge with one end-point on~$V(C_{2h})$ and the cycle~$C_{2h-1}$ must share at least one common vertex with~$C_{2h}$. If~$p\coloneqq \Abs{E(D)} \geq 2$ then~$C_{2h}$ must have at least two distinct vertices in common with~$V(C_{2h-1})$ and have at least two vertices disjoint from~$V(C_{2h-1})$ connected to terminal vertices, thus the assumption is already satisfied for~$p \geq 2$.
\end{remark}

The following lemma turns out to be a rather strong tool that implies a rather rigid behaviour of the linkage~$\LLL$ once it enters the swirl. 
\begin{lemma} \label{lem:minimal_counterexample_no_bumps}
    Let~$e =(u,v) \in E(C_i)$ for some~$0 \leq i \leq 2h-1$. Then, either~$u \in V(C_{i+1})$ or~$ v \in V(C_{i+1})$.
\end{lemma}
\begin{proof}
For~$i = 0~$ it is obvious---otherwise the edge is solely contained in~$C_0$ but then both vertices are either of degree two as a contradiction to \cref{thm:irrelevant_cycle_claim9} and \cref{thm:irrelevant_cycle_claim3}, or they are of degree four in~$C_0$ as a contradiction to \cref{thm:irrelevant_cycle_nodeg4}.

\smallskip

For~$i>0$, \cref{thm:irrelevant_cycle_nodeg2} implies that~$u,v \in V(C_{i-1}) \cup V(C_{i+1})$. We ought to prove that not both are contained in~$V(C_{i-1})$. We prove it for~$i=1,2$ and then by induction for~$i \geq 3$.

 \smallskip
 
Let~$i=1$ and assume towards a contradiction that~$u,v \in V(C_0)$. Using the Euler-embedding then~$(u,v) + vC_0u$ form a cycle and since~$C_0$ and~$C_1$ are both vertex-disjoint cycles by \cref{thm:irrelevant_cycle_claim6} it bounds a face~$F$. If there is a vertex~$w \in V(C_0)$ with~$v,w,u$ visited in that order along~$C_0$, then~$w\notin V(C_1)$ by the restricted embedding ($C_1$ has no edge inside~$F$ for it is a face). Then~$w$ is solely part of~$V(C_0)$ as a contradiction to either it being degree four as given by \cref{thm:irrelevant_cycle_nodeg2} or~$C_0$ being a vertex disjoint cycle \cref{thm:irrelevant_cycle_nodeg4}; see the illustration in \cref{fig:min_count_base} on the left.

\smallskip

For ease of argumentation we prove another special case, namely~$i=2$ which is basically the model for the remaining cases. So towards a contradiction assume that~$u,v \in V(C_1)$. Then there exists~$L \in \LLL$ that uses~$(u,v) \in E(C_2)$ by \cref{thm:irrelevant_cycle_claim9} and thus by \cref{thm:irrelevant_cycle_claim3} together with the fact that~$u,v \in V(C_1)$ we deduce that~$(u',u,v,v') \subset L$ must be a sub-path with~$(u',u),(v,v') \in E(C_1)$. Clearly then~$u' \neq v'$ else we get a contradiction to \cref{thm:irrelevant_cycle_claim3} again for~$L$ would contain a two-path of~$C_1$. Using the Euler-embedding of the graph and the fact that the cycles are concentric we derive that~$u',v' \in V(C_{0})$. To see this note that otherwise one of both, call it~$x' \in \{u',v'\}$ were part of~$V(C_2)$ by \cref{thm:irrelevant_cycle_nodeg2} and since~$(u,v) \in E(C_2)$ it holds that~$u,v,x'$ are visited in that order by~$C_2$ (up to a cyclic rotation). In either case we get a region bounded by either~$(u',u)$ or~$(v',v)$ together with part of~$C_2$ which must contain another edge of~$C_1$, a contradiction to~$C_1$ and~$C_2$ being concentric vertex-disjoint cycles. 

But then, using that~$C_1$ is a cycle we deduce that there is some path~$P$ connecting~$v'$ to~$u'$ that must be drawn inside the region bounded by~$(u',u) + (u,v) + (v,v') + u'C_{0} v'$. In particular~$P$ contains an edge with both endpoints on~$C_0$; contradiction to the case~$i=1$, see the middle illustration in \cref{fig:min_count_2step}.

\smallskip

Finally, towards a contradiction assume that the lemma is wrong and let~$3 \leq i \leq 2h-1$ be minimal such that~$(u,v) \in E(C_i)$ with~$u,v \in V(C_{i-1})$. Then~$(u,v) + vC_{i-1}u$ forms a cycle by \cref{thm:irrelevant_cycle_claim4}. Again it forms a face for~$C_i$ is concentric with~$C_{i-1}$ and since they are Euler-embedded the edges of~$C_i$ visited prior and after~$(u,v)$ must lie outside the region bounded by~$(u,v) + vC_{i-1}u$. Let~$L \in \LLL$ be a path using~$(u,v)$. Then, since~$L$ contains no sub-path of length~$2$ of the same cycle by \cref{thm:irrelevant_cycle_claim3}, we know that there exist~$e_1,e_2 \in E(C_{i-1})$ such that~$(e_1,(u,v),e_2)$ is a sub-path of~$L$. Let~$e_1 = (u',u)$ and~$e_2 = (v,v')$. Clearly~$u' \neq v'$ for~$L$ must be a vertex-disjoint path by \cref{thm:irrelevant_cycle_claim9} together with \cref{lem:rigid_implies_no_intersection}. Let~$u'',v'' \in V(C_{i-1})$ be distinct with~$(u'',v'') \in E(C_{i-1})$ such that~$v',u'',v'',u'$ are visited in that order by~$C_{i-1}$ (where at least two of~$\{v',u'',u',v''\}$ are distinct by the above). By rigidity of~$\LLL$ there exists~$L'$ with~$(u'',v'') \in L'$ and again by \cref{thm:irrelevant_cycle_claim3} there must exist edges~$(u_1,u''),(v'',v_2) \in E(L) \cap E(C_{i-2})$ (for they cannot be with~$E(C_i)$ as~$u'',v''$ lie in the region bounded by~$e + v C_{i-1} u$ using the Euler-embedding). Finally since~$(u_1,u''), (v'',v_2) \in E(C_{i-2})$ we know that there must be some path~$P \subset C_{i-2}$ such that~$P$ connects~$v_2$ to~$u_2$ and~$P$ cannot have a vertex on~$u'' C_{i-1} v''$ for by the above choice the path~$u'' C_{i-1} v''$ is exactly the edge~$(u'',v'')$. Again using the embedding restriction we derive that~$P$ is drawn in the region bounded by the undirected cycle~$(u_1,u'') + (u'',v'') + (v'',v_1) + u_1C_{i-3}v_1$ (which exists for~$i \geq 3$). But then~$P$ is either the edge~$(v_1,u_1)$ or it must have another endpoint~$w \in V(C_{i-3})$ where~$u_1,w,v_1$ are visited by~$C_{i-3}$ in that order. Let~$w$ be chosen such that~$(v_1,w)$ is an edge of~$P$, then this edge has both endpoints on~$V(C_{i-3})$ and thus it satisfies the assumption of our claim; a contradiction to the minimality of~$i$. See the right-hand side of the illustration in \cref{fig:min_count_indstep} for a schematic representation of this step.
\end{proof}

\begin{figure}
\begin{subfigure}{.33\textwidth}
  \centering
  \includegraphics[width=\linewidth]{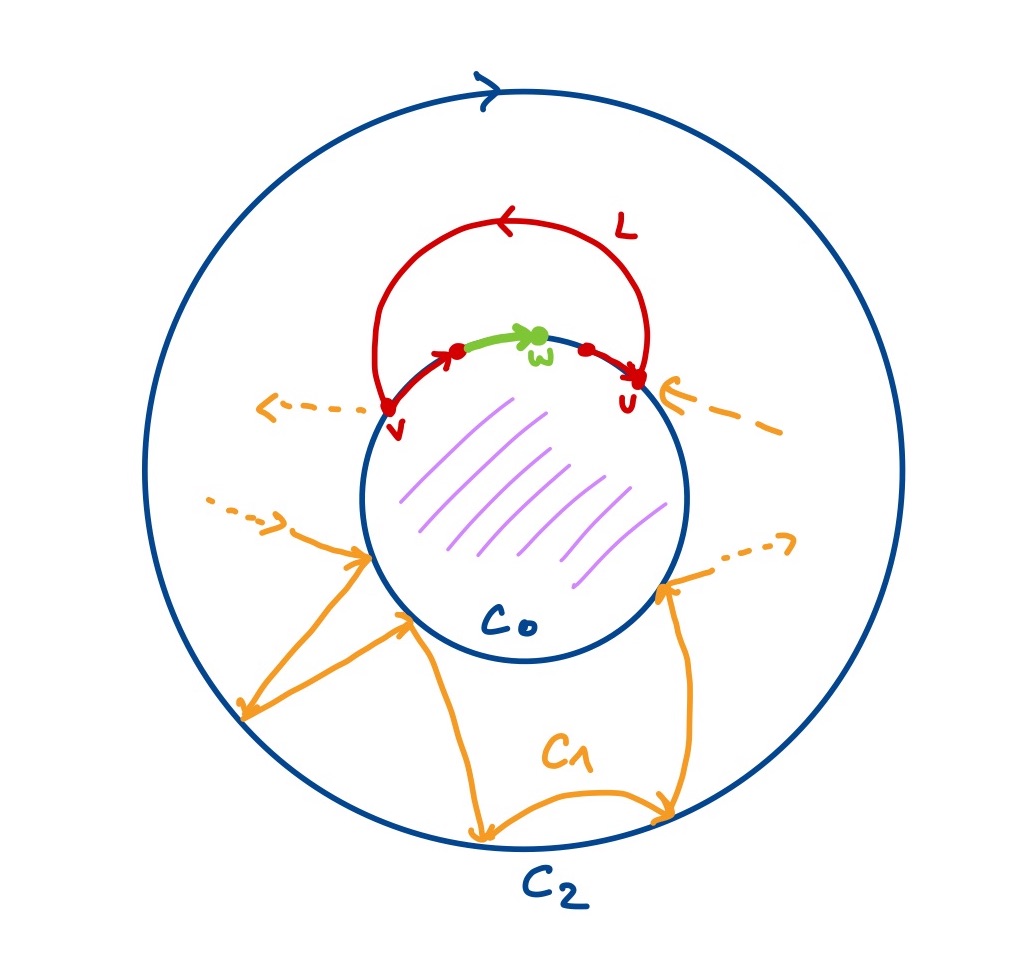}
  \caption{Case~$i=1$}
  \label{fig:min_count_base}
\end{subfigure}
\begin{subfigure}{.33\textwidth}
  \centering
  \includegraphics[width=\linewidth]{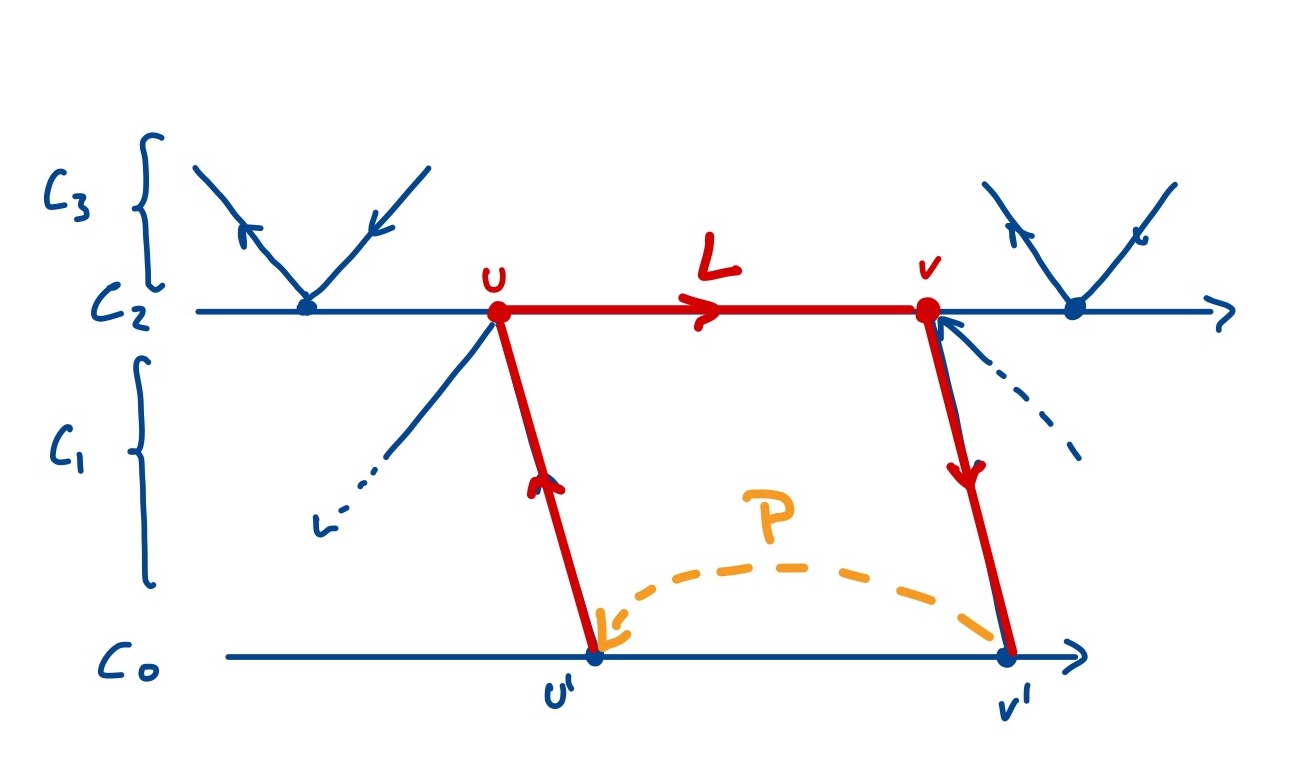}
  \caption{Case~$i=2$}
  \label{fig:min_count_2step}
\end{subfigure}
\begin{subfigure}{.33\textwidth}
  \centering
  \includegraphics[width=\linewidth]{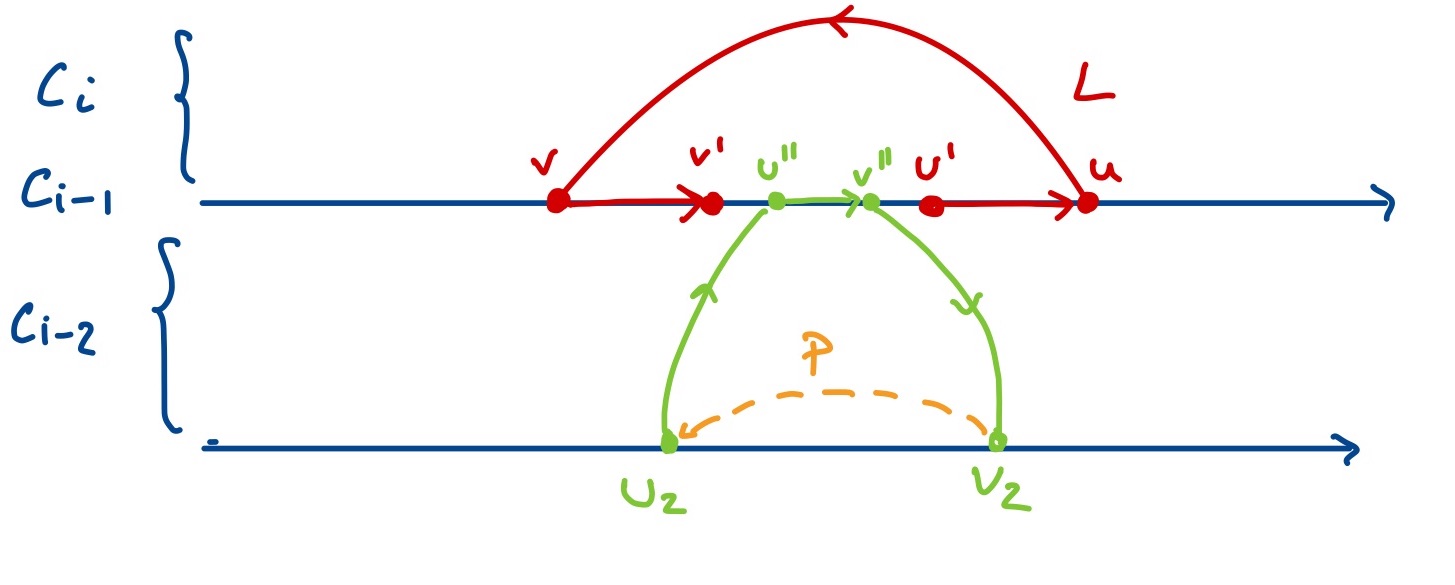}
  \caption{General case~$i\geq 3$}
  \label{fig:min_count_indstep}
\end{subfigure}
\label{fig:no_bumps_in_minimal_ex}
\caption{A schematic illustration of the different steps in the proof of \cref{lem:minimal_counterexample_no_bumps}}
\end{figure}

Let~$\SSS = \bigcup_{i=0}^{2h}C_i$. Using \cref{lem:minimal_counterexample_no_bumps} we get the following crucial observation. 

\begin{lemma}
    Let~$L\in \LLL$ be a path where~$\LLL$ is as in \cref{thm:irrelevant_cycle_minimal_counterexample}. Then there exists no sub-path~$(e_i,e_{i+1},e_i') \subset L$ such that~$e_i,e_i' \in E(C_i)$ and~$e_{i+1} \in E(C_{i+1})$ for~$1 \leq i \leq 2h-1$.
    \label{lem:minimal_counterexample_no_curves_in_paths}
\end{lemma}
\begin{proof}
    Assume the contrary and let~$e_i = (u_1,u'')$ and~$e_i'= (v'',v_1)$. First note that~$u'',v'' \in V(C_{i+1})$ and thus, similar to \cref{thm:irrelevant_cycle_claim4} using the Euler-embedding we derive that~$u_1,v_1 \notin V(C_{i+1})$, in particular~$u_1,v_1 \in V(C_{i-1})$. Then, again using the Euler-embedding (and the fact that cycles of the swirl~$\SSS$ are alternatingly oriented) we derive that~$e_i + e_{i+1} + e_i' + u_1 C_{i-1} v_1$ bounds a region. Finally let~$P$ be the path~$v_1C_{i}u_1$ as in the proof of \cref{lem:minimal_counterexample_no_bumps}. Again as in the proof of \cref{lem:minimal_counterexample_no_bumps} we derive the existence of an edge~$(v_1,w) \in E(C_i)$ with~$v_1,w \in V(C_{i-1})$ as a contradiction to \cref{lem:minimal_counterexample_no_bumps}
\end{proof}

The consequences of \cref{lem:minimal_counterexample_no_curves_in_paths} are huge. For example it tells us that each component of a path~$L$ in~$\SSS$---making up most of the counterexample~$G+D$---cannot `use' much of the graph: it goes straight down from~$C_{2h}$ to some level~$C_{i}$ for some~$0 \leq i \leq 2h$ and then straight back up to~$C_{2h}$. This means that paths behave very rigidly inside~$\SSS$. This is summarised by the following.

\begin{lemma}
    Let~$L \in \LLL$ be a path and let~$P \subset L$ be a maximal sub-path such that~$P$ uses solely edges in~$E(\SSS)$. Then~$P$ has length~$\leq 4h+1$ and there exists~$0 \leq i \leq 2h$ such that~$P$ can be written as~$P^i_{\text{down}} + e_i + P^i_{\text{up}}$ where~$P^i_{\text{down}}=(e_{2h},e_{2h-1},\ldots,e_{i+1})$ and~$P^i_{\text{up}}=(f_{i+1},f_{i+1},\ldots,f_{2h})$ for~$e_j,f_j \in E(C_j)$ and some~$e_i \in E(C_i)$.
\label{lem:minimal_counterexample_level_paths}
\end{lemma}

We call a path~$P$ adhering to the structure of \cref{lem:minimal_counterexample_level_paths} a \emph{level~$i$ path} meaning the obvious.

Summarising all of the above we get the following.
\begin{corollary}\label{cor:sunflower_graph_structure}
    Let~$G+D$ be an Eulerian graph satisfying the assumptions of \cref{thm:irrelevant_cycle_minimal_counterexample}. Then~$G$ is of the following form.
    \begin{itemize}
         \item[1.]~$V(G) = V(D) \cup V(S) \cup V(C_0)$, such that~${\SSS \coloneqq S \cup C_0}$ is a flat~$(2h+1)$-swirl of tree-width~$\geq \xi(p)$ where each swirl-cycle~$C \in \{C_0,C_1,\ldots,C_{2h}\}$ is vertex-disjoint.
        \item[2.]~$G +D = \SSS \cup \mathcal{I}$, where both~$\SSS$ and~$\mathcal{I}$ are edge-disjoint Eulerian graphs, and~$\mathcal{I}$ consists of edges with both endpoints in~$V(C_{2h})$, the edges of~$D$ and edges with one endpoint in~$V(D)$ and one end-point in~$V(C_{2h})$.
        \item[3.] Let~$e=(u,v)\in E(G-\SSS)~$ with~$u,v \in V(C_{2h})$. Then there exists a vertex~$w \in V(C_{2h})$ such that~$u,w,v$ are visited by~$C_{2h}$ in that order.
        \item[4.] The linkage~$\LLL$ is rigid. Further, let~$P \subset L \in \LLL$ be a maximal sub-path that is contained in~$\SSS$, i.e., not using edges of~$\III$. Then there exists~$0 \leq i \leq 2h$ such that~$P$ is a level~$i$ path.
    \end{itemize}
\end{corollary}

We call a graph adhering to the structure imposed by \cref{cor:sunflower_graph_structure} an~\emph{$h$-flower graph}, a name suggesting itself when looking at the exemplary \cref{fig:flower_graph} of a flower-graph. If the graph contains no router of size~$t$ but a flat swirl induced by a tile of size~$h(p)$ then we refer to it as an~$h(p;t)$-flower graph.

\begin{figure}
    \centering
\includegraphics[width=.4\linewidth,height=.4\linewidth]{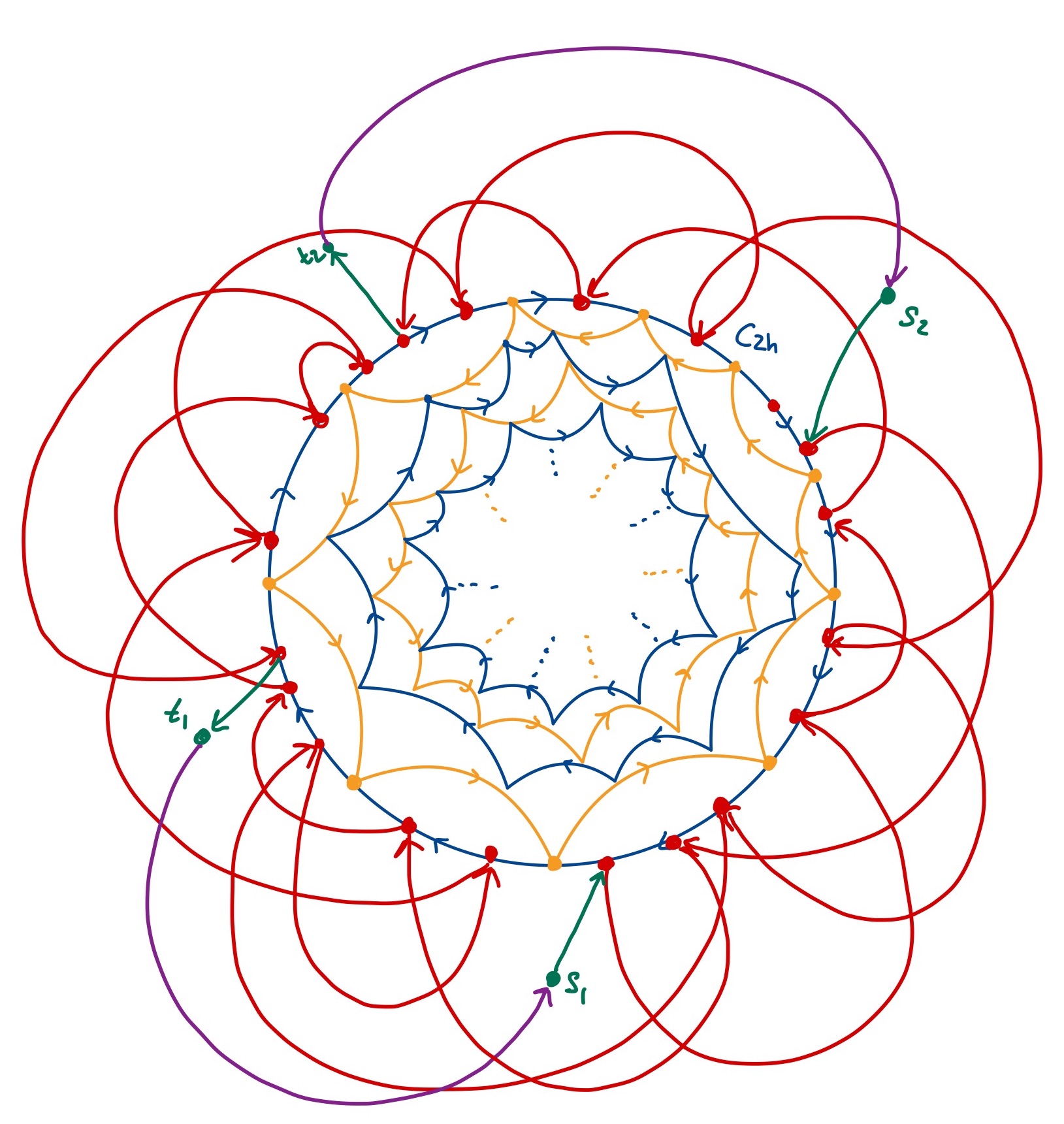}
    \caption{An example of a possible~$h$-flower graph in the case of~$p=2$.}
    \label{fig:flower_graph}
\end{figure}

\section{Navigating the Open Sea}
\label{sec:shippings}
In this section we will prove that Eulerian graphs of high tree-width admitting (weak) coastal maps cannot have rigid linkages. We start by giving a proof of the base-case, namely that an Eulerian graph~$G+D$ such that~$G$ can be Euler-embedded in some surface~$\Sigma$ (without boundary), does not admit a rigid~$p$-linkage by spotting an irrelevant cycle deeply nested in a swirl of the graph. We then leverage this argument to strong coastal maps and finally weak coastal maps, altogether proving the irrelevant cycle theorem for graphs of high tree-width admitting weak coastal maps. Note that the base-case, i.e., the main \cref{thm:shipping_in_open_sea} of the next section, was already proven by Johnson in \cite{Johnson2002}; our proof differs from his and highlights how to leverage our results from the previous sections derived as \cref{cor:sunflower_graph_structure}.

\paragraph{Some notation regarding surfaces.} Let~$\Sigma$ be a surface with boundary. Recall that we denote by~$\hat{\Sigma}$ the surface obtained by gluing discs to the cuffs and thus making~$\Sigma$ a manifold without boundary. 
\begin{definition}[Simpler surface]
  Let~$\Sigma,\Sigma'$ be two surfaces possibly with boundary. We say that~$\Sigma'$ is \emph{simpler} than~$\Sigma$ if either~$\hat{\Sigma'}$ has lower genus than~$\hat{\Sigma}$, i.e., at most as many handles and cross-caps, but at least one of the two less than~$\hat{\Sigma}$, or if they have the same genus but~$\Abs{c(\Sigma')} < \Abs{c(\Sigma)}$. 
\end{definition}
\begin{remark}
    In a nutshell~$\Sigma'$ is simpler than~$\Sigma$ if it has lower genus or the same genus but fewer holes.
\end{remark}
Given the definition of `simpler' for surfaces, it becomes clear what we mean by making an induction over surfaces: either augmenting (or decreasing depending on the way it is phrased) its genus or its number of holes---if the surface may have a boundary---in each step.  

\subsection{Shipping in a world without islands}
\label{subsec:shipping_in_the_open_sea}

We start with the case that~$G+D$ is Eulerian and~$G$ is Euler-embedded in a surface~$\Sigma$ \emph{without} boundary, that is, there are no coast lines and no islands to consider. We show that if~$G$ has large enough tree-width, then there exists no rigid linkage in~$G$ with pattern~$D$. 

The following theorem was proven by Johnson \cite[Theorem 7.2]{Johnson2002} in great detail (although assuming that~$G+D$ is completely Euler-embedded in~$\Sigma$) using a very different approach, namely classifying the homotopy classes of curves. We provide a different proof leveraging \cref{thm:irrelevant_cycle_minimal_counterexample,cor:sunflower_graph_structure}. 

\begin{theorem}
   For every surface~$\Sigma$ and every~$p \geq 0$ there exists~$\omega\coloneqq \omega(p,\Sigma)$ such that the following holds. If~$G+D$ is an Eulerian graph Euler-embedded in~$\Sigma$, such that there exists a disc~$\Delta \subseteq \Sigma$ such that~$G[\Delta]$ contains an induced flat~$\omega$-swirl. Then there exists no rigid~$p$-linkage in~$G$ with pattern~$D$, i.e., there exists no~$p$-shipping in~$G$.
    \label{thm:shipping_in_open_sea}
\end{theorem}

\begin{proof}
    Given a surface~$\Sigma$ we will define~$\omega(p,r)$ (instead of~$\omega(p,\Sigma)$) in terms of~$r \in \N$ being the genus of the surface~$\Sigma$, which implicitly defines~$\omega(p,\Sigma)$ via the genus of the surface.  
    To this extent we define~$\omega(p,r)$ and~$\tilde{\omega}(p,r)$ inductively to satisfy the following:
    \begin{align}
    \nonumber
        \omega(p,0) &\geq d_{\ref{thm:directed_irr_vertex_Marx}}(p),\\
         \label{eq:Shipping_in_open_sea}
        \omega(p,r) &\geq 4\omega(p+2\omega(p,r-1),r-1), \\
        \nonumber
        \tilde{\omega}(p,r) &\geq h_{\ref{thm:irrelevant_cycle_minimal_counterexample}}(p;\omega(p,r)) 
    \end{align} 
    where in the last row~$\omega(p,r)$ takes the role of~$\xi$ in \cref{thm:irrelevant_cycle_minimal_counterexample}.
    Note that since~$\omega(p,0)$ only depends on~$p$, the function~$\tilde{\omega}(p,r)$ only depends on~$p$ for fixed~$r$ and clearly it is well-defined.
    
    Let~$\Sigma$ be a surface of genus~$r \in \N$ and let~$\hat{\omega}(p,r)\coloneqq 2(\tilde{\omega}(p,r))\cdot p^2$. Then we claim that~$\hat{\omega}(p,r)$ satisfies the theorem. First note that by the pigeon-hole principle---looking at a~$2\tilde{\omega}(p,r)$-tiling---there exits a flat sub-swirl~$\SSS$ of size at least~$2\tilde{\omega}(p,r)$ such that the outline of~$\SSS'$ covers a disc~$\Delta_\SSS \subseteq \Delta$ such that~$G[\Delta_\SSS]$ does not contain any traps; in particular~$G[\Delta_\SSS]$ is Eulerian (compare to \cref{obs:from_eCl_to_swirls_and_routers_in_the_graph}). The proof is by induction on the genus of~$\Sigma$.

  \begin{claim}
      The theorem holds true if~$\Sigma$ is the sphere, i.e.,~$r = 0$.
      \label{thm:open_shippings_notsphere}
  \end{claim}
  \begin{ClaimProof}
      To this extent let~$\mathcal{S}\subset G$ be a flat swirl of size~$2\tilde\omega(p,r) \geq 2 \omega(p,0)$ induced by some tile~$T \subset \WWW$ of some cylindrical wall~$\WWW$ and Euler-embedded in some disc~$\Delta_\SSS$ such that~$G[\Delta_\SSS]$ is edge-disjoint from~$D$. (Note that technically any swirl embedded in the sphere is flat for it is a plane drawing). Thus, given that~$\Sigma$ is a sphere and using that~$2\omega(p,0) \geq d_{\ref{thm:directed_irr_vertex_Marx}}(p)$ by definition, the claim follows by switching to the line-graph~$\LL_G$ and applying \cref{obs:linegraph_of_swirl_is_swirl}, then using \cref{thm:directed_irr_vertex_Marx} and switching back to~$G$ after deleting a deeply-nested swirl-cycle in~$\LL_G$ which in turn is the same as deleting a deeply-nested swirl-cycle in~$G$ as given by \cref{obs:linegraph_of_swirl_is_swirl}. 
  \end{ClaimProof}
  This finishes the proof of the base-case of our induction. Now assume towards a contradiction that the theorem does not hold true for all surfaces~$\Sigma$. Let~$\Sigma$ be a surface (without boundary) of minimal genus~$r \geq 1$---the case~$r=0$ was proven above---such that the theorem does not hold true. In particular assume that the theorem holds true for all lower-genus surfaces~$\Sigma'$ with genus~$r' < r$ for all~$p \in \N$. Given~$\Sigma$ choose~$p \in \N$ minimal such that the theorem fails---note that that the theorem is trivially true for~$p=1$ thus we may assume that~$p\geq 2$. 
  
  \smallskip
  
  Finally let~$G+D$ be a minimal counterexample with respect to~$\Abs{E(G)} + \Abs{V(G)}$ to the theorem given~$\Sigma$ and~$p$, i.e., there exists a rigid~$p$-linkage~$\LLL$ of~$G$ with pattern~$D$ where~$G$ is embedded in~$\Sigma$ and contains a flat~$\geq 2\tilde\omega(p,r)$-swirl~$\SSS$ of tree-width~$\geq 2\tilde\omega(p,r)$ by the reasoning above embedded in some disc~$\Delta_\SSS$ away from~$V(D)$ and thus the graph~$G[\Delta_\SSS]$ is Eulerian and edge-disjoint from~$E(D)$. Note that all the assumptions of \cref{thm:irrelevant_cycle_minimal_counterexample} are satisfied. Thus, using the minimality of~$G+D$, \cref{thm:irrelevant_cycle_minimal_counterexample} implies that~$V(G) = V(\SSS)\cup V(D)$ where~$\SSS$ is a~$\geq {\omega}(p,r)$-insulation of tree-width~$\geq 2\omega(p,r)$ using the uniqueness of its embedding.
    Let~$C \coloneqq C_{2\omega(p,r)}$ be the outer-cycle of~$\SSS=\bigcup_{i=1}^{2 \omega(p,r)}C_i$, where~$C_i$ are the alternating cycles of the swirl~$\SSS$ embedded in~$\Sigma$ each bounding a disc~$\Delta_i$---for the cycles are vertex-disjoint by 1. of \cref{cor:sunflower_graph_structure} and the embedding in the disc is~$2$-cell, recall \cref{obs:faces_in_Euler_embeddings_bounded_by_cycle}---such that~$\Delta_1 \subseteq \ldots \subseteq \Delta_{2\omega(p,r)}$ for~$1\leq i \leq {2\omega(p,r)}$. Let~$\Delta$ be the disc bounded by~$C$ such that all edges not in~$E(\SSS)$---edges with terminals or with both end-points on the outer-cycle~$C$ of the swirl by 2. of \cref{cor:sunflower_graph_structure}---lie outside of~$\Delta$ and the swirl~$\SSS$ is contained in~$\Delta$ with~$C$ bounding the disc. Let~$\SSS' = \bigcup_{i=1}^{2\omega(p,r-1)}C_i \subseteq \SSS$, then~$\SSS' \subset \SSS$ is the `most deeply nested' sub-swirl of size~$2\omega(p,r-1)$ with outer-cycle~$C'\coloneqq C_{2\omega(p,r-1)}$; it is also an~$\omega(p,r-1)$-insulation by definition. 
    
    \smallskip
        
    Let~$G'$ be the graph obtained from~$G$ after splitting off along the linkage~$\LLL$ with respect to~$\SSS'$ and~$D$---only split off at vertices away from~$V(\SSS') \cup V(D)$---and let~$\Delta'$ be the respective disc bounded by the vertex-disjoint cycle~$C'$ analogously defined to~$\Delta$. Then the resulting linkage~$\LLL'$ is again rigid for~$G'$ by \cref{lem:splitting_off_linkages_remains_rigid}. We say that an edge is \emph{non-separating in~$G'$ (or~$G$)} if it is an edge of~$G'[\Sigma-\Delta']$ (or~$G[\Sigma-\Delta]$) with both end-points in the outer-cycle of the swirl~$\SSS'$ (or~$\SSS$) such that the respective closed curve in~$\Sigma/\Delta'$ (or~$\Sigma/\Delta$)---the surface arising from~$\Sigma$ after contracting~$\Delta'$ (or~$\Delta$) to a point~$c'$ (or~$c$)---is non-separating, that is, the surface remains connected after deleting the closed curve. Recall that the only edges of~$G'$ (or~$G$) that are drawn outside of~$\Delta'$ (or~$\Delta$) and that may actually `use' the genus of~$\Sigma$ have both their endpoints on the outer-cycle~$C'$ (or~$C$) of the respective swirl~$\SSS'$ (or~$\SSS$).
        
    \begin{claim}
        If~$\Sigma$ is not the sphere then there exist non-separating edges in~$G$ (and~$G'$).
        \label{thm:shipping_in_open_sea_claim3}
    \end{claim}
    \begin{ClaimProof}
        Towards a contradiction assume the contrary; note that the drawing is not~$2$-cell. Then we are not on the sphere and all edges are separating. Contract all the vertices in~$\Delta$ (or~$\Delta'$) to a single vertex~$c$ (or~$c'$ respectively) and call the resulting graph~$G/\Delta$ (or~$G'/\Delta)$. In the remainder we omit the case~$G'$ for it is analogous. By construction~$V(G/\Delta) = V(D) \cup \{c\}$, continue with splitting off at the degree-two vertices~$V(D)$ in the only possible way, then the remaining graph has a single vertex namely~$c$ with loops. Let~$r_c\coloneqq (e_{i_1},\ldots,e_{i_{2n}})$ be a clockwise ordering of the edges adjacent to~$c$ with respect to the the local orientation of~$\Sigma$ around~$c$---for example given by the orientation of the outer-cycle of~$\SSS$---where each edge appears exactly twice in the ordering, once for its head and once for its tail. One can think of this as the rotation of the vertex~$c$. Since every edge is separating by assumption (splitting off at degree-two vertices does not change this), one easily sees that for any two distinct edges~$e_1,e_2$ adjacent to~$c$ we have that~$e_1,e_2,e_2,e_1$ appear in this order in~$r_c$ up-to a cyclic rotation of the ordering. In particular there are no two edges forming a~$(e_1,e_2,e_1,e_2)$ pattern in the sequence---note here that the resulting graph~$G/\Delta$ is still embedded and said pattern cannot be embedded in a disc. Thus there exists at least one edge~$e$ adjacent to~$c$ with~$(e,e)\subset r_c$ being a sub-pattern, that is there is no edge~$e'$ with~$e,e',e$ appearing in that order on~$r_c$. We say that~$e$ is \emph{peripheral}. We delete~$e$ from~$G/\Delta$ and Euler-embed~$G/\Delta - e$ in a disc maintaining the same rotation as prescribed by~$r_c$ using induction. Clearly we can add~$e$ to that drawing maintaining an embedded drawing of~$G/\Delta$ in a disc with the prescribed rotation~$r_c$ at~$c$. Thus~$G/\Delta$ can be drawn in a disc and finally, reverting the contraction of~$\Delta$ in~$G$ (and~$\Sigma$ if one wishes) as well as the splitting off of~$V(D)$ we can draw~$G$ in a disc with the swirl remaining Euler-embedded, for we only changed the drawing outside of~$\Delta$. Since the rotation~$r_c$ remained the same and~$G$ was Euler-embedded in the first place, we deduce that~$G$ in turn is Euler-embedded in the disc for only~$e$ could destroy the embedding, but it cannot for the rotation remained the same; contradiction to~$\Sigma$ not being a sphere.
        
        \smallskip

        Note that if for~$G'$ we would find said disc then, since~$\LLL'$ is a rigid~$p$-linkage in~$G'$ and since~$G'$ contains a~${\omega}(p,r-1) \geq {\omega}(p,0)$ insulation, we again get a contradiction by \cref{thm:open_shippings_notsphere}.

    \end{ClaimProof}
    Thus there exist some edges~$e\in E(G)$ and~$e' \in E(G')$ that are non-separating where~$e'$ is obtained by contracting some sub-path~$P_e \subset L \in \LLL$ starting and ending in~$C'$ (the outer-cycle of~$\SSS')$ but otherwise internally disjoint from~$C'$ and containing~$e$. In particular, after contracting~$\Delta$ in~$G$ the edge~$e$ becomes a closed curve that is not null-homotopic (and similar for~$e'$ and~$\Delta'$).

    Let~$P_e \subset L \in \LLL$ be edge-minimal with both end-points~$u,v$ in~$V(C')$ such that it is else disjoint from~$\Delta'$ and such that when splitting off along~$P_e$ we get an edge~$e'$ that is non-separating in~$\Sigma$ (after contracting~$\Delta'$ to a point~$c'$ as above). Let~$G'$ be the graph resulting from splitting off along~$P_e$ only (the remaining linkage is still rigid by \cref{lem:splitting_off_linkages_remains_rigid} for~$P_e \subset L\in \LLL$). 
    
    \smallskip
    
    By 4. of \cref{cor:sunflower_graph_structure} we deduce that~$P_e$ consists of two sub-paths~$P_e^1,P_e^2 \subset L$, each being sub-paths of~$P^1,P^2 \subset \restr{L}{S}$ which in turn are level~$i_1$ and level~$i_2$ paths for some~$1 \leq i_1,i_2 \leq 2{\omega}(p,r)$. This implies that~$P_e$ visits at most every cycle~$C_{2{\omega}(p,r-1)+1},\ldots,C_{2{\omega}(p,r)}$ away from~$\SSS'$ twice (once in~$P^1$ and once in~$P^2)$. Let~$\ell$ be a cut-line in~$G[\Delta']$ of minimal length connecting~$u$ and~$v$, i.e., intersecting the minimum number of edges in the drawing, then~$\ell$ cuts at most~$2\omega(p,r-1)$ edges; draw~$\ell$ as the concatenation of two straight lines from~$u$ to the centre of the concentric cycles in~$\SSS'$ and back to~$v$ using the fact that all the cycles in~$\SSS'$ are each vertex-disjoint cycles and that there are no other edges in~$G[\Delta']$ apart from the cycle-edges in~$\SSS'$ using 1. and 2. of \cref{cor:sunflower_graph_structure}. Now draw~$\ell_{e'}$ close to~$e'$ so that it ends in two edges of~$C'$ (just take the edges that~$u$ and~$v$ are adjacent too such that~$\ell_{e'}$ can be drawn in parallel to~$e'$ without intersecting any other edges of the graph). Then~$\ell_{e'}+\ell$ forms a non-separating essential curve that cuts~$G'$ exactly in the~$\leq 2\omega(p,r-1)$ edges of~$S'$. We cut~$\Sigma$ and~$G'$ along~$\ell_{e'} + \ell$ to obtain a surface~$\Sigma'$ of genus~$\leq r-1$ (where we glue discs into the holes that arise from the cutting) together with a rigid~$\big(p+2\omega(p,r-1)\big)$-linkage by \cref{lem:rigid_linkages_after_cutting_edges}. Thus we get a graph~$G''$ embedded in~$\Sigma'$ with a rigid~$\big(p+2\omega(p,r-1)\big)$-linkage~$\LLL''$. Now the graph has tree-width~$\geq \frac{1}{2}\omega(p,r)$ since at least one side of the cut-line~$\ell + \ell_{e'}$ contains (almost) half of the original underlying undirected wall and certainly one fourth of it using the fact that~$P_e$ used at most two edges of each cycle away from~$\SSS'$ and~$\ell$ .
    
    The theorem then follows using the fact that~$\frac{1}{2}\omega(p,r) \geq \omega\big(p+2\omega(p,r-1),r-1\big)$ as given by equation \cref{eq:Shipping_in_open_sea}, and thus since~$\Sigma$ was minimal,~$\LLL''$ cannot be rigid in~$\Sigma'$ as a contradiction.
\end{proof}

\begin{remark}
    Note that \cref{thm:irrelevant_cycle_minimal_counterexample} does not itself assume the linkage to be rigid, but proves that it must be. 

 Further we want to emphasize that the existence of said disc~$\Delta$ containing a flat swirl in the assumptions of the theorem can indeed be \emph{derived} from the fact that the tree-width of~$G$ is large. To see this note that one can prove that there exists a map~$t:\N \to \N$ such that given some surface~$\Sigma$ of genus~$r$ no~$t'\geq t(r)$-router can be Euler-embedded in~$\Sigma$. Together with this fact we derive that~$\Sigma$ does not contain a~$t(r)$-router which together with the flat-swirl \cref{thm:flat_swirl_away_from_D} implies that it contains a large, say~$s(t)$-swirl. We claim that this implies the existence of a swirl embedded in a disc: for one it follows from the fact that by the work due to Robertson and Seymour \cite{GMXVI} there exists an underlying undirected large wall embedded in a disc of~$\Sigma$. Then, using the fact that~$G$ is Euler-embedded this undirected wall contains a large directed tile which in turn induces a canonical swirl using \cref{lem:no_jumps_is_swirl_gen}; it is indeed a swirl for it is Euler-embedded and thus `untangled'.  Another way to see this is that when constructing the embedding we can always start with a plane swirl embedded in a disc and extend the embedding while keeping the tree-width of the swirl high enough similar to the constructions provided in \cite{KawarabayashiTW2021}. Note here that for this theorem it is not relevant whether we can find the respective swirl embedded in a disc in~$fpt$-time (although we can). 
 
Finally, the above remark together with the results provided in \cref{sec:Swirls,sec:Routing} can easily be extended to a proof of \cref{thm:main} in the case that~$G+D$ is Euler-embeddable in~$\Sigma$.
\end{remark}

Having established the case that~$G$ can be Euler-embedded in a surface~$\Sigma$, we continue with the case that~$G+D$ admits a strong coastal map.

\subsection{Bounding the number of routes between strong coastlines} \label{subsec:bounding_jumps_in_strong_maps}

Recall the setting introduced in \cref{subsec:embedded-incidence-digraphs} as well as the results presented in \cref{subsec:rigid-linkages} and the definitions introduced in \cref{subsec:coastlines}. Given said setting---and in particular the pair~$(\Gamma,I)$ of pseudo-Eulerian graphs for~$G$---we will, throughout this section, talk about strong coastal maps~$(\Gamma, \CCC_1, \dots, \CCC_r, \mu)$ of~$G$ in~$\Sigma$ with~$r \geq 1$ and depth~$d\geq 0$ sights and will thus assume such a map to be given unless stated otherwise. Similarly we may assume for every cuff~$C \in c(\Sigma)$ that~$|C \cap V(\Gamma)| \geq 2$, for else we get trivial one-separations which can be handled using standard techniques. Similarly we assume that the depth is at least one for every cuff~$C$ for else we can simply cap the hole. Recall that as in \cref{subsec:embedded-incidence-digraphs} we assume~$\Gamma$ is pseudo-Euler embedded in~$\Sigma$ (see \cref{def:pseudo_euler_embedding}). Henceforth, let~$(\Gamma, \CCC_1, \dots, \CCC_r, \mu)$ be a \emph{strong coastal map (of~$G$ in~$\Sigma$)} with~$r\geq 1$ \emph{sights} and \emph{depth}~$d\geq 1$. Note that by the assumptions on~$G=\Gamma \cup I$ every non-terminal vertex on~$C \in c(\Sigma)$ has two edges in~$\Gamma$ and two edges in the strong island of~$C$; recall that if four edges are part of~$\Gamma$ then we can delete the vertex from~$I$ pushing the embedding of the vertex and its adjacent edges into the surface~$\Sigma$ away from the cuff. We emphasise that many of the reductions we will use throughout this section will transform possible~$2$-cell embeddings into non-$2$-cell embeddings; in particular we take no prior assumptions on the embedding if not stated otherwise where the usual \cref{obs:disc_embedd_is_2cell} that embeddings in disc are~$2$-cell remains true. 

\smallskip

The goal of this section is to prove that, given a strong coastal map of a high tree-width graph~$G$ and a rigid linkage~$\LLL$, we can bound the number of sub-paths with endpoints in islands. This will then help us to reduce the instance by cutting open the surface and reducing the instance to an instance without islands: a \emph{shipping in the open sea}, a case we have already dealt with in \cref{subsec:shipping_in_the_open_sea}. 

We start with two lemmas that are needed in the following proofs. The statements of the lemmas are analogous to the ones provided by Robertson and Seymour in \cite[Lemmas 6.1 and 6.2]{GMXXI} rephrased to fit our setting. In a nutshell they provide ways to guarantee linkedness between pairs of ports; note that although two ports~$p,p' \in \Port(G)$ represent a cut in~$I_\ZZZ$ given by~$\mu(p) \cup \mu(p')$, they do not yield cuts in~$G$ for~$\Gamma$ may still be connected to~$I$ exactly via said ports which are \emph{not} part of~$\mu(p)\cup \mu(p')$ as given by \WMII which will turn out to be a nuisance later.

\begin{lemma}
Let~$(\Gamma, \CCC_1, \ldots,\CCC_r, \mu)$ be a strong coastal map of~$G$ in~$\Sigma$ with~$r\geq 1$ sights and depth~$d \geq 0$. Let~$s,s' \in \Shore(G)$ where~$s$ has ends~$p_1,p_2$ and~$s'$ has ends~$p_1',p_2'$ for some~$p_1,p_2,p_1',p_2' \in \Port(\CCC_i)$ occurring in that order for some~$1\leq i \leq r$. Let~$d_i \coloneqq |\mu(p_1)|$. Further let~$P_1,\ldots,P_{d_i}$ be edge-disjoint~$\mu(p_1){-}\mu(p_2)$ paths in~$\mu(s)$ and let~$P_1',\ldots,P_{d_i}'$ be edge-disjoint~$\mu(p_1'){-}\mu(p_2')$ paths in~$\mu(s')$. Then the following hold true:
\begin{enumerate}
    \item~$E(P_i)\cap E(P_j') \subseteq \mu(p_2)\cap \mu(p_1')$ for all~$1 \leq i \leq j \leq d_i$ and for every~$e \in E(P_i)\cap E(P_j')$ it is an end of both paths,
    \item if~$p_2 = p_1'$ then we can renumber the paths such that~$P_1 \cup P_1',\ldots, P_{d_i} \cup P_{d_i}'$ are pairwise edge-disjoint paths of~$\mu(s)\cup \mu(s')$, and
    \item  if~$|{\mu(p_1) \cup \mu(p_2')}| = |{\mu(p_2)\cup\mu(p_1')}| = k$ for some~$k \geq d_i$, then~$|\{P_1,P_1',\ldots,P_{d_i},P_{d_i}'\}| = k$ and the paths form~$k$ edge-disjoint~$\big(\mu(p_1)\cup \mu(p_2')\big){-}\big(\mu(p_2)\cup \mu(p_1')\big)$ paths in~$\mu(s) \cup \mu(s')$.
\end{enumerate}
\label{lem:presentation_to_map_lemma_1}
\end{lemma}
\begin{proof}
    Denote by~$C_i \in c(\Sigma)$ the cuff with~$\CCC_i \subset \Zone(C_i)$.
    \begin{enumerate}
        \item Let~$e \in E(P_i) \cap E(P_j')$, then~$e \in \mu(s) \cap \mu(s')$ for the paths are paths in~$\mu(s),\mu(s')$ respectively. But~$s,p_2,p_1',s'$ appear in this order on~$\Zone(C_i)$ and thus the claim follows from (SM3).
        \item Since~$|\mu(p_i)| = |\mu(p_i')|=d_i$ for~$i=1,2$ and every edge of~$\mu(p_2)=\mu(p_1')$ is part of exactly one path in each collection, we deduce that there is an injective map~$\phi : \{P_1,\ldots,P_{d_i}\} \to \{P_1',\ldots,P_{d_i}'\}$ such that~$\pi(P_j) = (e_1,e_2) \iff \pi(\phi(P_j)) = (e_2,e_3)$ for every~$1\leq j\leq d_i$ where~$e_2 \in \mu(p_2)=\mu(p_1')$. Clearly every edge in~$\mu(p_2)$ is part of exactly one path in both collections for the paths are edge-disjoint. Then the claim follows from (SM3).
        \item If~$\mu(p_1)=\mu(p_2')$ then they all agree by (SM2) and (SM3) and the claim is trivial. If~$\mu(p_1) = \mu(p_1')$ then~$\mu(p_2)=\mu(p_1')$ and thus using~$|{\mu(p_1) \cup \mu(p_2')}| = |{\mu(p_2)\cup\mu(p_1')}|$ we deduce that~$\mu(p_1)=\mu(p_2')$.  In the same spirit but for a more finegrained analysis enumerate~$\mu(p_1)=\{e_1,\ldots,e_{d_i}\}$,~$\mu(p_2) = \{f_1,\ldots,f_{d_i}\}$ and~$\mu(p_1')=\{e_1',\ldots,e_{d_i}'\}$,~$\mu(p_2') = \{f_1',\ldots,f_{d_i}'\}$ such that~$\pi(P_j)=(e_j,f_j)$ and~$\mu(P_j') = (e_j',f_j')$ and if~$f_j = e_t'$ then~$j=t$. If~$e_j = f_j'$ then~$e_j=f_j=e_j'=f_j'$ by 1 and thus~$P_j = P_j'$. We deduce the following 
  ~$$A \coloneqq \{j \mid e_j = f_j',\: 1\leq j \leq d_i\} \subseteq \{j \mid f_j = e_j',\: 1\leq j \leq d_i\} \eqqcolon B,$$
        From the hypothesis we know that~$|\{e_1\ldots,e_{d_i},f_1',\ldots,f_{d_i}'\}| = |\{e_1'\ldots,e_{d_i}',f_1,\ldots,f_{d_i}\}|$ and hence~$B \subseteq A$, but then they are equal; this immediately implies~$|\{P_1,P_1',\ldots,P_{d_i},P_{d_i}'\}| = k$. In particular if two paths meet they are already equal, and thus all the paths in~$\{P_1,P_1',\ldots,P_{d_i},P_{d_i}'\}$ are pairwise edge-disjoint; concluding the proof of 3.
\end{enumerate}
\end{proof}
\begin{lemma}\label{lem:presentation_to_map_lemma_2}
Let~$(\Gamma, \CCC_1, \dots, 
\CCC_r, \mu)$ be a strong coastal map of~$G$ in~$\Sigma$ with~$r\geq 1$ sights and depth~$d \geq 0$. Let~$c_i\subseteq \CCC_i$  be a coastline with ends~$p_1,p_2 \in \Port(\CCC_i)$ for some~$1\leq i \leq r$. Then there are~$|\mu(p_1)|$ edge-disjoint~$\mu(p_1){-}\mu(p_2)$ paths in~$\mu(c_i)$. 

\end{lemma}
\begin{proof}
The proof is straightforward by gluing the edge-disjoint paths given by (SM5) along the common edges keeping them edge-disjoint by (SM3) (see also 2. of \cref{lem:presentation_to_map_lemma_1}).
\end{proof}

The main theorem of this section is analogous to \cite[Theorem 6.3]{GMXXI} and our proof follows the same line of argumentation as theirs by reducing `most of' our setting to theirs after some careful definitions and analysis. As already highlighted above, there are nuisances appearing that were no concern in the undirected vertex-disjoint setting that we will need to deal with later. Before stating the theorem we introduce some more lemmas needed for its proof, the first of which deals with simplifying the pseudo-Eulerian graph~$\Gamma$ and the~$p$-shipping~$\LLL$.
\begin{lemma}
    Let~$G+D$ be Eulerian and let~$G=\Gamma \cup I$ as above. Let~$(\Gamma, \CCC_1, \dots, \CCC_r, \mu)$ be a strong coastal map of~$G$ in~$\Sigma$ with~$r\geq 1$ sights and depth~$d\geq 0$. Let~$\LLL=\{L_1,\ldots,L_p\}$ be~$p$-shipping in~$G$ for some~$p \in \N$. Let~$G'$ be the graph obtained by splitting off along all two-edge sub-paths~$(e_1,e_2) \subset L_i$ such that~$|\{e_1,e_2\} \cap \Port(G)| \leq 1$. Then, if there is a vertex~$v \in V(\Gamma_\ZZZ)$ left of degree four, split the vertex into two vertices of degree two consistent with the linkage. Let~$\Gamma'$ be the respective graph obtained from~$\Gamma$ and let~$\LLL'$ be the respective linkage obtained from~$\LLL$. Then the following hold true:
    \begin{enumerate}
        \item~$G'+D$ is Eulerian and~$\LLL'$ is a~$p$-shipping with pattern~$D$,
        \item~$(\Gamma',\CCC_1, \dots, \CCC_r, \mu)$ is a strong coastal map of~$G'$ in~$\Sigma$,
        \item~$\Sigma$-routes (and~$\Sigma_\ZZZ$-routes) in~$\Gamma'$ are two-paths where all the vertices in~$V(\Gamma'_\ZZZ)$ are of degree two and the number of~$\Sigma$-routes (and~$\Sigma_\ZZZ$-routes) is the same for both coastal maps. 
    \end{enumerate}
    \label{lem:coastal_map_after_splitting_off_along_linkage}
\end{lemma}
\begin{proof}
    The proof of 1. follows at once from \cref{lem:splitting_off_linkages_remains_rigid} together with the fact that we do not split off at terminals (which are traps), for they are of degree one in~$\Gamma$. Note that~$|\{e_1,e_2\} \cap \Port(G)| \leq 1$ implies that we never split off at a vertex~$v$ drawn on~$\bd(\Sigma)$---$v \in V(\Gamma)\cap V(I)$---and thus, by renaming carefully after splitting off, we never lose a port, for we never split off along two ports. Finally note that after splitting off, every trap is either adjacent to a vertex on the cuff, or two traps are drawn in~$\Gamma$ and linked by an edge if the whole path was contained in~$\Gamma$. The latter would imply said path to be trapped; a contradiction to the fact that~$\LLL$ was a~$p$-shipping.
    
    Now the proof of 2. is imminent for~$V(\Gamma) \cap \nu^{-1}(\bd(\Sigma)) = V(\Gamma') \cap \nu^{-1}(\bd(\Sigma))$ and none of the ports are part of a splitting pair in the construction of~$G'$ as seen above. Hence the~$\Zone$ and~$I$ as well as~$I_\ZZZ$ remain equal. Also~$\Gamma'$ remains pseudo Euler-embedded since we split off along linkages in an Euler-embedded graph and when splitting a vertex into two vertices with respect to the linkage we again remain Euler-embedded for the vertex was not strongly planar. Also no new traps arise during the construction and every vertex on the boundary remains of degree two in~$\Gamma$. Note that splitting the vertex into two vertices with respect to the linkage is the same as splitting off along the linkage and re-introducing a vertex to subdivide the edge. The reason we do this is to not combine two ports of distinct (or the same) cuffs into a single one, thus not splitting off along a two-path consisting of two ports. 
    
    Finally 3. is valid, since for any~$\Sigma$-route~$P$ in~$\Gamma$ it holds~$P \subseteq L_i \cap \Gamma$ with no internal point in~$\bd(\Sigma)$ for some~$1 \leq i \leq p$. By definition of~$\Gamma'$ the path~$P$ is split off to a path~$P' = (e_1,e_2)$---where~$e_1\neq e_2$ by (SM1) and (SM2)---with~$e_1,e_2 \in \Port(G)=\Port(G')$ either in the same or two different island-zones. Thus~$P'$ is a~$\Sigma$-route in~$\Gamma'$. Clearly every~$\Sigma$-route in~$\Gamma'$ comes from a~$\Sigma$-route in~$\Gamma$. The claim for~$\Sigma_\ZZZ$-routes follows using \cref{obs:sigma_routes_are_port_paths}.
\end{proof}

In particular~$U(\Gamma') \setminus \bd(\Sigma)$ is a collection of disjoint lines which is in one-to-one correspondence with the disjoint lines in~$U(\Gamma_\ZZZ')$ using \cref{obs:sigma_routes_are_port_paths}.

The following is an easy consequence of the above.

\begin{corollary}
    Let~$G+D$ be Eulerian and let~$G=\Gamma \cup I$ as above. Let~$(\Gamma, \CCC_1, \dots, \CCC_r, \mu)$ be a strong coastal map of~$G$ in~$\Sigma$ with~$r\geq 1$ sights and depth~$d\geq 0$. Let~$\LLL=\{L_1,\ldots,L_p\}$ be a rigid~$p$-linkage in~$G$ for some~$p \in \N$. Let~$G'$ be the graph obtained by splitting off along all two-edge sub-paths~$(e_1,e_2) \subset L_i$ such that~$|\{e_1,e_2\} \cap \Port(G)| \leq 1$. Then, if there is a vertex~$v \in V(\Gamma_\ZZZ)$ left of degree four, split the vertex into two vertices of degree two consistent with the linkage. Finally if there is an edge~$e=(u,v)$ left in~$\Gamma$ such that~$u,v \in V(D)$ delete that edge. Let~$\Gamma'$ be the respective graph obtained from~$\Gamma$ and let~$\LLL'$ be the respective linkage obtained from~$\LLL$. Then the following hold true:
    \begin{enumerate}
        \item~$G'+D$ is Eulerian and~$\LLL'$ is a~$p'$-shipping with pattern~$D$ for some~$p'\leq p$,
        \item~$(\Gamma',\CCC_1, \dots, \CCC_r, \mu)$ is a strong coastal map of~$G'$ in~$\Sigma$,
        \item~$\Sigma$-routes (and~$\Sigma_\ZZZ$-routes) in~$\Gamma'$ are two-paths where all the vertices in~$V(\Gamma'_\ZZZ)$ are of degree two and the number of~$\Sigma$-routes (and~$\Sigma_\ZZZ$-routes) is the same for both coastal maps. 
    \end{enumerate}
    In particular every edge is adjacent to a boundary-vertex. \label{cor:coastal_map_after_splitting_off_along_rigid_is_shipping}
\end{corollary}
\begin{proof}
    This follows at once from \cref{lem:coastal_map_after_splitting_off_along_linkage} and the fact that if an edge~$e=(u,v)$ with~$u,v \in V(D)$ is left in~$\Gamma$, then there was a respective path~$L \in \LLL$ in~$\Gamma$ such that~$L$ was trapped and disjoint from~$\nu^{-1}(\bd(\Sigma))$. In particular~$L$ contains no sub-path that is a~$\Sigma$-route. Finally it is obvious that every edge must be adjacent to a boundary vertex for else we could continue splitting off.
\end{proof}

In the same spirit as the above we can split off paths that use two consecutive ports by slightly adapting the coastal map. We will need this in order to effectively count the number of~$\Sigma$-routes that come from a path entering the island and thus passing through the island-zone and not `bouncing back' since paths bouncing back do not really use the island and thus are not part of any useful cross for the routing. That is we transform~$\Sigma_\ZZZ$-bounces into~$\Sigma_\ZZZ$-routes.

\begin{lemma}\label{lem:splitting_at_pairs_of_ports_remains_coastal}
     Let~$G+D$ be Eulerian and let $\LLL=\{L_1,\ldots,L_p\}$ be a rigid~$p$-linkage in~$G$ for some~$p \in \N$. Let~$G=\Gamma \cup I$ where~$\Gamma$ is pseudo Euler-embedded. and let~$(\Gamma,\Zone, \CCC_1, \dots, \CCC_r, \mu)$ be a strong coastal map of~$G$ in~$\Sigma$ with~$r\geq 1$ sights and depth~$d\geq 0$. Let~$G'$ be the graph obtained by splitting off along a two-edge sub-path~$(e_1,e_2) \subset L_i$ such that~$e_1,e_2 \in \Port(G)$. Let~$\Gamma'$ be the respective graph obtained from~$\Gamma$ and let~$\LLL'$ be the respective linkage obtained from~$\LLL$. Then the following hold true:
    \begin{enumerate}
        \item~$G'+D$ is Eulerian and~$\LLL'$ is a rigid~$p$-linkage with pattern~$D$, and
        \item~There exists a strong coastal map~$(\Gamma',\Zone',\CCC_1', \dots, \CCC_r', \mu')$ of~$G'$ in~$\Sigma$ with~$r$ sights and depth~$d$.
    \end{enumerate}
    \label{lem:coastal_map_after_splitting_off_at_ports}
\end{lemma}
\begin{proof}
     The first part of the lemma is clear by \cref{obs:splitting_off_at_vertex_remains_Eulerian} and \cref{lem:splitting_off_linkages_remains_rigid}; we are left to prove the existence of a strong coastal map~$(\Gamma',\Zone',\CCC_1', \dots, \CCC_r', \mu')$. To this extent let~$p_1,p_2 \in \Port(G)$ such that~$(p_1,p_2) \subseteq L$ is a sub-path of~$L \in \LLL$ and let~$s \in \Shore(p_1) \cap \Shore(p_2)$ be their adjacent shore and let~$v \in V(G)$ be such that when seen as edges,~$p_1,p_2$ are adjacent to~$v$. Using \cref{obs:adj_ports_give_nondeserted_shore_and_are_even}---note that strong coastal maps are special cases of weak coastal maps---we derive that~$v \in \mu(s)$ and thus in particular~$s \in \Shore(\CCC_i)$ for some~$1 \leq i \leq r$. Without loss of generality assume that~$p_1,s,p_2$ are visited in that order by~$\CCC_i$ and let~$s^*_1,s^*_2 \in \Shore(G)$ be such that~$s^*_1,p_1,s,p_2,s^*_2$ are visited in that order by~$\Zone$ and~$p_i \in \Port(s)\cap \Port(s^*_i)$ for~$i=1,2$. We define the following.
    \begin{itemize}
        \item[(i)] First let~$G'$ be obtained from~$G$ after splitting off~$(p_1,p_2)$. Then we define~$\Gamma'$ by removing~$p_1,p_2$ from~$\Gamma$ and adding the new edge to~$\Gamma$. Finally we delete~$v$ from~$\Gamma'$ keeping it in~$I$, i.e., we slightly push the vertex into the island not drawing it on the boundary anymore; in particular~$I' = I$. Note that splitting off along~$(p_1,p_2)$ does not have any effect on the remaining ports~$\Port(G)$ where the newly introduced edge is between the two end-points of~$p_1,p_2$ that are not drawn on a cuff of~$\Sigma$, i.e., they are vertices of~$V(\Gamma)\setminus V(I)$.
        
        \item[(ii)] We define~$\Zone'$ from~$\Zone$ in the obvious way by deleting~$p_1,s,p_2$ from the sequence and combining~$s^*_1,s^*_2$ to a new shore~$s^*$.

        \item[(iii)] Let~$\CCC_j' \coloneqq \CCC_j$ for every~$j \in \{1,\ldots,r\}$ with~$j \neq i$. Let~$\CCC_i=(p^1,s^1,p^2,\ldots,s^t,p^{t+1})$ for some~$t \geq 1$ and~$p^i,p^{t+1} \in \Port(G)$ as well as~$s^i\in\Shore(G)$ for~$1 \leq i \leq t$. Let~$\iota \in \{1,\ldots,t\}$ be such that~$p^\iota = p_1$. There is three cases to consider next. If~$p^{\iota} \neq p^1$ and~$p^{\iota+1} \neq p^{t+1}$ then we define~$\CCC_i'$ from~$\CCC_i$ by deleting~$p^{\iota},s^{\iota},p^{\iota+1}$ from the sequence and combining both shores~$s_1^*$ and~$s_2^*$ to~$s^*$ in~$\CCC_i'$. If however say~$p^{\iota}=p^{1}$ then we define~$\CCC_i'\coloneqq (p^3,s^3,\ldots,p^t,s^t,p^{t+1})$ and let~$s^*$ be a deserted shore---recall that~$t+1 \in 2\NN$ using \cref{obs:adj_ports_give_nondeserted_shore_and_are_even}. 

        \item[(iv)] We let~$\mu'(\chi) \coloneqq \mu(\chi)$ for all~$\chi \in \big(\Port(G)\cup \Shore(G)\big) \setminus \big(\Port(\CCC_i) \cup \Shore(\CCC_i)\big)$. We continue with some cases. First the case that~$p_1=p^1$ and~$p_2 = p^{t+1}$. In that case then~$\mu(\CCC_i) \cap \Gamma_\ZZZ = \{p_l,p_r\}$ and thus after splitting off at~$p_l,p_r$ we deduce that~$G'[V(\mu(\CCC_i))]$ is a component of~$G'$, i.e., it is disconnected from the rest of the graph. We may thus add it to any other non-deserted shore of the coastal-map without altering anything about the properties concluding this case.
        
\smallskip

        Thus assume that at least one of both ports is no end of the coast line. Assume next that~$p_1=p^1$ (the case~$p_2 = p^{t+1}$ is analogous) and thus~$p_2 \neq p^{t+1}$. By \cref{obs:adj_ports_give_nondeserted_shore_and_are_even} we deduce that there exists~$s'$ and ports~$p_1',p_2' \in \Port(s')$ such that~$p_1,s,p_2,s^*_2,p_1',s',p_2'$ appear in this order on~$\CCC_i$ with~$p_1' \in \Port(s^*_2) \cap \Port(s') $. We define~$\mu'(s^*) \coloneqq \mu(s_1^*) = \emptyset$ which remains deserted shore by definition, and~$\mu'(s') \coloneqq \mu(s')\cup \mu(s^*_2) \cup (\mu(s)\setminus \{p_1,p_2\})$, combining the three consecutive shores; note that~$p_1'$ is one end of~$\CCC_i'$. Finally we set~$\mu'(\chi)=\mu(\chi)$ for all~$\chi \in \Shore(\CCC_i) \cup \Port(\CCC_i)$ with~$\chi \notin \{p_1,s,p_2,s^*_2\}$.

    \smallskip

        Thus let~$p_1 = p^{\iota}$ for some~$1<\iota<t$. Then we define~$\mu'(s^*) \coloneqq \mu(s_1^*) \cup (\mu(s)\setminus \{p_1,p_2\}) \cup \mu(s^*_2)$ by combining the three consectuive shores which are all part of~$\Shore(\CCC_i)$ by the assumption on~$\iota$, and~$\mu'(\chi) \coloneqq \mu(\chi)$ for all~$\chi \in \Shore(\CCC_i) \cup \Port(\CCC_i)$ with~$\chi \notin \{p_1,s,p_2,s^*_2\}$.
    \end{itemize}

    One easily verifies that~$(\Gamma',\Zone',\CCC_1',\ldots,\CCC_r',\mu')$ is a strong coastal map with~$r$ sights of depth~$d$. Note that~$p_1,p_2 \notin \mu(\CCC_i')$ for they are solely covered by~$\mu(s)$ as given by \cref{obs:adj_ports_give_nondeserted_shore_and_are_even} (i).
\end{proof}

Combining \cref{cor:coastal_map_after_splitting_off_along_rigid_is_shipping} and \cref{lem:splitting_at_pairs_of_ports_remains_coastal} we may assume that~$G+D$ has a piercing~$p$-shipping (recall \cref{def:shipping})

\begin{observation}\label{obs:reduced_coastal_map_for_bounding_jumps}
     Let~$G+D$ be Eulerian and let~$G=\Gamma \cup I$ for two pseudo-Eulerian graphs~$\Gamma,I$ where~$\Gamma$ is pseudo Euler-embedded. Let~$(\Gamma, \CCC_1, \dots, \CCC_r, \mu)$ be a strong coastal map of~$G$ in~$\Sigma$ with~$r\geq 1$ sights and depth~$d\geq 0$. Let~$\LLL=\{L_1,\ldots,L_p\}$ be a rigid~$p$-linkage in~$G$ for some~$p \in \N$. Then there exists an Eulerian graph~$G'+D$ and a rigid~$p'$-linkage for some~$p' \leq p$ satisfying the following.
    \begin{enumerate}
        \item~$G'+D$ admits a strong coastal map $(\Gamma',\CCC_1', \dots, \CCC_r', \mu')$ in~$\Sigma$ with~$r$ sights and depth~$d$, 
        \item $\Gamma'$ is a collection of paths where every edge is a port, and
        \item~$\LLL'$ is a piercing~$p'$-shipping in~$G'$.
    \end{enumerate}
\end{observation}
\begin{proof}
    This follows at once by first applying \cref{cor:coastal_map_after_splitting_off_along_rigid_is_shipping} and then applying \cref{lem:splitting_at_pairs_of_ports_remains_coastal} for any pair of consecutive ports that are part of a~$\Sigma_\ZZZ$-bounce.
\end{proof}

All of the above lemmas help us to transform our instance into a much simpler instance admitting a weak coastal map in the same surface together with a \emph{piercing}~$p$-shipping. 
In a next instance we introduce a result that allows us to simplify the surface, that is, we will deal with splitting edges---a tool we will use to cut the surface along cut-lines in order to get a simpler surface inductively (similar to the proof of \cref{thm:shipping_in_open_sea}). 

\begin{lemma}
     Let~$(\Gamma, \CCC_1, \dots, \CCC_r, \mu)$ be a strong coastal map of~$G$ in some surface~$\Sigma$ with~$r \geq 1$ sights of depth~$d \geq 0$ and let~$\LLL$ be a rigid~$k$-linkage in~$G$ for some~$k 
    \in \N$. Let~$p \in \CCC_1$ be an interior port of~$\CCC_1$ and let~$E_1 \coloneqq \mu(p)\cup \{p\}$. Let~$G'$ be obtained from~$G$ by splitting all the edges in~$E_1$. Then there exists a strong coastal map~$(\Gamma, \CCC_1^l(p), \CCC_1^r(p), \CCC_2 \ldots, \CCC_r, \mu')$ with~$r+1$ sights and depth~$d$ such that~$\LLL'$, which is obtained from~$\LLL$ after splitting the edges in~$E_1$, is a rigid~$k+(d+1)$-linkage in~$G'$. 

     If~$\LLL$ is a (piercing) shipping, then so is~$\LLL'$.
     \label{lem:cutting_coastal_maps_at_shores}
\end{lemma}
\begin{proof}
    Let~$(\Gamma,I)$ be as usual, in particular~$\Gamma$ is pseudo Euler-embedded in~$\Sigma$, and let~$(\Gamma_\ZZZ,I_\ZZZ)$ be the respective separation tailored to~$\Zone$.
    First note that splitting the edges in~$E_1$ results in a new Eulerian graph~$G' + D'$ using \cref{lem:rigid_linkages_after_cutting_edges}. Let~$v_l,v_r \in \Port(G)$ be the ends of~$\CCC_1$ in the order of~$\CCC_1$ with respect to the~$\Zone(C)$ of the respective cuff~$C$ and define~$\CCC_1^l(p),\CCC_1^r(p)$ as in \cref{def:left_right_coastlines} to be the left and right coast of~$\CCC_1$ at~$p$ respectively. By definition then it holds~$v_l \in \CCC_1^l(p)$ and~$v_r \in \CCC_1^r(p)$ and both coastlines are mutually disjoint up to~$p$ (where they agree). We define the set of vertices~$\delta_l$ to be the vertices in~$ V(\mu(\CCC_1^l(p)))$ incident to edges in~$\mu(p)$ and~$\delta_r$ analogously defined for~$V(\mu(\CCC_1^r(p))$. Then~$\delta_l \cap \delta_r = \emptyset$ by (SM1) together with the fact that~$\mu(\Port(G)) \subset E(G)$. 

    \smallskip
    
    Splitting the edge~$p \in E_1$ itself results in two new edges~$p_l$ and~$p_r$ adjacent to two new vertices~$x_l,x_r$ respectively. Say~$p=(u,v)$, then~$p_l=(u,x_l)$ and~$p_r = (x_r,v)$. Recall that by \cref{obs:ports_are_divided_by_separation_tailored_to_zone} it follows that~$\{u,v\} \cap \Gamma_\ZZZ \neq \emptyset$ and~$\{u,v\} \cap I_\ZZZ \neq \emptyset$; without loss of generality let~$u \in V(\Gamma_\ZZZ)$ and~$v \in V(I_\ZZZ)$; the other case is symmetric.
    
    Then, using (SM1) and (SM2) we derive that~$v \in \mu(s^p)$ for some~$s\in \Shore(p)$ adjacent to~$p$; in particular now~$p_r$ is incident to~$v\in \mu(s^p)$. Again without loss of generality (and for the sake of readability) assume that~$v \in \delta_r$, the case~$v \in \delta_l$ is analogous by renaming~$p_l$ and~$p_r$ an thus we call~$s^p=s_r$ for the sake of readability. 

\smallskip

    Next we define~$\Gamma'$ from~$\Gamma$ by deleting~$p$ and adding~$p_l,p_r$ and their new ends, where we keep the drawing by slightly pulling~$x_l$ away form~$x_r$ keeping both embedded in~$\Gamma'$---they are new traps.
    
    We define~$I'$ from~$I$ by removing all edges in~$\mu(p)$ and adding all of the edges arising from cutting the edges in~$\mu(p)$ as well as their newly introduced ends. We add none of these to~$\Gamma'$. Given~$\Gamma'$ and~$I'$ we next redefine the zone to~$\Zone'(C)$ by replacing~$p$ in the sequence by~$p_r$ instead, keeping the order. Again using the assumptions let~$s_l',p',s_l,p_r,s_r$ appear in this order on~$\Zone'(C)$ where~$s_l,s_r \in \Shore(p_r)$ are the adjacent shores to~$p_r$ and~$s_l,s_l' \in \Shore(p')$. From this we can now construct a new strong coastal map~$(\Gamma',\Zone',\CCC_1^l(p)',\CCC_1^r(p)',\CCC_2,\ldots,\CCC_r,\mu')$ with~$r+1$ sights and depth~$d$ satisfying~$\mu'(\CCC_1^l(p)') \cap \mu'(\CCC_1^r(p)') = \emptyset$ as follows. First one easily verifies that~$(\Gamma',I')$ are pseudo-Eulerian,~$\Gamma'$ is pseudo Euler-embedded, and~$\Zone'$ agrees with the \cref{def:ports_shores_zones}; in particular~$(\Gamma'_{\ZZZ'},I'_{\ZZZ'})$ is well-defined and satisfies \cref{obs:ports_are_divided_by_separation_tailored_to_zone}.

\smallskip

    We continue with defining~$\mu'$: first we set~$\mu'(s_l) \coloneqq \emptyset$---this guarantees \small{(SM6)}. Then we define~$\CCC_1^l(p)'$ to be the sub-coast with ends~$v_l,p'$ and~$\CCC_1^r(p)'$ analogously with ends~$p_r,v_r$. Finally~$\mu'$ is defined as follows. For every edge~$e\in \mu(p)$ it has exactly one incidence in~$\delta_l$ and one in~$\delta_r$ respectively by (SM4); let~$e_l,e_r$ be the respective edges after the splitting, where~$e_l$ has one end in~$\delta_l$ and~$e_r$ has one end in~$\delta_r$ and both have a newly introduce vertex each---say~$u^e_l$ for~$e_l$ and~$u^e_r$ for~$e_r$---that has yet to be assigned to some shore. Then we define~$\mu'(p')$ for~$p' \in \CCC_1^l(p)'$ say, by replacing the edge~$e \in \mu(p')$ with the new edge~$e_l$ and analogously for~$\CCC_1^r(p)'$ by~$e_r$. The depth remains~$d$ for we only replace the edges in each~$\mu(p')$ and never add a new one; this is due to the fact that there exists no~$s \in \Shore(\CCC_1^l(p)')\cup\Shore(\CCC_1^r(p)')$ containing both terminals of an edge in~$\mu(p)$ by (SM4): thus for each edge~$e \in \mu(p')$ we replace~$e$ with either~$e_r$ or~$e_l$ but never both. The new vertices~$u^e_l,u^e_r$ as well as~$x_l$ are all added to the graphs~$\mu(s_l'),\mu(s_r)$ for the shores~$s_l',s_r \in \Shore(G)$ such that~$s_l' \in \CCC_1^l(p)'$ and~$s_r \in \CCC_1^r(p)'$ with respect to~$\delta_l,\delta_r$, i.e.,~$\mu'(s_l') = \mu(s_l')\cup \{u^e_l \mid e \in \mu(p)\}$ and~$\mu'(s_r) = \mu(s_r)\cup \{u^e_r \mid e \in \mu(p)\}\cup \{p_r\}$. Note that we leave~$x_r,x_l \in V(\Gamma')$.
    
    One easily verifies that this yields a strong coastal map noting that (SM2) is assured by construction since the new port~$p_r$ is still part of~$\mu'(s_r)$ for~$p \in \mu(s_r)$ and we only rename the port. Note further that the incidence of~$p_r$ with~$v$ is captured in~$\mu(s_r)$ while the incidence of~$p_l$ with~$x_l$ is contained in~$\Gamma'$. 
    
\smallskip

    Finally if~$\LLL$ is a rigid~$k$-linkage in~$G$ then using \cref{lem:rigid_linkages_after_cutting_edges} we deduce that~$\LLL'$ is a rigid~$k+(d+1)$-linkage in~$G'$ for we split at most~$d+1$ edges. The last part of the lemma is imminent recalling the \cref{def:shipping} of (piercing) shipping.

\end{proof}

Before continuing with the main result of this section, we introduce \emph{direct routes} for ease of argumentation in the proofs to come. Recall \cref{def:sigma_routes,def:route-tracing_cut-line}.

\begin{definition}[Direct route]
    Let~$(\Gamma,\CCC_1,\ldots,\CCC_r,\mu)$ be a strong coastal map of~$G$ in~$\Sigma$. Then a \emph{direct route (in~$\Gamma)$} is a route-tracing cut-line~$\ell$ that is internally disjoint from~$U(\Gamma)$.
    \label{def:direct-route}
\end{definition}

The following theorem which is an analogue of \cite[Theorem 6.3]{GMXXI} is proved by induction on~$\Sigma$, where the base-case will be proved subsequently as \cref{lem:bounding_jumps_strong_maps_disc}. 

\begin{theorem}
    Let~$G+D$ be Eulerian and let~$\LLL$ be a rigid~$p$-linkage in~$G$ for some~$p \in \N$. For every surface~$\Sigma$ and all integers~$p,d,r \geq 0$ there exists~$\lambda \coloneqq \lambda(p,d,r;\Sigma) \geq 0$ such that, if~$G$ admits a strong coastal map~$(\Gamma, \CCC_1, \dots, \CCC_r, \mu)$  with~$r$ sights and depth~$d$ such that every~$\Sigma_\ZZZ$-bounce is a~$\Sigma_\ZZZ$-route, then there are at most~$\lambda$~$\Sigma_\ZZZ$-routes (or~$\Sigma$-routes).
\label{thm:bounding_jumps_strong_maps}
\end{theorem}
\begin{proof}

    Let~$\lambda \coloneqq \lambda(p,d,r;\Sigma)$ be defined inductively as follows:
    \begin{align*}
        \lambda(p,q,r;\Delta) &\coloneqq \lambda_{\ref{lem:bounding_jumps_strong_maps_disc}}(p,q,r), \text{ for any disc } \Delta, \\
       \lambda(p,q,r;\Sigma) &\geq \lambda(p + (2d+2),d,r+2;\Sigma') \text{ for every simpler }\Sigma',\text{ and}\\
         \lambda(p,q,r;\Sigma) &\geq \lambda(p + 2d +2,d,r+1;\Sigma_1)+ \lambda(p + 2d +2,d,r+1;\Sigma_2)\text{ for simpler surfaces }~\Sigma_1,\Sigma_2.
    \end{align*}
    Then we claim that~$\lambda$ satisfies the theorem. The proof follows by induction on~$\Sigma$ where the case that~$\Sigma$ has no sights ($r=0$) or no depth ($d=0$) is not considered for we can simply cap the surface in these cases and thus there is no boundaries to consider and we say that~$G+D$ has no~$\Sigma$-routes.
    
    The base-case, namely that~$\Sigma$ is a disc with~$r\geq 1$ sights of depth~$d\geq 1$---the anchor of the induction---is proven in \cref{lem:bounding_jumps_strong_maps_disc}. We thus have the following induction hypothesis.
    \begin{enumerate}
        \item[($\star$)] For every surface~$\Sigma'$ such that either~$\hat{\Sigma'}$ is simpler than~$\hat{\Sigma}$ or~$\hat{\Sigma'} \cong \hat{\Sigma}$ and~$c(\Sigma') < c(\Sigma)$ the following holds true: For all integers~$d,r \geq 0$ and~$p \geq 1$ there exists~$\lambda(p,r,d;\Sigma')\geq 0$ satisfying the theorem.
    \end{enumerate}
        To this extent let~$(\Gamma, \CCC_1, \dots, \CCC_r, \mu)$ be a strong coastal map of~$G$ in~$\Sigma$ of depth~$\leq d$ and with~$r$ sights. Let~$\LLL = \{L_1,\ldots,L_p\}$ be a rigid~$p$-linkage in~$G$ with pattern~$D$ such that every~$\Sigma_\ZZZ$-bounce is a~$\Sigma_\ZZZ$-route. Using \cref{cor:coastal_map_after_splitting_off_along_rigid_is_shipping} (and especially 3. thereof) we may assume that~$\Gamma$ consists solely of two-paths between cuffs and each such edge is part of some island-zone, i.e., is a port. In particular then we may assume that~$\LLL$ is a piercing~$p$-shipping.
        
        \begin{claim}
            Every cut-line~$\ell$ with both its ends in cuffs~$c_1,c_2 \in c(\Sigma)$ (probably the same) and~$\ell \cap U(\Gamma) = \{p_1,p_2\}$ for~$p_1,p_2 \in \Port(G)$ is separating.
        \label{thm:bounding_jumps_strong_maps_claim1}
        \end{claim}
        \begin{ClaimProof}
                        Let~$\ell$ be such a cut-line and assume that it is non-separating. In particular then~$\ell$ must either be between two cuffs or have both its ends in the same cuff, where in the latter case~$\ell$ must be \emph{essential}, that is, non-nullhomotopic after contracting the respective cuff in~$\Sigma$ to a single point. We proceed by showing how to get a coastal map with depth~$d'$ and~$r'$ sights on some surface~$\Sigma'$ satisfying the hypothesis~$(\star)$. Let~$E_i \coloneqq \mu(p_i)\cup\{p_i\}$ be the respective cut-sets for~$i = 1,2$ (note that they may be equal) assuming that the~$p_i$ are not ends of some~$\CCC_j$, for otherwise we could proceed without splitting at those ends. Let~$G_1,D_1,\LLL_1$ be obtained from~$G,D,\LLL$ by splitting the edges in~$E_1$, then~$G_1 + D_1$ is Eulerian and~$\LLL_1$ is a piercing~$(p+d+1)$-shipping as follows from applying \cref{lem:rigid_linkages_after_cutting_edges} at most~$d$ times. Using \cref{lem:cutting_coastal_maps_at_shores} we get a coastal map~$(\Gamma', \CCC_1^l, \CCC_1^r, \CCC_2 \dots, \CCC_r, \mu_1)$ with~$r+1$ sights and depth~$\leq d$ such that~$\mu_1(s_1) = \emptyset$. We repeat this process for~$E_2$ and end with an Eulerian graph~$G'+D'$ with a~$(p+(2d+2))$-shipping and a coastal map~$(\Gamma'', \CCC_1^l, \CCC_1^r, \CCC_2^l,\CCC_2^r,\CCC_3 \dots, \CCC_r, \mu')$ with~$r+2$ sights and depth~$\leq d$ (note that technically we could end up with less new coastlines if we were to split at~$\CCC_1^l$ or~$\CCC_1^r$ again, a case we will not consider for it is analogous). By construction~$\ell$ is now disjoint from~$U(\Gamma'')$ (after redrawing the graph by slightly pulling the newly introduced vertices away from the point of intersection with~$\ell$). Finally  we can cut~$\Sigma$ along~$\ell$ to combine the (possibly) two cuffs into a single cuff and call the surface~$\Sigma'$. Since~$\ell \cap U(\Gamma'') = \emptyset$, clearly~$\Gamma''$ is still embedded in~$\Sigma'$ and~$\LLL$ is still a piercing~$(p+2d+2)$-shipping. The new cuff (either from combining two cuffs or from enlarging one) has two new shores, i.e., the two copies of the line~$\ell$ which lie on opposite sides of the cuff, call the shores~$s_{\ell}^1,s_{\ell}^2$. The two shores connect~$\CCC_1^l$ with~$\CCC_2^l$ and~$\CCC_1^r$ with~$\CCC_2^r$ respectively (after probably switching~$l$ and~$r$). We simply leave both shores deserted. The claim follows by induction using~$\lambda(p,d,r;\Sigma) \geq \lambda(p+(2d+2),d,r+2;\Sigma')$, for~$\Sigma'$ is a simpler surface as we have either reduced the number of cuffs or we have reduced the genus of~$\Sigma$ if~$\ell$ was an essential cut-line with its two endpoints in the same cuff. Either way the theorem follows by definition of~$\lambda$ together with the hypothesis~$(\star)$ and thus either the theorem or the claim hold (hence assume the latter for the remainder of the proof).
        \end{ClaimProof}
        
        \begin{claim}
            It holds~$\Abs{c(\Sigma)} = 1$.
            \label{thm:bounding_jumps_strong_maps_claim2}
        \end{claim}
    \begin{ClaimProof}
        Towards a contradiction assume the contrary, i.e.,~$\Abs{c(\Sigma)} \geq 2$. Recall that~$\Gamma$ is pseudo Euler-embedded and consists solely of two-paths with both ends in cuffs. Since~$G$ is connected together with (SM1) and the previous \cref{thm:bounding_jumps_strong_maps_claim1}, there exists at least one two-path~$(e_1,e_2)$ with both edges in~$\Gamma$,~$e_1,e_2 \in \Port(G)$ and terminals in different cuffs. Using that~$\Gamma$ is embedded we may draw a direct route~$\ell$ with ends~$e_1,e_2$, in particular it is internally disjoint from~$U(\Gamma)$ by \cref{def:route-tracing_cut-line}. Further using the assumptions on the embedding, every port~$p\in \Zone(C)$ is drawn on a cuff and thus~$\ell$ is a direct route between the two cuffs that~$e_1$ and~$e_2$ are adjacent too, such that it only intersects~$e_1,e_2$ in~$U(\Gamma)$. By \cref{thm:bounding_jumps_strong_maps_claim1} the curve~$\ell$ must be separating, but it is not; contradiction.
    \end{ClaimProof}

    By \cref{thm:bounding_jumps_strong_maps_claim2} we deduce that all~$\Sigma$-routes are from a single cuff to itself, call the cuff~$C_\Sigma \in c(\Sigma)$. Next we prove that following \cref{thm:bounding_jumps_strong_maps_claim1} and \cref{thm:bounding_jumps_strong_maps_claim2} we deduce that~$\Sigma$ is a disc; the induction anchor.
    
    \begin{claim}
     Let~$\ell$ be a direct route. Then there exists a disc~$\Delta \subseteq \Sigma$ with~$\ell \subseteq \bd(\Delta)\cap \Sigma$.
         \label{thm:bounding_jumps_strong_maps_claim3}
    \end{claim}
    \begin{ClaimProof}
        As above recall that the direct route~$\ell$ is a cut-line having its ends in the single cuff~$C_\Sigma$ sharing no internal points with~$U(\Gamma)$, for both the ports are forming the ends of~$\ell$. 
       By \cref{thm:bounding_jumps_strong_maps_claim1} the direct route~$\ell$ is separating and by \cref{thm:bounding_jumps_strong_maps_claim2} it starts and ends in the same cuff. Let~$\Sigma_1$ and~$\Sigma_2$ be the two surfaces arising after cutting~$\Sigma$ along~$\ell$ (and the maps at the ports as in the previous claims), i.e.,~$\Sigma_1 \cup \Sigma_2 = \Sigma$ and~$\Sigma_1 \cap \Sigma_2 = \ell$. If neither of both is a disc, then both are simpler than~$\Sigma$. Let~$\Gamma_i =\Gamma[[[\Sigma_i]]]$ and~$G_i = \Gamma_i \cup \bigcup_{C_i \in c(\hat{\Sigma_i})}\mu(\Zone(C_i))$. Then again using \cref{lem:cutting_coastal_maps_at_shores} and \cref{lem:rigid_linkages_after_cutting_edges} we get two new coastal maps with~$\leq r+1$ sights, depth~$\leq d$ for~$G_i$ in~$\hat{\Sigma_i}$ and a~$p+(2d+2)$-shipping for~$i=1,2$; this yields at most~$\lambda(p+2d+2,d,r+1;\Sigma_1) +\lambda(p+2d+2,d,r+1;\Sigma_2)~$ distinct~$\Sigma$-routes. The theorem in this case follows using the definition of~$\lambda$ and hypothesis~$(\star)$. Thus we can assume that one, say~$\Sigma_1$, is a disc, which yields the claim by construction.
    \end{ClaimProof}

    Note that \cref{thm:bounding_jumps_strong_maps_claim3} implies that every~$\Sigma_\ZZZ$-route~$P_\ell$---and every~$\Sigma$-route which immediately yields a~$\Sigma_\ZZZ$-route---is homotopic in~$\Sigma$ to the boundary curve on~$C_\Sigma$ bounded by its endpoints: for~$\ell$ is homotopic to~$\bd(\Delta)\cap C_{\Sigma}$ by~\cref{thm:bounding_jumps_strong_maps_claim3} and clearly~$\nu(P_\ell)$ is homotopic to~$\ell$ in~$\Sigma$. Thus every two-path in~$\Gamma$, being a~$\Sigma_\ZZZ$-route, is part of a disc and each of those discs is partly bounded by the same cuff. Since all the ports~$p \in \Port(G)$ are drawn on the cuff~$C_\Sigma$, we deduce that all the discs satisfy the exact same properties as in \cite[Theorem 6.3]{GMXXI}. As shown by Robertson and Seymour in the proof of \cite[Theorem 6.3]{GMXXI} it follows that~$\Sigma$ itself is a disc; the proof then follows using the induction anchor.

\end{proof}

We now prove \cref{thm:bounding_jumps_strong_maps} in the case that~$\Sigma$ is a disc, which, after the preliminary work done above and some more observations regarding pseudo Euler-embeddings given below, follows along the same lines to the respective proof in \cite{GMXXI}. Unfortunately there are two small caveats that make the proof in our setting way more cumbersome: the fact that we have edge-disjoint paths and the nuisance called \emph{traps}. Note that in the setting of \cite{GMXXI}, the terminals can be assumed to be part of~$I$, more importantly disjoint from~$\Gamma$ and thus there are no traps; an assumption we cannot easily reproduce and thus a nuisance we have to deal with.

\begin{observation}
    Let~$(\Gamma,\CCC_1,\ldots,\CCC_r,\mu)$ be a strong coastal map of~$G$ in a disc~$\Delta$ after applying \cref{cor:coastal_map_after_splitting_off_along_rigid_is_shipping}. Then there exists an injective map~$\varphi: \{\ell \mid \ell \text{ is a~$\Sigma$-route}\} \to \{\ell \mid \ell \text{ is a direct route}\}$ such that~$\varphi(\ell_1), \varphi(\ell_2)$ are pairwise disjoint lines for distinct~$\Sigma$-routes~$\ell_1,\ell_2$.
    \label{lem:from_route_to_direct}
\end{observation}
\begin{proof}
This follows at once using \cref{obs:sigma_routes_are_port_paths} and the fact that~$\restr{\LLL}{\Gamma_\ZZZ}$ are vertex-disjoint paths by construction. Simply draw the direct routes in parallel to the~$\Sigma_\ZZZ$-routes.
    
\end{proof}

\cref{lem:from_route_to_direct} is an easy but important observation, for direct routes are routes with endpoints in island-zones (and in particular on cuffs), and thus, seeing edges drawn on~$\Sigma$ as points, they are pairwise 'point-disjoint', i.e., simply disjoint lines. This allows for the main arguments behind the proofs by Robertson and Seymour in \cite{GMXXI} to transfer `intuitively' to our setting. Note however that as mentioned above there are two new types of obstacles in our setting. First since we look at edge-disjoint paths, the direct routes do not necessarily separate~$\Gamma$ into two disjoint components since ports come in pairs that are adjacent to some common vertex of the boundary, leaving room to pass from one side of the direct route to the other without violating the embedding properties. This is in contrast to the setting in \cite{GMXXI} where every~$\Sigma$-route yields a separation in~$\Gamma$. The second problem stems from the fact that we have traps, that is~$V(D) \cap V(\Gamma) \neq \emptyset$. We will discuss this in more details during the respective proofs. Prior to tackling the proof we return our focus to traps and trapped routes as promised in \cref{subsec:embedded-incidence-digraphs} and start with an obvious observation that will come in handy shortly.

\begin{observation}\label{obs:traps_in_deg_two_are_same_face}
    Let~$\Gamma$ be a pseudo-Eulerian graph with maximum degree two such that each edge has an incidence with a boundary-vertex. Further assume that~$\Gamma$ is pseudo Euler-embedded in a disc~$\Delta$ with cuff~$C_\Sigma$. Then every trapped route has both its traps in a common face of the embedding.
\end{observation}
\begin{proof}
  First note that for any trap~$v \in V(\Gamma)$ there exists a unique maximal path starting (or ending) in~$v$ and ending (or starting) in some other trap~$u$ using the fact that all vertices have degree at most two; in particular that path is trapped. The remainder is straightforward using the embedding restrictions given by the specification for boundary-vertices in \cref{def:embedding_boundary_vertices}: trace the trapped route by keeping the face on the right side of the path say (using the orientation of the disc), then if the side switches we either get a self-intersection---impossible for we have a planar drawing---or a violation to the embedding being strongly planar.
\end{proof}

To highlight the similarities to the respective proofs in \cite{GMXXI} note however that we are now bounding `point-disjoint' two-paths using arguments that are highly depending on the drawing, i.e., we may equally well bound the number of undirected two-paths for they are the same, where the tools needed to do so can be used as black-boxes due to our previous work, hiding the cumbersomeness of directions. We want to highlight that the main ingredient to the following proof is \cref{thm:rigid_linkage_in_laminar_cuts_is_Menger_general} which in turn is analogous to \cite[Theorem 2.6]{GMXXI} and can be applied in the same spirit (with some caution); note here that the lemma is only interested in cuts induced in the pseudo-Eulerian graph~$G$ by some~$X \subset V(G)$ and not in~$G+D$ which is crucial. 

\smallskip

We give a detailed proof to the first of the following sequence of lemmas, highlighting how to adapt the respective proofs in \cite{GMXXI} to our setting, and addressing the nuisances that differentiate our setting from the undirected non-eulerian one; see \cref{fig:bounding_jumps_disc}. 

\begin{figure}
    \centering
    \includegraphics[width=.4\linewidth]{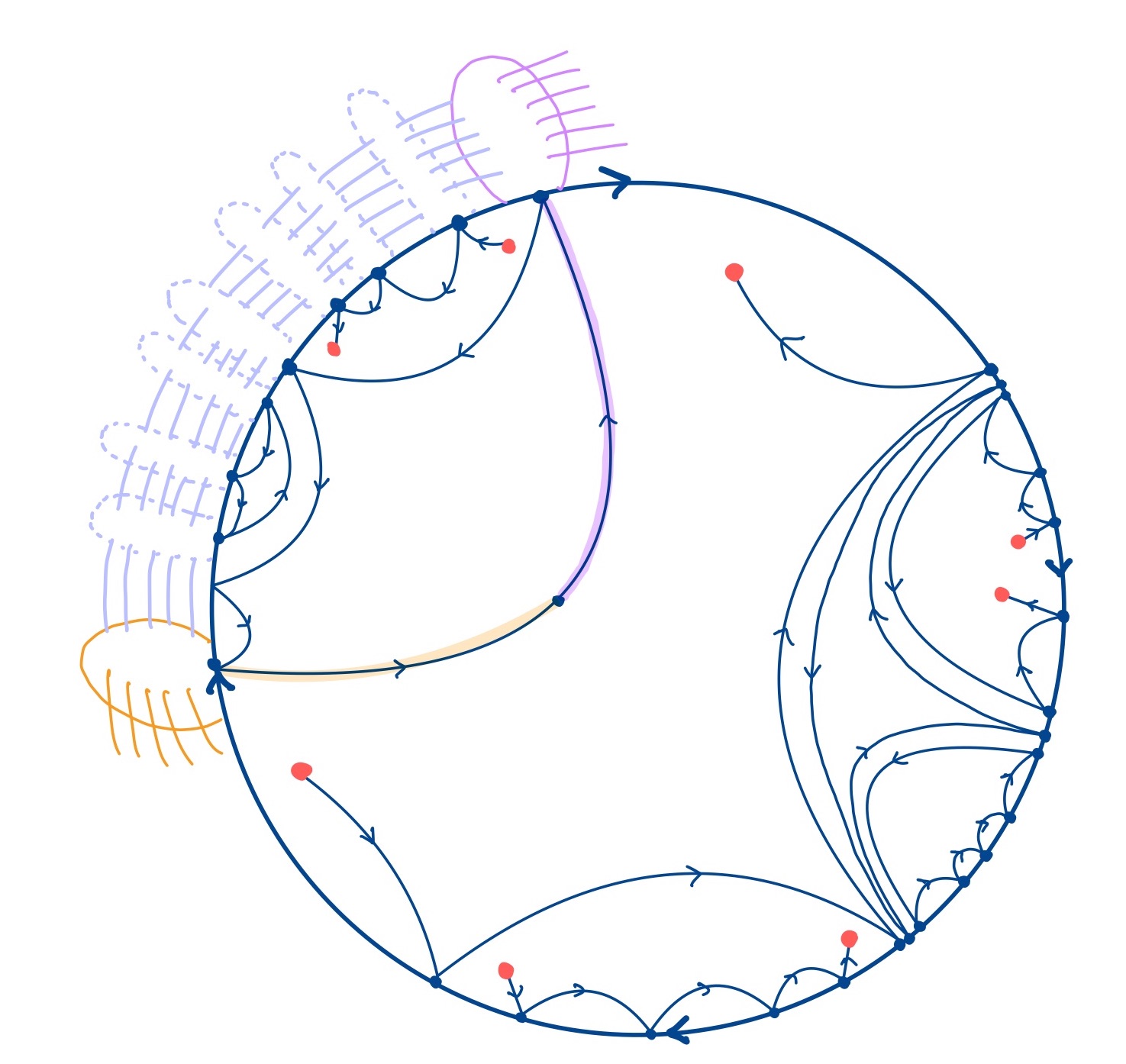}
    \caption{An exemplary illustration of the setting in \cref{lem:bounding_jumps_strong_maps_disc_r1}, where the traps are highlighted by thick red dots. Two ports have been highlighted by orange and purple highlight, and the graphs to the adjacent shores as well as the cuts given by the coastal maps and the ports have been schematically highlighted.}
    \label{fig:bounding_jumps_disc}
\end{figure}

\begin{lemma}[Adapted from {\cite[Theorem 6.4]{GMXXI}}]\label{lem:bounding_jumps_strong_maps_disc_r1}
    \cref{thm:bounding_jumps_strong_maps} holds true if~$\Sigma$ is a disc and there is only one sight, i.e.,~$r = 1$.
\end{lemma}
\begin{proof}
First using \cref{cor:coastal_map_after_splitting_off_along_rigid_is_shipping} we assume that~$\Gamma$ is a collection of paths where all the paths are two-paths and in particular all the paths consist of two ports, in particular~$U(\Gamma)$ is a collection of internally disjoint lines with both endpoints in the unique cuff~$C \in c(\Sigma)$. In particular then, since all the~$\Sigma_\ZZZ$-bounces are~$\Sigma_\ZZZ$-routes we deduce that~$\LLL$ is a piercing~$p$-shipping; this is crucial.

\smallskip

Choose~$n \geq f_{\ref{thm:rigid_linkage_in_laminar_cuts_is_Menger_general}}(p,2d)$ and let~$\lambda = 4n(n^n)+4np$. Let~$(\Gamma, \CCC, \mu)$ be a strong coastal map of~$G$ in~$\Sigma$ with depth~$\leq d$ (and one sight) where~$\Gamma$ is pseudo Euler-embedded in the disc~$\Sigma$. Further let~$\CCC = (p_1,s_1,p^1,s_1^*,p_2,\ldots,p_{t},s_{t},p^{t})$ for some~$t \in \N$; recall \cref{obs:evenly_many_ports}. Using \cref{lem:from_route_to_direct} it suffices to prove the statement for the respective direct routes given by the map~$\varphi$ and in particular it suffices to prove that the number of direct routes in general is bounded. Note that when counting the direct routes we may equally well forget about their orientation and simply count the number of respective curves on~$\Sigma$ (this is the same problem as in {\cite[Theorem 6.4]{GMXXI}}). Assume for the sake of contradiction that there are more than~$\lambda$ pairwise disjoint direct routes in~$\Gamma$. The ends of any such direct route~$\ell$ are two distinct ports~$p_1,p_2 \in \Port(G)$ (for there is  no edge with both endpoints in the same cuff after possibly splitting the edge as explained in the setting introduced in \cref{subsec:embedded-incidence-digraphs}), and thus~$\ell$ gives rise to two sequences~$J(\ell),K(\ell) \subseteq \Zone(C)$ that together cover the island-zone, i.e.,~$J(\ell) \cup K(\ell) =\Zone(C)$ and they only agree in the ports~$p_1,p_2$. Let~$u_1,u_2 \in \nu^{-1}(\bd(\Sigma))$ be the respective incidences of~$p_1,p_2$ in~$I_\ZZZ$, i.e.,~$u_1,u_2 \in V(I)\cap V(\Gamma)$. Further let~$p_1',p_2' \in \Port(G)$ be distinct from~$p_1,p_2$ such that~$p_i'\in E(\Gamma)$ is adjacent to~$u_i$ for~$i=1,2$, i.e.,~$u_i$ is the vertex adjacent to~$p_i',p_i$; note that these are all the incidences of~$u_i$ in~$\Gamma$ using the fact that~$\Gamma$ is pseudo Euler-embedded in the disc.

\smallskip
    
By the \cref{def:strong_coastal_map} of strong coastal maps we know that all the ports are covered by~$\CCC$ and thus without loss of generality~$K(\ell)\subseteq \CCC$ (and in particular~$J(\ell) \not \subseteq \CCC$; recall that~$\CCC$ misses at least one shore of the cuff by \cref{obs:disjointness_of_strong_coastlines}). Let~$\Delta(\ell)\subset \Sigma$ be the closed disc in~$\Sigma$ bounded by~$\ell \cup K(\ell)$ where~$K(\ell)$ can be seen as a sub-sequence of the cuff (recall that ports are drawn on the cuff, see also \cref{def:ports_shores_zones}). Let~$\delta(\ell) \coloneqq \Gamma[[\Delta(\ell)]]\cap \Gamma[[\Sigma\setminus \Delta(\ell)]]$ where~$\ell$ does not pass through any vertex of~$\Gamma$ for it is a cut-line.

\smallskip

Here we start deviating from the proof of \cite[Theorem 6.4]{GMXXI}. As a first step we want to get rid of trapped routes in~$\Gamma$ by proving that every trapped route has bounded length and thus in particular every trapped route yields only a bounded number of~$\Sigma_\ZZZ$-routes. Once we have dealt with trapped routes we are left to deal with the fact that our paths are only edge-disjoint and thus they may agree on vertices on the boundary. In both cases we will need what we will call \emph{loose} and \emph{tight lines}: Let~$\ell$ be a direct route, then we call~$\ell$ \emph{tight} if~$\delta(\ell) = \emptyset$ and otherwise we call it \emph{loose}. We further define
    \begin{align*}
        A(\ell) &\coloneqq \Gamma[[\Delta(\ell)]] \:\cup \:\mu(K(\ell)),\\
        B(\ell) &\coloneqq \Gamma[[\Sigma \setminus \Delta(\ell)]] \:\cup \:\mu(J(\ell)),\\
        X(\ell) &\coloneqq E(A(\ell) \cap B(\ell)).
    \end{align*}
For tight lines we get the following.
\begin{claim}\label{thm:bounding_jumps_strong_maps:claim1}
    Let~$\ell$ be a tight line with ends~$p_1,p_2 \in \Port(G)$. Then~$X(\ell)=\mu(p_1) \cup \mu(p_2)$.
\end{claim}
\begin{ClaimProof}
    It is clear that~$(A(\ell),B(\ell))$ is a separation and that~$\mu(p_1) \cup \mu(p_2) \subseteq X(\ell)$ by definition.
       Thus assume that there is an edge~$e \in A(\ell) \cap B(\ell)$ with~$e \notin \mu(p_1) \cup \mu(p_2)$. We claim that~$e \notin \mu(\CCC)$. To this extent note that~$\mu(p_1)\cup\mu(p_2)$ is an edge-cut between~$\mu(K(\ell))$ and~$\mu(J(\ell))$ using (SM4) and (SM3) of~$\CC$. Further using (SM4) and (SM3) again we deduce that if~$e \in \mu(\CCC)$, with ends~$e=(u,v)$ say, has both its endpoints in~$\mu(\CCC)$, then~$e \in \mu(p_1)\cup \mu(p_2)$. So we are left with the case that~$e \in \mu(\CCC)$ has one end-point in~$\mu(\CCC)$ and one in~$\Gamma_\ZZZ$, i.e.,~$e \in \Port(C)$ by \cref{obs:ports_are_divided_by_separation_tailored_to_zone} and the fact that~$\Gamma_\ZZZ \cap I_\ZZZ = \Port(G)$.
       
       Then since~$\ell$ is tight, and thus splits~$\Gamma$ (and in particular~$\Gamma_\ZZZ)$ into two disjoint components, we derive that~$e \notin \Gamma$. Since~$\Gamma \cup I_\ZZZ = G$ this concludes the proof. Note here that the edges with an end in~$I_\ZZZ$ and an end in~$\Gamma_\ZZZ$ are exactly the ports which are completely contained in~$\Gamma$.
\end{ClaimProof}
See \cref{fig:bounding_jumps_disc} where the highlighted two-path schematically shows a tight line, where~$A(\ell)$ can be seen to be the graph above the highlighted to-path together with the schematic representation of~$\mu(K(\ell))$. 

For loose lines we get a similar but slightly looser result.
\begin{claim}\label{thm:bounding_jumps_strong_maps:claim2}
    Let~$\ell$ be a loose line with ends~$p_1,p_2 \in \Port(G)$. Then~$X(\ell) \subseteq \mu(p_1) \cup \mu(p_2) \cup \{p_1,p_2\}$.
\end{claim}
\begin{ClaimProof}
    This follows as in \cref{thm:bounding_jumps_strong_maps:claim1} noting that the only possible edges of~$\Gamma$ that can be part of a cut, using the fact that~$\ell$ is a direct route, are indeed~$\{p_1,p_2\}$.
\end{ClaimProof}

As mentioned above we start by getting rid of trapped routes (this step is in a sense unnecessary for we could treat trapped routes as components in what is to follow.)
\begin{claim}\label{thm:bounding_jumps_strong_maps_claim_traps}
    Let~$L \subset \Gamma$ be a trapped route (not necessarily~$L \in \LLL$). Then~$L$ contains at most~$4n$ edge-disjoint~$\Sigma$-routes in~$\Gamma$.
\end{claim}
\begin{ClaimProof}
    Let~$L \subset \Gamma$ be a trapped route and let~$(u_{i_1},\ldots,u_{i_{4n+1}})$ be an ordering of~$V(L) \cap \nu^{-1}(C)$ with respect to~$\CCC$, that is~$u_{i_j} \in \mu(s_{i_j})$ for~$1 \leq i_j \leq t$ and~$1 \leq j \leq 4n+1$ using \cref{obs:adj_ports_give_nondeserted_shore_and_are_even}.
    
    Then for every pair~$(u_{i_j},u_{i_{j+1}})$ we have that~$L_{i_j}\coloneqq u_{i_j}Lu_{i_{j+1}}$ is a~$\Sigma$-route. Using the fact that~$\Gamma$ is pseudo Euler-embedded and the fact that there is no vertex of degree four left, \cref{obs:traps_in_deg_two_are_same_face} implies that all the edges of~$L$ are adjacent to some common face in the embedding~$(U,\nu)$. Again using the fact that~$\Gamma$ is pseudo Euler-embedded we deduce that there exists~$1 \leq k \leq 4n+1$ such that~$u_{i_k},\ldots,u_{i_{4n+1}},u_{i_1},\ldots,u_{i_{k-1}}$ appear in this order on~$\CCC$ (with obvious cyclic coordinates and corner cases if~$k\in \{1,4n+1\}$). We will look at~$(u_{i_1},\ldots,u_{i_{k-1}})$, where~$i_1 \leq \ldots \leq i_{k-1}$, and without loss of generality assume that~$k \geq 2n+2$, for else~$2n+1 - k \geq 2n$ and we can proceed analogously by taking~$(u_{i_k},\ldots,u_{i_{4n+1}})$. 

\smallskip
    Given the sequence~$(u_{i_1},\ldots,u_{i_{k-1}})$ we get a maximal sub-path~$L' \subset L$ visiting exactly the boundary-vertices~$u_{i_1},\ldots,u_{i_{k-1}}$. In particular by the maximality assumption then~$L'$ contains both ports adjacent to every~$u_{i_1},\ldots,u_{i_{k-1}}$. Now~$L'$ may contain at most two of the traps of~$L$ and both are either adjacent to~$u_{i_1}$ and~$u_{i_k}$ or to two consecutive vertices in the sequences~$(u_{i_1},\ldots,u_{i_{k-1}})$ simply using the fact that~$L$ is a path visiting said vertices in said order. Thus since~$k-1 \geq 2n+1$ there exists a sub-path~$L'' \subset L'$ visiting a sub-sequence of order~$n+1$, say~$(u_{i_1},\ldots,u_{i_{n+1}})$; for simplicity we write~$(u_1,\ldots,u_{n+1})$. In particular~$L''$ has~$n+3$ edges where only the terminals may be traps and thus~$L''$ has at least~$n+1$ edges that are disjoint from terminals of~$L$. We write~$L'' = (l(u_1),r(u_1),\ldots,l(u_{n+1}),r(u_{n+1}))$ where~$l(u_i),r(u_i) \in \Port(G)$ are the edges adjacent to~$u_i$ for every~$1 \leq i \leq n+1$. Using~$L''$ we can construct a sequence of~$n+3$ loose lines as follows.

    Define~$\ell_i$ to be the direct route with ends~$l(u_{1}),r(u_{i})$ for every~$1\leq i \leq n+1$ such that~$\ell_i$ is drawn in the same face as~$L$ (and thus~$L''$) by closely following~$L''$. Then by construction (drawing the~$\ell_i$ accordingly) it holds that~$\Delta(\ell_i) \subset \Delta(\ell_j)$ for every~$1\leq i < j \leq n+1$ where the lines are loose by construction. Applying \cref{thm:bounding_jumps_strong_maps:claim2} then~$X(\ell_i) \subseteq \mu(l(u_1)) \cup \mu(r(u_i)) \cup \{l(u_1),r(u_i)\}$ where we get equality for~$L''$ is a path from~$l(u_1)$ to~$r(u_i)$ with its ends in~$\Gamma[\Sigma\setminus \Delta(\ell_i)]$ by construction, and thus both ports are part of the cut, in particular of~$\delta(\ell_i)$, thus
    $$ X(\ell_i) = \mu(l(u_1)) \cup \mu(r(u_i)) \cup \{l(u_1),r(u_i)\},$$
    and again by construction~$\big(A(\ell_i) ,B(\ell_i)\big)_{1 \leq i \leq n+1}$ is a laminar cut-family with cut-edges given by~$X(\ell_i)$.
    
    Note that~$L''$ is a path connecting~$l(u_1)$ to~$r(u_i)$ in~$\Gamma[[\Delta(\ell_i)]]$. Then using \cref{lem:presentation_to_map_lemma_1,lem:presentation_to_map_lemma_2} together with the fact that~$L''$ connects~$\{l(u_1),r(u_i)\}$ and thus also~$r(u_i),r(u_{i+1})$ edge-disjointly from all the paths in~$I_\ZZZ$ and the linkages coming from \cref{lem:presentation_to_map_lemma_1}, we can apply \cref{thm:rigid_linkage_in_laminar_cuts_is_Menger_general} (note for example that we need 3. of \cref{lem:presentation_to_map_lemma_1} for 2. of \cref{thm:rigid_linkage_in_laminar_cuts_is_Menger_general} and we need \cref{lem:presentation_to_map_lemma_2} to extend the linkages to witness the Menger-linkage needed in \cref{thm:rigid_linkage_in_laminar_cuts_is_Menger_general}). And thus there exist~$1\leq i < i' \leq n+1$ such that~$A(\ell_i) \subseteq A(\ell_{i'})$ such that~$\restr{\LLL}{G[[B(\ell_i) \cap A(\ell_{i'})]]} = M^{ii'}$ where~$M^{ii'}$ is a linkage containing~$\restr{L''}G[[B(\ell_i) \cap A(\ell_{i'})]]$. This is a contradiction to the fact that~$\LLL$ was a piercing~$p$-shipping whereas~$\restr{L''}{G[[B(\ell_i) \cap A(\ell_{i'})]]}$ is a path containing at least two consecutive ports, and thus~$L''$ witnesses the existence of a~$\Sigma_\ZZZ$-bounce in~$\LLL$ as a contradiction to our assumptions and thus concluding the proof.
\end{ClaimProof}
Note here that the last step in the proof of the claim is why we introduced piercing shippings and bounces since if we try to apply the same arguments as in \cite[Theorem 6.4]{GMXXI} we run into the problem that loose lines do not yield cuts that are completely covered by the map: there can be an arbitrarily long sub-path of some~$L \in \LLL$ that runs in parallel with our nested direct routes and does not yield a contradiction for said path could be part of the Menger-linkage we get from \cref{thm:rigid_linkage_in_laminar_cuts_is_Menger_general}. Thus we needed to guarantee that the paths are \emph{piercing} and no path of the rigid linkage can run in parallel to the island. We will encounter this same argument again shortly.

\smallskip

With \cref{thm:bounding_jumps_strong_maps_claim_traps} at hand, and using the fact that there are at most~$p$ trapped routes in~$\Gamma$, there are at most~$4np$ distinct direct routes coming from trapped routes by applying \cref{thm:bounding_jumps_strong_maps_claim_traps} to every trapped route. Finally using \cref{obs:pseudo_Euler_can_be_capped} we immediately see that~$\Gamma^* \coloneqq \Gamma - \bigcup \LLL_{tr}$ is Eulerian where~$\LLL_{tr}$ is the linkage formed by the trapped routes in~$\Gamma$; it is unique for every vertex is of maximum degree two in~$\Gamma$ by our previous reduction. Thus we are left to prove that the Eulerian graph~$\Gamma^*$ contains at most~$n^n$ direct routes; we again refer to it as~$\Gamma$ for simplicity. Now~$\Gamma$ is Eulerian and it is pseudo Euler-embedded in a disc; its components are hence each Eulerian and can be traced via an Euler-cycle.

The following claim is analogous to the main argument in \cite[Theorem 6.4]{GMXXI} for tight lines agree with the respective lines in said proof.
\begin{claim}\label{thm:bounding_jumps_strong_maps_claim_nested_tight_lines}
    There exists no sequence of~$n+1$ tight lines~$\ell_0,\ldots,\ell_n$ such that~$\Delta(\ell_i) \subset \Delta(\ell_j)$ and~$\Gamma[\Delta(\ell_j) \setminus \Delta(\ell_i)] \neq \emptyset$ for~$0 \leq i < j \leq n$.
\end{claim}
\begin{ClaimProof}
    The proof follows analogously to the proof of \cref{thm:bounding_jumps_strong_maps_claim_traps} by noting that~$(A(\ell_i),B(\ell_i))_{0 \leq i \leq n}$ is a laminar cut-family and applying \cref{thm:bounding_jumps_strong_maps:claim2} together with \cref{lem:presentation_to_map_lemma_1,lem:presentation_to_map_lemma_2} and \cref{thm:rigid_linkage_in_laminar_cuts_is_Menger_general} deducing that there exist~$i,i'$ such that~$\Gamma[\Delta(\ell_{i'}) \setminus \Delta(\ell_i)] = \emptyset$ as opposed to the assumption of the theorem (this follows from the fact that the respective Menger-linkage~$M^{ii'}$ is completely contained in~$I_\ZZZ$ and thus since~$\LLL$ was rigid there cannot be edges left in~$\Gamma$.)
\end{ClaimProof}

We use the above to bound the number of components; in a sense the components here play the role of the edges in \cite[Theorem 6.4]{GMXXI}.

\begin{claim}\label{thm:bounding_jumps_strong_maps_claim_not_many_components}
    The graph~$\Gamma$ contains at most~$n^n$ disjoint components.
\end{claim}
\begin{ClaimProof}
    Let~$C_1,\ldots,C_{n^n+1} \subset \Gamma$ be~$n+1$ disjoint components where~$C_i$ already encodes the Euler-cycle of the component for every~$1 \leq i \leq n^n+1$. To every component~$C_i$ we identify an interval~$\kappa_i \subset [1,t]$ such that~$\kappa_i \in \{[l,r] \mid u_l,u_r \in V(C_i)\}$ is maximal; recall that~$u_l,u_r$ come from the order of~$\CCC$. First note that the endpoints~$u_l,u_r$ stemming from~$\kappa_i=[l,r]$ are connected by a two-path in~$C_i$ using the pseudo Euler-embedding of~$C_i$ in~$\Delta$ (otherwise the components must have a cross in itself); call the two-path~$e_i$ (it behaves like an edge). Let~$\kappa \coloneqq \{\kappa_i \mid 1 \leq i \leq n^n +1\}$ be the set of intervals and let~$\iota_1,\iota_2 \in \kappa$. Then using the fact that~$\Gamma$ is planar it holds either~$\iota_1 \cap \iota_2 = \emptyset$ or~$\iota_1 \subseteq \iota_2$ or~$\iota_2 \subseteq \iota_1$. In particular using Ramsey's theorem we either find a sequence of~$n+1$ distinct intervals that are nested or~$n+1$ intervals that are pair-wise disjoint; let~$\iota_0,\ldots,\iota_n \in \kappa$ be said intervals and~$e_{i_j} \in E(C_{i_j})$ the respective edges for~$1 \leq i_j \leq n^n + 1 $ and~$0 \leq j \leq n$ where~$\iota_{j} = \kappa_{i_j}$. 

    \begin{enumerate}
        \item[Case 1.] Assume that~$\iota_0 \subset \ldots \subset \iota_n$. Then draw a tight line~$\ell_{i_j}$ in parallel to and ending in the same edges for every two-path~$e_{i_j}$ with~$0 \leq j \leq n$. By construction then~$C_{i_j} \subset \Gamma[\Delta(\ell_{i_j})]$ (using the fact that by definition the intervals respect the single deserted shore of~$\CCC$). Since the intervals are nested we deduce that~$\Delta(\ell_{i_{j'}}) \subset \Delta(\ell_{i_j})$ and since all the components are disjoint we further deduce that~$C_{i_{j'}} \cap \Gamma[\Delta(\ell_{i_j})] = \emptyset$ for all~$j' > j$. Altogether then this is a contradiction to \cref{thm:bounding_jumps_strong_maps_claim_nested_tight_lines}.

        \item[Case 2.] Assume that~$\iota_0,\ldots,\iota_n$ are pair-wise disjoint. Then for every pair of distinct intervals~$\iota_i,\iota_j$ it holds that either both ends of~$\iota_i$ are smaller than both ends of~$\iota_j$ or vice-versa for every~$0 \leq i,j \leq n$; without loss of generality let the intervals be sorted via~$\iota_0 \prec \ldots, \prec \iota_n$ meaning the obvious. Then we get a sequence of~$n+1$ nested tight lines~$\ell_0,\ldots,\ell_n$ such that~$\Delta(\ell_0) \subset \ldots \subset \Delta(\ell_n)$ and~$C_{i_\ell} \subset \Gamma[\Delta(\ell_j)]$ and $C_{i_{\ell'}} \not\subset \Gamma[\Delta(\ell_j)]$ for all~$0 \leq \ell \leq j \leq n$ and all~$0 \leq j <\ell' \leq n$. We get the lines in the same way as in the proof of \cref{thm:bounding_jumps_strong_maps_claim_traps} by taking as a first tight line the one that `envelops' the component~$C_{i_0}$ and extending the line to a tight-line adding more and more components. 

        Again this yields a sequence of nested tight lines and is in contradiction with \cref{thm:bounding_jumps_strong_maps_claim_nested_tight_lines}.
    \end{enumerate}
This proves the claim.
\end{ClaimProof}

In a next and final step we bound the number of direct routes contained in each component of~$\Gamma$. But this is analogous to the case of trapped routes.

\begin{claim}\label{thm:bounding_jumps_strong_maps_claim_bounded_per_component}
    Let~$C \subseteq \Gamma$ be a component of~$\Gamma$. Then~$C$ contains at most~$4n$ direct routes.
\end{claim}
\begin{ClaimProof}
    This follows analogously to the case of trapped routes given in \cref{thm:bounding_jumps_strong_maps_claim_traps} for we only need to find a sub-path of length~$n+1$ that yields a sequence of disjoint intervals and then argue exactly as for the path~$L''$ found in said proof.
\end{ClaimProof}

Combining \cref{thm:bounding_jumps_strong_maps_claim_traps}, \cref{thm:bounding_jumps_strong_maps_claim_not_many_components} and \cref{thm:bounding_jumps_strong_maps_claim_bounded_per_component} concludes the proof.

\end{proof}
\begin{remark}
   In the proof we could have skipped the extra step dealing with trapped routes for trapped routes yield components and the proof is analogous to the proof that components are not too big. However it seems intuitively easier to get rid of trapped routes in order to convince one-self that one easily gets tight lines that separate components. If there were traps left one would have to pay attention how one draws the cut-line to always include the the traps which are not drawn on the cuff. In general the setting that~$\Gamma$ is Eulerian is easier to think about, another reason to argue why we can get rid of trapped routes first.
\end{remark}

Similar to the above, the proofs of the following results transfer almost verbatim from the provided proofs in \cite{GMXXI} keeping in mind that one has to be cautious with traps and the fact that we are edge-disjoint as we have done in the proof of \cref{thm:bounding_jumps_strong_maps}, and are thus omitted (for it would bloat up the rest of this section without gaining much new insights, as we established all the needed tools and made the relevant arguments above to adapt the respective proofs).

\begin{lemma}[Analogous to {\cite[Theorem 6.5]{GMXXI}}]
    \cref{thm:bounding_jumps_strong_maps} holds true if~$\Sigma$ is a disc and~$r = 2$.
    
\end{lemma}
\begin{proof}
    This is analogous to the respective proof of \cite[Theorem 6.5]{GMXXI} reducing it to~$r=1$ using the same insights as gathered in the proof of \cref{lem:bounding_jumps_strong_maps_disc_r1}. Note that there are only~$p$ trapped routes to consider, and every component that is not a trapped route but visits a vertex of each of the two coast lines does bound a face concentric to the cuff such that there are no other edges inside that face (compare this to the regions~$r_i$ in the proof of \cite[Theorem 6.5]{GMXXI}).
\end{proof}

\begin{lemma}[Analogous to {\cite[Theorem 6.6]{GMXXI}}]
    \cref{thm:bounding_jumps_strong_maps} holds true if~$\Sigma$ is a disc and~$r = 3$.
\end{lemma}
\begin{proof}
    This is analogous to the respective proof of \cite[Theorem 6.6]{GMXXI} reducing it to~$r=2$ and using the insights gathered in the proof of \cref{lem:bounding_jumps_strong_maps_disc_r1}.
\end{proof}
\begin{lemma}[Analogous to {\cite[Theorem 6.7]{GMXXI}}]
        \cref{thm:bounding_jumps_strong_maps} holds true if~$\Sigma$ is a disc with~$r\geq 1$.
        \label{lem:bounding_jumps_strong_maps_disc}
\end{lemma}
\begin{proof}This is analogous to the respective proof of \cite[Theorem 6.7]{GMXXI} inductively reducing it to~$r' \in \{1,2,3\}$ if~$r\geq 4$; cases that have been taken care of above.
\end{proof}

\subsection{Shipping between strong coastlines}

This rather short section is devoted to the proof that, given a strong coastal map of a high tree-width graph~$G$ with a bounded number of sights of bounded depth, then~$G+D$ cannot have a rigid linkage. In particular, there exists an irrelevant cycle in~$G$. Again we will use the setting introduced in \cref{subsec:embedded-incidence-digraphs} and in particular assume a drawing~$(U,\nu)$ of the graph to be given.

\begin{theorem}
    For every surface~$\Sigma$ and every~$p,r,d \in \N$ there exists~$\theta(p,r,d;\Sigma)$ such that the following holds true. Let~$G+D$ be Eulerian such that~$G$ has a strong coastal map~$\CC\coloneqq (\Gamma,\CCC_1,\ldots,\CCC_r,\mu)$ with~$r$ sights of depth~$d$ in~$\Sigma$ where~$\Gamma$ contains some swirl~$S$ of size~$\geq \theta(p,r,d;\Sigma)$ induced by a tile that is Euler-embedded in a disc~$\Delta \subseteq \Sigma$---in particular~$\tw(\Gamma) \geq \theta$. Then~$G$ has no rigid~$p$-linkage~$\LLL$ with pattern~$D$.
    \label{thm:no_rigid_linkages_on_strong_maps}
\end{theorem}
\begin{proof}
   Recall that~$G= \Gamma \cup I$ where both are pseudo-Eulerian and~$\Gamma$ is pseudo Euler-embedded. First apply \cref{lem:splitting_at_pairs_of_ports_remains_coastal} iteratively to pairs of ports that are consecutively visited by some path~$L \in \LLL$, noting that this cannot destroy much of the swirl for most of the swirl lies away from the boundary of~$\Sigma$ since it is drawn in a disc~$\Delta$ (this is the only reason why we assume the existence of a disc, note however that this is not explicitly needed since one can prove that if the swirl is large enough, then only a bounded number of its cycles can be adjacent to the boundary of~$\Sigma$). Thus we may assume that all~$\Sigma_\ZZZ$-bounces are~$\Sigma_\ZZZ$-routes; this is crucial.

   \smallskip
   
   Next let~$\mathcal{I}$ be the graphs of the islands of~$\CC$ (recall \cref{def:strong_island}), that is the collection of induced graphs in~$\mu(\Shore(G))$ and let~$I_\ZZZ = \bigcup \mathcal{I}$ be the partial graph resulting from combining all of the islands containing the edges in~$\Port(G)$ and let~$\Gamma_\ZZZ$ be the partial graph~$\Gamma-\nu^{-1}(\bd(\Sigma))$, i.e.,~$(\Gamma_\ZZZ,I_\ZZZ)$ is the separation tailored to~$\Zone$ as given by \cref{def:separation_tailored_to_zone}. Then~$G = I_\ZZZ \cup \Gamma_\ZZZ$ and~$I_\ZZZ\cap \Gamma_\ZZZ = \Port(G)$ by \cref{lem:separation_tailored_to_zone}, in particular~$\Port(G)$ is a cut in~$G$ and~$V(\Gamma_\ZZZ)\cap V(I_\ZZZ) = \emptyset$ by construction. So~$X\coloneqq V(\Gamma_\ZZZ)$ induces said cut.

    Assume for the sake of contradiction that the theorem is false and let~$\LLL$ be a rigid linkage in~$G$. Let~$G'\coloneqq G[[X]]$ and~$D' \coloneqq D_X^{\LLL}$ as by \cref{def:euler_extended_cut} and \cref{lem:Eulerian_extension_properties_rigid}, then~$G'+D'$ is Eulerian. Let~$\LLL' \coloneqq \restr{\LLL}{G[[X]]}$ be the induced rigid linkage of order~$\leq p + |\Port(G)|$ by \cref{lem:cuts_and_rigid_linkages}. By construction~$G'$ can be Euler-embedded in~$\hat{\Sigma}$, using a sub-drawing of~$\Gamma$.
    \begin{claim}
        It holds~$|\Port(G)| \leq 2\lambda_{\ref{thm:bounding_jumps_strong_maps}}(p,d,r;\Sigma) +2p$.
    \end{claim}
    \begin{ClaimProof}
        
        Since~$\LLL'$ is rigid, every edge~$e \in \Port(G)$ is an end of some~$\Sigma_\ZZZ$-route and every~$\Sigma_\ZZZ$-route contains exactly two ports.
        Using \cref{obs:sigma_routes_are_port_paths} then \cref{thm:bounding_jumps_strong_maps} implies that~$|\Port(G)| \leq 2\lambda_{\ref{thm:bounding_jumps_strong_maps}}(p,d,r;\Sigma)$.
    \end{ClaimProof}
    Thus~$\LLL'$ is a~$\leq 2\lambda_{\ref{thm:bounding_jumps_strong_maps}}(p,d,r;\Sigma) +p$ linkage in~$G'$ by \cref{lem:cuts_and_rigid_linkages}. Finally, using \cref{thm:shipping_in_open_sea} and choosing~$\theta(p,r,d;\Sigma) \geq f_{\ref{thm:shipping_in_open_sea}}\big(2\lambda_{\ref{thm:bounding_jumps_strong_maps}}(p,d,r;\Sigma) +p,\Sigma\big)$ the theorem follows. 
\end{proof}
\begin{remark}
    As mentioned in the proof, the assumption that~$\SSS$ is drawn in a disc is just a convenience assumption and not needed for the proof. It turns out that this is a valid assumption so it does not hurt to assume it (see \cref{thm:Johnsons_structure_theorem} for example).
\end{remark}

\subsection{Shipping between weak coastlines}
The goal of this section is to \emph{split} weak coastal maps and deform them into strong coastal maps for which we have already proven that there cannot exist rigid linkages unless the graph has low tree-width (more precisely the graph contains no large flat swirl). Note here that the main difference between weak and strong coastal maps is that deserted shores are not `actually' deserted in the case of weak coastal maps, i.e., it may happen that~$\mu(s) \neq \empty$ for a deserted shore~$s$ in a weak coastal map.

\begin{definition}[Split weak coastal map]\label{def:split_weak_map}
    Let~$G+D$ be an Eulerian graph and~$\Sigma$ some surface. Let~$\MMM \coloneqq (\Gamma, \CCC_1, \dots, \CCC_r, \mu)$ be a weak coastal map of~$G$ in~$\Sigma$. We say that~$\MMM$ is \emph{split} if for each deserted shore~$s \in \Shore(G) \setminus \bigcup \{\Shore(\CCC_i) \sth 1 \leq i \leq r \}$ with distinct incident ports~$p_l,p_r\in \Port(s)$ the graphs~$\mu(p_l)$ and~$\mu(p_r)$ are edge-disjoint and thus disjoint (recall that they are vertex-less). 
\end{definition}

A direct consequence to the \cref{def:split_weak_map} of split maps reads as follows.

\begin{observation} \label{obs:splitmap_deserted_shore_empty}
   Let~$\MMM \coloneqq (\Gamma, \CCC_1, \dots, \CCC_r, \mu)$ be a weak coastal map of~$G$ in~$\Sigma$. Let~$s \in \Shore(G)$ be a deserted shore. Then~$\mu(s) = \emptyset$.
\end{observation}
\begin{proof}
    First note that~$V(\mu(s)) = \emptyset$ by \cref{obs:deserted_shore_vertexless}. The claim then follows at once using \cref{lem:edges_define_coast lines} and the fact that~$\mu(p_l) \cap \mu(p_r) = \emptyset$ for the adjacent ports~$p_l,p_r \in \mu(s)$ by \cref{def:split_weak_map}. 
\end{proof}

In light of the above we have the following key result.

\begin{lemma}
    \label{lem:split_weak_map_is_strong}
    Let~$\MMM \coloneqq (\Gamma, \CCC_1, \dots, \CCC_r, \mu)$ be a weak coastal map. If~$\MMM$ is split then~$\MMM$ is strong. 
\end{lemma}
\begin{proof}
    Let~$H_i \coloneqq \mu(\CCC_i)$ for~$1 \leq i \leq r$. We will verify each of the axioms for a strong map.
    \begin{enumerate}
        \item[{\small(SM1)}]~$G = \Gamma \cup \bigcup_{i=1}^rH_i$ is clear from (WM1). Assume there is~$\chi \in H_i\cap H_j$ for some~$1 \leq i,j \leq r$. Let~$p_l,p_r$ be the ends of~$\CCC_i$ and~$p_l',p_r'$ be the ends of~$\CCC_j$ for some~$1\leq i < j\leq r$. By (WM1) we know that~$\CCC_i,\CCC_j \subset \Zone(C)$ for a common cuff~$C \in c(\Sigma)$. Further since~$\MMM$ is split we know that~$\mu(p_l)\cup\mu(p_r)$ is disjoint from~$\mu(p_l')\cup\mu(p_r')$, since the coastlines are mutually disjoint and thus there is at least two deserted shore~$s_1^\ast,s_2^\ast$ with~$s_1^\ast,p_l,p_r,s_2^\ast,p_l',p_r'$ appearing in this order on~$\Zone(C)$. Assume that~$\chi \in \mu(\chi_i)\cap \mu(\chi_j)$ for some~$\chi_i \in \CCC_i$ and~$\chi_j \in \CCC_j$ then~$p_l,\chi_i,p_r,p_l',\chi_2,p_r'$ appear in this order in~$\Zone(C)$. By (WM3) then~$\chi \in \big(\mu(p_l)\cup\mu(p_r)\big)\cap \big(\mu(p_l')\cup \mu(p_r')\big)$; a contradiction to them being disjoint. 
        
        \item[{\small(SM2)}] This is obvious by \WMII.
       
        \item[{\small(SM3)}] Let~$s,s' \in \Shore(\CCC_i)$ and~$p \in \Port(\CCC_i)$ lying between them (with respect to~$\CCC_i$) where~$\CCC_i \subset \Zone(C)$ for some~$C \in c(\Sigma)$. Let~$s^\ast$ be a deserted shore in~$C$ which exists for there are at least two ports on each cuff following our assumption that~$|C \cap V(\Gamma)|\geq 2$. Clearly~$s^\ast,s,p,s'$ occur in that order in~$\Zone(C)$. Let~$p_1,p_2$ be the ends of~$s^\ast$ such that~$p_1,s,p,s',p_2$ appear in that order in~$\Zone(C)$. Since~$\MMM$ is split we know that~$\mu(p_1)\cap \mu(p_2) = \emptyset$. Using \WMIII once for~$p_1,s,p,s'$ and once for~$s,p,s',p_2$ we deduce~$\mu(s)\cap \mu(s') \subseteq \mu(p)$.
        \item[{\small(SM4)}] Let~$p \in \Port(\CCC_i)$ be an interior port for some~$1 \leq i \leq r$ where~$\CCC_i$ is a coast of~$C \in c(\Sigma)$. By \WMIV there exist~$s_l,s_r \in \Shore(C)$ such that~$s_l,p,s_r$ appear in that order on~$\Zone(C)$ and both shores cover one end of~$e$; note that~$e \in \mu(s_l) \cap \mu(s_r)$. Using \cref{lem:edges_define_coast lines} and the fact that the map is split we deduce that~$\CCC(e)$ is a sub-sequence of~$\CCC_i$ and thus~$s_l,s_r \in \Shore(\CCC_i)$. Note here that split maps do not allow for edges between two different coast lines~$\CCC_i,\CCC_j$ regardless if they are of the same cuff for any~$1 \leq i\neq j \leq r$, i.e,~$\mu(s) = \emptyset$ for deserted shores. 
        
        \item[{\small(SM5)}] This is obvious by \WMV. 
        \item[{\small(SM6)}] This follows from \cref{obs:splitmap_deserted_shore_empty}.
    \end{enumerate}
\end{proof}

We define a splitting operation on weak maps as follows.

\paragraph{Defining~$\Split(\MMM,e)$ by splitting a weak map at an edge~$e$:} Let~$\CC \coloneqq (\Gamma,\CCC_1,\ldots,\CCC_r,\mu)$ be a weak map and let~$C \in c(\Sigma)$ be some cuff. Let~$s^\ast$ be a deserted shore in~$\Zone(C)$---which exists for every cuff---and let its ends be~$p_l,p_r \in \Zone(C)$ such that~$p_l,s^\ast,p_r$ appear in that order on~$\Zone(C)$. In particular both~$p_l$ and~$p_r$ are ends of two coastlines,~$\CCC_1,\CCC_2$ say (possibly the same). Recall that all the coasts of a cuff~$C$ must cover all the ports of~$\Zone(C)$.
Let~$e \in \mu(p_l) \cap \mu(p_r)$ with ends~$u_e$ and~$v_e$ respectively (note that~$e \notin \Port(G)$ using \WMII). Then, using the first part of \WMIV, there exists~$p_e^\ast \in \Port(\CCC_i)$ for some~$1 \leq i \leq r$ with~$e \notin \mu(p_e^\ast)$ and hence there exist mutually disjoint left and right coasts~$C_l(e),C_r(e) \subset \Zone(C)$ where~$C_l(e)$ has ends~$p_l,p_e^\ast$ and~$C_r(e)$ has ends~$p_r,p_e^\ast$ respectively (these are different from the ones in \cref{def:left_right_coastlines}). By the second part of \WMIV we may assume that~$u_e \in \mu(C_l(e))$ and~$v_e \in \mu(C_r(e))$; else rename them accordingly. Then we may define~$\ast: \mu(p_l)\cap\mu(p_r) \to \Port(C)$---after a choice of~$p_e^\ast$ for each edge---in the obvious way mapping~$e$ to~$p_e^\ast$ and thus unambiguously defining~$C_l(e),C_r(e)$ given the ends (after renaming them)~$u_e$ and~$v_e$ of the edge~$e$. Given~$\ast$ we can then define the maps~$\delta_x: \mu(p_l)\cap\mu(p_r) \to V(G)$ for~$x \in \{l,r\}$ via~$\delta_l(e) = u_e$ and~$\delta_r(e) = v_e$. 

We immediately get the following.
\begin{observation}\label{obs:deltax_is_well-defined}
    Let~$s \in \Shore(G)$ be a deserted shore with ends~$p_l,p_r$ such that~$p_l,s,p_r$ appear in this order on~$\Zone(C)$ for some~$C \in c(\Sigma)$. Let~$e \in \mu(p_l)\cap \mu(p_r)$ and let~$C_l(e)$ and~$C_r(e)$ be defined as above. Then~$u_e,v_e \in \mu(C_l(e)) \cup \mu(C_r(e))$ and~$u_e \in \mu(C_l(e)) \iff v_e\in \mu(C_r(e))$.
\end{observation}
\begin{proof}
    This follows at once from the second part of \WMIV.
\end{proof}
\begin{remark}
    As mentioned above we may assume that~$u_e \in \mu(C_l(e))$ for else we can just rename the ends of the edge accordingly.
\end{remark}

In light of \cref{obs:deltax_is_well-defined} let~$s_l \in \Shore(C_l(e))$ and~$s_r \in \Shore(C_r(e))$ be such that~$u_e \in \mu(s_l)$ and~$v_e \in \mu(s_r)$. Denote by~$s_x'$ be the shore in~$C_x(e)$ adjacent to~$p_x$ for~$x \in \{l,r\}$, i.e.,~$s_x'$ is an end of~$C_x(e)$ by definition. Recall the \cref{def:coastline_defined_by_edge} of~$\CCC(e)$.

\begin{observation}\label{obs:left_and_right_ends_are_good}
    Let~$C^x(e)$ be the maximal sub-sequence of~$C_x(e)$ such that~$e \in \mu(\chi)$ for every~$\chi$ appearing in~$C^x(e)$ and every~$x\in \{l,r\}$.
    It holds that~$C^l(e),s^\ast,C^r(e) \subset \CCC(e)$ are all sub-sequences of~$\CCC(e)$ and they are visited in that order. In particular we have that~$s_l,s_l',s_r',s_r$ are visited in that order on~$\CCC(e)$ (where~$s_l = s_l'$ and~$s_r = s_r'$ are possible a priori).
\end{observation}
\begin{proof}
    The claim is obvious by definition of~$C_l(e),C_r(e)$ and~$\CCC(e)$ together with \cref{lem:edges_define_coast lines}.
\end{proof}

Finally we can split~$e$ by introducing new edges~$e_l,e_r$ with new terminals~$u_e',v_e'$ such that~$e_l=(u_e',u_e)$ and~$e_r=(v_e',v_e)$ (assuming~$u_e$ was the head and~$v_e$ the tail, otherwise the relations are inverted) and define~$\mu'$ with respect to~$\delta_l,\delta_r$ similar as in \cref{lem:cutting_coastal_maps_at_shores} with the difference that we do not change the~$\Zone(C)$ for it does not hold~$p_l=p_r$ and thus~$p_l \notin \mu(p_l)\cap \mu(p_r)$ by \WMII.

More precisely we define~$\mu'$ via
\begin{align}
    &\mu'(p) \coloneqq \mu(p) \text{ if } p\in\Port(G) \text{ and } e\notin E(\mu(p)),\\
    &\mu'(p) \coloneqq (\mu(p) \setminus \{e\})\cup \{e_x\} \text{ if } p\in C_x(e) \text{ and } e\in\mu(p), \label{eq:mu'_2}\\
    &\mu'(s) \coloneqq \mu(s) \text{ if } s \in \CCC_1\cup\ldots\cup\CCC_r \text{ and } e \notin \mu(s),\\
    &\mu'(s) \coloneqq (\mu(s)\setminus \{e\}) \cup \{e_x\} \text{ if } s \in (\CCC_1\cup\ldots\cup\CCC_r)\cap C_x(e)\text{ and } e \in \mu(s), \label{eq:mu'_5},\\
    &\mu'(s_l') \coloneqq (\mu(s_l)\setminus\{e\}) \cup \{u_e',e_l\},\\
    &\mu'(s_r') \coloneqq (\mu(s_r)\setminus\{e\}) \cup \{v_e',e_r\}
\end{align}
for~$x \in \{l,r\}$. We define~$\Split(\MMM,e) =  \Split(\MMM,e,p_1,p_2,\ast) \coloneqq(\Gamma,\CCC_1,\ldots,\CCC_r,\mu')$. 

\smallskip
With the definition of~$\Split(\MMM,e)$ at hand, we prove the following.

\begin{lemma}\label{lem:splitting-weak-maps}
Let~$\MMM$ be a weak coastal map of~$G$ in~$\Sigma$. Let~$C \in c(\Sigma)$ and let~$e \in \mu(p_l)\cap\mu(p_r)$ with terminals~$u_e,v_e$ for two distinct ports~$p_l,p_r \in \Port(C)$ such that they are ends to a common deserted shore~$s^\ast$. Then~$\Split(\MMM,e)$ is a weak coastal map of~$G'$ in~$\Sigma$ where~$G'$ is the graph obtained from splitting~$e$.
\end{lemma}
\begin{proof}
    We will verify each of the axioms for weak coastal maps. Note that~$\Shore(G) = \Shore(G')$ and~$\Port(G)=\Port(G')$ by construction. In particular~$\Gamma$ remains pseudo Euler-embedded. Let~$u_e',v_e' \in V(G)$ denote the newly introduced terminals after splitting~$e$ as in the definition of the split operation and adopt the notation introduced in the paragraph above.
    \begin{enumerate}
        \item[{\small(WM1)}] As we never split at ports, since~$\mu(p)\cap\mu(p')\cap\{p,p'\} = \emptyset$ for distinct~$p,p' \in \Port(G)$ by (WM1), we deduce that~$\Gamma \subset G'$. By definition of~$\mu'$ clearly~$\mu'(\chi) \subset G'$ for every~$\chi \in \Port(G)\cup\Shore(G)$. It is easily seen that every~$x \in V(G') \cup E(G')$ is part of some~$\mu(\chi)$, this is trivial for all~$x \in G'\cap G$, and by definition of~$\mu'$, for the newly introduced vertices~$u_e',v_e' \in V(G')$ we have~$u_e' \in \mu'(s_l')$ and~$v_e' \in \mu'(s_r')$ for~$s_l',s_r' \in \Shore(G)$ being the non-deserted shores adjacent to~$p_l,p_r$ respectively. It is clear that~$\mu'(\chi_1)\cap\mu'(\chi_2) = \emptyset$ for~$\chi,\chi' \in \Port(G)\cup\Shore(G)$ in different cuffs.
       
        \item[{\small(WM2)}] The first property is obvious for~$\Shore(G) = \Shore(G')$ and~$\Port(G)=\Port(G')$ and we did not change~$\Gamma_\ZZZ$. In particular note that we did not remove or add any ports in any~$\mu(s)$ for~$s\in\Shore(G)$ as~$e \notin \Port(G)$. Let~$s \in \Shore(\CCC_i)$ be a non-deserted shore for some~$1 \leq i \leq r$ with ends~$p,p'$. Then~$\mu'(p)\cup\mu'(p') \subseteq \mu(s)$ follows at once from the definition of~$\mu'$: To see this let~$e' \in \mu'(p)\cup\mu'(p')$ then either~$e'\in \mu(p)\cup \mu(p')$, whence the claim follows immediately, or~$e' \in\{e_l,e_r\}$ is an edge in~$E(G')\setminus E(G)$. But then~$e' \in \mu'(p)\cup\mu'(p')$ implies that~$e \in \mu(p)\cup \mu(p')$ by definition of~$\mu'$ (see \cref{eq:mu'_2}). Clearly it holds that~$p,p' \in C_x(e)$ for some~$x\in\{l,r\}$ for they are both part of the same coastline~$\CCC_i$ and we never split at a coastline (for we only split at deserted shores between coastlines). The claim now follows from \cref{eq:mu'_5} in the definition of~$\mu'$ above. 
        
        \item[{\small(WM3)}] Let~$C \in c(\Sigma)$ be some cuff and~$\chi_1,\chi_2 \in \Zone(C)$ and let~$p,p' \in \Port(C)$ such that~$\chi,p,\chi',p'$ occur on~$\Zone(C)$ in this order. Then by \WMIII we know that~$\mu(\chi)\cap \mu(\chi') \subseteq \mu(p)\cup \mu(p')$. Note that either~$\mu'(\chi)\cap \mu'(\chi') = \mu(\chi)\cap \mu(\chi')$ or~$\mu'(\chi)\cap \mu'(\chi') = \left(\mu(\chi)\cap \mu(\chi') \cup \{e_x\}\right) \setminus \{e\}$ for some~$x \in \{l,r\}$ using the definition of~$\mu'$ (see \cref{eq:mu'_2} and \cref{eq:mu'_5} which are the relevant cases) and noting that no vertex is part of two shores by definition (and \WMI). Without loss of generality assume~$e_l \in \mu'(\chi)\cap \mu'(\chi')$; the case~$e_r \in \mu'(\chi)\cap \mu'(\chi')$ is analogous and the case~$\mu'(\chi)\cap \mu'(\chi') = \mu(\chi)\cap \mu(\chi')$ is analogous too by letting~$e$ take the role of~$e_l$ in the following.  
    
        We claim that~$e_l \in \mu'(\chi)\cap \mu'(\chi')$ implies~$e_l \in \mu'(p)\cup \mu(p)$. To see this let~$e_l \in \mu'(\chi)\cap \mu'(\chi')$, implying~$\chi,\chi' \in C_l(e)$ by definition (see \cref{eq:mu'_5}). This implies that~$C_l(e) \cap \{p,p'\} \neq \emptyset$, say~$p \in C_l(e)$, for either~$\chi,p,\chi'$ or~$\chi',p',\chi$ must be visited in that order on~$C_l(e)$ since~$C_l(e) \cup C_r(e)$ cover all the ports of~$\Zone(C)$ and are each coastlines. Further, using~$e_l \in \mu'(\chi)\cap\mu'(\chi')$ and the definition of~$\mu'$, we deduce that~$e \in \mu(\chi)\cap \mu(\chi')$ and hence~$e \in \mu(p)\cup \mu(p')$ using \WMIII. Thus again by definition then~$e_l \in \mu(p)$ and thus~$e_l \in \mu(p) \cup \mu'(p)$.

        \item[{\small(WM4)}] The first assertion is obvious for~$\ast$ remains the same for all edges in~$E(G)\cap E(G')$ and~$\mu'(p^\ast)\cap\{e_l,e_r\} = \emptyset$ by definition. Thus me may define~$\ast(e_l) = p^\ast = \ast(e_r)$. The second assertion is clear for all edges in~$E(G') \cap E(G)$ for there is no change to them. We are left to prove it for~$e_l$ say (the case~$e_r$ is symmetrical). Let~$C^l(e),C^r(e)$ as in \cref{obs:left_and_right_ends_are_good}. Using this same \cref{obs:left_and_right_ends_are_good} together with \WMIV we derive that~$s_l$ must be an end of~$C^l(e)$ and that there is no non-deserted shore~$s \in \Shore(C)$ left such that~$s,s_l,s_l'$ are visited in that order by~$\CCC(e)$. By definition of~$s_l'$ then there is no shore~$s \in \Shore(C^l(e))$ such that~$s,s_l,s_l'$ or~$s_l,s_l',s$ are visited in that order by~$C^l(e)$. In particular then for every shore~$s \in \Shore(C)$ with~$e_l \in \mu'(s)$ we deduce that~$s_l,s,s_l'$ are visited in that order on~$C^l(e)$ and thus on~$C_l(e)$. This immediately implies \WMIV.

         \item[{\small(WM5)}] First note that for every~$p \in \Port(G)$ either~$\mu(p) = \mu'(p)$ or if~$e \in \mu(p)$ then~$\mu'(p) = (\mu(p)\setminus \{e\}) \cup \{e_x\}$ for some~$x \in \{l,r\}$ by \cref{eq:mu'_2}. This implies that the cardinality~$|\mu'(p)|=d_i$ remains as is for every port~$p \in \Zone(\CCC_i)$ and every~$1 \leq i \leq r$. For the remainder let~$s \in \Shore(\CCC_i)$ be some shore of a coast line and let~$p,p' \in \Port(s)$ be its adjacent ports. Note that the case where both~$\mu(p)=\mu'(p')$ and~$\mu(p')=\mu'(p')$ is trivial for then also~$\mu(s)=\mu'(s)$ by definition of~$\mu$ and the paths are the same. It is also clear if both are not equal, i.e.,~$\mu(p) \neq \mu'(p)$ and~$\mu(p') \neq \mu'(p')$, for then by definition of~$\mu'$ one of the paths in~$\mu(s)$ has to be~$e \in \mu(p)\cap \mu(p')$ and said path gets replaced by the path~$e' \in \mu'(p)\cap \mu'(p') \subset \mu'(s')$, noting that we never split at a coast line~$\CCC_i$ and thus either~$\CCC_i \subseteq C_l(e)$ or~$\CCC_i \subseteq C_r(e)$. The other paths are equal in both~$\mu(s)$ and~$\mu'(s')$. Thus, without loss of generality, assume that~$\mu(p) = \mu'(p)$ but~$\mu(p') \neq \mu'(p')$. In particular, by definition of~$\mu'$,~$e \notin \mu(p)$ but~$e \in \mu(p')$. By \WMII there exists a path~$P \in \mu(s)$ with directed pattern~$\pi(P) = (f,e)$ for some~$f \neq e$ with~$f \in \mu(p)=\mu'(p)$. Hence, using \WMIV, one endpoint of~$e$ must lie in~$\mu(s)$, say~$u_e \in \mu(s)$ (and thus~$\CCC_i \subseteq C_l(e)$). But then by definition of~$\mu'$ we deduce that~$e_l \in \mu'(p')$ and thus there exists a path~$P' \in \mu'(s)$ with pattern~$(f,e_l)$, for~$\mu(s)\setminus{e} \subset \mu'(s)$ (see \cref{eq:mu'_5}).  

         \item[\WMVI] This is trivial by definition, for we remove the edges we split from~$\mu(s^\ast)$.
    \end{enumerate}
\end{proof}

A direct consequence of the above is that using at most~$r\cdot d$ consecutive splitting operations, every weak coastal map can be transformed into a split coastal map which in turn is a strong coastal map by \cref{lem:split_weak_map_is_strong}. We summarise this as follows.

\begin{corollary}
   Let~$\MMM$ be a weak coastal map of~$G=\Gamma \cup I$ in~$\Sigma$ with~$r\geq 1$ sights and depth~$d\geq 1$. Then there exist~$t \leq r\cdot d$ edges~$e_1,\ldots,e_t$ such that~$\Split(\MMM,e_1,\ldots,e_t)$ is a strong coastal map of~$G'=\Gamma'\cup I'$ in~$\Sigma$ with~$r$ sights and depth~$d$ such that there exists a~$p$-shipping in~$G$ if and only if there is a~$p+(r\cdot d)$-shipping in~$G'$.
\label{cor:shipping_from_weak_to_strong_map}
\end{corollary}

Finally we derive the main theorem of this section. 
\begin{theorem}
    For every surface~$\Sigma$ and every~$p,r,d \in \N$ there exists~$\rho(p,r,d;\Sigma)$ such that the following holds true. Let~$G+D$ be Eulerian such that~$G$ has a weak coastal map~$\CC\coloneqq (\Gamma,\CCC_1,\ldots,\CCC_r,\mu)$ in~$\Sigma$ with~$r$ sights and depth~$d$, where~$\Gamma$ contains some swirl~$\SSS$ of size~$\geq \rho(p,r,d;\Sigma)$ induced by some tile~$T\subset \WWW$ of some cylindrical wall, Euler-embedded in a disc~$\Delta \subseteq \Sigma$---in particular~$\tw(\Gamma) \geq \rho(p,r,d;\Sigma)$. Then~$G$ has no rigid~$p$-linkage.
    \label{thm:no_rigid_linkages_on_weak_maps}
\end{theorem}
\begin{proof}
    We claim that~$\rho(p,r,d;\Sigma) \geq \theta_{\ref{thm:no_rigid_linkages_on_strong_maps}}(p+(r\cdot d),r,d;\Sigma)$ satisfies the theorem. First use \cref{cor:shipping_from_weak_to_strong_map} to get an equivalent instance together with a split coastal map. Then, using \cref{lem:split_weak_map_is_strong} and the fact that no edge of the swirl~$\SSS$ is split for it lies away from the islands of the coastal map, we are left with a strong coastal map of a pseudo-Eulerian graph~$G'=\Gamma' \cup I'$ with~$\Gamma'$ containing a large swirl~$\SSS$ and~$G'$ containing a~$(p+r\cdot d)$-shipping. Finally,  \cref{thm:no_rigid_linkages_on_strong_maps} concludes the proof.
\end{proof}

\section{A Coastal Map for Flower Graphs}
\label{sec:structure_thms}

As mentioned earlier, Johnson has proven a structure theorem for directed internally~$6$ edge-connected Eulerian graphs in his dissertation \cite{Johnson2002}. A graph~$G$ is called \emph{internally~$6$ edge-connected} if for every~$X\subset V(G)$ inducing a~$2$-cut or a~$4$-cut it holds that~$\Abs{X} = 1$. In the following we assume the reader to be familiar with the notion of \emph{vortices}, where a \emph{vortex} of depth~$d \in N$ can be though of as a graph~$H \subset G$ together with a cyclic order on~$C \subset V(H)$ such that there exists no linkage~$\LLL$ in~$H$ of size~$\geq d+1$ with its endpoints in~$C$; the idea being that one looks for a `quasi-embedding' of~$G$ into~$G = \Gamma \cup I$ where~$\Gamma$ is embedded in a surface~$\Sigma$ and every component~$I' \subset I$ is part of a vortex where the order is given on~$C= V(I')\cap V(\Gamma)$. That is when talking about `embeddings up-to vortices' on some surface~$\Sigma$ with boundary, the vertices~$C$ of a single vortex~$(H,C)$ are drawn on a single cuff~$\CCC \in c(\Sigma)$ with respect to the cyclic order~$C$; vortices are intuitively what we call weak islands (see \cref{def:weak_island}). The main theorem of Johnsons dissertation reads as follows (adapted marginally to use the terms presented in this paper).

\begin{theorem}[Theorem~$17.1$ in \cite{Johnson2002}]
    Let~$n$ be a positive integer. There exist integers~$c,\:g,\:s,\:v$, and~$d_v$ such that the following holds. Let~$D$ be an internally~$6$ edge-connected eulerian digraph. Suppose~$D$ contains an eulerian subdigraph~$G$ which immerses a swirl\footnote{Johnson refers to swirls as \emph{medial grids}.} of size at least~$g$. Then either~$D$ immerses a router\footnote{Johnson refers to routers as \emph{circuit cliques}.} of size~$n$ or there exists a surface~$\Sigma$ and an eulerian digraph~$D'$ such that:
    \begin{itemize}
        \item[1.] No router of size at least~$n$ embeds in~$\Sigma$
        \item[2.] $D'$ is obtained from~$D$ by exchanging at most~$s$ edges\footnote{Johnson refers to this operation as \emph{switching}.}
        \item[3.] $D'$ Euler-embeds in~$\Sigma$ with at most~$v$ vortices, each of depth at most~$c$ where every vortex is surrounded by~$d_v$ edge-disjoint cycles of alternating orientation drawn in~$\Sigma$ forming an insulation in~$\hat{\Sigma}$ \footnote{The number~$d_v$ is not stated in the theorem but follows from his proof using that the graph Euler-embeds with high representativity (see also Lemma 6.9 in his dissertation).}
        \item[4.] The embedding can be chosen so that there is a closed disc disjoint from every vortex and every changed edge\footnote{Johnson calls these \emph{original discs}.} containing an eulerian subdigraph~$G$ which immerses a swirl of size~$n$.
    \end{itemize}
    \label{thm:Johnsons_structure_theorem}
\end{theorem}
\begin{remark}
    We will use \cref{thm:Johnsons_structure_theorem} for we did not want to provide a lengthy but rather straightforward proof using the techniques established in \cite{KawarabayashiTW2021}, proving that~$h$-flower graphs adhere to the structure of \cref{thm:pre_structure_theorem_euler}. It seems unnecessary to use the full strength of \cref{thm:Johnsons_structure_theorem}, for~$h$-flower graphs are already Euler-embedded up to edges with both end-points on the outer-cycle. It is tedious but straightforward---we have a proof-sketch at hand---to extend the embedding inductively by looking for planar (or cross-cap) transactions on the `society' formed by the outer-cycle of the swirl (using the language of \cite{KawarabayashiTW2021}) where the \cref{thm:flat_swirl_away_from_D} serves as the starting point of the embedding that one tries to extend. One then shows that each introduction of a handle or cross-cap helps in finding a large router, similar as Kawarabayashi, Thomas, and Wollan find clique-minors when extending their embeddings in \cite{KawarabayashiTW2021}. We aim to provide a proof in the future that proves the more general structure theorem for Eulerian directed graphs (not only for flower-graphs) similar to the one provided by Johnson and mentioned in \cref{thm:Johnsons_structure_theorem}, but leveraging different techniques and getting rid of the degree assumptions and the internally~$6$ edge-connectedness therefor losing the assumption that the whole graph is `quasi-embeddable' but rather its parts up-to two- and four-cuts are. This then leads to a global structure theorem, where Johnsons theorem may be seen as a local version.
\end{remark}

Let~$G+D$ be an~$h$-flower graph as in \cref{cor:sunflower_graph_structure}. Then, after contracting the edges between~$V(D)$ and the outer-cycle~$C_{2h}$ of the swirl~$\SSS \subseteq G$ in the flower-graph, the graph is~$4$-regular and easily seen to be internally~$6$ edge-connected; otherwise apply  \cref{lem:reducing_to_no_small_cuts} to get a smaller counterexample. In particular we may use \cref{thm:Johnsons_structure_theorem} for~$G+D$. By deleting the vertices incident to edges that have been switched in \cref{thm:Johnsons_structure_theorem} we immediately derive the following.

\begin{theorem}
\label{thm:pre_structure_theorem_euler}
For integers~$t,p\in \N$ there exist functions~$h(p;t),\rho(t),\nu(t)$ and~$d(t)$ such that the following holds. Let~$G+D$ be an~$h(p;t)$-flower graph satisfying \cref{cor:sunflower_graph_structure} such that~$G$ does not contain a~$t$-router.
Then there exists
a set~$A\subseteq V(G)\setminus V(D)$ of size at most~$\rho(t)$ and
a surface~$\Sigma$ of Euler genus at most~$\rho(t)$ such that~$\restr{G}{\bar{A}}$ can be Euler-embedded in~$\Sigma$ up-to at most~$\nu(t)$ vortices
of depth at most~$\rho(t)$ each surrounded by at least~$d(t)$ alternating cycles embedded in~$\Sigma$. Further there exists a disc~$\Delta \subseteq \Sigma$ such that~$G[\Delta]$ contains a flat~$p$-swirl induced by a~$p$-tile.
\end{theorem}
\begin{proof}
    First use \cref{lem:reducing_to_no_small_cuts} and then contract the edges using terminals in~$G$ and update~$D$. Use \cref{thm:Johnsons_structure_theorem} on~$G+D$ and let~$\rho'(t)=\max(c_{\ref{thm:Johnsons_structure_theorem}},s_{\ref{thm:Johnsons_structure_theorem}})$ of the respective numbers in \cref{thm:Johnsons_structure_theorem}. Let~$d(t) \coloneqq (d_v)_{\ref{thm:Johnsons_structure_theorem}} + 2\rho'(t)$ and let~$h(p;t) \coloneqq g_{\cref{thm:Johnsons_structure_theorem}}(p)$ where~$p$ (and~$t$) take the role of~$n$ in \cref{thm:Johnsons_structure_theorem}. Let~$E' \subset E(G)$ be the edges that have been changed (or switched as Johnson would say) during the embedding process in \cref{thm:Johnsons_structure_theorem}. Then~$\Abs{E'} \leq \rho'(t)$; note that although the way the statement is written may not suggest it, in the proof of \cref{thm:Johnsons_structure_theorem} the size of the Euler-embedded swirl that can be found, as well as the representativity of the graph which yields~$d_v$, are both independent of~$\rho'$. 
    
    If an edge~$e \in E(D)$ in the demand graph was changed then simply undo it, for we do not embed edges in~$D$ in our setting; undo all of the contractions with vertices in~$V(D)$ again. Then for each of the edges in~$E' \setminus E(D)$ that have been changed, let~$A$ be the set of adjacent vertices in~$G$ disjoint from~$V(D)$; note that vertices in~$V(D)$ have degree one in~$G$ and there is no need to delete them for they can always be embedded where, after undoing the contraction, they are not part of switched edges anymore anyways. Then~$\Abs{A}\leq 2\rho'(t)$.  Finally let~$\rho(t) \coloneqq 2\rho'(t)$. This concludes the proof for~$G[\bar{A}]$ and thus~$\restr{G}{\bar{A}}$ can be embedded accordingly.
\end{proof}

Using \cref{thm:pre_structure_theorem_euler} and the results in \cref{subsec:charting_an_island} we can finally prove the main result of this section: how to obtain a coastal map for~$h$-flower graphs in the absence of large routers.

First subdivide edges in~$E(D)$ to make the terminals degree one again (we contracted them earlier). Then look at the Euler-restriction~$\restr{G}{\bar{A}}$ resulting from the cut induced by~$A$ with order bounded by~$4\rho(t)$ for our vertices are of degree at most four; note that this introduces new traps in the embedding.

For ease of argumentation we give the following definition.
\begin{definition}[Equivalent graphs]
    Let~$G,G'$ be two graphs. We call~$G$ and~$G'$ \emph{equivalent} if for every demand graph~$D$ such that~$G+D$ is Eulerian, then~$G'+D$ is Eulerian and both are equivalent instances of the Eulerian edge-disjoint paths problem. More precisely,~$G+D$ is a \emph{YES}-instance if and only if~$G'+D$ is a \emph{YES}-instance.
\end{definition}

Now using the results from \cref{subsec:charting_an_island} we derive the following. 

\begin{theorem}
    For every integers~$ t , p \in \N$ there exist functions~$h(p;t),\rho(t),d(t),r(t),t(p)$ and~$g(t)$ such that the following holds. Let~$G+D$ be an~$h(p;t)$-flower graph with~$\Abs{E(D)} = p$ such that~$G$ contains no~$t$-router. Then there exists a set~$X \subset V(G)\setminus V(D)$ of size~$\rho(t)$ and an Eulerian graph~$G'$ equivalent to the Euler-restriction~$\restr{G}{\bar{X}}$ such that~$G'$ has a weak coastal map~$\CC = (\Gamma,\Zone,\CCC_1,\ldots,\CCC_r(t), \mu)$ with~$r(t)$ sights of depth~$d(t)$ in a surface~$\Sigma$ of genus at most~$g(t)$ such that~$\Gamma$ contains a flat~$p$-swirl embedded in some disc~$\Delta \subseteq \Sigma$.
\label{thm:coastal_map_theorem_for_sunflower_graphs}
\end{theorem}
\begin{proof}
    The proof follows by first applying \cref{thm:pre_structure_theorem_euler} and letting~$X = A$, where~$A$ is the set of `apices' defined in said theorem. Then we apply \cref{thm:structure-thm-to-coastal-maps} to~$\restr{G}{\bar{X}} +D'$---where~$D' \supset D$ is any choice of demand graph extending~$D$---noting that, given the Euler-embedding of \cref{thm:pre_structure_theorem_euler}, for every single vortex~$I$ the vertices of~$I$ drawn on~$\Sigma$ are all drawn on the same cuff~$C_I \in c(\Sigma)$.
\end{proof}

We are ready to prove the main result needed in the proof of \cref{thm:main}, combining all the results presented in this paper: the irrelevant cycle theorem for the \textsc{Eulerian Edge-Disjoint Paths Problem}.

\section{The irrelevant Cycle Theorem for Eulerian Digraphs}
\label{sec:irrelevant_cycle_theorem}

We are ready to prove the main ingredient to Step 4. of our algorithm provided in \cref{sec:proof_structure} to prove \cref{thm:main}: the irrelevant cycle theorem in the case that we find a flat swirl. In a sense, it is the part we left out in the proof of \cref{thm:irrelevant_cycle_minimal_counterexample}. Its proof is a straightforward but rather tedious combination of the main results acquired through \cref{sec:Routing,sec:structure_thms,sec:shippings,sec:structure_of_min_examples,sec:charting_eulerian_digraphs}.

For ease of argumentation we provide a refinement of \cref{thm:irrelevant_cycle_in_router_general} to the case of rigid linkages. 

\begin{theorem}
There exists a function~$t(p)$ satisfying the following. Let~$(G,S,T;D)$ be an instance of the Eulerian edge-disjoint paths problem with~$\Abs{E(D)} = p \in \N$. Let~$\RRR\coloneqq C_1\cup \ldots \cup C_{t(p)} \subset G$ be a~$t(p)$-router in~$G$. Then there exists \emph{no} rigid linkage~$\LLL=\{L_1,\ldots,L_p\}$ in~$G$ with pattern~$D$.
    \label{thm:irrelevant_cycle_in_router_rigid}
\end{theorem}
\begin{proof}
    This follows at once from \cref{thm:irrelevant_cycle_in_router_general}.
\end{proof}
 Recall the \cref{def:insulation} of~$h$-insulations.

\begin{theorem}[Irrelevant Cycle Theorem]
    For every integer~$p \geq 0$ and every surface~$\Sigma$ of genus~$g(p)\in \N$ there exist~$h\coloneqq h(p)$ with the following property. Let~$G+D$ be Eulerian such that~$|E(D)| \leq p$ and every vertex in~$V(G)\setminus V(D)$ has degree four while every vertex in~$V(D)$ has degree two in~$G+D$. Further assume that~$G= \Gamma \cup K$ where~$V(\Gamma) \cap V(D) = \emptyset$ and~$\Gamma$ is Euler-embedded in~$\Sigma$ such that~$\Gamma$ contains an~$h$-insulation~$S$. Let~$C \subset \Gamma$ be an Eulerian subgraph~$h$-insulated from~$V(\Gamma \cap K)$. Then there exists a~$p$-linkage~$\LLL$ in~$G$ with pattern~$D$ if and only if there exists a~$p$-linkage~$\LLL'$ in~$G-C$ with the same pattern.
    \label{thm:irrelevant_cycle}
\end{theorem}
\begin{proof}
    Let~$$\xi(p) \coloneqq \max\Bigg(h_{\ref{thm:coastal_map_theorem_for_sunflower_graphs}}\Big(p,t_{\ref{thm:irrelevant_cycle_in_router_rigid}}(p), \rho_{\ref{thm:no_rigid_linkages_on_weak_maps}}\big(p+2\rho_{\ref{thm:coastal_map_theorem_for_sunflower_graphs}}(p),r_{\ref{thm:coastal_map_theorem_for_sunflower_graphs}}(p),d_{\ref{thm:coastal_map_theorem_for_sunflower_graphs}}(p);\Sigma_{\ref{thm:coastal_map_theorem_for_sunflower_graphs}}\big)\Big)\Bigg).$$
   Let~$h_{\ref{thm:irrelevant_cycle_minimal_counterexample}}(p;\xi)$ be chosen with respect to~$\xi(p)$. We claim that~$h$ satisfies the theorem, for suppose it does not and let~$G+D$ be a minimal counterexample with respect to~$\Abs{E(G)}+\Abs{V(G)}$. That is, there exist~$\LLL,G,K,C,S$ as in the theorem where~$S = \bigcup_{i=1}^{2h}C_i$ is the respective collection of cycles in the~$h$-insulation drawn in some disc~$\Delta \subseteq \Sigma$ away from~$V(D)$. Then
   \cref{thm:irrelevant_cycle_minimal_counterexample} (and \cref{cor:sunflower_graph_structure}) implies that~$G$ is an~$h(p)$-flower graph and~$\LLL$ is rigid; in particular~$S$ is a flat swirl with~$S[G+D] \cap E(D) = \emptyset$ using the insulation.

    First we conclude that there exists no~$t_{\ref{thm:irrelevant_cycle_in_router_rigid}}(p)$-router in~$G$ by \cref{thm:irrelevant_cycle_in_router_rigid}. Using \cref{thm:coastal_map_theorem_for_sunflower_graphs}, the structure theorem tailored to flower graphs, we get a set~$X\subset V(G)$ with~$\Abs{X} \leq \rho_{\ref{thm:coastal_map_theorem_for_sunflower_graphs}}(p)$ and a weak coastal map with~$\leq r_{\ref{thm:coastal_map_theorem_for_sunflower_graphs}}(p)$ sights of depth~$\leq d_{\ref{thm:coastal_map_theorem_for_sunflower_graphs}}(p)$ on some surface~$\Sigma$ of genus~$\leq g_{\ref{thm:coastal_map_theorem_for_sunflower_graphs}}(t_{\ref{thm:irrelevant_cycle_in_router_rigid}}(p))$ for the Euler-restriction~$\restr{G}{\bar{X}}$ together with a disc~$\Delta \subseteq \Sigma$ such that~$\Gamma[\Delta]$ contains a flat Euler-embedded~$h(p)$-swirl away from~$V(D)\cup E(D)$. Note here that~$\restr{G}{\bar{X}} + D_X^\LLL = \bar{X}^\LLL + D_X^\LLL = G[[\bar{X}]] + D_X^\LLL$ using \cref{def:Euler-restriction,def:euler_extended_cut} and \cref{obs:rigid_linkage_in_euler_restrictions,lem:Eulerian_extension_properties_rigid}. Using the fact that~$\LLL$ is rigid, \cref{def:euler_extended_cut} and \cref{lem:cuts_and_rigid_linkages} yield that~$\restr{\LLL}{G[[\bar{X}]]}$ is a rigid~$(p+2\rho_{\ref{thm:coastal_map_theorem_for_sunflower_graphs}}(p))$-linkage in~$G[[\bar{X}]] + D_X^\LLL$.
     The claim now follows by \cref{thm:no_rigid_linkages_on_weak_maps} using the fact that~$h(p) \geq \rho_{\ref{thm:no_rigid_linkages_on_weak_maps}}(p+2\rho_{\ref{thm:coastal_map_theorem_for_sunflower_graphs}}(p),r_{\ref{thm:coastal_map_theorem_for_sunflower_graphs}}(p),d_{\ref{thm:coastal_map_theorem_for_sunflower_graphs}}(p);\Sigma_{\ref{thm:coastal_map_theorem_for_sunflower_graphs}})$.

\end{proof}

All in all this concludes the correctness proof of \cref{thm:main}---recall the algorithm presented in \cref{sec:proof_structure}. We summarise the general steps of the proof. Given an Eulerian graph~$G+D$ encoding an instance of the edge-disjoint Eulerian paths problem first decide whether or not~$G$ has high undirected tree-width with respect to~$p \coloneqq \Abs{E(D)}$; if not then proceed with a standard dynamic program e.g. Courcelles theorem \cite{Courcelle1990,ArnSLJSD91}. In a next instance embed a large cylindrical wall found using \cref{thm:dir_wall_away_from_D} in order to apply \cref{thm:flat_swirl_away_from_D} to find a large router grasped by the wall or a large flat swirl~$\SSS$ in the graph away from~$D$, i.e., such that~$\SSS[G+D]\cap E(D) = \emptyset$. If we find a large router apply \cref{thm:irrelevant_cycle_in_router_general} to spot an irrelevant cycle in the router, delete it and restart the algorithm. Thus assume we find a large flat swirl in the graph. Use \cref{thm:embedded_flat_swirl} in order to reduce the instance~$G+D$ to an equivalent instance where we may assume the flat swirl and its attachment-extension~$\SSS[G+D]$ to be Euler-embeddable in a disc; do so and let~$\Gamma \coloneqq \SSS[G+D]$. Hence we are in the setting that~$G= \Gamma \cup K$ where~$\Gamma$ is Euler-embedded in a disc containing a large Euler-embedded flat swirl~$\SSS'$, in particular a large insulation. Finally \cref{thm:irrelevant_cycle} implies that we can spot an irrelevant cycle in~$\SSS'$ which allows us to reduce the instance further and continue inductively. 

\section{Conclusion}
As our main result we have proven that the \textsc{Eulerian Edge-Disjoint Paths Problem} on directed graphs can be solved in~$fpt$-time with respect to the number of demand-edges and in particular the number of terminal-pairs. It is folklore that the demand-graph setting and the general setting are equivalent when parameterized by the number of terminals. We want to note here that our proof can be adapted to the setting without demand graphs; the core observation for this is \cref{lem:exhaustive_yields_pseudo-eulerian}, but it needs some more slight arguments on which edges to delete given a non-exhaustive solution in order to make it exhaustive. This can be done by first closing the solution paths to edge-disjoint cycles; ideas we have not pursued further in this paper but are easily derivable from it. As mentioned at several points in the paper we will write up a detailed proof of a general structure theorem for Eulerian directed graphs similar to the one stated in \cref{thm:Johnsons_structure_theorem} proven in Johnsons dissertation, tailored to the setting introduced in this paper. The proof goes by leveraging techniques from \cite{KawarabayashiTW2021} as well as techniques and results we introduced in \cref{sec:Swirls}, in order to keep this work as self-contained as possible. That is we will present a paper reworking \cref{sec:structure_thms} in greater detail, extending the results from \cref{sec:Swirls} to the broader setting of general Eulerian digraphs.

\newpage

\bibliographystyle{alphaurl}

\end{document}